\documentclass[a4paper,11pt]{article}

\input{preamble.tex}
\usepackage[font=small,skip=0.5em]{caption}

\newcommand{\disjpath}{\textsf{Disjoint Paths}}

\begin{document}

\title{An Exponential Time Parameterized Algorithm for \pdp\footnote{A preliminary version of this paper appeared in the proceedings of {\em STOC 2020}.}}

\author{
Daniel Lokshtanov\thanks{University of California, Santa Barbara, USA. \texttt{daniello@ucsb.edu}}
 \and Pranabendu Misra\thanks{Max Planck Institute for Informatics, Saarbrucken, Germany. \texttt{pmisra@mpi-inf.mpg.de}}
 \and Michal Pilipczuk\thanks{Institute of Informatics, University of Warsaw, Poland. \texttt{michal.pilipczuk@mimuw.edu.pl}}
 \and Saket Saurabh\thanks{The Institute of Mathematical Sciences, HBNI, Chennai, India. \texttt{saket@imsc.res.in}}
 \and Meirav Zehavi\thanks{Ben-Gurion University, Beersheba, Israel. \texttt{meiravze@bgu.ac.il}} 
}

\maketitle

\begin{abstract} 

In the \disjpath\ problem, the input is an undirected graph $G$ on $n$ vertices and a set of $k$ vertex pairs, $\{s_i,t_i\}_{i=1}^k$, and the task is to find $k$ pairwise vertex-disjoint paths such that the $i$'th path connects $s_i$ to $t_i$.  In this paper, we give a parameterized algorithm with running time $2^{\OO(k^2)}n^{\OO(1)}$ for  \pdp{}, the variant of the problem where the input graph is required to be planar. Our algorithm is based on the unique linkage/treewidth reduction theorem for planar graphs by Adler et al.~[JCTB 2017], the algebraic cohomology based technique of Schrijver~[SICOMP 1994] and one of the key combinatorial insights developed by Cygan et al.~[FOCS 2013] in their algorithm for {\sf Disjoint Paths} on directed planar graphs. To the best of our knowledge our algorithm is the first parameterized algorithm to exploit that the treewidth of the input graph is small in a way completely different from the use of dynamic programming.

\end{abstract}

\newpage
\pagestyle{plain}
\setcounter{page}{1}


\section{Introduction}\label{sec:intro}

\vspace{-0.5em}

In the \disjpath\ problem, the input is an undirected graph $G$ on $n$ vertices and a set of $k$ pairwise disjoint vertex pairs, $\{s_i,t_i\}_{i=1}^k$, and the task is to find $k$ pairwise vertex-disjoint paths connecting $s_i$ to $t_i$ for each $i\in\{1,\ldots,k\}$. 
The  \disjpath\ problem is a fundamental routing problem that finds applications in VLSI layout and virtual circuit routing, and has a central role in Robertson and Seymour’s Graph Minors series. We refer to surveys such as~\cite{frank1990packing,schrijver2003combinatorial} for a detailed overview. 
%
%
%
%
The \disjpath\ problem was shown to be NP-complete by Karp (who attributed it to Knuth) in a followup paper~\cite{karp1975computational} to his initial list of 21 NP-complete problems~\cite{DBLP:conf/coco/Karp72} . It remains NP-complete even if $G$ is restricted to be a grid \cite{lynch1975equivalence, kramer1984complexity}. On directed graphs, the problem remains NP-hard even for $k=2$~\cite{DBLP:journals/tcs/FortuneHW80}. For undirected graphs, Perl and Shiloach~\cite{DBLP:journals/jacm/PerlS78} designed a polynomial time algorithm for the case where $k=2$.  Then, the  seminal work of 
Robertson and Seymour~\cite{DBLP:journals/jct/RobertsonS95b} showed that the problem is polynomial time solvable for every fixed $k$. In fact, they showed that it is {\em fixed parameter tractable (FPT)} by designing an algorithm with running time $f(k)n^3$. The currently fastest parameterized algorithm for \disjpath\ has running time $h(k)n^2$~\cite{DBLP:journals/jct/KawarabayashiKR12}. However, all we know about $h$ and $f$ is that they are computable functions. That is, we still have no idea about what the running time dependence on $k$ really is. Similarly, the problem appears difficult in the realm of approximation, where one considers the optimization variant of the problem where the aim is to find disjoint paths connecting as many of the  $\{s_i,t_i\}$ pairs as possible. Despite substantial efforts, the currently best known approximation algorithm 
remains a simple greedy algorithm that achieves approximation ratio $\cO(\sqrt n)$.

The  \disjpath\ problem has received particular attention when the input graph is restricted to be planar~\cite{DBLP:journals/jct/AdlerKKLST17,ding1992disjoint,DBLP:journals/siamcomp/Schrijver94,DBLP:conf/focs/CyganMPP13}. Adler et al.~\cite{DBLP:journals/jct/AdlerKKLST17} gave an algorithm for \disjpath\ on planar graphs (\pdp) with running time $2^{2^{\cO(k)}}n^2$, giving at least a concrete form for the dependence of the running time on $k$ for planar graphs. Schrijver~\cite{DBLP:journals/siamcomp/Schrijver94} gave an algorithm for \disjpath\ on {\em directed} planar graphs with running time $n^{\cO(k)}$, in contrast to the NP-hardness for $k=2$ on general directed graphs. Almost 20 years later, Cygan et al.~\cite{DBLP:conf/focs/CyganMPP13} improved over the algorithm of Schrijver and showed that {\sf Disjoint Paths} on directed planar graphs is FPT by giving an algorithm with running time $2^{2^{\cO(k^2)}} n^{\cO(1)}$. The \pdp\ problem is well-studied also from the perspective of approximation algorithms, with a recent burst of activity~\cite{ChuzhoyK15,ChuzhoyKL16,ChuzhoyKN17,ChuzhoyKN18,DBLP:conf/icalp/ChuzhoyKN18}. Highlights of this work include an approximation algorithm with factor $\cO(n^{9/19} \log ^{\cO(1)}n)$~\cite{ChuzhoyKL16} and, under reasonable complexity-theoretic assumptions, hardness of approximating the problem within a factor of $2^{\Omega{( \frac{1}{(\log \log n)^2} ) } }$~{\cite{ChuzhoyKN18}.

In this paper, we consider the parameterized complexity of \pdp{}.~Prior~to~our work, the fastest known algorithm was the $2^{2^{\cO(k)}}n^2$ time algorithm of Adler et al.~\cite{DBLP:journals/jct/AdlerKKLST17}. Double exponential dependence on $k$ for a natural problem on planar graphs is something of an outlier--the majority of problems that are FPT on planar graphs enjoy running times of the form $2^{\cO(\sqrt{k}\ \mathrm{polylog}\ \!k)} n^{\cO(1)}$ (see, e.g.,~\cite{DBLP:journals/jacm/DemaineFHT05,DBLP:conf/focs/FominLMPPS16,DBLP:journals/jacm/FominLS18,DBLP:conf/soda/KleinM14,DBLP:journals/talg/PilipczukPSL18}).
This, among other reasons (discussed below), led Adler~\cite{AdlerOpen13} to pose as an open problem in GROW 2013\footnote{The conference version of~\cite{DBLP:journals/jct/AdlerKKLST17} appeared in 2011, before~\cite{AdlerOpen13}. The document~\cite{AdlerOpen13} erroneously states the open problem for \disjpath\ instead of for \pdp{}---that \pdp{} is meant is evident from the statement that a $2^{2^{\cO(k)}}n^{\cO(1)}$ time algorithm is known. } whether \pdp{} admits an algorithm with running time $2^{k^{\OO(1)}}n^{\OO(1)}$. 
By integrating tools with origins in algebra and topology, we resolve this problem in the affirmative. 
In particular, we prove the following. 

\vspace{-0.25em}

\begin{theorem}\label{thm:main}
The {\rm {\pdp}} problem is solvable in time $2^{\OO(k^2)}n^{\OO(1)}$.\footnote{In fact, towards this we implicitly design a $w^{\OO(k)}$-time algorithm, where $w$ is the treewidth of the input~graph.}
\end{theorem}

\vspace{-0.25em}

In addition to its value as a stand-alone result, our algorithm should be viewed as a piece of an on-going effort of many researchers to make the Graph Minor Theory of Robertson and Seymour algorithmically efficient. The graph minors project is abound with powerful algorithmic and structural results, such as the algorithm for \disjpath{}~\cite{DBLP:journals/jct/RobertsonS95b}, {\sf Minor Testing}~\cite{DBLP:journals/jct/RobertsonS95b}  (given two undirected graphs, $G$ and $H$ on $n$ and $k$ vertices, respectively, the goal is to check whether $G$ contains $H$ as a minor), the structural decomposition~\cite{DBLP:journals/jct/RobertsonS03a} and the Excluded Grid Theorem~\cite{DBLP:journals/jct/RobertsonST94}. Unfortunately, all of these results suffer from such bad hidden constants and dependence on the parameter $k$ that they have gotten their own term--``galactic algorithms''~\cite{Lipton2013}.

It is the hope of many researchers that, in time, algorithms and structural results~from~Graph Minors can be more algorithmically efficient, perhaps even practically applicable. Substantial progress has been made in this direction, examples include the simpler decomposition theorem of Kawarabayashi and Wollan~\cite{DBLP:conf/stoc/KawarabayashiW11}, the faster algorithm for computing the structural decomposition of Grohe et al.~\cite{DBLP:conf/soda/GroheKR13}, the improved unique linkage theorem of  Kawarabayashi and Wollan~\cite{DBLP:conf/stoc/KawarabayashiW10}, the linear excluded grid theorem on minor free classes of Demaine and Hajiaghayi~\cite{DBLP:journals/combinatorica/DemaineH08}, paving the way for the theory of Bidimensionality~\cite{DBLP:journals/jacm/DemaineFHT05}, and the polynomial grid minor theorem of Chekuri and Chuzhoy~\cite{DBLP:journals/jacm/ChekuriC16}.
The algorithm  for \disjpath\ is a cornerstone of the entire Graph Minor Theory, and a vital ingredient in the  $g(k)n^3$-time algorithm for {\sf Minor Testing}. Therefore, {\em efficient} algorithms for \disjpath\ and {\sf Minor Testing} are necessary and crucial ingredients in an algorithmically efficient Graph Minors theory. This makes obtaining $2^{{\sf poly}(k)}n^{\cO(1)}$ time algorithms for  \disjpath{} and  {\sf Minor Testing} a tantalizing and challenging~goal.

Theorem~\ref{thm:main} is a necessary basic step towards achieving this goal---a $2^{{\sf poly}(k)}n^{\cO(1)}$ time algorithms for \disjpath{} on general graphs also has to handle planar inputs, and it is easy to give a reduction from \pdp\ to {\sf Minor Testing} in such a way that a $2^{{\sf poly}(k)}n^{\cO(1)}$ time algorithm for {\sf Minor Testing} would imply a $2^{{\sf poly}(k)}n^{\cO(1)}$ time algorithms for \pdp{}.
In addition to being a necessary step in the formal sense, there is strong evidence that an efficient algorithm for the planar case will be useful for the general case as well---indeed the algorithm for \disjpath{} of Robertson and Seymour~\cite{DBLP:journals/jct/RobertsonS95b} relies on topology and essentially reduces the problem to surface-embedded graphs. Thus, an efficient algorithm for \pdp{} represents a speed-up of the base case of the algorithm for \disjpath{} of Robertson and Seymour. Coupled with the other recent advances~\cite{DBLP:journals/jacm/ChekuriC16,DBLP:journals/jacm/DemaineFHT05,DBLP:journals/combinatorica/DemaineH08,DBLP:conf/soda/GroheKR13,DBLP:conf/stoc/KawarabayashiW10,DBLP:conf/stoc/KawarabayashiW11}, this gives some hope that $2^{{\sf poly}(k)}n^{\cO(1)}$ time algorithms for {\sf Disjoint Paths} and {\sf Minor Testing} may be within reach. 






\vspace{-1em}

\paragraph{Known Techniques and Obstacles in Designing a $2^{{\sf poly}(k)}$ Algorithm.}  All known algorithms for both {\sf Disjoint Paths} and {\sf Planar Disjoint Pat}hs have the same high level structure. 
In particular, given a graph $G$ we distinguish between the cases of $G$ having ``small'' or ``large'' treewidth.
In case  the treewidth is large, we distinguish between two further cases: either $G$ contains a ``large'' clique minor or it does not.
This results in the following case distinctions. 

\vspace{-0.25em}

\begin{enumerate}
\setlength{\itemsep}{-2pt}
\item {\bf Treewidth is small.} Let the treewidth of $G$ be $w$. Then, we use the known dynamic programming algorithm with running time 
$2^{\cO(w \log w)}n^{\cO(1)}$~\cite{scheffler1994practical} to solve the problem.  
It is important  to note that, assuming the Exponential Time Hypothesis (ETH), there is no algorithm for \disjpath\ running in time $2^{o(w \log w)}n^{\cO(1)}$~\cite{DBLP:journals/siamcomp/LokshtanovMS18}, 
nor an algorithm for 
{\sf Planar Disjoint Paths} running in time $2^{o(w)}n^{\cO(1)}$~\cite{DBLP:journals/tcs/BasteS15}. 

\item {\bf Treewidth is large and $G$ has a large clique minor.} In this case, we use the good routing property of the clique to find an irrelevant vertex and delete it without changing the answer to the problem. Since this case will not arise for graphs embedded on a surface or for planar graphs, we do not discuss it in more detail.

\item {\bf  Treewidth is large and $G$  has no large clique minor .} Using a fundamental structure theorem for minors called the ``flat wall theorem'',  we can conclude that $G$ contains a large planar piece of the graph and a vertex $v$ that is sufficiently insulated in the middle of it. Applying the unique linkage theorem~\cite{DBLP:journals/jct/RobertsonS12}  to this vertex, we conclude that it is irrelevant and remove it. 
For planar graphs, one can use the unique linkage theorem of Adler et al.~\cite{DBLP:journals/jct/AdlerKKLST17}. In particular,  we use the following result: 
\begin{quote}
Any instance of \disjpath\ consisting of a planar graph with treewidth at least $82 k^{3/2}2^k$ and $k$ terminal pairs contains a vertex $v$ such that every solution to \disjpath\  can be replaced by an equivalent one whose paths avoid $v$.

\vspace{-0.25em}

\end{quote}
This result says that if the treewidth of the input planar graph is (roughly) $\Omega(2^k)$, then we can find an irrelevant vertex and remove it. A natural question is whether we can guarantee an irrelevant vertex even if the treewidth is $\Omega({\sf poly}(k))$. Adler and Krause~\cite{DBLP:journals/corr/abs-1011-2136} exhibited a planar graph 
$G$ with $k+1$ terminal pairs such that $G$ contains a $(2^k + 1) \times (2^k + 1)$ grid as a subgraph,  \disjpath\ on this input has a unique solution, and  the solution uses all vertices of $G$; in particular, {\em no vertex of $G$ is irrelevant}. This implies that the irrelevant vertex technique 
can only guarantee a treewidth of $\Omega(2^k)$, even if the input~graph~is~planar. 
\end{enumerate}

\vspace{-0.55em}

\noindent Combining items (1) and (3), we conclude that the known methodology for \disjpath\ can 
only guarantee an algorithm with running time $2^{2^{\cO(k)}}n^2$ for \pdp. 
Thus, 
a $2^{{\sf poly}(k)}n^{\cO(1)}$-time algorithm for \pdp\  appears to require entirely new ideas. 
As this obstacle was known to Adler et al.~\cite{AdlerOpen13}, it is likely to be the main motivation for Adler to pose the existence of a $2^{{\sf poly}(k)}n^{\cO(1)}$ time algorithm for \pdp\ as an open problem.

\vspace{-0.8em}

\paragraph{Our Methods.}  Our algorithm is based on a novel combination of two techniques that do not seem to give the desired outcome when used on their own. The first ingredient is the treewidth reduction theorem of Adler et al.~\cite{DBLP:journals/jct/AdlerKKLST17} that proves that given an instance of \pdp, the treewidth can be brought down to $2^{\cO(k)}$ (explained in item (3) above).  This by itself is sufficient for an FPT algorithm (this is what Adler et al.~\cite{DBLP:journals/jct/AdlerKKLST17} do), but as explained above, it seems hopeless that it will bring a $2^{{\sf poly}(k)}n^{\cO(1)}$-time algorithm. 


We circumvent the obstacle 
by using  an algorithm for a more difficult problem with a worse running time, namely,  Schrijver's $n^{\cO(k )}$-time algorithm for {\sf Disjoint Paths} on directed planar graphs~\cite{DBLP:journals/siamcomp/Schrijver94}. Schrijver\rq{}s algorithm has two steps: a ``guessing'' step where one (essentially) guesses the homology class of the solution paths, and then a surprising homology-based algorithm that, given a homology class, finds a solution in that class (if one exists) in polynomial time. Our key insight is that for \pdp, if the instance that we are considering has been reduced according to the procedure of Adler et al.~\cite{DBLP:journals/jct/AdlerKKLST17}, then we only need to iterate over $2^{\cO(k^2)}$ homology classes in order to find the homology class of a solution, if one exists. The proof of this key insight is highly non-trivial, and builds on a cornerstone ingredient of the recent FPT algorithm of Cygan et al.~\cite{DBLP:conf/focs/CyganMPP13} for {\sf Disjoint Paths} on directed planar graphs. To the best of our knowledge, this is the first algorithm that finds the exact solution to a problem that exploits that the treewidth of the input graph is small in a way that is different from doing dynamic programming. 
A technical overview of our methods will appear in the next section. 
In our opinion, a major strength of the paper is that it breaks not only a barrier in running time, but also a longstanding methodological barrier. Since there are many algorithms that use the irrelevant vertex technique in some way, there is reasonable hope that they could benefit from the methods developed in this work.

We remark that we have made no attempt to optimize the polynomial factor in this paper. Doing that, and in particular achieving linear dependency on $n$ while keeping the dependency  on $k$ single-exponential, is the natural next question for future research. In particular, this might require to ``open-up'' the black boxes that we use, whose naive analysis yields a large polynomial dependency on $n$, but there is no reason to believe that it cannot be made linear---most likely, this will require extensive independent work on these particular ingredients. Having both the best dependency on $k$ and the best dependency on $n$ simultaneously may be critical to achieve a practical exact algorithm for large-scale instances.

\setlength{\belowcaptionskip}{-17pt}
\newcommand{\WUse}{\ensuremath{\mathsf{PreSegGro}}}
\newcommand{\Segga}{\ensuremath{\mathsf{SegGro}}}

\begin{figure}
    \begin{center}
        \includegraphics[width=0.45\textwidth]{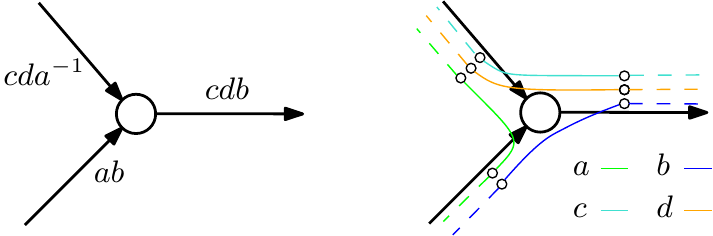}
        \caption{Flow at a vertex and its reduction.}
        \label{fig:flowStitch}
    \end{center}
\end{figure}

\vspace{-0.5em}

\section{Overview}\label{sec:overview}

\vspace{-0.5em}

\noindent{\bf Homology.} In this overview, we explain our main ideas in an {\em informal} manner. Our starting point is Schrijver's view \cite{DBLP:journals/siamcomp/Schrijver94} of a collection of ``non-crossing'' (but possibly not vertex- or even edge-disjoint) sets of walks as flows. 
To work with flows (defined immediately), we deal with directed graphs. (In this context, undirected graphs are treated  as directed graphs by replacing each edge by two parallel arcs of opposite directions.) Specifically, we denote an instance of \dpdp\ as a tuple $(D,S,T,g,k)$ where $D$ is a directed plane graph, $S,T\subseteq V(D)$, $k=|S|$ and $g: S\rightarrow T$ is bijective. Then, a {\em solution} is a set ${\cal P}$ of pairwise vertex-disjoint directed paths in $D$ containing, for each vertex $s\in S$, a path directed from~$s$~to~$g(s)$.

In the language of flows, each arc of $D$ is assigned a word with letters in $T\cup T^{-1}$ (that is, we treat the set of vertices $T$ also as an alphabet), where $T^{-1}=\{t^{-1}: t\in T\}$. This collection of words is denoted by $(T \cup T^{-1})^*$ and let $1$ denote the empty word.
A word is {\em reduced} if, for all $t\in T$, the letters $t$ and $t^{-1}$ do not appear consecutively.
Then, a {\em flow} is an assignment of reduced words to arcs that satisfies two constraints. First, when we concatenate the words assigned to the arcs incident to a vertex $v\notin S\cup T$ in clockwise order, where words assigned to ingoing arcs are reversed and their letters negated, the result (when reduced) is the empty word $1$ (see Fig.~\ref{fig:flowStitch}). This is an algebraic interpretation of the standard flow-conservation constraint. Second, when we do the same operation with respect to a vertex $v\in S\cup T$, then when the vertex is in $S$, the result is $g(s)$ (rather than the empty word), and when it is in $T$, the result is $t$. There is a natural association of flows to solutions: for every $t\in T$, assign the letter $t$ to all arcs used by the path from $g^{-1}(t)$ to $t$.


Roughly speaking, Schrijver proved that if a flow $\phi$ is given along with the instance $(D,S,T,g,k)$, then in {\em polynomial time} we can either find a solution or determine that there is no solution ``similar to $\phi$''. Specifically, two flows are {\em homologous} (which is the notion of similarity) if one can be obtained from the other by a {\em set} of ``face operations'' defined as follows.

\vspace{-0.25em}

\begin{definition}\label{def:homologyOverview}
Let $D$ be a directed plane graph with outer face $f$, and denote the set of faces of $D$ by $\cal F$. Two flows $\phi$ and $\psi$ are {\em homologous} if there exists a function $h: {\cal F}\rightarrow (T\cup T^{-1})^*$ such that {\em (i)} $h(f)=1$, and {\em (ii)} for every arc $e\in A(D)$, $h(f_1)^{-1}\cdot \phi(e)\cdot h(f_2)=\psi(e)$ where $f_1$ and $f_2$ are the faces at the left-hand side and the right-hand side of $e$, respectively.
\end{definition} 

\vspace{-0.25em}

Then, a slight modification of Schrijver's theorem~\cite{DBLP:journals/siamcomp/Schrijver94} readily gives the following corollary.

\vspace{-0.25em}

\begin{corollary}\label{prop:schOverview}
There is a polynomial-time algorithm that, given an instance $(D,S,T,g,$ $k)$ of \dpdp, a flow $\phi$ and a subset $X\subseteq A(D)$, either finds a solution of $(D-X,S,T,g,k)$ or decides that there is no solution of it such that the ``flow associated with it'' and $\phi$ are homologous in $D$.
\end{corollary}

\begin{figure}
    \begin{center}
        \includegraphics[scale=0.375]{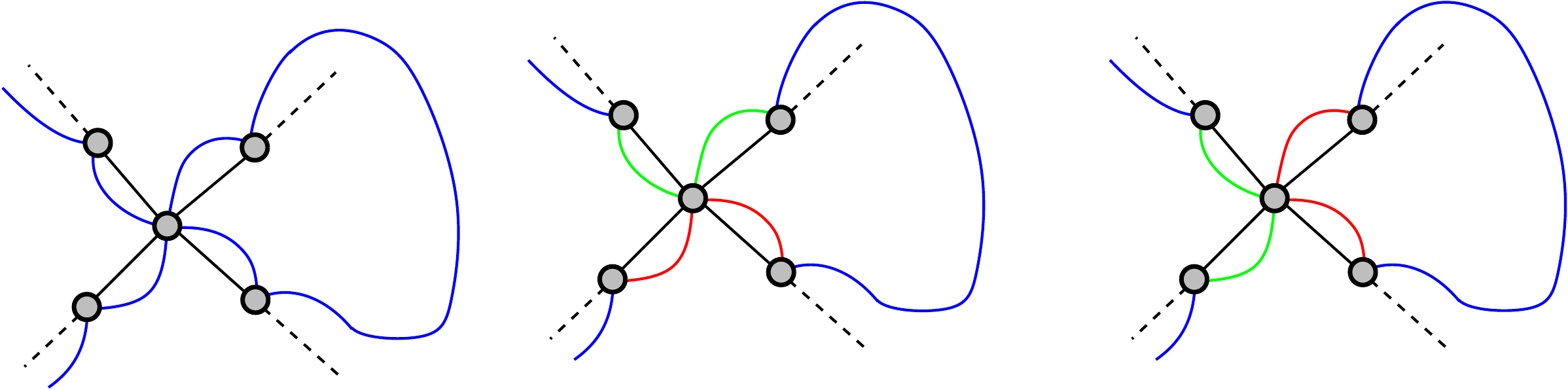}
        \caption{Two different ways of extracting a walk from a flow.}
        \label{fig13}
    \end{center}
\end{figure} 

\noindent{\bf Discrete Homotopy and Our Objective.} While the language of flows and homology can be used to phrase our arguments, it also makes them substantially longer and somewhat obscure because it brings rise to multiple technicalities. For example, different sets of non-crossing walks may correspond to the same flow (see Fig.~\ref{fig13}). Instead, we define a notion of {\em discrete homotopy}, inspired by (standard) homotopy. Specifically, we deal only with collections of non-crossing {\em and edge-disjoint} walks, called {\em weak linkages}.  Then, two weak linkages are {\em discretely homotopic} if one can be obtained from the other by using ``face operations'' that push/stretch its walks across faces and keep them non-crossing and edge-disjoint (see Fig.~\ref{fig0203}). More precisely, discrete homotopy is an equivalence relation that consists of three face operations, whose precise definition (not required to understand this overview) can be found in Section \ref{sec:discreteHomotopy}. We note that the order in which face operations are applied is important in discrete homotopy (unlike homology)---we cannot stretch a walk across a face if no walk passes its boundary, but we can execute operations that will move a walk to that face, and then stretch it. In Section \ref{sec:discreteHomotopy}, we translate Corollary \ref{prop:schOverview} to discrete homotopy (and undirected graphs) to derive the following~result.

\vspace{-0.25em}

\begin{lemma}\label{lem:discreteHomotopyOverview}
There is a polynomial-time algorithm that, given an instance $(G,S,T,g,k)$ of \pdp, a weak linkage $\cal W$ in $G$ and a subset $X\subseteq E(G)$, either finds a solution of $(G-X,S,T,g,k)$ or decides that no solution of it is discretely homotopic to $\cal W$ in $G$.
\end{lemma}

\vspace{-0.25em}

\begin{figure}
    \begin{center}
        \includegraphics[width=0.475\textwidth]{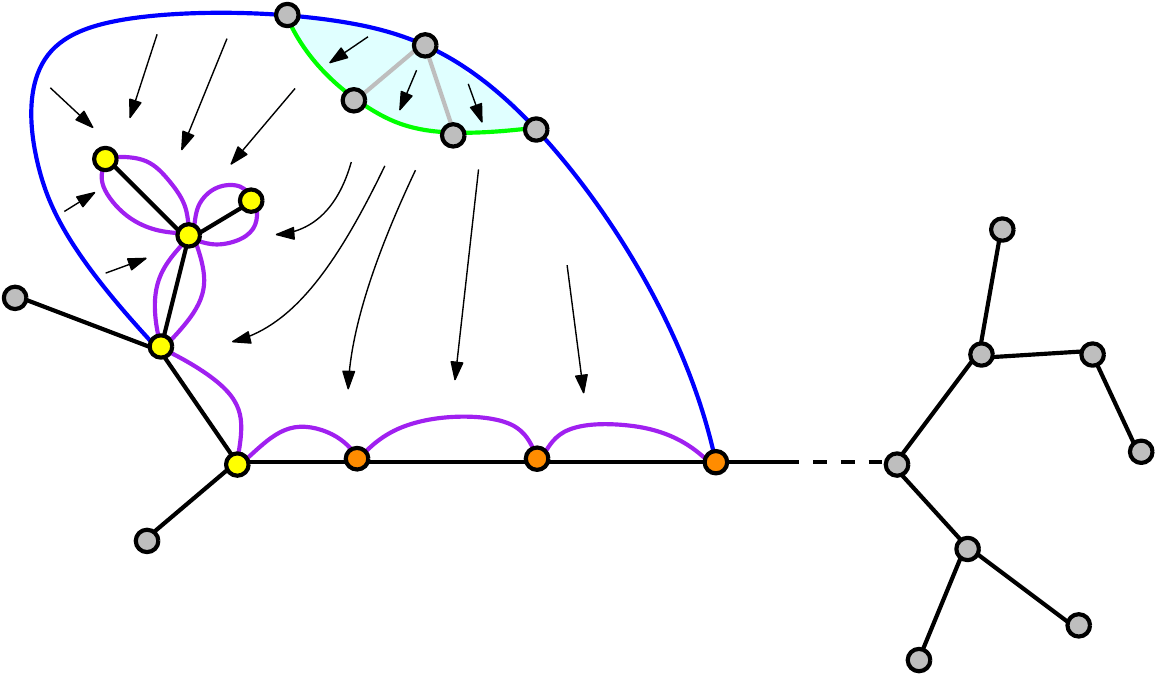}
        \caption{Moving a walk of a weak linkage (in blue) onto the Steiner tree (the walk in purple) with ``face operations''(e.g. a sub-path of the blue path is pushed giving the green sub-path).}
        \label{fig0203}
    \end{center}
\end{figure}

In light of this result, our objective is reduced to the following task.

\vspace{-0.4em}

\begin{quote}
\framebox{
\begin{minipage}{0.85\textwidth}
Compute a collection of weak linkages such that if there exists a solution, then there also exists a solution ({\em possibly a different one!}) that is discretely homotopic to one of the weak linkages in our collection. To prove Theorem \ref{thm:main}, the size of the collection should be upper bounded by $2^{\OO(k^2)}$.
\end{minipage}
}
\end{quote}

\smallskip
\noindent{\bf Key Player: Steiner Tree.} A key to the proof of our theorem is a very careful construction (done in three steps in Section \ref{sec:steiner}) of a so-called {\em Backbone Steiner tree}. We use the term Steiner tree to refer to any tree in the {\em radial completion} of $G$ (the graph obtained by placing a vertex on each face and making it adjacent to all vertices incident to the face) whose set of leaves is precisely $S\cup T$. In the first step, we consider an arbitrary Steiner tree as our Steiner tree $R$. Having $R$ at hand, we have a more focused goal: we will zoom into weak linkages that are ``pushed onto $R$'', and we will only generate such weak linkages to construct our collection. Informally, a weak linkage is {\em pushed onto $R$} if all of the edges used by all of its walks are {\em parallel to} edges of $R$. We do not demand that the edges belong to $R$, because then the goal described immediately cannot be achieved---instead, we make $4n+1$ parallel copies of each edge in the radial completion (the number $4n+1$ arises from considerations in the ``pushing process''), and then impose the above weaker demand. Now, our goal is to show that, if there exists a solution, then there also exists one that can be pushed onto $R$ by applying face operations (in discrete homotopy) so that it becomes {\em identical} to one of the weak linkages in our collection (see Fig.~\ref{fig0203}).

At this point, one remark is in place. Our Steiner tree $R$ is a subtree of the radial completion of $G$ rather than $G$ itself. Thus, if there exists a solution discretely homotopic to one of the weak linkages that we generate, it might not be a solution in $G$. We easily circumvent this issue by letting the set $X$ in Lemma \ref{lem:discreteHomotopyOverview} contain all ``fake'' edges.

\medskip
\noindent{\bf Partitioning a Weak Linkage Into Segments.} For the sake of clarity, before we turn to present the next two steps taken to construct $R$, we begin with the (non-algorithmic) part of the proof where we analyze a (hypothetical) solution towards pushing it onto $R$. Our most basic notion in this analysis is that of a {\em segment}, defined as follows (see Fig.~\ref{fig0408}). 

\vspace{-0.25em}

\begin{definition}\label{def:segmentOverview}
For a walk $W$ in the radial completion of $G$ that is edge-disjoint from $R$, a {\em segment} is a maximal subwalk of $W$ that does not ``cross'' $R$.
\end{definition}

\vspace{-0.25em}

\begin{figure}
    \begin{center}
        \includegraphics[width=0.575\textwidth]{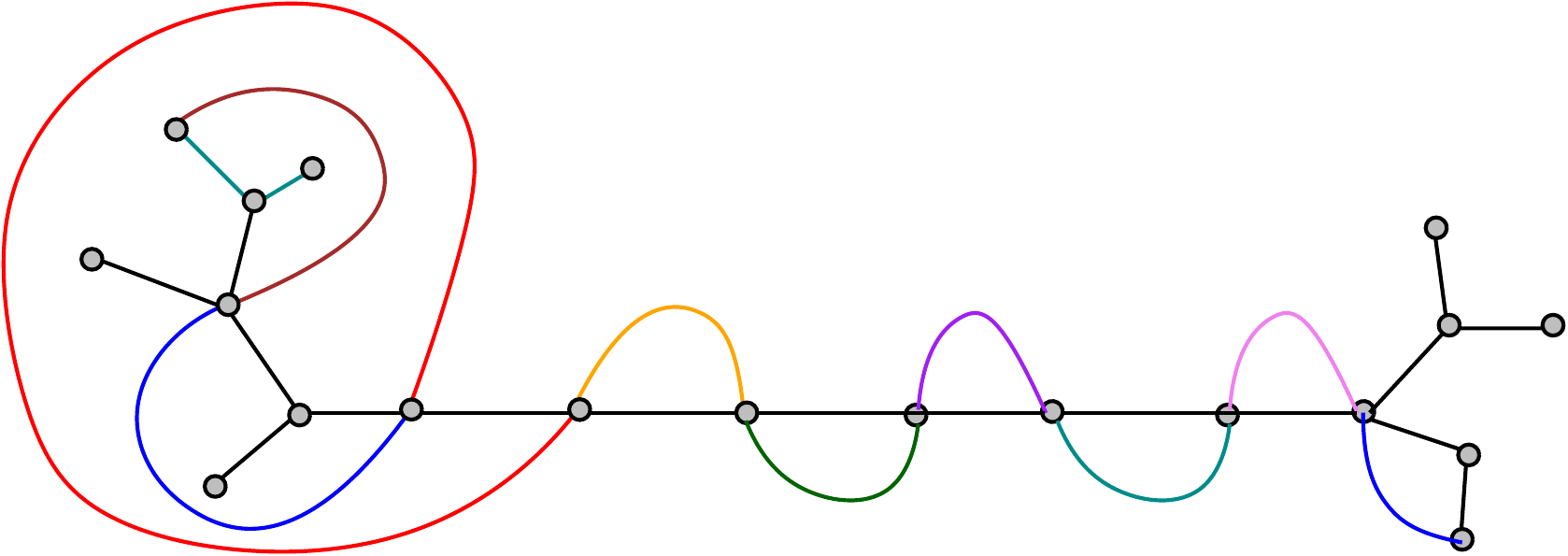}
        \caption{Segments arising from the crossings of a walk with the Steiner tree.}
        \label{fig0408}
    \end{center}
\end{figure}

Let $\Seg(W)$ denote the set of segments of $W$. Clearly, $\Seg(W)$ is a partition of $W$. Ideally, we would like to upper bound the number of segments of (all the paths of) a solution by $2^{\OO(k^2)}$. However, this will not be possible because, while $R$ is easily seen to have only $\OO(k)$ vertices of degree $1$ or at least $3$, it can have ``long'' maximal degree-2 paths which can give rise to numerous segments (see Fig.~\ref{fig0408}). To be more concrete, we say that a maximal degree-2 path of $R$ is {\em long} if it has more than $2^{ck}$ vertices (for some constant $c$), and it is {\em short} otherwise. Then, as the paths of a solution are vertex disjoint, the following observation is immediate.

\vspace{-0.25em}

\begin{observation}\label{obs:shortPathsOverview}
Let $\cal P$ be a solution. Then, its number of segments that have at least one endpoint on a short path, or a vertex of degree other than $2$, of $R$, is upper bounded by $2^{\OO(k)}$.
\end{observation}

\vspace{-0.25em}

To deal with segments crossing only long paths, several new ideas are required. In what follows, we first explain how to handle segments going across different long paths, whose number {\em can} be bounded (unlike some of the other types of segments we will encounter).

\medskip
\noindent{\bf Segments Between Different Long Paths.} To deal with such segments, we modify $R$ (in the second step of its construction). For each long path $P$ with endpoints $u$ and $v$, we will compute two minimum-size vertex sets, $S_u$ and $S_v$, such that $S_u$ separates (i.e., intersects all paths with one endpoint in each of the two specified subgraphs) the following subgraphs in the radial completion of $G$: {\em (i)} the subtree of $R$ that contains $u$ after the removal of a vertex $u_1$ of $P$ that is ``very close'' to $u$, and {\em (ii)} the subtree of $R$ that contains $v$ after the removal of a vertex $u_2$ that is ``close'' to $u$. The condition satisfied by $S_v$ is symmetric (i.e. $u$ and $v$ switch their roles; see Fig.~\ref{fig05}). Here, ``very close'' refers to distance $2^{c_1k}$ and ``close'' refers to distance $2^{c_2k}$ on the path, for some constants $c_1<c_2$. Let $u'$ and $v'$ be the vertices of $P$ in the intersection with the separators $S_u$ and $S_v$ respectively. (The selection of $u'$ not to be $u$ itself is of use in the third modification of $R$.)

\begin{figure}
    \begin{center}
        \includegraphics[width=0.875\textwidth]{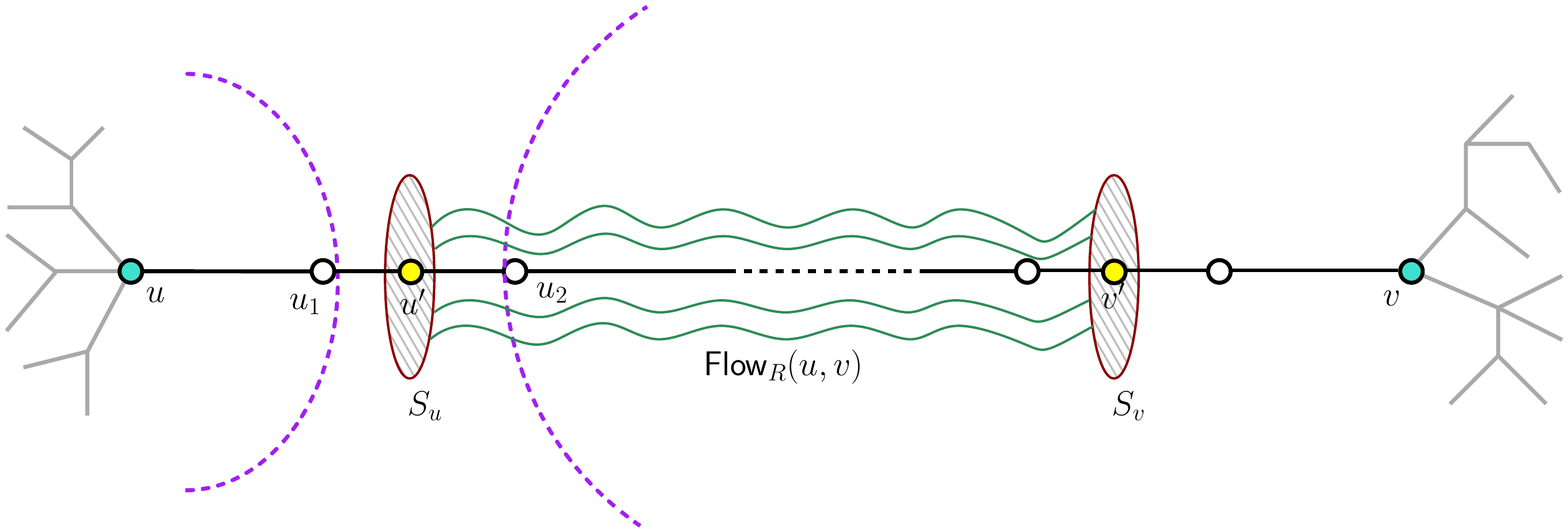}
        \caption{Separators and flows for a long maximal degree-2 path $P$ in $R$.}
        \label{fig05}
    \end{center}
\end{figure}

To utilize these separators, we need their sizes to be upper bounded by $2^{\OO(k)}$. For our initial $R$, such small separators may not exist. However, the modification we present now will guarantee their existence. Specifically, we will ensure that $R$ does not have any {\em detour}, which roughly means that each of its maximal degree-2 paths is a shortest path connecting the two subtrees obtained once it is removed. More formally, we define a detour as follows (see Fig.~\ref{fig:undetour}).

\begin{figure}
    \begin{center}
        \includegraphics[scale=0.475]{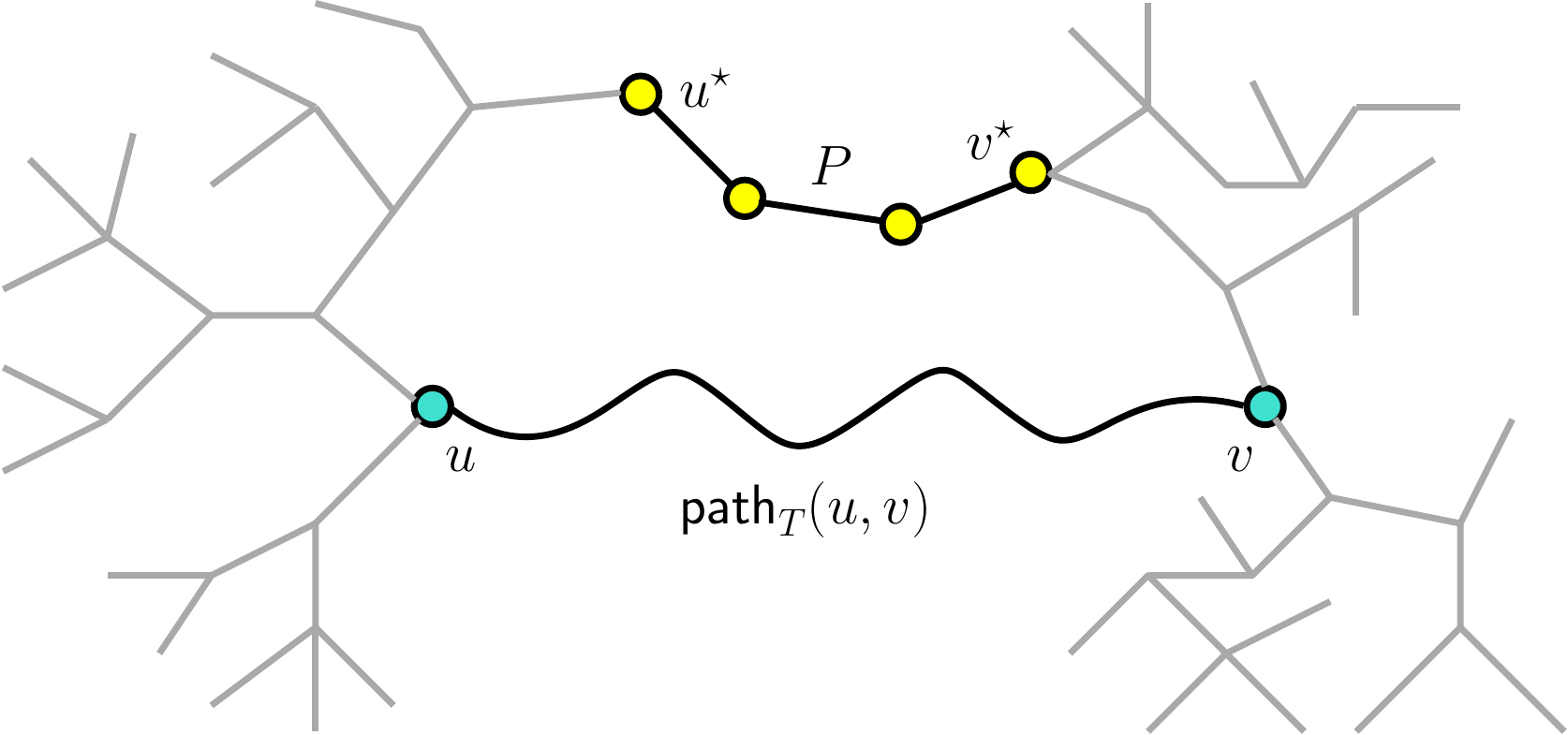}
        \caption{Detours in the Steiner tree.}
        \label{fig:undetour}
    \end{center}
\end{figure}

\vspace{-0.25em}

\begin{definition}\label{def:detourOverview}
A {\em detour} in $R$ is a pair of vertices $u,v\in V_{\geq 3}(R)\cup V_{=1}(R)$ (i.e. the non-degree $2$ vertices in R) that are endpoints of a maximal degree-2 path $L$ of $R$, and a path $P$ in the radial completion of $G$, such that {\em (i)} $P$ is shorter than $L$, {\em (ii)} one endpoint of $P$ belongs to the component of $R-V(L)\setminus\{u,v\}$ containing $u$, and {\em (iii)} one endpoint of $P$ belongs to the component of $R-V(L)\setminus\{u,v\}$ containing $v$.
\end{definition}

\vspace{-0.25em}

By repeatedly ``short-cutting'' $R$, a process that terminates in a linear number of steps, we obtain a new Steiner tree $R$ with no detour. Now, if the separator $S_u$ is large, then there is a large number of vertex-disjoint paths that connect the two subtrees separated by $S_u$, and all of these paths are ``long'', namely, of length at least $2^{c_2k}-2^{c_1k}$. Based on a result by Bodlaender et al.~\cite{DBLP:journals/jacm/BodlaenderFLPST16} (whose application requires to work in the radial completion of $G$ rather than $G$ itself), we show that the existence of these paths implies that the treewidth of $G$ is large. Thus, if the treewidth of $G$ were small, all of our separators would have also been small. Fortunately, to guarantee this, we just need to invoke the following known result in a preprocessing step: 

\vspace{-0.25em}

\begin{proposition}[\cite{DBLP:journals/jct/AdlerKKLST17}]\label{prop:twReductionOverview}
There is a $2^{\OO(k)}n^2$-time algorithm that, given an instance $(G,S,T,g,k)$ of \pdp, outputs an equivalent instance $(G',S,T,g,k)$ of \pdp\ where $G'$ is a subgraph of $G$ whose treewidth is upper bounded by $2^{ck}$ for some constant $c$.
\end{proposition}

\vspace{-0.25em}

Having separators of size $2^{\OO(k)}$, because segments going across different long paths must intersect these separators (or have an endpoint at distance $2^{\OO(k)}$ in $R$ from some endpoint of a maximal degree-2 path), we immediately deduce the following.

\vspace{-0.25em}

\begin{observation}\label{obs:goingAcrossDiffPathsOverview}
Let $\cal P$ be a solution. Then, its number of segments that have one endpoint on one long path, and a second endpoint on a different long path, is upper bounded by $2^{\OO(k)}$.
\end{observation}

\noindent{\bf Segments with Both Endpoints on the Same Long Path.} We are now left with segments whose both endpoints belong to the same long path, which have two different kinds of behavior: they may or may not {\em spiral} around $R$, where spiraling means that the two endpoints of the segment belong to different ``sides'' of the path (see Fig.~\ref{fig0408} and Fig.~\ref{fig:rollbackSpirals}). By making sure that at least one vertex in $S\cup T$ is on the outer face of the radial completion of $G$, we ensure that the cycle formed by any non-spiraling segment together with the subpath of $R$ connecting its two endpoints does not enclose all of $S\cup T$; specifically, we avoid having to deal with segments as the one in Fig.~\ref{fig07}.

While it is tempting to try to devise face operations that transform a spiraling segment into a non-spiraling one, this is not always possible. In particular, if the spiral ``captures'' a path $P$ (of a solution), then when $P$ and the spiral are pushed onto $R$, the spiral is not reduced to a simple path between its endpoints, but to a walk that ``flanks'' $P$. Due to such scenarios, dealing with spirals (whose number we are not able to upper bound) requires special attention. Before we turn to this task, let us consider the non-spiraling segments.

\medskip
\noindent{\bf Non-Spiraling Segments.} To achieve our main goal, we aim to push a (hypothetical) solution onto $R$ so that the only few parallel copies of each edge will be used. Now, we argue that non-spiraling segments do not pose a real issue in this context. To see this, consider a less refined partition of a solution where some non-spiraling segments are ``grouped'' as follows (see Fig.~\ref{fig0408}).

\begin{figure}
    \begin{center}
        \includegraphics[width=0.575\textwidth]{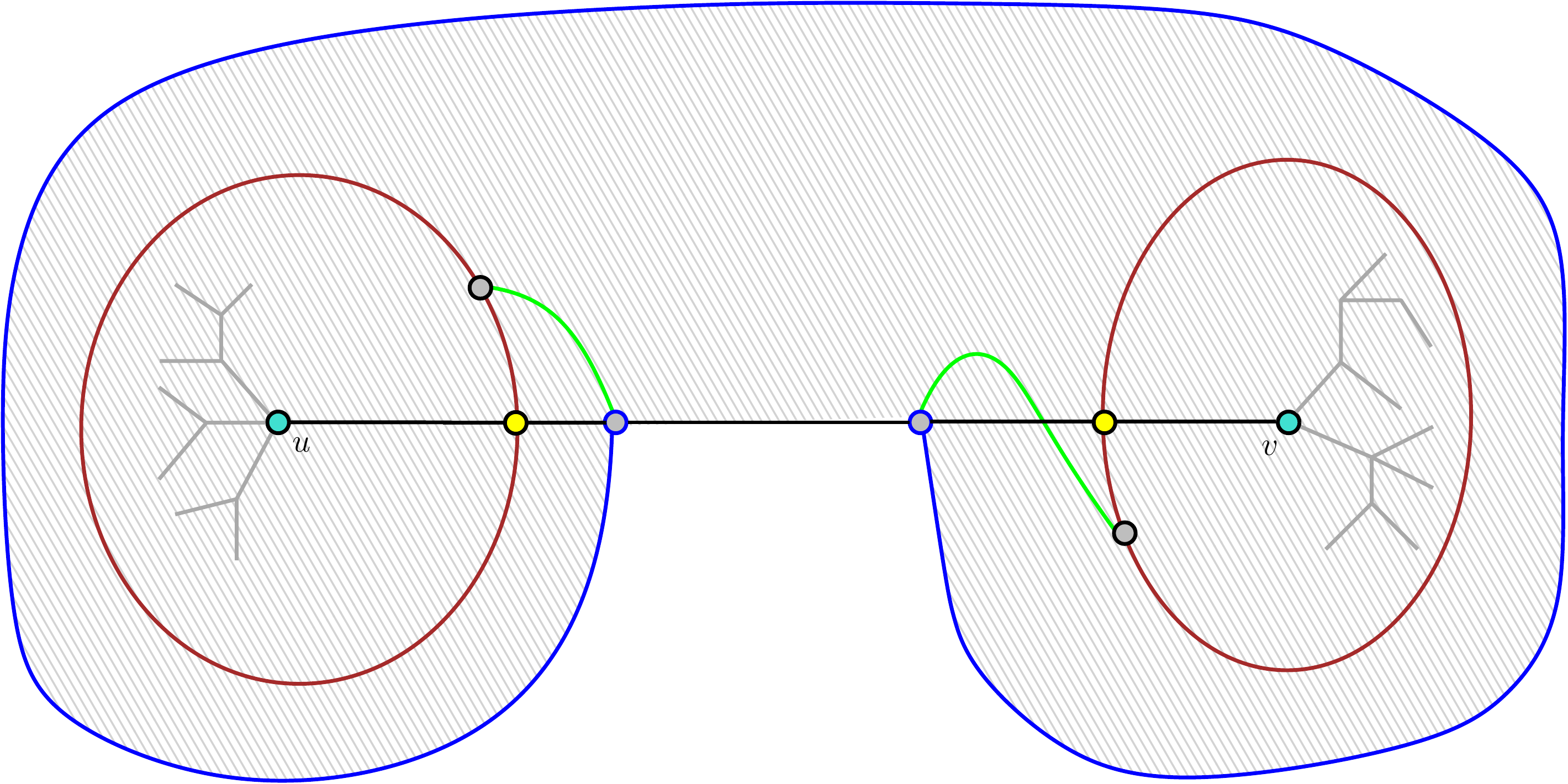}
        \caption{A bad segment that contains all of $S \cup T$ in its cycle.}
        \label{fig07}
    \end{center}
\end{figure}

 \vspace{-0.25em}

\begin{definition}\label{def:usefulSubwalkWeaklyOverview}
A subwalk of a walk $W$ is a {\em preliminary group} of $W$ if either {\em (i)} it has endpoints on two different maximal degree-2 paths of $R$ or an endpoint in $V_{=1}(R)\cup V_{\geq 3}(R)$ or it is spiraling, or {\em (ii)} it is the union of an inclusion-wise maximal collection of segments not of type {\em (i)}.
\end{definition}

\vspace{-0.25em}

The collection of preliminary groups of $W$ is  denoted by $\WUse(W)$.  Clearly, it is a partition of $W$. For a weak linkage $\cal W$, $\WUse({\cal W})=\bigcup_{W\in{\cal W}}\WUse(W)$. Then, 

\vspace{-0.25em}

\begin{observation}\label{obs:usefulSubwalkWeaklyOverview}
Let $\cal W$ be a weak linkage. The number of type-(ii) preliminary groups in $\WUse({\cal W})$ is at most $1$ plus the number of type-(i) preliminary groups  in $\WUse({\cal W})$.
\end{observation}

\vspace{-0.25em}

Roughly speaking, a type-(i) preliminary group is easily pushed onto $R$ so that it becomes merely a simple path (see Fig.~\ref{fig0408}). Thus, by Observation \ref{obs:usefulSubwalkWeaklyOverview}, all type-(ii) preliminary groups of a solution in total do not give rise to the occupation of more than $x+1$ copies of an edge, where $x$ is the number of type-(i) preliminary groups.

\medskip
\noindent{\bf Rollback Spirals and Winding Number.} Unfortunately, the number of spirals can be huge. Nevertheless, we can pair-up {\em some} of them so that they will ``cancel'' each other when pushed onto $R$ (see Fig.~\ref{fig:rollbackSpirals}), thereby behaving like a type-(ii) preliminary group.
Intuitively, we pair-up two spirals of a walk if one of them goes from the left-side to the right-side of the path, the other goes from the right-side to the left-side of the same path, and ``in between'' them on the walk, there are only type-(ii) preliminary groups and spirals that have already been paired-up. We refer to paired-up spirals as {\em rollback spirals}. (Not all spirals can be paired-up in this manner.) This gives rise to the following strengthening of Definition \ref{def:usefulSubwalkWeaklyOverview}.

\vspace{-0.25em}

\begin{definition}\label{def:usefulSubwalkOverview}
A subwalk of a walk $W$ is called a {\em group} of $W$ if either {\em (i)} it is a non-spiral type-(i) preliminary group, or {\em (ii)} it is the union of an inclusion-wise maximal collection of segments not of type {\em (i)} (i.e., all endpoints of the segments in the group are internal vertices of the same maximal degree-2 path of $R$). The {\em potential} of a group is (roughly) $1$ plus its number of non-rollback~spirals.
\end{definition}

\vspace{-0.25em}

Now, rather than upper bounding the total number of spirals, we only need to upper bound the number of non-rollback spirals. To this end, we use the notion of {\em winding number} (in Section \ref{sec:winding}), informally defined as follows. Consider a solution $\cal P$, a path $Q\in {\cal P}$, and a long path $P$ of $R$ with separators $S_u$ and $S_v$. As $S_u$ and $S_v$ are minimal separators in a triangulated graph (the radial completion is triangulated), they are cycles, and as at least one vertex in $T$ belongs to the outer face, they form a ring (see Fig.~\ref{fig:windingPaths}). Each maximal subpath of $Q$ that lies inside this ring  can either {\em visit} the ring, which means that both its endpoints belong to the same separator, or {\em cross} the ring, which means that its endpoints belong one to $S_u$ and the other to $S_v$ (see Fig.~\ref{fig:windingPaths}). Then, the (absolute value of the) {\em winding number} of a crossing subpath is the number of times it ``winds around'' $P$ inside the ring (see Fig.~\ref{fig:windingPaths}). 
At least intuitively, it should be clear that winding numbers and non-rollback spirals are related. In particular, each ring can only have $2^{\OO(k)}$ visitors and crossings subpaths (because the size of each separator is $2^{\OO(k)}$), and we only have $\OO(k)$ rings to deal with. Thus, it is possible to show that if the winding number of every crossing subpath is upper bounded by $2^{\OO(k)}$, then the total number of non-rollback spirals is upper bounded by $2^{\OO(k)}$ as well. The main tool we employ to bound the winding number of every crossing path is the following known result (rephrased to simplify the overview).

\begin{figure}
    \begin{center}
    \includegraphics[scale=0.425]{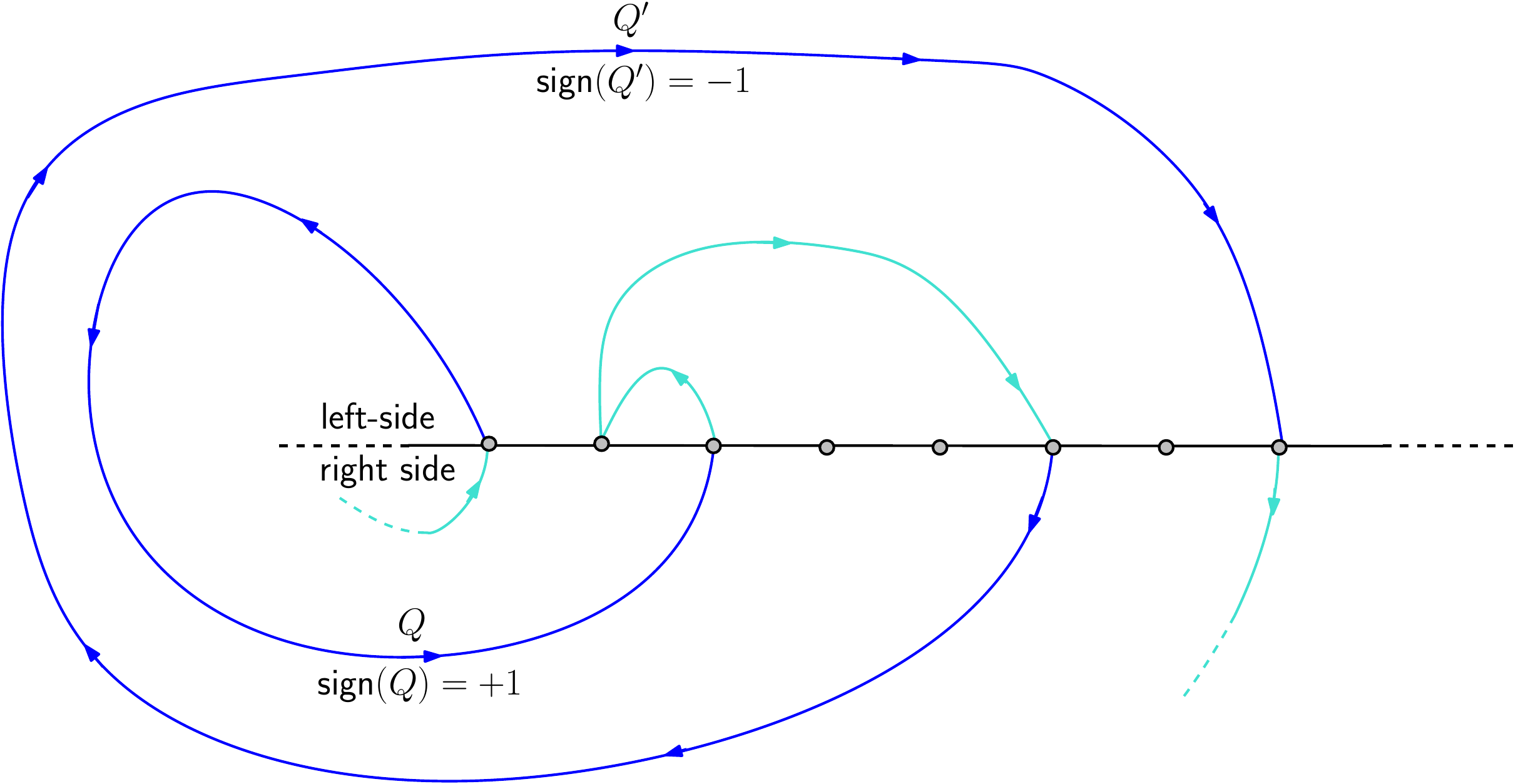}
    \caption{Rollback spirals.}
    \label{fig:rollbackSpirals}
    \end{center}
\end{figure}

\vspace{-0.25em}

\begin{proposition}[\cite{DBLP:conf/focs/CyganMPP13}]\label{prop:ring-reroutingOverview}
Let $G$ be a graph embedded in a ring  with a crossing path $P$. Let $\Pp$ and $\Qq$ be two collections of vertex-disjoint crossings paths of the same size. (A path in $\Pp$ can intersect a path in $\Qq$, but not another path in $\Pp$.) Then, $G$ has a collection of crossing paths $\Pp'$ such that {\em (i)} for every path in $\Pp$, there is a path in $\Pp'$ with the same endpoints and vice versa, and {\em (ii)} the maximum difference between (the absolute value of) the winding numbers with respect to $P$ of any path in $\Pp'$ and any path in $\Qq$ is at most $6$.
\end{proposition}

\vspace{-0.25em}

To see the utility of Proposition~\ref{prop:ring-reroutingOverview}, suppose momentarily that none of our rings has visitors. Then, if we could ensure that for each of our rings, there is a collection $\Qq$ of vertex-disjoint paths of {\em maximum size} such that the winding number of each path in $\Qq$ is a constant, Proposition \ref{prop:ring-reroutingOverview} would have the following implication: if there is a solution, then we can modify it within each ring to obtain another solution such that each crossing subpath of each of its paths will have a constant winding number (under the supposition that the rings are disjoint, which we will deal with later in the overview), see Fig.~\ref{fig:windingPaths}. Our situation is more complicated due to the existence of visitors---we need to ensure that the replacement $\Pp'$ does not intersect them. On a high-level, this situation is dealt with by first showing how to ensure that visitors do not ``go too deep'' into the ring on either side of it. Then, we consider an ``inner ring'' where visitors do not exist, on which we can apply Proposition \ref{prop:ring-reroutingOverview}. Afterwards, we are able to bound the winding number of each crossing path by $2^{\OO(k)}$ (but not by a constant) in the (normal) ring.  

\begin{figure}[t]
    \begin{center}
    \includegraphics[scale=0.375]{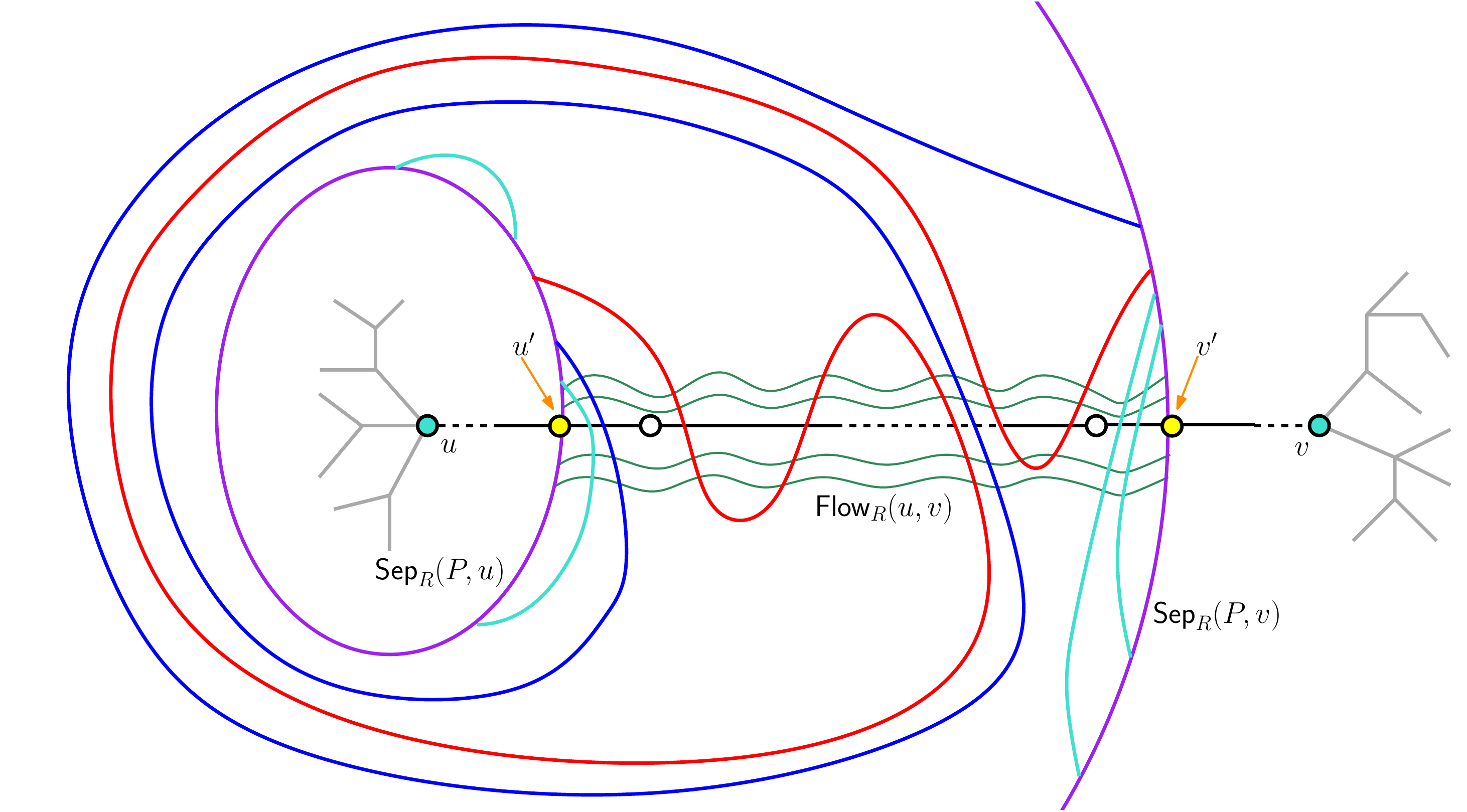}
    \includegraphics[scale=0.375]{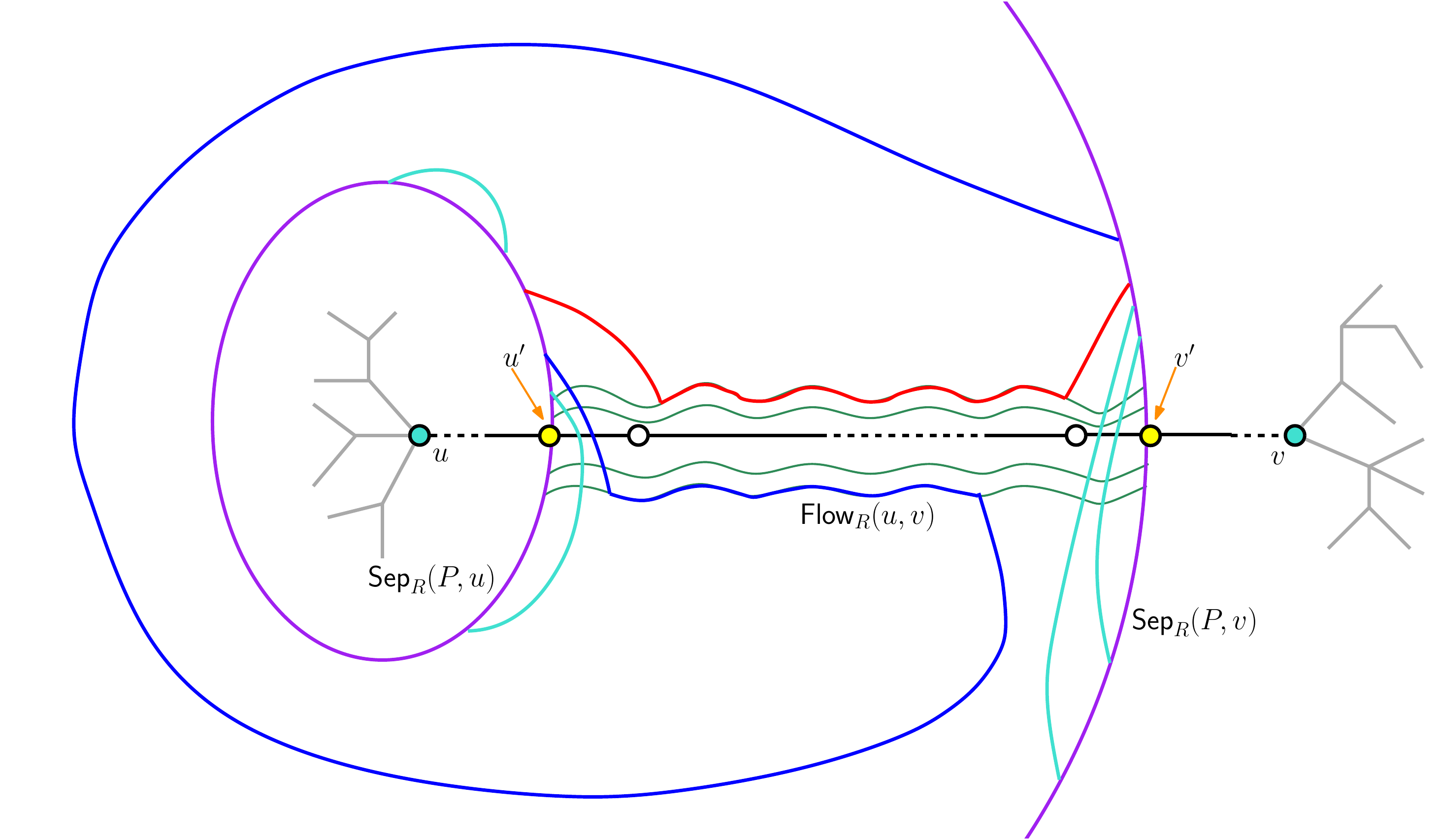}
    \caption{A solution winding in a ring (top), and the ``unwinding'' or it (bottom).}
    \label{fig:windingPaths}
    \end{center}
\end{figure}



\medskip
\noindent{\bf Modifying $R$ within Rings.} To ensure the existence of the aforementioned collection $\Qq$ for each ring, we need to modify $R$. To this end, consider a long path $P$ with separators $S_u$ and $S_v$, and let $P'$ be the subpath of $P$ inside the ring defined by the two separators. We compute a maximum-sized collection of vertex-disjoint paths $\flow(u,v)$ such that each of them has one endpoint in $S_u$ and the other  in $S_v$.\footnote{This flow has an additional property: there is a tight collection of $\Cc(u,v)$ of concentric cycles separating $S_u$ and $S_v$ such that paths in $\flow(u,v)$ do not ``oscillate'' too much between any two cycles in the collection. Such a maximum flow is said the be \emph{minimal} with respect to $\Cc(u,v)$.}
Then, we prove a result that roughly states the following.

\vspace{-0.25em}

\begin{lemma}
There is a path $P^\star$ in the ring defined by $S_u$ and $S_v$ with the same endpoints as $P'$ crossing each path in $\flow(u,v)$ at most once. Moreover, $P^\star$ is computable in linear time.
\end{lemma}

\vspace{-0.25em}

Having $P^\star$ at hand, we replace $P'$ by $P^\star$. This is done for every maximal degree-2 path, and thus we complete the construction of $R$. However, at this point, it is not clear why after we perform  these replacements, the separators considered earlier remain separators, or that we even still have a tree. Roughly speaking, a scenario as depicted in Fig.~\ref{fig:badRings} can potentially happen. To show that this is not the case, it suffices to prove that there cannot exist a vertex that belongs to two different rings. Towards that, we apply another preprocessing operation: we ensure that the radial completion of $G$ does not have $2^{ck}$ (for some constant $c$) {\em concentric cycles} that contain no vertex in $S\cup T$ by using another result by Adler et al.~\cite{DBLP:journals/jct/AdlerKKLST17}. Informally, a sequence of concentric cycles is a sequence of vertex-disjoint cycles where each one of them is contained inside the next one in the sequence. Having no such sequences, we prove the following.

\vspace{-0.25em}

\begin{lemma}\label{lem:closeToROverview}
Let $R'$ be any Steiner tree. For every vertex $v$, there exists a vertex in $V(R')$ whose distance to $v$ (in the radial completion of $G$) is $2^{ck}$ for some constant $c$.
\end{lemma}

\vspace{-0.25em}

\begin{figure}
    \begin{center}
        \includegraphics[scale=0.325]{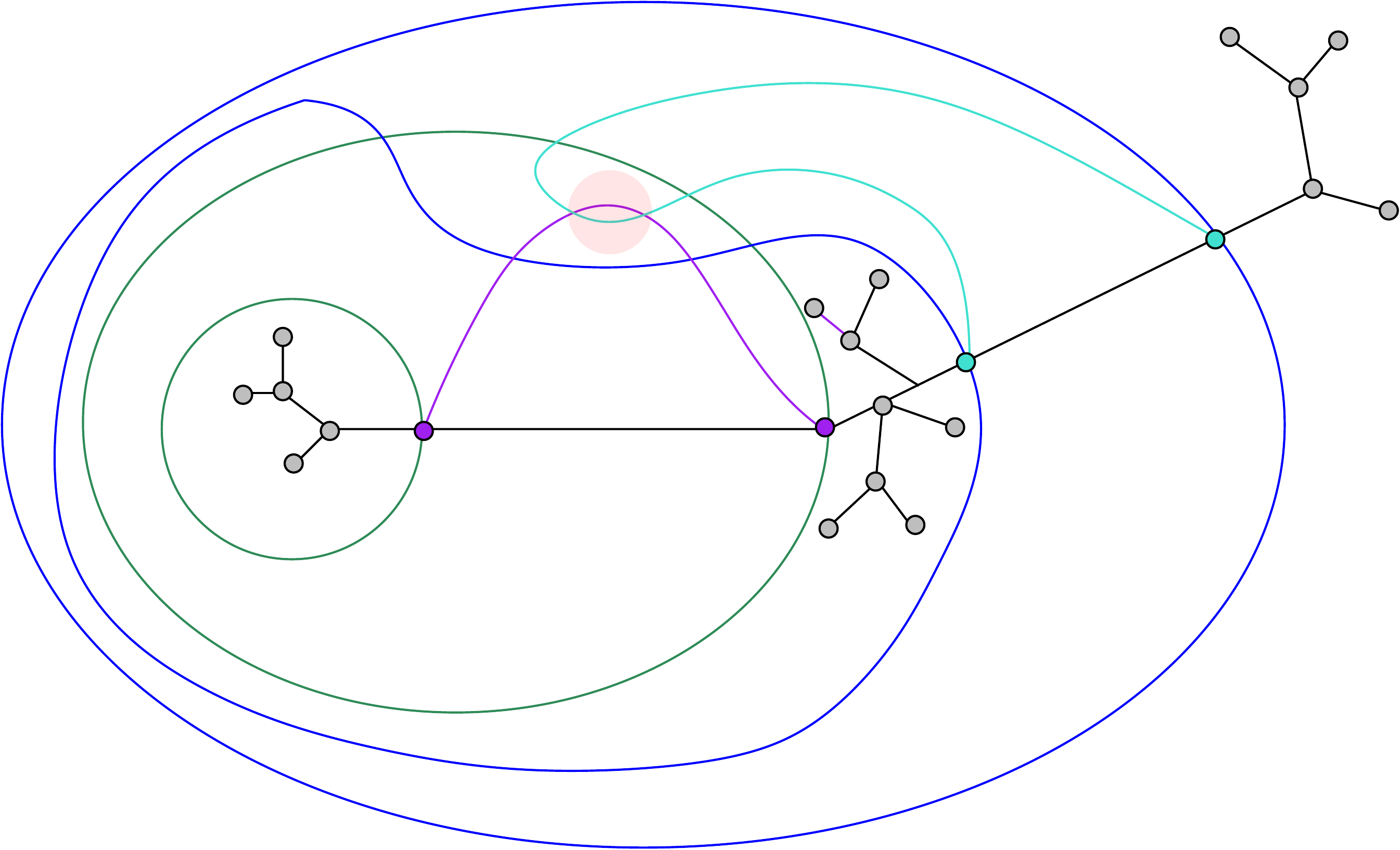}
        \caption{The green and blue rings intersect, which can create cycles in $R$ when~replacing~paths.}
        \label{fig:badRings}
    \end{center}
\end{figure}

To see why intuitively this lemma is correct, note that if $v$ was ``far'' from $R'$ in the {\em radial completion of $G$}, then in $G$ itself $v$ is surrounded by a large sequence of concentric cycles that contain no vertex in $S\cup T$. Having Lemma \ref{lem:closeToROverview} at hand, we show that if a vertex belongs to a certain ring, then it is ``close'' to at least one vertex of the restriction of $R$ to that ring. In turn, that means that if a vertex belongs to two rings, it can be used to exhibit a ``short'' path between one vertex in the restriction of $R$ to one ring and another vertex in the restriction of $R$ to the second ring. By choosing constants properly, this path is shown to exhibit a detour in $R$, and hence we reach a contradiction. (In this argument, we use the fact that for every vertex $u$, towards the computation of the separator, we considered a vertex $u'$ of distance $2^{c_1k}$ from $u$---this subpath between $u$ and $u'$ is precisely that subpath that we will shortcut.)

\medskip
\noindent{\bf Pushing a Solution Onto $R$.} So far, we have argued that if there is a solution, then there is also one such that the sum of the potential of all of the groups of all of its paths is at most $2^{\OO(k)}$. Additionally, we discussed the intuition why this, in turn, implies the following result.

\vspace{-0.25em}

\begin{lemma}\label{lem:finalOverview}
If there is a solution $\cal P$, then there is a weak linkage pushed onto $R$ that is discretely homotopic to $\cal P$ and uses at most $2^{\OO(k)}$ copies of every edge.
\end{lemma}

\vspace{-0.25em} 

The formal proof of Lemma \ref{lem:finalOverview} (in Section \ref{sec:pushing}) is quite technical. On a high level, it consists of three phases. First, we push onto $R$ all {\em sequences} of the solution---that is, maximal subpaths that touch (but not necessarily cross) $R$ only at their endpoints. Second, we eliminate some U-turns of the resulting weak linkage (see Fig.~\ref{fig16}), as well as ``move through'' $R$ segments with both endpoints being internal vertices of the same maximal degree-2 path of $R$ and crossing it in opposing directions (called {\em swollen segments}). At this point, we are able to bound by $2^{\OO(k)}$ the number of segments of the pushed weak linkage. Third, we eliminate all of the remaining U-turns, and show that then, the number of copies of each edge used must be at most $2^{\OO(k)}$. We also modify the pushed weak linkage to be of a certain ``canonical form'' (see Section \ref{sec:pushing}).

\medskip
\noindent{\bf Generating a Collection of Pushed Weak Linkages.}  In light of Lemma \ref{lem:finalOverview} and Proposition \ref{prop:schOverview}, it only remains to generate a collection of $2^{\OO(k^2)}$ pushed weak linkages that includes all pushed weak linkages (of some canonical form) using at most $2^{\OO(k)}$ copies of each edge. (This part, along with the preprocessing and construction of $R$, are the algorithmic parts of~our~proof.)

This part of our proof is essentially a technical modification and adaptation of the work of Schrijver \cite{DBLP:journals/siamcomp/Schrijver94} (though we need to be more careful to obtain the bound $2^{\OO(k^2)}$). Thus, we only give a brief description of it in the overview. Essentially, we generate pairs of a {\em pairing} and a {\em template}: a pairing assigns, to each vertex $v$ of $R$ of degree $1$ or at least $3$, a set of pairs of edges incident to $v$ to indicate that copies of these edges are to be visited consecutively (by at least one walk of the weak linkage under construction); a template further specifies, for each of the aforementioned pairs of edges, how many times copies of these edges are to be visited consecutively (but not which copies are paired-up). Clearly, there is a natural association of a pairing and a template to a pushed weak linkage. Further, we show that to generate all pairs of pairings and templates associated with the weak linkages we are interested in, we only need to consider pairings that in total have $\OO(k)$ pairs and templates that assign numbers bounded by $2^{\OO(k)}$ (because we deal with weak linkages using $2^{\OO(k)}$ copies of each edge):

\vspace{-0.25em} 

\begin{lemma}
There is a collection of $2^{\OO(k^2)}$ pairs of pairings and templates that, for any canonical pushed weak linkage $\cal W$ using only $2^{\OO(k)}$ copies of each edge, contains a pair (of a pairing and a template) ``compatible'' with $\cal W$. Further, such a collection is computable in~time~$2^{\OO(k^2)}$.
\end{lemma}

\vspace{-0.25em} 

\begin{figure}[t]
    \begin{center}
        \includegraphics[scale=0.65]{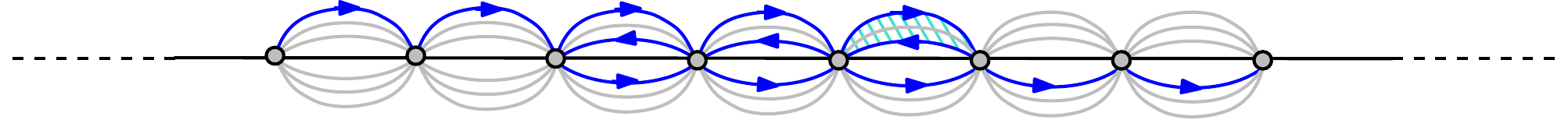}
        \caption{A walk going back and forth along a path of $R$, which gives rise to U-turns.}
        \label{fig16}
    \end{center}
\end{figure}

Using somewhat more involved arguments (in Section \ref{sec:reconstruction}), we also prove the following.

\vspace{-0.25em} 

\begin{lemma}
Any canonical pushed weak linkage is ``compatible'' with exactly one pair of a pairing and a template. Moreover, given a pair of a pairing and a template, if a canonical pushed weak linkage compatible with it exists, then it can be found in time polynomial in~its~size.
\end{lemma}

\vspace{-0.25em} 

These two lemmas complete the proof: we can indeed generate a collection of $2^{\OO(k^2)}$ pushed weak linkages containing all canonical pushed weak linkages using only $2^{\OO(k)}$ copies of any edge.


\section{Preliminaries}\label{sec:prelims}

Let $A$ be a set of elements. A cyclic ordering $<$ on $A$ is an ordering $(a_0,a_1,\ldots,a_{|A|-1})$ of the elements in $A$ such that, by enumerating $A$ in {\em clockwise order} starting at $a_i\in A$, we refer to the ordering $a_i,a_{(i+1)\mod |A|},\ldots,a_{(i+|A|-1)\mod |A|}$, and by enumerating $A$ in {\em counter-clockwise order} starting at $a_i\in A$, we refer to the ordering $a_i,a_{(i-1)\mod |A|},\ldots,a_{(i-|A|+1)\mod |A|}$. We consider all cyclic orderings of $A$ that satisfy the following condition to be the equivalent (up to cyclic shifts): the enumeration of $A$ in cyclic clockwise order starting at $a_i$, for any $a_i\in A$, produces the same sequence. For a function $f: X\rightarrow Y$ and a subset $X'\subseteq X$, we denote the restriction of $f$ to $X'$ by $f|_{X'}$.

\paragraph{Graphs.} 
Given an undirected graph $G$, we let $V(G)$ and $E(G)$ denote the vertex set and edge set of $G$, respectively. Similarly, given a directed graph (digraph) $D$, we let $V(D)$ and $A(D)$ denote the vertex set and arc set of $D$, respectively. Throughout the paper, we deal with graphs without self-loops but with parallel edges. 
 Whenever it is not explicitly written otherwise, we deal with undirected graphs. Moreover, whenever $G$ is clear from context, denote $n=|V(G)|$.

For a graph $G$ and a subset of vertices $U\subseteq V(G)$, the subgraph of $G$ induced by $U$, denoted by $G[U]$, is the graph on vertex set $U$ and edge set $\{\{u,v\}\in E(G): u,v\in U\}$. Additionally, $G-U$ denotes the graph $G[V(G)\setminus U]$. For a subset of edges $F\subseteq E(G)$, $G-F$ denotes the graph on vertex set $V(G)$ and edge set $E(G)\setminus F$. For a vertex $v\in V(G)$, the set of neighbors of $v$ in $G$ is denoted by $N_G(v)$, and for a subset of vertices $U\subseteq V(G)$, the open neighborhood of $U$ in $G$ is defined as $N_G(U)=\bigcup_{v\in U}N_G(v)\setminus U$. Given three subsets of vertices $A,B,S\subseteq V(G)$, we say that $S$ {\em separates} $A$ from $B$ if $G-S$ has no path with an endpoint in $A$ and an endpoint in $B$. For two vertices $u,v\in V(G)$, the {\em distance} between $u$ and $v$ in $G$ is the length (number of edges) of the shortest path between $u$ and $v$ in $G$ (if no such path exists, then the distance is $\infty$), and it is denoted by $\dist_G(u,v)$; in case $u=v$, $\dist_G(u,v)=0$. For two subsets $A,B\subseteq V(G)$, define $\dist_G(A,B)=\min_{u\in A, v\in B}\dist_G(u,v)$. A {\em{linkage}} of order $k$ in $G$ is an ordered family $\Pp$ of $k$ vertex-disjoint paths in $G$. Two linkages $\Pp=(P_1,\ldots,P_k)$ and $\Qq=(Q_1,\ldots,Q_k)$ are {\em{aligned}} if for all $i\in\{1,\ldots,k\}$, $P_i$ and $Q_i$ have the same endpoints.

For a tree $T$ and $d\in\mathbb{N}$, let $V_{\geq d}(T)$ (resp.~$V_{=d}(T)$) denote the set of vertices of degree at least (resp.~exactly $d$ in $T$. For two vertices $u,v\in V(T)$, the unique subpath of $T$ between $u$ and $v$ is denoted by $\pathT_T(u,v)$. We say that two vertices $u,v\in V(T)$ are {\em near each other} if $\pathT_T(u,v)$ has no internal vertex from $V_{\geq 3}(T)$, and  call $\pathT_T(u,v)$ a {\em degree-2 path}. In case $u,v\in V_{\geq 3}(T)\cup V_{=1}(T)$, $\pathT_T(u,v)$ is called a {\em maximal degree-2 path}.

\paragraph{Planarity.}
A planar graph is a graph that can be embedded in the Euclidean plane, that is, there exists a mapping  from every vertex to a point on a plane, and from every edge to a plane curve on that plane, such that the extreme points of each curve are the points mapped to the endpoints of the corresponding edge, and all curves are disjoint except on their extreme points. A {\em plane graph} $G$ is a planar graph with a fixed embedding. 
Its faces are the regions bounded by the edges, including the outer infinitely large region. For every vertex $v\in V(G)$, we let $E_G(v)=(e_0,e_1,\ldots,e_{t-1})$ for $t\in \mathbb{N}$ where $e_0,e_1,\ldots,e_{t-1}$ are the edges incident to $v$ in clockwise order (the decision which edge is $e_0$ is arbitrary). A planar graph $G$ is  {\em triangulated} if the addition of any edge (not parallel to an existing edge) to $G$ results in a non-planar graph. A plane graph $G$ that is triangulated is $2$-connected, and each of its faces is a simple cycle that is a triangle or a cycle that consists of two parallel edges (when the graph is not simple). As we will deal with triangulated graphs, the following proposition will come in handy.

\begin{proposition}[Proposition 8.2.3 in \cite{DBLP:books/daglib/0030489}]\label{prop:sepCycle}
Let $G$ be a triangulated plane graph. Let $A,B\subseteq V(G)$ be disjoint subsets such that $G[A]$ and $G[B]$ are connected graphs. Then, for any minimal subset $S\subseteq V(G)\setminus (A\cup B)$ that separates $A$ from $B$, it holds that $G[S]$ is a cycle.\footnote{Here, the term cycle also refers to the degenerate case where $|S|=1$.}
\end{proposition}

The {\em radial graph} (also known as the {\em face-vertex incidence graph}) of a plane graph $G$ is the planar graph $G'$ whose vertex set consists of $V(G)$ and a vertex $v_f$ for each face $f$ of $G$, and whose edge set consists of an edge $\{u,v_f\}$ for every vertex $u\in V(G)$ and face $f$ of $G$ such that $u$ is incident to (i.e.~lies on the boundary of) $f$. 
The {\em radial completion} of $G$ is the graph $G'$ obtained by adding the edges of $G$ to the radial graph of $G$. The graph $G'$ is planar, and we draw it on the plane so that its drawing coincides with that of $G$ with respect to $V(G)\cup E(G)$. Moreover, $G'$ is triangulated and, under the assumption that $G$ had no self-loops, $G'$ also has no self-loops (since all new edges in $G'$ have one endpoint in $V(G)$ and the other endpoint in $V(G') \setminus V(G)$).
For a plane graph $G$, the \emph{radial distance} between two vertices $u$ and $v$ is one less than the minimum length of a sequence of vertices that starts at $u$ and ends at $v$, such that every two consecutive vertices in the sequence lie on a common face.\footnote{We follow the definition of radial distance that is given in \cite{DBLP:conf/soda/JansenLS14} in order to cite a result in that paper verbatim (see Proposition \ref{prop:concentricAtAllDists}). We remark that in \cite{DBLP:journals/jacm/BodlaenderFLPST16}, the definition of a radial distance is slightly different.}  We denote the radial distance by ${\sf rdist}_G(u,v)$. This definition extends to subsets of vertices: for $X,Y \subseteq V(G)$, ${\sf rdist}_G(X,Y)$ is the minimum radial distance over all pairs of vertices in $X \times Y$.

For any $t\in\mathbb{N}$, a sequence ${\cal C}=(C_1,C_2,\ldots,C_t)$ of $t$ cycles in a plane graph $G$ is said to be {\em concentric} if for all $i\in\{1,2,\ldots,t-1\}$, the cycle $C_i$ is drawn in the strict interior of $C_{i+1}$ (excluding the boundary, that is, $V(C_i)\cap V(C_{i+1})=\emptyset$). The {\em length} of ${\cal C}$ is $t$. For a subset of vertices $U\subseteq V(G)$, we say that $\cal C$ is $U$-free if no vertex of $U$ is drawn in the strict interior~of~$C_t$.

\paragraph{Treewidth.}
Treewidth is a measure of how ``treelike'' is a graph, formally defined as follows.
\begin{definition}[{\bf Treewidth}]\label{def:treewidth}
A \emph{tree decomposition} of a graph $G$ is a pair $(T,\beta)$ of a tree $T$ and $\beta:V(T) \rightarrow 2^{V(G)}$, such that
\vspace{-0.5em}
\begin{enumerate}
\itemsep0em 
\item\label{item:twedge} for any edge $\{x,y\} \in E(G)$ there exists a node $v \in V(T)$ such that $x,y \in \beta(v)$, and
\item\label{item:twconnected} for any vertex $x \in V(G)$, the subgraph of $T$ induced by the set $T_x = \{v\in V(T): x\in\beta(v)\}$ is a non-empty tree.
\end{enumerate}
The {\em width} of $(T,\beta)$ is $\max_{v\in V(T)}\{|\beta(v)|\}-1$. The {\em treewidth} of $G$, denoted by $\tw(G)$, is the minimum width over all tree decompositions of $G$.
\end{definition}

The following proposition, due to Cohen-Addad et al.~\cite{DBLP:conf/stoc/Cohen-AddadVKMM16}, relates the treewidth of a plane graph to the treewidth of its radial completion.

\begin{proposition}[Lemma 1.5 in ~\cite{DBLP:conf/stoc/Cohen-AddadVKMM16}]\label{prop:twRadial}
Let $G$ be a plane graph, and let $H$ be its radial completion. Then, $\tw(H)\leq \frac{7}{2}\cdot \tw(G)$.\footnote{More precisely, Cohen-Addad et al.~\cite{DBLP:conf/stoc/Cohen-AddadVKMM16} prove that the {\em branchwidth} of $G$ is at most twice the branchwidth of $H$. Since the treewidth (plus 1) of a graph is lower bounded by its branchwidth and upper bounded by $\frac{3}{2}$ its branchwidth (see \cite{DBLP:conf/stoc/Cohen-AddadVKMM16}), the proposition follows.} 
\end{proposition}

\subsection{Homology and Flows}\label{sec:prelimsHomology}

For an alphabet $\Sigma$, denote $\Sigma^{-1}=\{\alpha^{-1}: \alpha\in \Sigma\}$. For a symbol $\alpha\in\Sigma$, define $(\alpha^{-1})^{-1}=\alpha$, and for a word $w=a_1a_2\cdots a_t$ over $\Sigma\cup \Sigma^{-1}$, define $w^{-1}=a_t^{-1}a_{t-1}^{-1}\cdots a_1^{-1}$ and $w^1=w$. The empty word (the unique word of length $0$) is denoted by $1$. We say that a word $w=a_1a_2\cdots a_t$ over $\Sigma\cup \Sigma^{-1}$ is {\em reduced} if there does not exist $i\in\{1,2,\ldots,t-1\}$ such that $a_i=a_{i+1}^{-1}$. We denote the (infinite) set of reduced words over $\Sigma\cup \Sigma^{-1}$ by $\RW(\Sigma)$. The {\em concatenation} $w\circ\widehat{w}$ of two words $w=a_1a_2\cdots a_t$ and $\widehat{w}=b_1b_2\cdots b_\ell$ is the word $w^\star=a_1a_2\cdots a_tb_1b_2\cdots b_\ell$. The {\em product} $w\cdot\widehat{w}$ of two words $w=a_1a_2\cdots a_t$ and $\widehat{w}=b_1b_2\cdots b_\ell$ in $\RW(\Sigma)$ is a word $w^\star$ defined as follows:
\[w^\star=a_1a_2\cdots a_{t-r}b_{r+1}b_{r+2}\cdots b_\ell\]
where $r$ is the largest integer in $\{1,2,\ldots,\min(t,\ell)\}$ such that, for every $i\in\{1,2,\ldots,r\}$, $b_i=a_{t+1-i}^{-1}$. Note that $w^\star$ is a reduced word, and  the product operation is associative. The {\em reduction} of a word $w=a_1a_2\cdots a_t$ over $\Sigma\cup \Sigma^{-1}$ is the (reduced) word $w^\star=a_1\cdot a_2\cdots a_t$.


\begin{definition}[{\bf Homology}]\label{def:homology}
Let $D$ be a directed plane graph with outer face $f$, and denote the set of faces of $D$ by $\cal F$. Let $\Sigma$ be an alphabet. Two functions $\phi,\psi: A(D)\rightarrow \RW(\Sigma)$ are {\em homologous} if there exists a function $h: {\cal F}\rightarrow \RW(\Sigma)$ such that $h(f)=1$, and for every arc $e\in A(D)$, we have $h(f_1)^{-1}\cdot \phi(e)\cdot h(f_2)=\psi(e)$ where $f_1$ and $f_2$ are the faces at the left-hand side and the right-hand side of $e$, respectively.
\end{definition}

The following observation will be useful in later results.
\begin{observation}\label{obs:homProp}
    Let $\alpha,\beta,\gamma: A(D) \rightarrow \RW(\Sigma)$ be three functions such that $\alpha,\beta$ and $\beta,\gamma$ are pairs of homologous functions. Then, $\alpha, \gamma$ is also a pair of homologous functions.
\end{observation}
\begin{proof}
    Let $f$ and $g$ be the functions witnessing the homology of $\alpha, \beta$ and $\beta,\gamma$, respectively.
    Then, it is easy to check that the function $h = g \circ f$ (i.e.~the composition of $g$ and $f$) witnesses the homology of $\alpha,\gamma$.
\end{proof}

Towards the definition of flow, we denote an instance of \dpdp\ by a tuple $(D,S,T,g,k)$ where $D$ is a directed plane graph, $S,T\subseteq V(D)$, $k=|S|$ and $g: S\rightarrow T$. We assume that $g$ is bijective because otherwise the given instance is a \no-instance,
A {\em solution} of an instance $(D,S,T,g,k)$ of \dpdp\ is a set ${\cal P}$ of pairwise vertex-disjoint directed paths in $D$ that contains, for every vertex $s\in S$, a path directed from $s$ to $g(s)$. When we write $\RW(T)$, we treat $T$ as an alphabet---that is, every vertex in $T$ is treated as a symbol.

\begin{definition}[{\bf Flow}]\label{def:flow}
Let $(D,S,T,g,k)$ be an instance of \dpdp. Let $\phi: A(D)\rightarrow \RW(T)$ be a function. For any vertex $v\in V(D)$, denote the concatenation $\phi(e_1)^{\epsilon_1}\circ\phi(e_2)^{\epsilon_2}\cdots\phi(e_r)^{\epsilon_r}$ by $\conc(v)$, where $e_1,e_2,\ldots,e_r$ are the arcs incident to $v$ in clockwise order where the first arc $e_1$ is chosen arbitrarily, and for each $i\in\{1,2,\ldots,r\}$, $\epsilon_i=1$ if $v$ is the head of $e_i$ and $\epsilon_i=-1$ if $v$ is the tail of $e_i$. Then, the function $\phi$ is a {\em flow} if:\footnote{We note that there is slight technical difference between our definition and the definition in Section~3.1 in~\cite{DBLP:journals/siamcomp/Schrijver94}. There, a flow must put only a single alphabet ($v$ or $v^{-1}$ in $T \cup T^{-1}$) on the arcs incident on vertices in $S \cup T$.}
\begin{enumerate}
\item For every vertex $v\in V(D)\setminus(S\cup T)$, the reduction of $\conc(v)$ is $1$.

\item For every vertex $v\in S$, {\em (i)} $\conc(v)=a_0a_1\ldots a_{\ell-1}$ is a word of length $\ell\geq 1$, and {\em (ii)} there exists $i\in\{0,1,\ldots,\ell-1\}$ such that the reduction of $a_ia_{(i+1)\mod \ell}\cdots a_{(i+\ell-1)\mod \ell}$ equals $a_i$, where $a_i=g(v)$ if the arc $e$ associated with $a_i$ has $v$ as its tail, $a_i=g(v)^{-1}$ otherwise.

\item For every vertex $v\in T$, {\em (i)} $\conc(v)=a_0a_1\ldots a_{\ell-1}$ is a word of length $\ell\geq 1$, and {\em (ii)} there exists $i\in\{0,1,\ldots,\ell-1\}$ such that the reduction of $a_ia_{(i+1)\mod \ell}\cdots a_{(i+\ell-1)\mod \ell}$ equals $a_i$,  where $a_i=v$ if the arc $e$ associated with $a_i$ has $v$ as its head, $a_i=v^{-1}$ otherwise.
\end{enumerate}
\end{definition}

In the above definition, the conditions on the reduction of $\conc(v)$ for each vertex $v \in V(D)$ are called \emph{flow conservation} constraints. Informally speaking, these constraints resemble ``usual'' flow conservation constraints and ensure that every two walks that carry two different alphabets do not cross. 
The association between solutions to $(D,S,T,g,k)$ and flows is defined as follows.

\begin{definition}\label{def:solToFlow}
Let $(D,S,T,g,k)$ be an instance of \dpdp. Let ${\cal P}$ be a solution of $(D,S,T,g,k)$. The {\em flow $\phi: A(D)\rightarrow \RW(T)$ associated with $\cal P$} is defined as follows. For every arc $e\in A(D)$, define $\phi(e)=1$ if there is no path in $\cal P$ that traverses $e$, and $\phi(e)=t$ otherwise where $t\in T$ is the end-vertex of the (unique) path in $\cal P$ that traverses $e$. 
\end{definition}

The following proposition, due to Schrijver \cite{DBLP:journals/siamcomp/Schrijver94}, also holds for the above definition of flow.

\begin{proposition}[Proposition 5 in \cite{DBLP:journals/siamcomp/Schrijver94}]\label{prop:homology}
There exists a polynomial-time algorithm that, given a 
instance $(D,S,T,g,k)$ of \dpdp\ and a flow $\phi$, either finds a solution of $(D,S,T,g,k)$ or determines that there is no solution of $(D,S,T,g,k)$ such that the flow associated with it and $\phi$ are homologous.
\end{proposition}

We need a slightly more general version of this proposition because we will work with an instance of \dpdp\ where $D$ contains some ``fake'' edges that emerge when we consider the radial completion of the input graph---the edges added to the graph in order to attain its radial completion are considered to be fake.

\begin{corollary}\label{cor:homology}
There exists a polynomial-time algorithm that, given a 
instance $(D,S,T,g,k)$ of \dpdp, a flow $\phi$ and a subset $X\subseteq A(D)$, either finds a solution of $(D-X,S,T,g,k)$ or determines that there is no solution of $(D-X,S,T,g,k)$ such that the flow associated with it and $\phi$ are homologous.\footnote{Note that $\phi$ and homology concern $D$ rather than $D-X$.}
\end{corollary}

\begin{proof}
Given $(D,S,T,g,k)$, $\phi$ and $X$, the algorithm constructs an equivalent instance $(D',S,T,g,$ $k)$ of {\sc Directed \pdp} and a flow $\phi'$ as follows. Each arc $(u,v) \in X$ is replaced by a new vertex $w$ and two new arcs (whose drawing coincides with the former drawing of $(u,v)$), $(u,w)$ and $(v,w)$, and we define $\phi'(u,w) = \phi(u,v)$ and $\phi'(v,w) = \phi(u,v)^{-1}$. For all other arcs $a \in A(D) \cap A(D')$, $\phi'(a) = \phi(a)$. It is immediate to verify that $\phi'$ is a flow in $D'$, and that $(D',S,T,g,k)$ admits a solution that is homologous to $\phi'$ if and only if $(D,S,T,g,k)$ admits a solution that is disjoint from $X$ and homologous to $\phi$. Indeed, any solution of one of these instances is also a solution of the other one. Now, we apply Proposition~\ref{prop:homology} to either obtain a solution to $(D',S,T,g,k)$ homologous to $\phi'$ or conclude that no such solution exists.
\end{proof}


\section{Preprocessing to Obtain a Good Instance}\label{sec:preprocessing}

We denote an instance of \pdp\ similarly to an instance of \dpdp\ (in Section \ref{sec:prelims}) except that now the graph is denoted by $G$ rather than $D$ to stress the fact that it is undirected. Formally, an instance of \pdp\ is a tuple $(G,S,T,g,k)$ where $G$ is a plane graph, $S,T\subseteq V(G)$, $k=|S|$ and $g: S\rightarrow T$ is bijective.  
Moreover, we say that $(G,S,T,g,k)$ is {\em nice} if every vertex in $S\cup T$ has degree $1$ and $S\cap T=\emptyset$. The vertices in $S\cup T$ are called {\em terminals}. Let $H_G$ be the radial completion of $G$.
We choose a plane embedding of $H_G$ so that one of the terminals, $t^\star\in T$, will lie on the outer face.\footnote{This can be ensured by starting with an embedding of $H_G$ on a sphere, picking some face where $t^\star$ lies as the outer face, and then projecting the embedding onto the plane.}
A {\em solution} of an instance $(G,S,T,g,k)$ of \pdp\ is a set ${\cal P}$ of pairwise vertex-disjoint paths in $G$ that contains, for every vertex $s\in S$, a path with endpoints $s$ and~$g(s)$.

The following proposition eliminates all long sequences of $S\cup T$-free concentric cycles.

\begin{proposition}[\cite{DBLP:journals/jct/AdlerKKLST17}]\label{prop:concentric}
There exists a $2^{\OO(k)}n^2$-time algorithm that, given an instance $(G,S,T,g,$ $k)$ of \pdp, outputs an equivalent instance $(G',S,T,g,k)$ of \pdp\ where $G'$ is a subgraph of $G$ that has no sequence of $S\cup T$-free concentric cycles whose length is larger than $2^{ck}$ for some fixed constant $c\geq 1$.
\end{proposition}

Additionally, the following proposition reduces the treewidth of $G$. In fact, Proposition \ref{prop:concentric} was given by Adler et al.~\cite{DBLP:journals/jct/AdlerKKLST17} as a step towards the proof of the following proposition. However, while the absence of a long sequence of concentric cycles implies that the treewidth is small, the reversed statement is not correct (i.e.~small treewidth does not imply the absence of a long sequence of concentric cycles). Having small treewidth is required but not sufficient for our arguments, thus we cite both propositions.

\begin{proposition}[Lemma 10 in \cite{DBLP:journals/jct/AdlerKKLST17}]\label{prop:twReduction}
There exists a $2^{\OO(k)}n^2$-time algorithm that, given an instance $(G,S,T,g,k)$ of \pdp, outputs an equivalent instance $(G',S,T,g,k)$ of \pdp\ where $G'$ is a subgraph of $G$ whose treewidth is upper bounded by $2^{ck}$ for some fixed constant $c\geq 1$.\footnote{While the running time in Lemma 10 in \cite{DBLP:journals/jct/AdlerKKLST17} is stated to be $2^{2^{\OO(k)}}n^2$, the proof is easily seen to imply the bound $2^{\OO(k)}n^2$ in our statement. The reason why Adler et al.~\cite{DBLP:journals/jct/AdlerKKLST17} obtain a double-exponential dependence on $k$ when they solve \pdp\ is {\em not} due to the computation that attains a tree decomposition of width $\tw=2^{\OO(k)}$, but it is because that upon having such a tree decomposition, they solve the problem in time $2^{\OO(\tw)}n=2^{2^{\OO(k)}}n$.}
\end{proposition}

The purpose of this section is to transform an arbitrary instance of \pdp\ into a so called ``good'' instance, defined as follows.

\begin{definition}\label{def:goodInstance}
An instance $(G,S,T,g,k)$ of \pdp\ is {\em good} if it is nice, at least one terminal $t^\star\in T$ belongs to the outer faces of both $G$ and its radial completion $H_G$, 
the treewidth of $G$ is upper bounded by $2^{ck}$, and $G$ has no $S\cup T$-free sequence of concentric cycles whose length is larger than $2^{ck}$. Here, $c\geq 1$ is the fixed constant equal to the maximum among the fixed constants in Propositions \ref{prop:concentric} and \ref{prop:twReduction}.  
\end{definition}

Towards this transformation, note that given an instance $(G,S,T,g,k)$ of \pdp\ and a terminal $v\in S$, we can add to $G$ a degree-1 vertex $u$ adjacent to $v$, and replace $v$ by $u$ in $S$ and in the domain of $g$. This operation results in an equivalent instance of \pdp. Furthermore, it does not increase the treewidth of $G$ (unless the treewidth of $G$ is $0$, in which it increases to be $1$). The symmetric operation can be done for any terminal $v\in T$. By repeatedly applying these operations, we can easily transform $(G,S,T,g,k)$ to an equivalent nice instance. Moreover, the requirement that at least one terminal $t^\star\in T$ belongs to the outer faces of both $G$ and its radial completion can be made without loss of generality by drawing $G$ appropriately in the first place. Thus, we obtain the following corollary of Propositions \ref{prop:concentric} and~\ref{prop:twReduction}. Note that $k$ remains unchanged.

\begin{corollary}\label{cor:twReduction}
There exists a $2^{\OO(k)}n^2$-time algorithm that, given an instance $(G,S,T,g,k)$ of \pdp, outputs an equivalent good instance $(G',S',T',g',k)$ of \pdp\ where $|V(G')|=\OO(|V(G)|)$.
\end{corollary}

We remark that our algorithm, presented in Section \ref{sec:algorithm}, will begin by applying Corollary~\ref{cor:twReduction}. To simplify arguments ahead, it will be convenient to suppose that every edge in $H_G$ has $4|V(G)|+1=4n+1$ parallel copies. Thus, we slightly abuse the notation $H_G$ and use it to refer to $H_G$ enriched with such a number of parallel copies of each edge. For a pair of adjacent vertices $u,v\in V(H_G)$, we will denote the $4n+1$ parallel copies of edges between them by $e_{-2n},e_{-2n+1},\ldots,e_{-1},e_0,e_1,e_2,\ldots,e_{2n}$ where $e=\{u,v\}$, such that when the edges incident to $u$ (or $v$) are enumerated in cyclic order, the occurrences of $e_i$ and $e_{i+1}$ are consecutive for every $i\in\{-2n,-2n+1,\ldots,2n-1\}$, and $e_{-2n}$ and $e_{2n}$ are the outermost copies of $e$. Thus, for every $i\in\{-2n+1,-2n+2,\ldots,2n-1\}$, $e_i$ lies on the boundary of exactly two faces: the face bounded by $e_{i-1}$ and $e_i$, and the face bounded by $e_i$ and $e_{i+1}$. When the specification of the precise copy under consideration is immaterial, we sometimes simply use the notation $e$.


\section{Discrete Homotopy}\label{sec:discreteHomotopy}

The purpose of this section is to assert that rather than working with homology (Definition \ref{def:homology}) or the standard notion of homotopy, to obtain our algorithm it will suffice to work with a notion called {\em discrete homotopy}.  Working with discrete homotopy will {\em substantially shorten and simplify} our proof, particularly Section \ref{sec:pushing}. 
Translation from discrete homotopy to homology is straightforward, thus readers are invited to skip the proofs in this section when reading the paper for the first time. We begin by defining the notion of a weak linkage. This notion is a generalization of a linkage (see Section \ref{sec:prelims}) that concerns walks rather than paths, and which permits the walks to intersect one another in vertices. Here, we only deal with walks that may repeat vertices but which do not repeat edges. Moreover, weak linkages concern walks that are {\em non-crossing}, a property defined as follows (see Fig.~\ref{fig:crossing}).

\begin{figure}
    \begin{center}
        \includegraphics[scale=0.5]{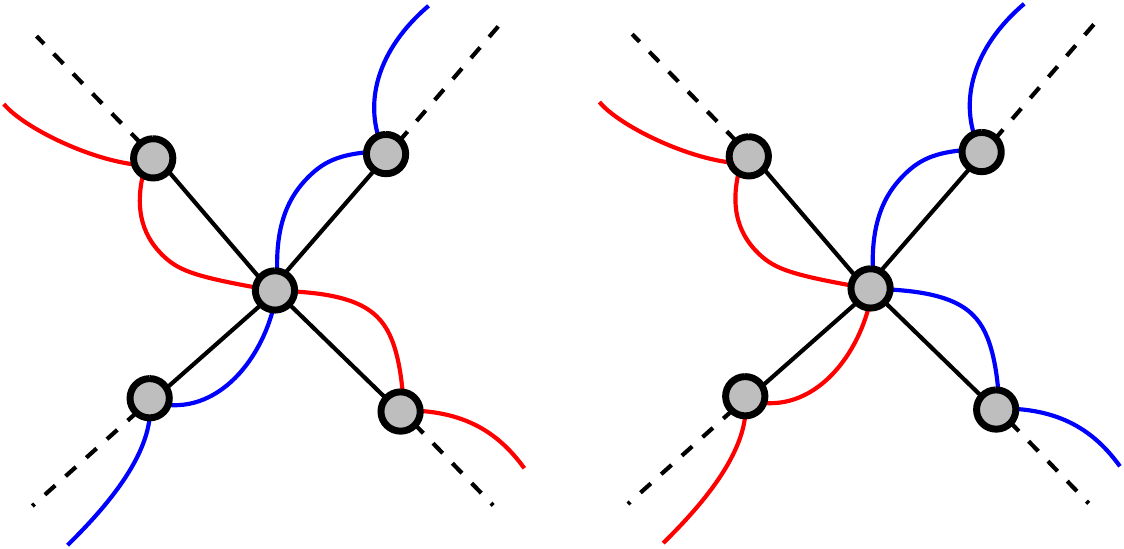}
        \caption{Crossing (left) and non-crossing (right) walks.}
        \label{fig:crossing}
    \end{center}
\end{figure}

\begin{definition}[{\bf Non-Crossing Walks}]\label{def:nonCrossingWalks}
Let $G$ be a plane graph, and let $W$ and $W'$ be two edge-disjoint walks in $G$. A {\em crossing} of $W$ and $W'$ is a tuple $(v,e,\widehat{e},e',\widehat{e}')$ where $e,\widehat{e}$ are consecutive in $W$, $e',\widehat{e}'$ are consecutive in $W'$, $v\in V(G)$ is an endpoint of $e,\widehat{e},e'$ and $\widehat{e}'$, and when the edges incident to $v$ are enumerated in clockwise order, then exactly one edge in $\{e',\widehat{e}'\}$ occurs between $e$ and $\wh{e}$.
We say that $W$ is {\em self-crossing} if, either it has a repeated edge, or it has two edge-disjoint subwalks that are crossing.
\end{definition}

We remark that when we say that a collection of edge-disjoint walks is non-crossing, we mean that none of its walks is self-crossing and no pair of its walks has a crossing.

\begin{definition}[{\bf Weak Linkage}]\label{def:weakLinkage}
Let $G$ be a plane graph. A {\em{weak linkage}} in $G$ of order $k$ is an ordered family of $k$ edge-disjoint non-crossing walks in $G$. Two weak linkages ${\cal W}=(W_1,\ldots,W_k)$ and $\Qq=(Q_1,\ldots,Q_k)$ are {\em{aligned}} if for all $i\in\{1,\ldots,k\}$, $W_i$ and $Q_i$ have the same endpoints.

Given an instance $(G,S,T,g,k)$ of \pdp, a weak linkage ${\cal W}$ in $G$ (or $H_G$) is {\em sensible} if its order is $k$ and for every terminal $s\in S$, ${\cal W}$ has a walk with endpoints $s$ and $g(s)$.
\end{definition}

The following observation is clear from Definitions \ref{def:nonCrossingWalks} and \ref{def:weakLinkage}.
\begin{observation}
    Let $G$ be a plane graph, and let $\cal W$ be a weak linkage in $G$.
    Let $e_1,e_2$ and $e_3,e_4$ be two pairs of edges in $E({\cal W})$ that tare all distinct and incident on a vertex $v$, and there is some walk in $\cal W$ where $e_1,e_2$ are consecutive, and likewise for $e_3,e_4$. Then, in a clockwise enumeration of edges incident to $v$, the pairs $e_1,e_2$ and $e_3,e_4$ do not cross, that is, they do not occur as $e_1, e_3, e_2, e_4$ in clockwise order (including cyclic shifs).
\end{observation}

We now define the collection of operations applied to transform one weak linkage into another weak linkage aligned with it (see Fig.~\ref{fig:faceOp}). We remark that the face push operation is not required for our arguments, but we present it here to ensure that discrete homotopy defines an equivalence relation (in case it will find other applications that need this property).

\begin{figure}
    \begin{center}
        \includegraphics[scale=0.55]{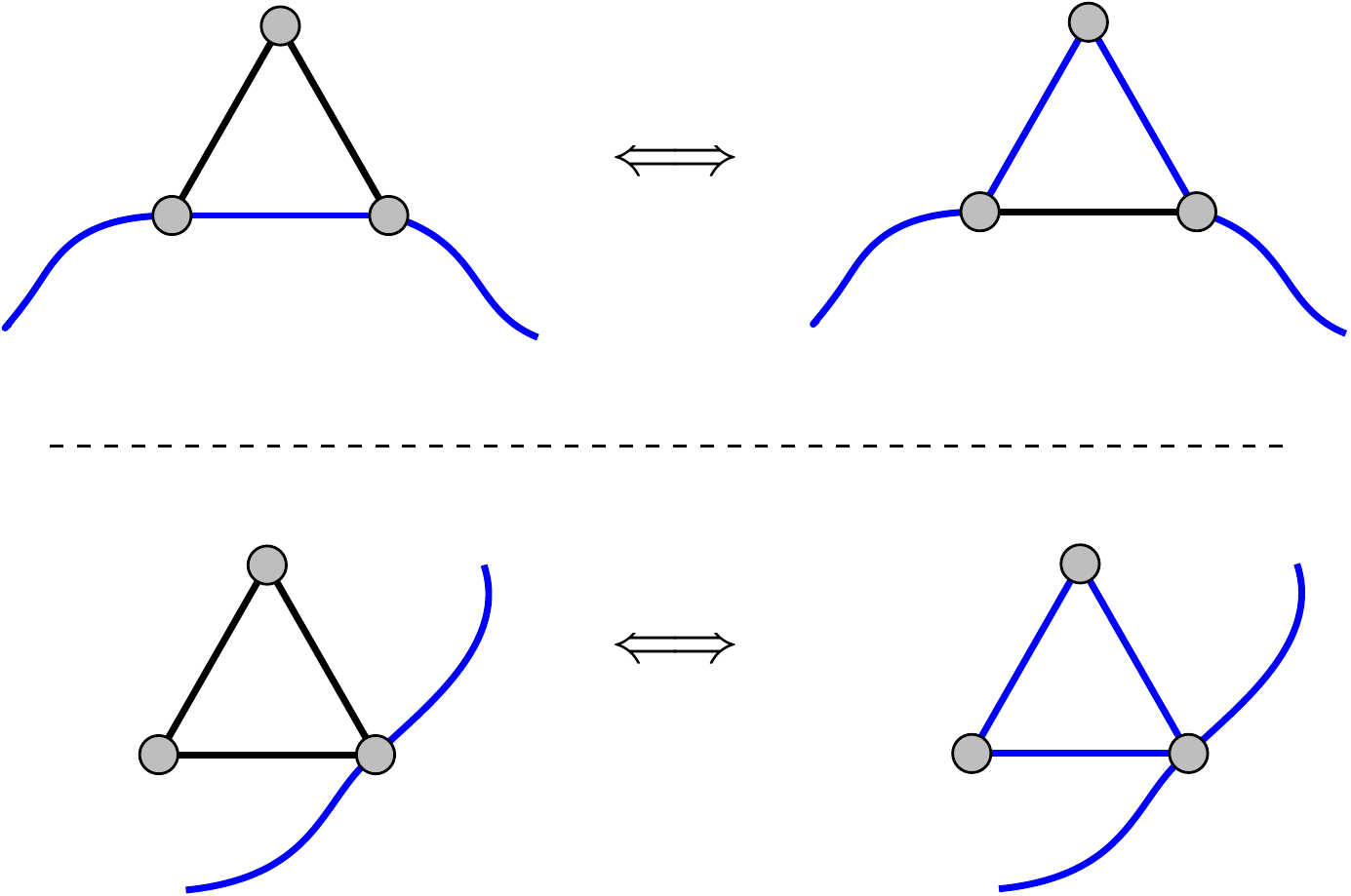}
        \caption{Face operations.}
        \label{fig:faceOp}
    \end{center}
\end{figure}

\begin{definition}[{\bf Operations in Discrete Homotopy}]\label{def:discreteHomotopyOperations}
Let $G$ be a triangulated plane graph with a weak linkage  ${\cal W}$, and a face $f$ that is not the outer face with boundary cycle $C$. Let~$W\!\in\! {\cal W}$. 
\begin{itemize}
\item {\bf \textsf{Face Move}.} Applicable to $(W,f)$ if there exists a subpath $P$ of $C$ such that {\em (i)} $P$ is a subwalk of $W$, {\em (ii)} $1\leq |E(P)|\leq |E(f)|-1$, and {\em (iii)} no edge in $E(C)\setminus E(P)$ belongs to any walk in $\cal W$. Then, the face move operation replaces $P$ in $W$ by the unique subpath of $C$ between the endpoints of $P$ that is edge-disjoint from $P$.

\item {\bf \textsf{Face Pull}.} Applicable to $(W,f)$ if $C$  is a subwalk $Q$ of $W$. Then, the face pull operation replaces $Q$ in $W$ by a single occurrence of the first vertex in $Q$.

\item {\bf \textsf{Face Push}.} Applicable to $(W,f)$ if {\em (i)} no edge in $E(C)$ belongs to any walk in $\cal W$, and {\em (ii)} there exist two consecutive edges $e,e'$ in $W$ with common vertex $v\in V(C)$ (where $W$ visits $e$ first, and $v$ is visited between $e$ and $e'$) and an order (clockwise or counter-clockwise) to enumerate the edges incident to $v$ starting at $e$ such that the two edges of $E(C)$ incident to $v$ are enumerate between $e$ and $e'$, and for any pair of consecutive edges of $W'$ for all $W'\in{\cal W}$ incident to $v$, it does not hold that one is enumerate between $e$ and the two edges of $E(C)$ while the other is enumerated between $e'$ and the two edges of $E(C)$. Let $\widetilde{e}$ be the first among the two edges of $E(C)$ that is enumerated. Then, the face push operation replaces the occurrence of $v$ between $e$ and $e'$ in $W$ by the traversal of $C$ starting~at~$\widetilde{e}$.
\end{itemize}
\end{definition}

We verify that the application of a single operation results in a weak linkage.

\begin{observation}
Let $G$ be a triangulated plane graph with a weak linkage  ${\cal W}$, and a face $f$ that is not the outer face. Let $W\in {\cal W}$ with a discrete homotopy operation applicable to $(W,f)$. Then, the result of the application is another weak linkage aligned to $\cal W$.
\end{observation}


Then, discrete homotopy is defined as follows.

\begin{definition}[{\bf Discrete Homotopy}]\label{def:discreteHomotopy}
Let $G$ be a triangulated plane graph with weak linkages ${\cal W}$ and ${\cal W}'$. Then, $\cal W$ is {\em discretely homotopic} to ${\cal W}'$ if there exists a finite sequence of discrete homotopy operations such that when we start with $\cal W$ and apply the operations in the sequence one after another, every operation is applicable, and the final result is~${\cal W}'$.
\end{definition}

We verify that discrete homotopy gives rise to an equivalence relation.

\begin{lemma}
Let $G$ be a triangulated plane graph with weak linkages  ${\cal W},{\cal W}'$ and ${\cal W}''$. Then, {\em (i)} $\cal W$ is discretely homotopic to itself, {\em (ii)} if ${\cal W}$ is discretely homotopic to ${\cal W}'$, then ${\cal W}'$ is discretely homotopic to ${\cal W}$, and {\em (iii)} if ${\cal W}$ is discretely homotopic to ${\cal W}'$ and ${\cal W}'$ is discretely homotopic to ${\cal W}''$, then $\cal W$ is discretely homotopic to ${\cal W}''$.
\end{lemma}

\begin{proof}
    Statement $(i)$ is trivially true.
    The proof of statement $(ii)$ is immediate from the observation that each discrete homotopy operation has a distinct inverse. Indeed, every face move operation is invertible by a face move operation (applied to the same walk and cycle). Additionally, every face pull operation is invertible by a face push operation (applied to the same walk and cycle), and vice versa.
 Hence, given the sequence of operations to transform $\cal W$ to $\cal W'$, say $\phi$, the sequence of operations to transform $\cal W'$ to $\cal W$ is obtained by first writing the operations of $\phi$ in reverse order and then inverting each of them. Finally, Statement $(iii)$ follows by considering the sequence of discrete homotopy operations obtained by concatenating the sequence of operations to transform $\cal W$ to $\cal W'$ and the sequence of operations to transform $\cal W'$ to $\cal W''$.
\end{proof}

Towards the translation of discrete homotopy to homology, we need to associate a flow with every weak linkage and thereby extend Definition \ref{def:solToFlow}.

\begin{definition}
Let $(D,S,T,g,k)$ be an instance of \dpdp. Let ${\cal W}$ be a sensible weak linkage in $D$. The {\em flow $\phi: A(D)\rightarrow \RW(T)$ associated with $\cal W$} is defined as follows. For every arc $e\in A(D)$, define $\phi(e)=1$ if there is no walk in $\cal W$ that traverses $e$, and $\phi(e)=t$ otherwise where $t\in T$ is the end-vertex of the (unique) walk in $\cal W$ that traverses $e$. 
\end{definition}

Additionally, because homology concerns directed graphs, we need the following notation. Given a graph $G$, we let $\vec{G}$ denote the directed graph obtained by replacing every edge $e\in E(G)$ by two arcs of opposite orientations with the same endpoints as $e$. Notice that $G$ and the underlying graph of $\vec{G}$ are not equal (in particular, the latter graph contains twice as many edges as the first one). Given a weak linkage $\cal W$ in $G$, the weak linkage in $\vec{G}$ that corresponds to $\cal W$ is the weak linkage obtained by replacing each edge $e$ in each walk in $\cal W$, traversed from $u$ to $v$,  by the copy of $e$ in $\vec{G}$ oriented from $u$ to $v$.

Now, we are ready to translate discrete homotopy to homology.

\begin{lemma}\label{lem:discreteHomotopyToHomology}
Let $(G,S,T,g,k)$ be an instance of \pdp\ where $G$ is triangulated. Let ${\cal W}$ be a sensible weak linkage in $G$. Let ${\cal W}'$ be a weak linkage discretely homotopic to ${\cal W}'$. Let $\widehat{\cal W}$ and $\widehat{\cal W}'$ be the weak linkages corresponding to ${\cal W}$  and ${\cal W}'$ in $\vec{G}$, respectively. Then, the flow associated with $\widehat{\cal W}$ is homologous to the flow associated with $\widehat{\cal W}'$.
\end{lemma}

\begin{proof}
    Let $\phi$ and $\psi$ be the flows associated with $\wh{\cal W}$ and $\wh{\cal W}'$, respectively, in $\vec{G}$.
    Consider a sequence $O_1, O_2, \ldots O_\ell$ of discrete homotopy operations that, starting from ${\cal W}$, result in ${\cal W}'$. We prove the lemma by induction on $\ell$. Consider the case when $\ell = 1$. Then, the sequence contains only one discrete homotopy operation, which a face move, face pull or face push operation. Let this operation be applied to a face $f$ and a walk ${W} \in {\cal W}$, where the walk ${W}$ goes from $s \in S$ to $g(s) \in T$. Let $C$ be the boundary cycle of $f$ in $G$, and let $\vec{C}$ denote the collection of arcs in $G$ obtained from the edges of $C$. After this discrete homotopy operation, we obtain a walk ${W}' \in {\cal W}'$, which differs from ${W}$ only in the subset of edges $C$. All other walks are identical between ${\cal W}$ and ${\cal W}'$. Hence, $\wh{\cal W}$ and $\wh{\cal W}'$ differ in $\vec{G}$ only in a subset of $\vec{C}$. Then observe that the flows $\phi$ and $\psi$ are identical everywhere in $A(\vec{G})$ except for a subset of $\vec{C}$. 
    More precisely, Let $P = \vec{C}\cap \wh{W}$ and $P' = \vec{C} \cap \wh{W}'$. Then  $\phi(e) = g(s)$ if $e \in P$ and $\phi(e) = 1$ if $e \in \vec{C} - P$; a similar statement holds for $\psi$ and $P'$. Furthermore, it is clear from the description of each of the discrete homotopy operations that $P$ and $P'$ have no common edges and $P \cup P'$ is the (undirected)\footnote{That is, the underlying undirected graph of $\vec{C}$ is a cycle.} cycle $\vec{C}$ in $\vec{G}$.
    
    It only remains to describe the homology between the flows $\phi$ and $\psi$, which is exhibited by a function $h$ on the faces of $\vec{G}$. Then $h$ assigns $1$ to all faces of $\vec{G}$ that lie in the exterior of $\vec{C}$, and $g(s)$ to all the faces that lie in the interior of $\vec{C}$.   
    Note that $h$ assigns $1$ to the outer face of $\vec{G}$.
    It is easy to verify that $h$ is indeed a homology between $\psi$ and $\phi$, that is, for any edge $e \in A(\vec{G})$ it holds that $\psi(e) = h(f_1)^{-1} \cdot \phi(e) \cdot h(f_2)$, where $f_1$ and $f_2$ are the faces on the left and the right of $e$ with respect to its orientation. This proves the case where $\ell = 1$.
    
    Now for $\ell > 1$, consider the weak linkage $\wh{\cal W}^\star$ obtained from $\wh{\cal W}$ after applying the sequence $O_1, O_2, \ldots, O_{\ell - 1}$. Then by the induction hypothesis, we can assume that the flow associated with $\wh{\cal W}^\star$, say $\psi^\star$ is homologous to $\phi$. Further, applying $O_\ell$ to $\wh{\cal W}^\star$ gives us $\wh{\cal W}'$, and hence the flows $\psi^\star$ and $\phi$ are also homologous. Hence, by Observation~\ref{obs:homProp} the flows $\phi$ and  $\psi$ are~homologous.
\end{proof}

Having Corollary \ref{cor:homology} and Lemma \ref{lem:discreteHomotopyToHomology} at hand, we prove the following theorem.

\begin{lemma}\label{lemma:discreteHomotopy}
There exists a polynomial-time algorithm that, given an instance $(G,S,T,g,k)$ of \pdp\ where $G$ is triangulated, a sensible weak linkage $\cal W$ in $G$ and a subset $X\subseteq E(G)$, either finds a solution of $(G-X,S,T,g,k)$ or determines that no solution of $(G-X,S,T,g,k)$ is discretely homotopic to $\cal W$ in $G$.
\end{lemma}

\begin{proof}
    We first convert the given instance of \pdp\ into an instance of \dpdp\ as follows. We convert the graph $G$ into the digraph $\vec{G}$, as described earlier. Then we construct $\vec{X}$ from $X$ by picking the two arcs of opposite orientation for each edge in $X$. Then we convert the sensible weak linkage $\cal W$ into a weak linkage $\wh{\cal W}$ in $\vec{G}$. Finally, we obtain the flow $\phi$ in $\vec{G}$ associated with $\wh{\cal W}$. 
    Next, we apply Corollary~\ref{cor:homology} to the instance $(\vec{G},S,T,g,k)$, $\vec{X}$ and $\phi$. Then either it returns a solution $\wh{\Pp}$ that is disjoint from $\vec{X}$, or that there is no solution that is homologous to $\phi$ and disjoint from $\vec{X}$. In the first case, $\wh{\Pp}$ can be easily turned into a solution $\Pp$ for the undirected input instance that is disjoint from $X$. 
    In the second case, 
    we can conclude that the undirected input instance has no solution that is discretely homotopic to $\cal W$. Indeed, if this were not the case, then consider a solution $\Pp$ to $(G-X, S,T,g,k)$ that is discretely homotopic to $\cal W$. Then we have a solution $\wh{\Pp}$ to the directed instance that is disjoint to $\vec{X}$.
    Hence, by Lemma~\ref{lem:discreteHomotopyToHomology}, the flow associated with $\wh{\Pp}$ is homologous to $\phi$, the flow associated with $\wh{W}$. 
    Hence, $\wh{\Pp}$ is a solution to the instance $(\vec{G},S,T,g,k)$ that is disjoint from $\vec{X}$ and whose flow is homologous to $\phi$. But this is a contradiction to Corollary~\ref{cor:homology}.
%
\end{proof}

As a corollary to this lemma, we derive the following result.

\begin{corollary}\label{cor:discreteHomotopy}
There exists a polynomial-time algorithm that, given an instance $(G,S,T,g,k)$ of \pdp\ and a sensible weak linkage $\cal W$ in $H_G$, either finds a solution of $(G,S,T,g,$ $k)$ or decides that no solution of $(G,S,T,g,k)$ is discretely homotopic to $\cal W$ in $H_G$.
\end{corollary}

\begin{proof}
    Consider the instance $(H_G,S,T,g,k)$ along with the set $X=E(H_G) \setminus E(G)$ of forbidden edges. We then apply Lemma~\ref{lemma:discreteHomotopy} to $(H_G,S,T,g,k)$, $X$ and $\cal W$ (note that $H_G$ is triangulated). If we obtain a solution to this instance, then it is also a solution in $G$ since it traverses only edges in $E(H_G)\setminus X = E(G)$.
    Else, we correctly conclude that there is no solution of $(H_G-X,S,T,g,k)$ (and hence also of $(G,S,T,g,k)$) that is discretely homotopic to $\cal W$~in~$H_G$.
\end{proof}


\section{Construction of the Backbone Steiner Tree}\label{sec:steiner}

In this section, we construct a tree that we call a backbone Steiner tree $R=R^3$  in $H_G$. Recall that $H_G$ is the radial completion of $G$ enriched with $4|V(G)|+1$ parallel copies for each edge. These parallel copies will not be required during the construction of $R$, and therefore we will treat $H_G$ as having just one copy of each edge. Hence, we can assume that $H_G$ is a simple planar graph, and then $E(H_G) = O(n)$ where $n$ is the number of vertices in $G$. We denote $H=H_G$ when $G$ is clear from context. The tree $R$ will be proven to admit the following property: if the input instance is a \yes-instance, then it admits a solution ${\cal P}=(P_1,\ldots,P_k)$ that is discretely homotopic to a weak linkage ${\cal W}=(W_1,\ldots,W_k)$ in $H$ aligned with ${\cal P}$ that uses at most $2^{\OO(k)}$ edges parallel to those in $R$, and none of the edges not parallel to those in $R$. We use the term Steiner tree to refer to any subtree of $H$ whose set of leaves is precisely $S\cup T$. To construct the backbone Steiner tree $R=R^3$, we start with an arbitrary Steiner tree $R^1$ in $H$. Then over several steps, we modify the tree to satisfy several useful properties.

\subsection{Step I: Initialization} We initialize $R^1$ to be an arbitrarily chosen Steiner tree. Thus, $R^1$ is a subtree of $H$ such that $V_{=1}(R^1)=S\cup T$. The following observation is immediate from the definition of a Steiner tree. 

\begin{observation}\label{obs:leaIntSteiner}
Let $(G,S,T,g,k)$ be an instance of \pdp. Let $R'$ be a Steiner tree. Then, $|V_{=1}(R')|=2k$ and $|V_{\geq 3}(R')|\leq 2k-1$.
\end{observation}

Before we proceed to the next step, we claim that every vertex of $H$ is, in fact, ``close'' to the vertex set of $R_1$. For this purpose, we need the following proposition by Jansen et al.~\cite{DBLP:conf/soda/JansenLS14}.

\begin{proposition}[Proposition 2.1 in \cite{DBLP:conf/soda/JansenLS14}]\label{prop:concentricAtAllDists}
Let $G$ be a plane graph and with disjoint subsets $X,Y\subseteq V(G)$ such that $G[X]$ and $G[Y]$ are connected graphs and ${\sf rdist}_G(X,Y) = d \geq 2$. For any $r\in\{0,1,\ldots,d-1\}$, there is a cycle $C$ in $G$ such that all vertices $u\in V(C)$ satisfy ${\sf rdist}_G(X,\{u\}) = r$, and such that $V(C)$ separates $X$ and $Y$ in $G$. 
\end{proposition}

Additionally, we need the following simple observation. 

\begin{observation}\label{obs:dist-eq-rdist-in-triangle-graph}
Let $G$ be a triangulated plane graph. Then, for any pair of vertices $u,v \in V(G)$, $\dist_G(u,v) = {\sf rdist}_G(u,v)$.
\end{observation}

\begin{proof}
Let ${\sf rdist}_G(u,v) = t$, and consider a sequence of vertices $u=x_1, x_2, \ldots, x_{t+1}=v$ that witnesses this fact---then, every two consecutive vertices in this sequence have a common face. Since $G$ is triangulated, we have that $\{x_i,x_{i+1}\} \in E(G)$ for every two consecutive vertices $x_i,x_{i+1}$, $1 \leq i \leq t$. Hence, $x_1, x_2, \ldots, x_{t+1}$ is a walk from $u$ to $v$ in $G$ with $t$ edges, and therefore $\dist_G(u,v) \leq {\sf rdist}_G(u,v)$. Conversely, let $\dist_G(u,v) = \ell$; then, there is a path with $\ell$ edges from $u$ to $v$ in $G$, which gives us a sequence of vertices $u=y_1, y_2, \ldots, y_{\ell+1} = v$ where each pair of consecutive vertices forms an edge in $G$. Since $G$ is planar, each such pair of consecutive vertices $y_i,y_{i+1}$,$1 \leq i \leq \ell$, must have a common face. Therefore, ${\sf rdist}_G(u,v) \leq \dist_G(u,v)$. 
\end{proof}

It is easy to see that Observation \ref{obs:dist-eq-rdist-in-triangle-graph} is not true for general plane graph. However, this observation will be useful for us because the graph $H$, where we construct the backbone Steiner tree, is triangulated.
We now present the promised claim, whose proof is based on Proposition \ref{prop:concentricAtAllDists}, Observation \ref{obs:dist-eq-rdist-in-triangle-graph} and the absence of long sequences of $S\cup T$-free concentric cycles in good instances. Here, recall that $c$ is the fixed constant in Corollary~\ref{cor:twReduction}. We remark that, for the sake of clarity, throughout the paper we denote some natural numbers whose value depends on $k$ by notations of the form $\alpha_{\mathrm{subscript}}(k)$ where the subscript of $\alpha$ hints at the use of the value.

\begin{lemma}\label{lem:closeToR}
Let $(G,S,T,g,k)$ be a good instance of \pdp. Let $R'$ be a Steiner tree. For every vertex $v\in V(H)$, it holds that $\dist_H(v,V(R'))\leq \alpha_{\mathrm{dist}}(k):=4\cdot 2^{ck}$.
\end{lemma}

\begin{proof}
Suppose, by way of contradiction, that $\dist_H(v^\star,V(R'))>\alpha_{\mathrm{dist}}(k)$ for some vertex $v^\star\in V(H)$. Since $H$ is the (enriched) radial completion of $G$, it is triangulated. By Observation~\ref{obs:dist-eq-rdist-in-triangle-graph}, 
${\sf rdist}_H(u,v) = \dist_H(u,v)$ for any pair of vertices $u,v \in V(H)$. Thus, ${\sf rdist}_H(v^\star,V(R'))>\alpha_{\mathrm{dist}}(k)$.
 By Proposition \ref{prop:concentricAtAllDists}, for any $r\in\{0,1,\ldots,\alpha_{\mathrm{dist}}(k)\}$, there is a cycle $C_r$ in $H$ such that all vertices $u\in V(C_r)$ satisfy ${\sf rdist}_H(v^\star,u)={\sf dist}_H(v^\star,u) = r$, and such that $V(C_r)$ separates $\{v^\star\}$ and $V(R')$ in $H$. 
In particular, these cycles must be pairwise vertex-disjoint, and each one of them contains either {\em (i)} $v^\star$ in its interior (including the boundary) and $V(R')$ in its exterior (including the boundary), or {\em (ii)} $v^\star$ in its exterior (including the boundary) and $V(R')$ in its interior (including the boundary). We claim that only case (i) is possible. Indeed, suppose by way of contradiction that $C_i$, for some $r\in\{0,1,\ldots,\alpha_{\mathrm{dist}}(k)\}$, contains $v^\star$ in its exterior and $V(R')$ in its interior. Because the outer face of $H$ contains a terminal $t^\star\in T$ and $t^\star\in V(R')$, we derive that $t^\star\in V(C_i)$.  Thus, ${\sf rdist}_H(v^\star,t^\star)= i\leq \alpha_{\mathrm{dist}}(k)$. However, because $t^\star\in V(R')$, this is a contradiction to the supposition that $\dist_H(v^\star,V(R'))>\alpha_{\mathrm{dist}}(k)$. Thus, our claim holds true. From this, we conclude that ${\cal C}=(C_0,C_1,\ldots,C_{\alpha_{\mathrm{dist}}(k)})$ is a $V(R')$-free sequence of concentric cycles in $H$. Since $S\cup T\subseteq V(R')$, it is also $S\cup T$-free. 

Consider some odd integer $r\in\{1,2,\ldots,\alpha_{\mathrm{dist}(k)}\}$. Note that every vertex $u\in V(C_r)$ that does not belong to $V(G)$ lies in some face $f$ of $G$, and that the two neighbors of $u$ in $C_r$ must belong to the boundary of $f$ (by the definition of radial completion). Moreover, each of the vertices on the boundary of $f$ is at distance (in $H$) from $u$ that is the same, larger by one or smaller by one, than the distance of $f$ from $u$, and hence none of these vertices can belong to any $C_i$ for $i>r+1$ as well as $i<r+1$. For every  $r\in\{1,2,\ldots,\alpha_{\mathrm{dist}}(k)-1\}$ such that $r\mod 3=1$, define $C'_r$ as some cycle contained in the closed walk obtained from $C_r$ by replacing every vertex $u\in V(C_r)\setminus V(G)$, with neighbors $x,y$ on $C_r$, by a path from $x$ to $y$ on the boundary of the face of $G$ that corresponds to $u$. In this manner, we obtain an $S\cup T$-free sequence of concentric cycles in $G$ whose length is at least $2^{ck}$. However, this contradicts the supposition that $(G,S,T,g,k)$ is good.
\end{proof}

\subsection{Step II: Removing Detours}
In this step, we modify the Steiner tree to ensure that there exist no ``shortcuts'' via vertices outside the Steiner tree. This property will be required in subsequent steps to derive additional properties of the Steiner tree.
To formulate this, we need the following definition (see Fig.~\ref{fig:undetour}).


\begin{definition}[{\bf Detours in Trees}]\label{def:detour}
A subtree $T$ of a graph $G$  {\em has a detour} if there exist two vertices $u,v\in V_{\geq 3}(T)\cup V_{=1}(T)$ that are near each other, and a path $P$ in $G$, such that 
\begin{enumerate}
\item $P$ is shorter than $\pathT_T(u,v)$, and
\item one endpoint of $P$ belongs to the connected component of $T-V(\pathT_T(u,v))\setminus\{u,v\}$ that contains $u$, and the other endpoint of $P$ belongs to the connected component of $T-V(\pathT_T(u,v))\setminus\{u,v\}$ that contains $v$.
\end{enumerate}
Such vertices $u,v$ and path $P$ are said to {\em witness the detour}. Moreover, if $P$ has no internal vertex from $(V(T)\setminus V(\pathT_T(u,v)))\cup\{u,v\}$ and its endpoints do not belong to $V_{=1}(T)\setminus\{u,v\}$, then $u,v$ and $P$ are said to {\em witness the detour compactly}.
\end{definition}
 
We compute a witness for a detour as follows. Note that this lemma also implies that, if there exists a detour, then there exists a compact witness rather than an arbitrary one.

\begin{lemma}\label{lem:computeDetour}
There exists an algorithm that, given a good instance $(G,S,T,g,k)$ of \pdp\ and a Steiner tree $R'$, determines in time $\OO(k^2 \cdot n)$ whether $R'$ has a detour. In case the answer is positive, it returns $u,v$ and $P$ that witness the detour compactly.
\end{lemma}

\begin{proof}
 Let $Q = \pathT_{R'}(u,v) - \{u,v\}$ for some two vertices $u,v\in V_{\geq 3}(T)\cup V_{=1}(T)$ that are near each other. Then, $R' - V(Q)$ contains precisely two connected components: $R'_u$ and $R'_v$ that contain $u$ and $v$, respectively. Consider a path $P$ of minimum length between vertices $x \in V(R'_u)$ and $y \in V(R'_v)$ in $H$, over all choices of $x$ and $y$. Further, we choose $P$ so that contains as few vertices of $(S \cup T) \setminus \{u,v\}$ as possible. 
Suppose that $|E(P)| \leq |E(\pathT_{R'}(u,v))| - 1$. Then, we claim that $P$ is a compact detour witness. To prove this claim, we must show that $(i)$~no internal vertex of $P$ lies in $(V(R')\setminus V(\pathT_{R'}(u,v)))\cup\{u,v\} = V(R'_u) \cup V(R'_v)$, and $(ii)$~the endpoints of $P$ do not lie in $V_{=1}(R') \setminus \{u,v\} = (S \cup T) \setminus \{u,v\}$. The first property follows directly from the choice of $P$. Indeed, if $P$ were a path from $x \in V(R'_u)$ to $y \in V(R'_v)$, which contained an internal vertex $z \in V(R'_u)$, then the subpath $P'$ of $P$ with endpoints $z$ and $y$ is a strictly shorter path from $V(R'_u)$ to $V(R'_v)$ (the symmetric argument holds when $z \in V(R'_v)$).
    
For the second property, we give a proof by contradiction. To this end, suppose that some terminal $w \in (S \cup T) \setminus \{u,v\}$ belongs to $P$. Necessarily, $w \in V(R'_u)\cup V(R'_v)$ (by the definition of a Steiner tree). Without loss of generality, suppose that $w \in V(R'_u)$. By the first property, $w$ must be an endpoint of $P$. Let $z \in V(R'_v)$ be the other endpoint of $P$. Because the given instance is good, $w$ has degree $1$ in $G$, thus we can let $n(w)$ denote its unique neighbor in $G$. Observe that  $w$  lies on only one face of $G$, which contains both $w$ and $n(w)$. Hence, $w$ is adjacent to exactly two vertices in $H$: $n(w)$ and a vertex $f(w)\in V(H)\setminus V(G)$. Furthermore, $\{n(w),f(w)\} \in E(H)$, i.e.~$w,n(w)$ and $f(w)$ form a triangle in $H$. Thus, $P$ contains exactly one of $n(w)$ or $f(w)$ (otherwise, we can obtain a strictly shorter path connecting $w$ and the other endpoint of $P$ that contradicts the choice of $P$).  Let $a(w)\in\{n(w),f(w)\}$ denote the neighbor of $w$ in $P$, and note that, by the first property, $a(w)\not\in V(R'_u)$. Note that it may be the case that $a(w) = z$. 
Since $w$ is a leaf of $R'$, exactly one of $n(w)$ and $f(w)$ is adjacent to $w$ in $R'$, and we let $b(w)$ denote this vertex. Because $w\neq u$, we have that $V(R'_u)$ contains but is not equal to $\{w\}$, and therefore  $b(w) \in V(R'_u)$. 
 In turn, by the first property, this means that $a(w) \neq b(w)$ (because otherwise $a(w) \neq z$ and hence it is an internal vertex of $P$, which cannot belong to $V(R'_u)$). Because $w, a(w)$ and $b(w)$ form a triangle in $H$, we obtain a path $P'\neq P$ in $H$ by replacing $w$ with $b(w)$ in $P$. Observe that $P'$ connects the vertex $b(w) \in V(R'_u)$ to the vertex $z \in V(R'_v)$. Furthermore,  because $|E(P')| = |E(P)|$, and $P'$ contains strictly fewer vertices of $(S \cup T) \setminus \{u,v\}$ compared to $P$, we contradict the choice of $P$. Therefore, $P$ also satisfies the second property, and we conclude that $u,v,P$ compactly witness a detour in~$R'$.
    
We now show that a compact detour in $R'$ can be computed in $\OO(k^2 \cdot n)$ time. First, observe that if there is a detour witnessed by some $u,v$ and $P$, then $u,v \in V_{\geq 3}(R') \cup V_{=1}(R')$. By Observation \ref{obs:leaIntSteiner}, $|V_{\geq 3}(R') \cup V_{=1}(R')| \leq 4k$. Therefore, there are at most $16k^2$ choices for the vertices $u$ and $v$. We consider each choice, and test if there is detour for it in linear time as follows. Fix a choice of distinct vertices $u,v \in V_{\geq 3}(R') \cup V_{=1}(R')$, and check if they are near each other in $R'$ in $O(|V(R)|)$ time by validating that each internal vertex of $\pathT_{R'}(u,v)$ has degree $2$. If they are not near each other, move on to the next choice. Otherwise, consider the path $Q = \pathT_{R'}(u,v) - \{u,v\}$, and the trees $R'_u$ and $R'_v$ of $R'-V(Q)$ that contain $u$ and $v$, respectively. Now, consider the graph $\widetilde{H}$ derived from $H$ by first deleting $(V(Q) \cup S \cup T) \setminus \{u,v\}$ and then introducing a new vertex $r$ adjacent to all vertices in $V(R'_u)$. We now run a breadth first search (BFS) from $r$ in $\widetilde{H}$. This step takes $\OO(n)$ time since $|E(\widetilde{H})| = \OO(n)$ (because $H$ is planar). From the BFS-tree, we can easily compute a shortest path $P$ between a vertex $x \in V(R'_u)$ and a vertex $y \in V(R'_v)$. Observe that $V(P) \cap (S \cup T) \subseteq \{u,v\}$ by the construction of $\widetilde{H}$. If $|E(P)| < |E(\pathT_{R'}(u,v))|$, then we output $u,v,P$ as a compact witness of a detour in $R'$. Else, we move on to the next choice of $u$ and $v$. If we fail to find a witness for all choices of $u$ and $v$, then we output that $R'$ has no detour. Observe that the total running time of this process is bounded by $\OO(k^2 \cdot n)$. This concludes the proof.
\end{proof}

Accordingly, as long as $R^1$ has a detour, compactly witnessed by some vertices $u,v$ and a path $P$, we modify it as follows: we remove the edges and the internal vertices of $\pathT_{R^1}(u,v)$, and add the edges and the internal vertices of $P$. We refer to a single application of this operation as {\em undetouring} $R^1$. For a single application, because we consider compact witnesses rather than arbitrary ones, we have the following observation.

\begin{observation}\label{obs:undetour}
Let $(G,S,T,g,k)$ be a good instance of \pdp\ with a Steiner tree $R'$. The result of undetouring $R'$ is another Steiner tree with fewer edges than $R'$.
\end{observation}

\begin{proof}
Consider a compact detour witness $u,v,P$ of $R'$. Then, $|E(P)| < |E(\pathT_{R'}(u,v))|$. Let $Q = \pathT_{R'}(u,v) - \{u,v\}$, and let $R'_u$ and $R'_v$ be the two trees of $R' - V(Q)$ that contain $u$ and $v$, respectively.
Consider the graph $\tilde{R}$ obtained from $(R' - V(Q)) \cup P$ by iteratively removing any leaf vertex that does not lie in $S \cup T$.
We claim that the graph $\widetilde{R}$ that result from undetouring $R'$ (with respect to $u,v,P$) is a Steiner tree with strictly fewer edges than $R'$.
Clearly,  $\widetilde{R}$ is connected because $P$ reconnects the two trees $R'_u$ and $R'_v$ of $R' - V(Q)$. Further, as $P$ contains no internal vertex from $(V(R')\setminus V(\pathT_{R'}(u,v)))\cup\{u,v\} = V(R'_u) \cup V(R'_v)$, and $R'_u$ and $R'_v$ are trees, $\widetilde{R}$ is cycle-free. 
Additionally, all the vertices in $S \cup T$ are present in $\widetilde{R}$ by construction and they remain leaves due to the compactness of the witness. Hence, $\widetilde{R}$ is a Steiner tree in $G$. 
Because $|E(P)| < |E(\pathT_{R'}(u,v))|$, it follows that $\widetilde{R}$ contains fewer edges than $R'$.
\end{proof}

Initially, $R^1$ has at most $n-1$ edges. Since every iteration decreases the number of edges (by Observation \ref{obs:undetour}) and can be performed in time $\OO(k^2 \cdot n)$ (by Lemma \ref{lem:computeDetour}), we obtain the following result.

\begin{lemma}\label{obs:undetourExhaustive}
Let $(G,S,T,g,k)$ be a good instance of \pdp\ with a Steiner tree $R'$. An exhaustive application of the operation undetouring $R'$ can be performed in time $\OO(k^2 \cdot n^2)$,  and results in a Steiner tree that has no detour. 
\end{lemma}

We denote the Steiner tree obtained at the end of Step II by $R^2$.

\subsection{Step III: Small Separators for Long Paths} 
We now show that any two parts of $R^2$ that are ``far'' from each other can be separated by small separators in $H$. This is an important property used in the following sections to show the existence of a ``nice'' solution for the input instance.
Specifically, we consider a ``long'' maximal degree-2 path in $R^2$ (which has no short detours in $H$), and show that there are two separators of small cardinality, each ``close'' to one end-point of the path. The main idea behind the proof of this result is that, if it were false, then the graph $H$ would have had large treewidth (see Proposition~\ref{prop:radialDisTw}), which contradicts that $H$ has bounded treewidth (by Corollary~\ref{cor:twReduction}).
We first define the threshold that determines whether a path is long or short.

\begin{definition}[{\bf Long Paths in Trees}]\label{def:longPath}
Let $G$ be a graph with a subtree $T$. A subpath of $T$ is {\em $k$-long} if its length is at least $\alpha_{\mathrm{long}}(k):= 10^4\cdot 2^{ck}$, and {\em $k$-short} otherwise.
\end{definition}

As $k$ will be clear from context, we simply use the terms long and short. Towards the computation of two separators for each long path, we also need to define which subsets of $V(R^2)$ we would like to separate.

\begin{definition}[{\bf $P'_u,P''_u,A_{R^2,P,u}$ and $B_{R^2,P,u}$}]
Let $(G,S,T,g,k)$ be a good instance of \pdp. Let $R^2$ be a Steiner tree that has no detour. For any long maximal degree-2 path $P$ of $R^2$ and for each endpoint $u$ of $P$, define $P'_u$, $P''_u$ and $A_{R^2,P,u},B_{R^2,P,u}\subseteq V(R^2)$ as follows.
\begin{itemize}
\item $P'_u$ (resp.~$P''_u$) is the subpath of $P$ consisting of the $\alpha_{\mathrm{pat}}(k):=100\cdot 2^{ck}$ (resp.~$\alpha_{\mathrm{pat}}(k)/2 = 50\cdot 2^{ck}$) vertices of $P$ closest to $u$.
\item $A_{R^2,P,u}$ is the union of $V(P''_u)$ and the vertex set of the connected component  of $R^2-(V(P'_u)\setminus \{u\})$ containing $u$.
\item $B_{R^2,P,u}=V(R^2)\setminus(A_{R^2,P,u}\cup V(P'_u))$.
\end{itemize}
\end{definition}

\begin{figure}
    \begin{center}
        \includegraphics[width=\textwidth]{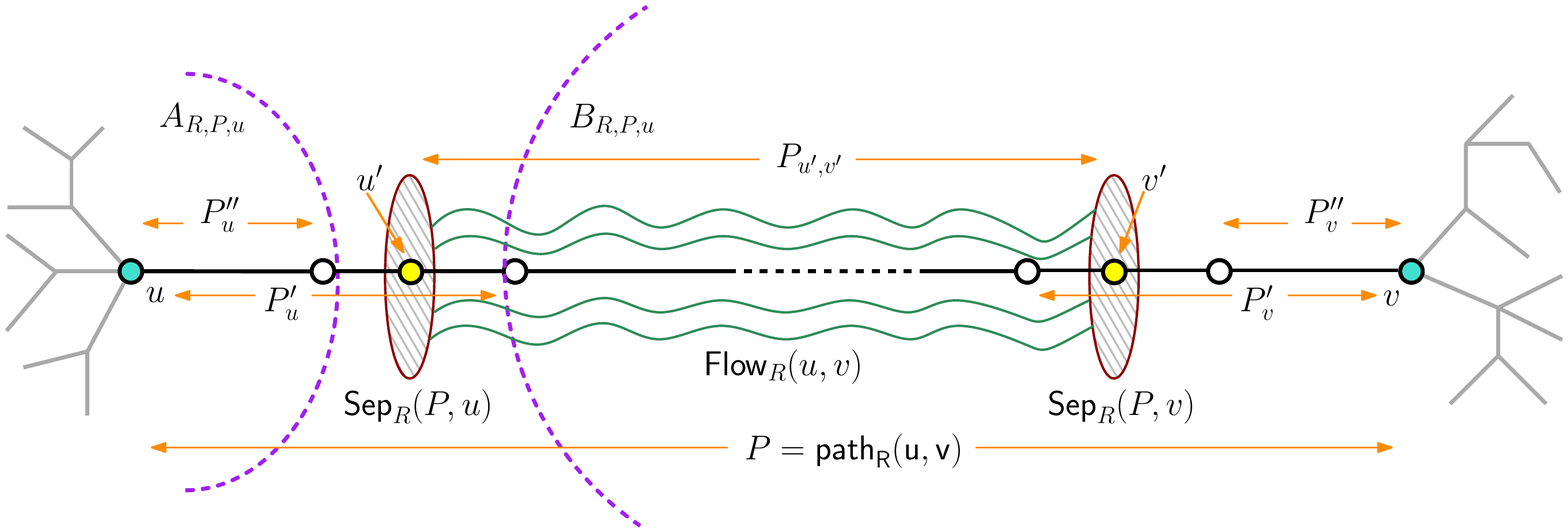}
        \caption{Separators and flows for long degree-2 paths.}
        \label{fig:treeflow1}
    \end{center}
\end{figure}

For each long maximal degree-2 path $P$ of $R^2$ and for each endpoint $u$ of $P$, we compute a ``small'' separator $\Sep_{R^2}(P,u)$ as follows. Let $A=A_{R^2,P,u}$ and $B=B_{R^2,P,u}$. Then, compute a subset of $V(H)\setminus (A\cup B)$ of minimum size that separates $A$ and $B$ in $H$, and denote it by $\Sep_{R^2}(P,u)$ (see Fig.~\ref{fig:treeflow1}). Since $A\cap B=\emptyset$ and there is no edge between a vertex in $A$ and a vertex in $B$ (because $R^2$ has no detours), such a separator exists. Moreover, it can be computed in time $\OO(n|\Sep_{R^2}(P,u)|)$: contract each set among $A$ and $B$ into a single vertex and then obtain a minimum vertex $s-t$ cut by using Ford-Fulkerson algorithm.

To argue that the size of $\Sep_{R^2}(P,u)$ is upper bounded by $2^{\OO(k)}$, we make use of the following proposition due to Bodlaender et al.~\cite{DBLP:journals/jacm/BodlaenderFLPST16}.

\begin{proposition}[Lemma 6.11 in \cite{DBLP:journals/jacm/BodlaenderFLPST16}]\label{prop:radialDisTw}
Let $G$ be a plane graph, and let $H$ be its radial completion. Let $t\in\mathbb{N}$. Let $C,Z,C_1,Z_1$ be disjoint subsets of $V(H)$ such that
\begin{enumerate}\setlength\itemsep{0em}
\item $H[C]$ and $H[C_1]$ are connected graphs,
\item $Z$ separates $C$ from $Z_1\cup C_1$ and $Z_1$ separates $C\cup Z$ from $C_1$ in $H$,
\item $\dist_{H}(Z,Z_1)\geq 3t+4$, and 
\item $G$ contains $t+2$ pairwise internally vertex-disjoint paths with one endpoint in $C\cap V(G)$ and the other endpoint in $C_1\cap V(G)$.
\end{enumerate}
Then, the treewidth of $G[V(M)\cap V(G)]$ is larger than $t$ where $M$ is the union of all connected components of $H\setminus (Z\cup Z_1)$ having at least one neighbor in $Z$ and at least one neighbor in~$Z_1$. 
\end{proposition}

Additionally, the following immediate observation will come in handy.

\begin{observation}\label{obs:radialRadial}
Let $G$ be a plane graph. Let $H$ be the radial completion of $G$, and let $H'$ be the radial completion of $H$. Then, for all $u,v\in V(H)$, $\dist_H(u,v)\leq \dist_{H'}(u,v)$.
\end{observation}

\begin{proof}
Note that $H$ is triangulated. Thus, for all $u,v\in V(H)$ and path $P$ in $H'$ between $u$ and $v$, we can obtain a path between $u$ and $v$ whose length is not longer than the length of $P$ by replacing each vertex $w\in V(H')\setminus V(H)$ by at most one vertex of the boundary of the face in $H$ that $w$ represents.  
\end{proof}

We now argue that $\Sep_{R^2}(P,u)$ is small.

\begin{lemma}\label{lem:sepSmall}
Let $(G,S,T,g,k)$ be a good instance of \pdp. Let $R^2$ be a Steiner tree that has no detour, $P$ be a long maximal degree-2 path of $R^2$, and $u$ be an endpoint of $P$. Then, $|\Sep_{R^2}(P,u)|\leq \alpha_{\mathrm{sep}}(k):=\frac{7}{2}\cdot 2^{ck}+2$.
\end{lemma}

\begin{proof}
Denote $P'=P'_u,P''=P''_u,A=A_{R^2,P,u}$ and $B=B_{R^2,P,u}$. Recall that $H$ is the  radial completion of $G$ (enriched with parallel edges), and let $H'$ denote the radial completion of $H$. Towards an application of Proposition \ref{prop:radialDisTw}, define $C=A$, $C_1=B$, $Z=N_{H'}(C)$, $Z_1=N_{H'}(C_1)$ and $t=\frac{7}{2}\cdot 2^{ck}$. Since $R^2$ is a subtree of $H$, it holds that $H[C]$ and $H[C_1]$ are connected, and therefore $H'[C]$ and $H'[C_1]$ are connected as well. From the definition of $Z$ and $Z_1$, it is immediate that $Z$ separates $C$ from $Z_1\cup C_1$ and $Z_1$ separates $C\cup Z$ from $C_1$ in $H$. 
Clearly, $C\cap C_1=\emptyset$, $C\cap Z=\emptyset$, and $C_1\cap Z_1=\emptyset$. We claim that, in addition, $Z\cap C_1=\emptyset$, $Z_1\cap C=\emptyset$ and $Z\cap Z_1=\emptyset$. To this end, it suffices to show that $\dist_{H'}(Z,Z_1)\geq 3t+4$. Indeed, because $Z=N_{H'}(C)$ and $Z_1=N_{H'}(C_1)$, we have that each inequality among $Z\cap C_1\neq\emptyset$, $Z_1\cap C=\emptyset$ and $Z\cap Z_1=\emptyset$, implies that $\dist_{H'}(Z,Z_1)\leq 2$.

Lastly, we show that $\dist_{H'}(Z,Z_1)\geq 3t+4$. As $\dist_{H'}(C,C_1)\leq \dist_{H}(Z,Z_1)+2$, it suffices to show that $\dist_{H'}(C,C_1)\geq 3t+6$. Because $C\cup C_1\subseteq V(H)$, Observation \ref{obs:radialRadial} implies that $\dist_{H'}(C,C_1)\geq\dist_{H}(C,C_1)$. Hence, it suffices to show that $\dist_{H}(C,C_1)\geq 3t+6$. However, $\dist_{H}(C,C_1)\geq |E(P')|-|E(P'')|$ since otherwise we obtain a contradiction to the supposition that $R^2$ has no detour. This means that $\dist_{H}(C,C_1)\geq \alpha_{\mathrm{pat}}(k)/2-1\geq 3t+6$ as required.

Recall that $\Sep_{R^2}(P,u)$ is a subset of $V(H)\setminus (C\cup C_1')$ of minimum size that separates $C$ and $C_1$ in $H$. We claim that $|\Sep_{R^2}(P,u)|\leq \alpha_{\mathrm{sep}}(k)$. Suppose, by way of contradiction, that $|\Sep_{R^2}(P,u)|>\alpha_{\mathrm{sep}}(k)=t+2$. By Menger's theorem, the inequality $|\Sep_{R^2}(P,u)|>\alpha_{\mathrm{sep}}(k)$ implies that $H$ contains $t+2$ pairwise internally vertex-disjoint paths with one endpoint in $C\subseteq V(H)$ and the other endpoint in $C_1\subseteq V(H)$. From this, we conclude that all of the conditions in the premise of Proposition \ref{prop:radialDisTw} are satisfied. Thus,  the treewidth of $H[V(M)\cap V(H)]$ is larger than $t$ where $M$ is the union of all connected components of $H'\setminus (Z\cup Z_1)$ having at least one neighbor in $Z$ and at least one neighbor in~$Z_1$. However, $H[V(M)\cap V(H)]$ is a subgraph of $H$, which means that the treewidth of $H$ is also larger than $t$. By Proposition \ref{prop:twRadial}, this implies that the treewidth of $G$ is larger than $2^{ck}$. This contradicts the supposition that $(G,S,T,g,k)$ is good. From this, we conclude that $|\Sep_{R^2}(P,u)|\leq \alpha_{\mathrm{sep}}(k)$.
\end{proof}

Recall that $\Sep_{R^2}(P,u)$ is computable in time $\OO(n|\Sep_{R^2}(P,u)|)$. Thus, by Lemma \ref{lem:sepSmall}, we obtain the observation below. We remark that the reason we had to argue that the separator is small is not due to this observation, but because the size bound will be crucial in later sections.

\begin{observation}\label{obs:sepTime}
Let $(G,S,T,g,k)$ be a good instance of \pdp. Let $R^2$ be a Steiner tree that has no detour, $P$ be a long maximal degree-2 path of $R^2$, and $u$ be an endpoint of $P$. Then, $\Sep_{R^2}(P,u)$ can be computed in time $2^{\OO(k)}n$.
\end{observation}

Moreover, we have the following immediate consequence of Proposition \ref{prop:sepCycle}.

\begin{observation}\label{obs:sepIsCycle}
Let $(G,S,T,g,k)$ be a good instance of \pdp. Let $R^2$ be a Steiner tree that has no detour, $P$ be a long maximal degree-2 path of $R^2$, and $u$ be an endpoint of $P$. Then, $H[\Sep_{R^2}(P,u)]$ is a cycle.
\end{observation}

\subsection{Step IV: Internal Modification of Long Paths} In this step, we replace the ``middle'' of each long maximal degree-2 path $P = \pathT_{R^2}(u,v)$ of $R^2$ by a different path $P^\star$. This ``middle'' is defined by the two separators obtained in the previous step. Let us informally explain the reason behind this modification. In Section \ref{sec:winding} we will show that, if the given instance $(G,S,T,g,k)$ admits a solution (which is a collection of disjoint paths connecting $S$ and $T$), then it also admits a ``nice'' solution that ``spirals'' only a few times around parts of the constructed Steiner tree. This requirement is crucial, since it is only such solutions $\cal P$ that are discretely homotopic to weak linkages ${\cal W}$ in $H$ aligned with ${\cal P}$ that use at most $2^{\OO(k)}$ edges parallel to those in $R$, and none of the edges not parallel to those in $R$. 
To ensure the existence of nice solutions, we show how an arbitrary solution can be rerouted to avoid too many spirals. This rerouting requires a collection of vertex-disjoint paths between $\Sep_{R^2}(P,u)$ and $\Sep_{R^2}(P,v)$ which itself does not spiral around the Steiner tree. The replacement of $P$ by $P^\star$ in the Steiner tree, described below, will ensure this property.

To describe this modification, we first need to assert the statement in the following simple lemma, which partitions every long maximal degree-2 path $P$ of $R^2$ into three parts (see Fig.~\ref{fig:treeflow1}).

\begin{lemma}\label{lem:threeParts}
Let $(G,S,T,g,k)$ be a good instance of \pdp. Let $R^2$ be a Steiner tree with no detour, and $P$ be a long maximal degree-2 path of $R^2$ with endpoints $u$ and $v$. Then, there exist vertices $u'=u'_P\in\Sep_{R^2}(P,u)\cap V(P)$ and $v'=v'_P\in\Sep_{R^2}(P,v)\cap V(P)$ such that:
\begin{enumerate}
\item The subpath $P_{u,u'}$ of $P$ with endpoints $u$ and $u'$ has no internal vertex from $\Sep_{R^2}(P,u)\cup\Sep_{R^2}(P,v)$, and $\alpha_{\mathrm{pat}}(k)/2\leq |V(P_{u,u'})|\leq \alpha_{\mathrm{pat}}(k)$.
Additionally, the subpath $P_{v,v'}$ of $P$ with endpoints $v$ and $v'$ has no internal vertex from $\Sep_{R^2}(P,u)\cup\Sep_{R^2}(P,v)$, and $\alpha_{\mathrm{pat}}(k)/2\leq |V(P_{v,v'})|\leq \alpha_{\mathrm{pat}}(k)$.
\item Let $P_{u',v'}$ be the subpath of $P$ with endpoints $u'$ and $v'$. Then, $P=P_{u,u'}-P_{u',v'}-P_{v',v}$.
\end{enumerate}
\end{lemma}

\begin{proof}
We first prove that there exists a vertex $u'\in\Sep_{R^2}(P,u)\cap V(P)$ such that the subpath $P_{u,u'}$ of $P$ between $u$ and $u'$ has no internal vertex from $\Sep_{R^2}(P,u)\cup\Sep_{R^2}(P,v)$. To this end, let $P'=P'_u$, and let $\widetilde{P}$ denote the subpath of $P$ that consists of the $\alpha_{\mathrm{pat}}(k)+1$ vertices of $P$ that are closest to $u$. Let $A=A_{R^2,P,u}$ and $B=B_{R^2,P,u}$. Recall that $\Sep_{R^2}(P,u)\subseteq V(H)\setminus (A\cup B)$ separates $A$ and $B$ in $H$. Since $\widetilde{P}$ is a path with the endpoint $u$ in $A$ and the other endpoint in $B$, it follows that $\Sep_{R^2}(P,u)\cap V(\widetilde{P})\neq\emptyset$. 
 Accordingly, let $u'$ denote the vertex of $P'$ closest to $u$ that belongs to $\Sep_{R^2}(P,u)$. Then, $u'\in\Sep_{R^2}(P,u)\cap V(P)$ and the subpath $P_{u,u'}$ of $P$ between $u$ and $u'$ has no internal vertex from $\Sep_{R^2}(P,u)$. As the number of vertices of $P_{u,u'}$ is between those of $P''_u$ and $P'$, the inequalities $\alpha_{\mathrm{pat}}(k)/2\leq |V(P_{u,u'})|\leq \alpha_{\mathrm{pat}}(k)$ follow. It remains to argue that $P_{u,u'}$ has no internal vertex from $\Sep_{R^2}(P,v)$. Because $\Sep_{R^2}(P,v)\subseteq V(H)\setminus (A_{R^2,P,v}\cup B_{R^2,P,v})$, the only vertices of $P$ that $\Sep_{R^2}(P,v)$ can possibly contain are the $\alpha_{\mathrm{pat}}(k)$ vertices of $P$ that are closest to $v$. Since $P$ is long, none of these vertices belongs to $P'$, and hence $P_{u,u'}$ (which is a subpath of $P'$) has no internal vertex from $\Sep_{R^2}(P,v)$.

Symmetrically, we derive the existence of a vertex $v'\in\Sep_{R^2}(P,v)\cap V(P)$ such that the subpath $P_{v,v'}$ of $P$ between $v$ and $v'$ has no internal vertex from $\Sep_{R^2}(P,u)\cup\Sep_{R^2}(P,v)$.

Lastly, we prove that $P=P_{u,u'}-P_{u',v'}-P_{v',v}$. Since $P_{u,u'},P_{u',v'}$ and $P_{v',v}$ are subpaths of $P$ such that $V(P)=V(P_{u,u'})\cup V(P_{u',v'})\cup V(P_{v,v'})$, it suffices to show that {\em (i)} $V(P_{u,u'})\cap V(P_{u',v'})=\{u'\}$, {\em (ii)} $V(P_{v,v'})\cap V(P_{u',v'})=\{v'\}$, and {\em (iii)} $V(P_{u,u'})\cap V(P_{v,v'})=\emptyset$. Because $P$ is long and $|V(P_{u,u'})|,|V(P_{v,v'})|\leq \alpha_{\mathrm{pat}}(k)$, it is immediate that item {\em (iii)} holds. For item {\em (i)}, note that $V(P_{u,u'})\cap V(P_{u',v'})$ can be a strict superset of $\{u'\}$ only if $P_{u',v'}$ is a subpath of $P_{u,u'}$; then, $v'\in V(P_{u,u'})$, which means that $P_{u,u'}$ has an internal vertex from $\Sep_{R^2}(P,v)$ and results in a contradiction. Thus, item {\em (i)} holds. Symmetrically, item {\em (ii)} holds as well.
\end{proof}

In what follows, when we use the notation $u'_P$, we refer to the vertex in Lemma \ref{lem:threeParts}. Before we describe the modification, we need to introduce another notation and make an immediate observation based on this notation.

\begin{definition}[{\bf $\widetilde{A}_{R^2,P,u}$}]
Let $(G,S,T,g,k)$ be a good instance of \pdp. Let $R^2$ be a Steiner tree that has no detour, $P$ be a long maximal degree-2 path of $R^2$, and  $u$ be an endpoint of $P$. Then, $\widetilde{A}_{R^2,P,u}=(V(P_{u,u'_P})\setminus\{u'_P\})\cup A_{R^2,p,u}$. 
\end{definition}

\begin{observation}\label{obs:separateComps}
Let $(G,S,T,g,k)$ be a good instance of \pdp. Let $R^2$ be a Steiner tree that has no detour, and $P$ be a long maximal degree-2 path of $R^2$ with endpoints $u$ and $v$. Then, there exists a single connected component $C_{R^2,P,u}$ in $H-(\Sep_{R^2}(u,P)\cup\Sep_{R^2}(v,P))$ that contains $\widetilde{A}_{R^2,P,u}$ and a different single connected component $C_{R^2,P,v}$ in $H-(\Sep_{R^2}(u,P)\cup\Sep_{R^2}(v,P))$ that contains $\widetilde{A}_{R^2,P,v}$.
\end{observation}

We proceed to describe the modification. For brevity, let $S_u = \Sep_{R^2}(P,u)$ and $S_v = \Sep_{R^2}(P,v)$. Recall that there is a terminal $t^\star\in T$ that lies on the outer face of $H$ (and $G$). By Observation \ref{obs:sepIsCycle}, $S_v$ and $S_v$ induce two cycles in $H$, and $t^\star$ lies in the exterior of both these cycles. Assume w.l.o.g.~that $u$ lies in the interior of both $S_u$ and $S_v$, while $v$ lies in the exterior of both $S_u$ and $S_v$. Then, $S_u$ belongs to the strict interior of $S_v$. We construct a sequence of concentric cycles between $S_u$ and $S_v$ as follows.
\begin{lemma}\label{lemma:CC_cons}
Let $(G,S,T,g,k)$ be a good instance of \pdp. Let $R^2$ be a Steiner tree with no detour, and $P$ be a long maximal degree-2 path of $R^2$ with endpoints $u$ and $v$.  Let $S_u = \Sep_{R^2}(P,u)$ and $S_v = \Sep_{R^2}(P,v)$, where $S_u$ lies in the strict interior of $S_v$. Then, there is a sequence of  concentric cycles $\Cc(u,v) = (C_1,C_2,\ldots,C_p)$ in $G$ of length $p\geq 100 \alpha_{\rm sep}(k)$ such that $S_u$ is in the strict interior of $C_1$ in $H$, $S_v$ is in the strict exterior of $C_p$ in $H$, and there is a path $\eta$ in $H$ with one endpoint $v_0\in S_u$ and the other endpoint $v_{p+1}\in S_v$, such that the intersection of $V(\eta)$ with $V(G)\cup S_u\cup S_v$ is $\{v_0, v_1, \ldots,  v_{p+1}\}$ for some $v_i \in V(C_i)$ for every $i\in\{1,\ldots,p\}$.
Furthermore, $\Cc(u,v)$ can be computed in linear time.
\end{lemma}
\begin{proof}
Towards the computation of  $\Cc(u,v)$, delete all vertices that lie in the strict interior of $S_u$ or in the strict exterior of $S_v$, as well as all vertices of $V(H) \setminus (V(G)\cup S_u\cup S_v)$.  Denote the resulting graph by $G^{+}_{u,v}$, and note that it has a plane embedding in the ``ring'' defined by $H[S_u]$ and $H[S_v]$. Observe that $S_u,S_v\subseteq V(G^{+}_{u,v})$, where $S_v$ defines the outer face of the embedding of $G^{+}_{u,v}$.
Thus, any cycle in this graph that separates $S_u$ and $S_v$ must contain $S_u$ in its interior and $S_v$ in its exterior.
Furthermore, $G_{u,v} = G^{+}_{u,v} - V(H)$ is an induced subgraph of $G$, which consists of all vertices of $G$ that, in $H$, lie in the strict exterior of $S_u$ and in the strict interior of $S_v$ simultaneously, or lie in $S_v \cup S_v$. In particular, any cycle of $G_{u,v}$  is also a cycle in $G$.

Now, $\Cc(u,v)$ is computed as follows. Start with an empty sequence, and the graph $G_{u,v} - (S_u \cup S_v)$.
As long as there is a cycle in the current graph such that all vertices of $S_u$ are in the strict interior of $C$ with respect to $H$, remove vertices of degree at most $1$ in the current graph until no such vertices remain, and append the outer face of the current graph as a cycle to the constructed sequence.
It is clear that this process terminates in linear time, and that by the above discussion, it constructs a sequence of  concentric cycles $\Cc(u,v) = (C_1,C_2,\ldots,C_p)$ in $G$ such that $S_u$ is in the strict interior of $C_1$ in $H$, $S_v$ is in the strict exterior of $C_p$ in $H$. 

To assert the existence of a path $\eta$ in $H$ with one endpoint $v_0\in S_u$ and the other endpoint $v_{p+1}\in S_v$, such that the intersection of $V(\eta)$ with $V(G)\cup S_u\cup S_v$ is $\{v_0,v_1, \ldots,  v_{p+1}\}$ for some $v_i \in V(C_i)$ for every $i\in\{1,\ldots,p\}$, we require the following claim.
\begin{claim}\label{claim:CC_cons}
Let  $C_{p+1}=H[S_v]$. For every $i\in \{1,2,\ldots,p\}$ and every vertex $w\in V(C_i)$, there exists a vertex $w'\in V(C_{i+1})$ such that $w$ and $w'$ lie on a common face in $G^{+}_{u,v}$. Moreover, there exist vertices $w\in S_u$ and $w'\in V(C_1)$ that lie on a common face in $G^{+}_{u,v}$.
\end{claim}

\noindent{\em Proof of Claim \ref{claim:CC_cons}.} Consider $i\in \{1,2,\ldots,p\}$ and a vertex $w\in V(C_i)$. We claim that $\rdist(w,V(C_{i+1})) \leq 1$, i.e.~there must be $w' \in V(C_{i+1})$ such that $w,w'$ have a common face in $G^{+}_{u,v}$. By way of contradiction, suppose that $\rdist(w,V(C_{i+1}))\geq 2$. Then, by Proposition~\ref{prop:concentricAtAllDists}, there is a cycle $C$ that separates $w$ and $V(C_{i+1})$ in $G^{+}_{u,v}$ such that $\rdist(w,w'') = 1$ for every vertex $w'' \in V(C)$. Here, $w$ lies in the strict interior of $C$, and $C$ lies in the strict interior of $C_{i+1}$. Further, $C$ is vertex disjoint from $C_{i+1}$, since $\rdist(w, w') \geq 2$ for every $w' \in V(C_{i+1})$. Now, consider the outer face of $G[V(C_i)\cup V(C)]$. By the construction of $C_i$, this outer face must be $C_i$. However, $w\in V(C_i)$ cannot belong to it, hence we reach a contradiction.

For the second part, we claim that $\rdist(S_u,V(C_{1})) \leq 1$, i.e.~there must be $w\in S_u$ and  $w' \in V(C_{1})$ such that $w,w'$ have a common face in $G^{+}_{u,v}$. By way of contradiction, suppose that $\rdist(S_u,V(C_{1}))\geq 2$. Then, by Proposition~\ref{prop:concentricAtAllDists}, there is a cycle $C$ that separates $S_u$ and $V(C_{1})$ in $G^{+}_{u,v}$ such that $\rdist(S_u,w'') = 1$ for every vertex $w'' \in V(C)$. Further, $C$ is vertex disjoint from $C_{1}$, since $\rdist(S_u, w') \geq 2$ for every $w' \in V(C_{1})$. However, this is a contradiction to the termination condition of the construction of $\Cc(u,v)$.  $\diamond$

\smallskip
Having this claim, we construct $\eta$ as follows. Pick vertices $v_0\in S_u$ and $v_1\in V(C_1)$ that lie on a common face in $G^{+}_{u,v}$. Then, for every $i\in\{2,\ldots,p+1\}$, pick a vertex $v_i\in V(C_i)$ such that $v_{i-1}$ and $v_i$ lie on a common face in $G^{+}_{u,v}$. Thus, for every $i\in\{0,1,\ldots,p\}$, we have that $v_i$ and $v_{i+1}$ are either adjacent in $H$ or there exists a vertex $u_i\in V(H)\setminus V(G)$ such that $u_i$ is adjacent to both $v_i$ and $v_{i+1}$. Because $\Cc(u,v) = (C_1,C_2,\ldots,C_p)$ is a sequence of concentric cycles in $G$ such that $S_u$ is in the strict interior of $C_1$ and $S_v$ is in the strict exterior of $C_p$, the $u_i$'s are distinct.  Thus, $\eta=v_0-u_0-v_1-u_1-v_2-u_2-\cdots-v_p-u_p-v_{p+1}$, where undefined $u_i$'s are dropped, is a path as required.

Finally, we argue that $p\geq 100\cdot \alpha_{\rm sep}(k)$. Note that $100 \alpha_{\rm sep}(k)=100(\frac{7}{2}\cdot 2^{ck} + 2) \leq 400 \cdot 2^{ck}$, thus it suffices to show that $p\geq 400 \cdot 2^{ck}$. 
To this end, we obtain a lower bound on the radial distance between $S_u$ and $S_v$ in $G^{+}_{u,v}$. Recall that $|S_u|,|S_v| \leq \alpha_{\rm sep}(k) = \frac{7}{2}\cdot 2^{ck} + 2 \leq 4\cdot 2^{ck}$. Let $P = \pathT_{R^2}(u,v)$, and recall that its length is at least $\alpha_{\rm long}(k) = 10^4 2^{ck}$.
Since $R^2$ has not detour, $P$ is a shortest path in $H$ between $u$ and $v$, thus for any two vertices in $V(P)$, the subpath of $P$ between them is a shortest path between them. Now, recall the vertices $u'= u'_P,v'= v'_P$ (defined in Lemma~\ref{lem:threeParts}), and denote the subpath between them by $P'$. By construction,
$|E(P')| \geq \alpha_{\rm long}(k) - 2 \cdot \alpha_{\rm pat}(k) = (10^4 - 200)\cdot 2^{ck}$.

We claim that the radial distance between $S_u$ and $S_v$ in $G^{+}_{u,v}$ is at least $|E(P')|/2 - |S_u| - |S_v|$.
Suppose not, and consider a sequence of vertices in $G^{+}_{u,v}$ that witnesses this fact: $x_1, x_2, x_3, \ldots, x_{p-1}, x_p\in V(G^{+}_{u,v})$ where $x_1 \in S_u$, $x_p \in S_v$ , $p < |E(P')|/2 - |S_v| - |S_u|$, and every two consecutive vertices lie on a common face.  Consider a shortest such sequence, which visits each face of $G^{+}_{u,v}$ at most once.
In particular, $x_1$ and $x_p$ are the only vertices of $S_u\cup S_v$ in this sequence. Then, we can extend this sequence on both sides to derive another sequence of vertices of $G^{+}_{u,v}$ starting at $u'$ and ending at $v'$ such that the prefix of the new sequence is a path in $G^{+}_{u,v}[S_u]$ from $u'$ to $x_1$, the midfix is $x_1, x_2, x_3, \ldots, x_{p-1}, x_p$, and the suffix is a path in $G^{+}_{u,v}[S_v]$ from $x_p$ to $v'$.
Further, the length of the new sequence of vertices is smaller than $|E(P')|/2$.
Hence, the radial distance between $u'$ and $v'$ in $H'=H[V(G)\cup S_u\cup S_v]$ (the graph derived from $G^{+}_{u,v}$ by reintroducing the vertices of $G$ that lie inside $S_u$ or outside $S_v$) is smaller than $|E(P')|/2$, and let it be witnessed by a sequence $Q[u',v']$.
As $H$ is the radial completion of $G$,  observe that $Q[u',v']$ gives rise to a path $Q$ in $H$ between $u'$ and $v'$ of length smaller than $|E(P')|$. 
However, then $u',v'$ and $Q$ witness a detour in $R^2$, which is a contradiction.
Hence, the radial distance between $S_u$ and $S_v$ in $G^{+}_{u,v}$ is at least
\begin{align*}
 &~~ \frac{(10^4 - 200)}{2}\cdot 2^{ck} - (|S_u| + |S_v|) \\
\geq &~~ 4900\cdot 2^{ck} - 2\alpha_{\rm sep}(k) \\
= &~~ 4900 \cdot 2^{ck} - (7 \cdot 2^{ck} + 4) ~~ \geq ~~ 400 \cdot 2^{ck}.
\end{align*}

Now, observe that $S_u$ and $S_v$ are connected sets in $G^{+}_{u,v}$ and $S_v$ forms the outer-face of $G^{+}_{u,v}$.
Then, by Proposition~\ref{prop:concentricAtAllDists}, we obtain a collection of at least $400 \cdot 2^{ck}$ disjoint cycles in $G^{+}_{u,v}$, where each cycle separates $S_u$ and $S_v$. 
Note that these cycles are disjoint from $S_u \cup S_v$, and hence they lie in $G$.
Moreover, each of them contains $S_u$ in its strict interior, and $S_v$ in its strict exterior. Thus, it is clear that the sequence $\Cc(u,v)$ computed above must contain at least $400 \cdot 2^{ck} > 100 \alpha_{\rm sep}(k)$ cycles.
\end{proof}

Recall the graph $G_{u,v}$, which is an induced subgraph of $G$, which consists of all vertices of $G$ that, in $H$, lie in the strict exterior of $S_u$ and in the strict interior of $S_v$ simultaneously, or lie in $S_v \cup S_v$.
With $\Cc(u,v)$ at hand, we compute a maximum size collection of disjoint paths from $S_u$ to $S_v$ in $G_{u,v}$ that minimizes the number of edges it traverses outside $E(\Cc(u,v))$. In this observation, the implicit assumption that $\ell\leq \alpha_{\rm sep}(k)$ is justified by Lemma \ref{lem:sepSmall}.
\begin{observation}\label{obs:disjPTime}
Let the maximum flow between $S_u\cap V(G)$ and $S_v\cap V(G)$ in $G_{u,v}$ be $\ell\leq \alpha_{\rm sep}(k)$.  Given the sequence $\Cc(u,v)$ of Lemma \ref{lemma:CC_cons},  a collection $\flow_{R^2}(u,v)$ of $\ell$ vertex-disjoint paths in $G_{u,v}$ from $S_u\cap V(G)$ to $S_v\cap V( G)$ that minimizes $|E(\flow_{R^2}(u,v)) \setminus E(\Cc(u,v))|$ is computable in time $\OO(n^{3/2} \log^3 n)$.
\end{observation}
\begin{proof}
We determine $\ell\leq \alpha_{\rm sep}(k)$ in time $2^{\OO(k)}n$ by using Ford-Fulkerson algorithm. Next, we define a weight function $w$ on $E(G_{u,v})$ as follows:    
    $$
    w(e) = \begin{cases}
            0 & \text{if $e \in E(\Cc(u,v))$}\\
            1 & \text{otherwise}
           \end{cases} 
    $$
We now compute a minimum cost flow between $S_u\cap V(G)$ and $S_v\cap V(G)$ of value $\ell$ in $G_{u,v}$ under the weight function $w$. This can be done in time $O(n^{3/2} \log^3 n)$ by \cite[Theorem 1]{CK12}, as the cost of such a flow is bounded by $O(n)$.
Clearly, the result is a collection $\flow_{R^2}(u,v)$ of $\ell$ vertex-disjoint paths from $S_u\cap V(G)$ to $S_v\in V(G)$ minimizing $|E(\flow_{R^2}(u,v)) \setminus E(\Cc(u,v))|$.
\end{proof}

Having $\flow_{R^2}(u,v)$ at hand, we proceed to find a certain path between $u'_P$ and $v'_P$ that will be used to replace $P_{u'_P,v'_P}$. We remark that all vertices and edges of $\flow_{R^2}(u,v)$ lie between $S_u$ and $S_v$ in the plane embedding of $H$. The definition of this path is given by the following lemma, and the construction will make it intuitively clear that the paths in $\flow_{R^2}(u,v)$ do not ``spiral around'' the Steiner tree once we replace $P_{u'_P,v'_P}$ with $P_{u'_P,v'_P}^\star$. Note that in the lemma, we consider a path $P$ in $H$, while the paths in $\flow_{R^2}(u,v)$ are in $G$.

\begin{lemma}\label{lem:pathThroughFlow}
Let $(G,S,T,g,k)$ be a good instance of \pdp. Let $R^2$ be a Steiner tree with no detour, and $P$ be a long maximal degree-2 path of $R^2$, with endpoints $u$ and $v$. Let $u'=u'_P$ and $v'=v'_P$. Then, there exists a path $P^\star_{u',v'}$ in $H-(V(C_{R^2,P,u})\cup V(C_{R^2,P,u}))$ between $u'$ and $v'$ with the following property: there do not exist three vertices $x,y,z\in V(P^\star_{u',v'})$ such that {\em (i)} $\dist_{P^\star_{u',v'}}(u',x)<\dist_{P^\star_{u',v'}}(u',y)<\dist_{P^\star_{u',v'}}(u',z)$, and {\em (ii)} there exist a path in $\flow_{R^2}(u,v)$ that contains $x$ and $z$ and a different path in $\flow_{R^2}(u,v)$ that contains $y$. Moreover, such a path $P^\star_{u',v'}$ can be computed in time $\OO(n)$.
\end{lemma}

\begin{proof}
Let $P^\star_{u',v'}$ be a path in $H-(V(C_{R^2,P,u})\cup V(C_{R^2,P,u}))$ between $u'$ and $v'$ that minimizes the number of paths $Q\in \flow_{R^2}(u,v)$ for which there exist at least one triple $x,y,z\in V(P^\star_{u',v'})$ that has the two properties in the lemma and $x,z\in V(Q)$. Due to the existence of $P_{u',v'}$, such a path $P^\star_{u',v'}$ exists. We claim that this path $P^\star_{u',v'}$ has no triple $x,y,z\in V(P^\star_{u',v'})$ that has the two properties in the lemma. Suppose, by way of contradiction, that our claim is false, and let $x,y,z\in V(P^\star_{u',v'})$ be a triple that has the two properties in the lemma. Note that, when traversed from $u'$ to $v'$, $P^\star_{u',v'}$ first visits $x$, then visits $y$ and afterwards visits $z$. Let $Q'$ be the path in $\flow_{R^2}(u,v)$ that contains $x$ and $z$. Let $x'$ and $z'$ be the first and last vertices of $Q'$ that are visited by $P^\star_{u',v'}$. Then, replace the subpath of $P^\star_{u',v'}$ between $x'$ and $z'$ by the subpath of $Q'$ between $x'$ and $z'$. This way we obtain a path $P'$ in $H-(V(C_{R^2,P,u})\cup V(C_{R^2,P,u}))$ between $u'$ and $v'$ for which there exist fewer paths $Q\in \flow_{R^2}(u,v)$, when compared to $P^\star_{u',v'}$, for which there exists at least one triple $x,y,z\in V(P')$ that has the two properties in the lemma and such that $x,z\in V(Q)$. As we have reached a contradiction, we conclude that our initial claim is correct.
While the proof is existential, it can clearly be turned into a linear-time algorithm. 
\end{proof}

The following is a direct corollary of the above lemma.
\begin{corollary}\label{cor:pathThroughFlow}
    For each path $Q \in \flow_{R^2}(u,v)$, there are at most two edges in $P^\star_{u',v'}$ such that one endpoint of the edge lies in $V(Q)$ and the other lies in $V(G) \setminus V(Q)$.
\end{corollary}

Having Lemma \ref{lem:pathThroughFlow} at hand, we modify $R^2$ as follows: for every long maximal degree-2 path $P$ of $R^2$ with endpoints $u$ and $v$, replace $P_{u'_P,v'_P}$ by $P^\star_{u'_P,v'_P}$. Denote the result of this modification by $R=R^3$. We refer to $R^3$ as a {\em backbone Steiner tree}. 
Let us remark that the backbone Steiner tree $R^3$ is always accompanied by the separators $\Sep_{R^2}(P,u)$ and $\Sep_{R^2}(P,v)$, and the collection $\flow_{R^2}(u,v)$  for every long maximal degree-2 path $P=\pathT_{R^2}(u,v)$ of $R^2$. These separators and flows will play crucial role in our algorithm.
In the following subsection, we will prove that $R^3$ is indeed a Steiner tree, and in particular it is a tree. Additionally and crucially, we will prove that the separators computed previously remain separators. Let us first conclude the computational part by stating the running time spent so far. From Lemma \ref{obs:undetourExhaustive}, Observations \ref{obs:sepTime} and \ref{obs:disjPTime}, and by Lemma \ref{lem:pathThroughFlow}, we have the following result.

\begin{lemma}\label{lem:goodSteinerTreeComputeTime}
Let $(G,S,T,g,k)$ be a good instance of \pdp. Then, a backbone Steiner tree $R$ can be computed in time $2^{\OO(k)}n^{3/2} \log^3 n$.
\end{lemma}

\subsection{Analysis of $R^3$ and the Separators $\Sep_{R^2}(P,u)$} 
Having constructed the backbone Steiner tree $R^3$, we turn to analyse its properties. Among other properties, we show that useful properties of $R^2$ also transfer to $R^3$.
We begin by proving that the two separators of each long maximal degree-2 path $P$ of $R^2$ partition $V(H)$ into five ``regions'', and that the vertices in each region are all close to the subtree of $R$ that (roughly) belongs to that region. Specifically, the regions are $V(C_{R^2,P,u})$, $\Sep_{R^2}(P,u)$, $V(C_{R^2,P,v})$, $\Sep_{R^2}(P,v)$, and $V(H)\setminus \big(V(C_{R^2,P,u})\cup V(C_{R^2,P,v})\cup\Sep_{R^2}(P,u)\cup\Sep_{R^2}(P,v) \big)$, and our claim is as follows.

\begin{lemma}\label{lem:regionsClose}
Let $(G,S,T,g,k)$ be a good instance of \pdp. Let $R^2$ be a Steiner tree that has no detour, $P$ be a long maximal degree-2 path of $R^2$, and $u$ and $v$ be its endpoints.
\begin{enumerate}
\item\label{item:threeRegionsClose1} For all $w\in \Sep_{R^2}(P,u)$, it holds that $\dist_H(w,u'_P)\leq \alpha_{\mathrm{sep}}(k)$.
\item\label{item:threeRegionsClose2} For all $w\in \Sep_{R^2}(P,v)$, it holds that $\dist_H(w,v'_P)\leq \alpha_{\mathrm{sep}}(k)$.
\item\label{item:threeRegionsClose3} For all $w\in V(C_{R^2,P,u})$, it holds that $\dist_H(w,\widetilde{A}_{R^2,P,u}\cup\{u'_P\})\leq \alpha_{\mathrm{dist}}(k)+\alpha_{\mathrm{sep}}(k)$.
\item\label{item:threeRegionsClose4} For all $w\in V(C_{R^2,P,v})$, it holds that $\dist_H(w,\widetilde{A}_{R^2,P,v}\cup\{v'_P\})\leq \alpha_{\mathrm{dist}}(k)+\alpha_{\mathrm{sep}}(k)$.
\item\label{item:threeRegionsClose5} For all $w\in V(H)\setminus (V(C_{R^2,P,u})\cup V(C_{R^2,P,v})\cup\Sep_{R^2}(P,u)\cup\Sep_{R^2}(P,v))$, it holds that $\dist_H(w,V(R^2)\setminus (\widetilde{A}_{R^2,P,u}\cup\widetilde{A}_{R^2,P,v}))\leq \alpha_{\mathrm{dist}}(k)+\alpha_{\mathrm{sep}}(k)$.
\end{enumerate}
\end{lemma}

\begin{proof}
First, note that Conditions \ref{item:threeRegionsClose1} and \ref{item:threeRegionsClose2} follow directly from Lemma \ref{lem:sepSmall} and Observation \ref{obs:sepIsCycle}.

For Condition \ref{item:threeRegionsClose3}, consider some vertex $w\in V(C_{R^2,P,u})$. By Lemma \ref{lem:closeToR}, $\dist_H(w,V(R^2))\leq \alpha_{\mathrm{dist}}(k)$. Thus, there exists a path $Q$ in $H$ with $w$ as one endpoint and the other endpoint $x$ in $V(R^2)$ such that the length of $Q$ is at most $\alpha_{\mathrm{dist}}(k)$. In case $x\in \widetilde{A}_{R^2,P,u}\cup\{u'_P\}$, we have that $\dist_H(w,\widetilde{A}_{R^2,P,u})\leq \alpha_{\mathrm{dist}}(k)$, and hence the condition holds. Otherwise, by the definition of $\Sep_{R^2}(P,u)$, the path $Q$ must traverse at least one vertex from $\Sep_{R^2}(P,u)$. Thus, $\dist_H(w,\Sep_{R^2}(P,u))\leq \alpha_{\mathrm{dist}}(k)$. Combined with Condition \ref{item:threeRegionsClose1}, we derive that $\dist_H(w,\widetilde{A}_{R^2,P,u}\cup\{u'_P\})\leq \alpha_{\mathrm{dist}}(k)+\alpha_{\mathrm{sep}}(k)$. The proof of Condition \ref{item:threeRegionsClose4} is symmetric.

The proof of Condition \ref{item:threeRegionsClose5} is similar. Consider some vertex $w\in V(H)\setminus (V(C_{R^2,P,u})\cup V(C_{R^2,P,v})\cup\Sep_{R^2}(P,u)\cup\Sep_{R^2}(P,v))$. As before, there exists a path $Q$ in $H$ with $w$ as one endpoint and the other endpoint $x$ in $V(R^2)$ such that the length of $Q$ is at most $\alpha_{\mathrm{dist}}(k)$. In case $x\in \widetilde{A}_{R^2,P,u}\cup\{u'_P\}$, we are done. Otherwise, the path $Q$ must traverse at least one vertex from $\Sep_{R^2}(P,u)\cup\Sep_{R^2}(P,v)$. Specifically, if $x\in V(C_{R^2,P,u})$, then it must traverse at least one vertex from $\Sep_{R^2}(P,u)$, and otherwise $x\in V(C_{R^2,P,v})$ and it must traverse at least one vertex from $\Sep_{R^2}(P,u)$. Combined with Conditions \ref{item:threeRegionsClose1} and \ref{item:threeRegionsClose2}, we derive that $\dist_H(w,V(R^2)\setminus (\widetilde{A}_{R^2,P,u}\cup\widetilde{A}_{R^2,P,v}))\leq \alpha_{\mathrm{dist}}(k)+\alpha_{\mathrm{sep}}(k)$.
\end{proof}

An immediate corollary of Lemma \ref{lem:regionsClose} concerns the connectivity of the ``middle region'' as follows. (This corollary can also be easily proved directly.)

\begin{corollary}\label{cor:connectivityMidRegion}
Let $(G,S,T,g,k)$ be a good instance of \pdp. Let $R^2$ be a Steiner tree that has no detour, $P$ be a long maximal degree-2 path of $R^2$, and $u$ and $v$ be its endpoints. Then, $H[V(H)\setminus (V(C_{R^2,P,u})\cup V(C_{R^2,P,v}))]$ is a connected graph.
\end{corollary}

\begin{proof}
By Lemma \ref{lem:regionsClose} and the definition of $\Sep_{R^2}(P,u)$ and $\Sep_{R^2}(P,v)$, for every vertex in $V(H)\setminus (V(C_{R^2,P,u})\cup V(C_{R^2,P,v}))$, the graph $H$ has a path from that vertex to some vertex in $\Sep_{R^2}(P,u)\cup\Sep_{R^2}(P,v)$ that lies entirely in $H[V(H)\setminus (V(C_{R^2,P,u})\cup V(C_{R^2,P,v}))]$. Thus, the corollary follows from Observation \ref{obs:sepIsCycle}.
\end{proof}

Next, we utilize Lemma \ref{lem:regionsClose} and Corollary \ref{cor:connectivityMidRegion} to argue that the ``middle regions'' of different long maximal degree-2 paths of $R^2$ are distinct. 
Recall that we wish to reroute a given solution to be a solution that ``spirals'' only a few times around the Steiner tree.  This lemma allows us to independently reroute the solution in each of these ``middle regions''. In fact, we prove the following stronger statement concerning these regions. The idea behind the proof of this lemma is that if it were false, then $R^2$ admits a detour, which is contradiction.

\begin{figure}
    \begin{center}
        \includegraphics[width=0.8\textwidth]{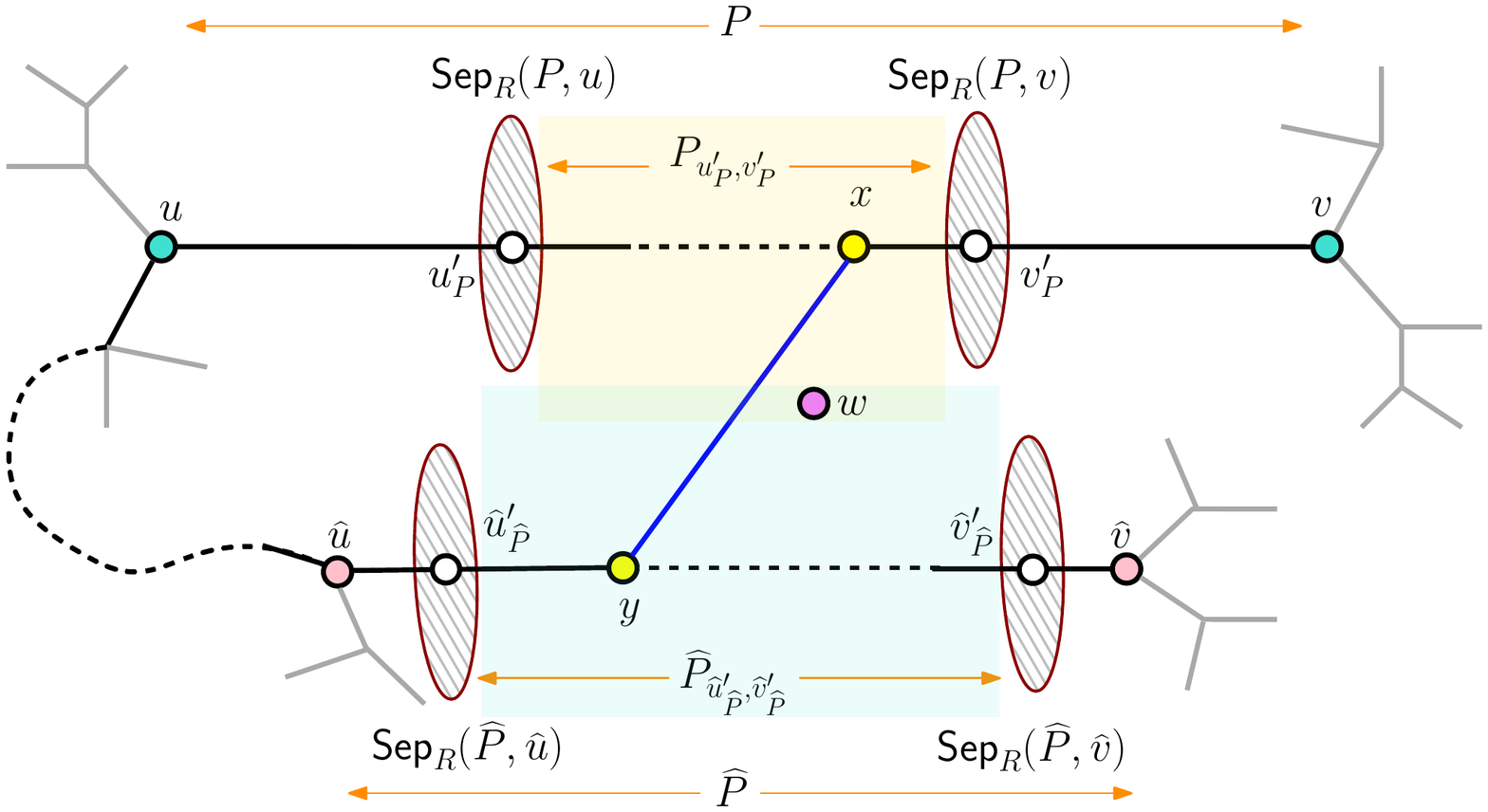}
        \caption{Illustration of Lemma~\ref{lem:distinctMidRegions}.}
        \label{fig:treeflow2}
    \end{center}
\end{figure}

\begin{lemma}\label{lem:distinctMidRegions}
Let $(G,S,T,g,k)$ be a good instance of \pdp. Let $R^2$ be a Steiner tree that has no detour. Additionally, let $P$ and $\widehat{P}$ be two distinct long maximal degree-2 paths of $R^2$. Let $u$ and $v$ be the endpoints of $P$, and $\widehat{u}$ and $\widehat{v}$ be the endpoints of $\widehat{P}$. Then, one of the two following conditions holds:
\begin{itemize}\setlength\itemsep{0em}
\item $V(H)\setminus (V(C_{R^2,P,u})\cup V(C_{R^2,P,v})) \subseteq V(C_{R^2,\widehat{P},\widehat{u}})$.
\item $V(H)\setminus (V(C_{R^2,P,u})\cup V(C_{R^2,P,v})) \subseteq V(C_{R^2,\widehat{P},\widehat{v}})$.
\end{itemize}
\end{lemma}

\begin{proof}
We first prove that the intersection of $V(H)\setminus (V(C_{R^2,P,u})\cup V(C_{R^2,P,v}))$ with $V(H)\setminus (V(C_{R^2,\widehat{P},\widehat{u}})\cup V(C_{R^2,\widehat{P},\widehat{v}}))$ is empty. To this end, suppose by way of contradiction that there exists a vertex $w$ in this intersection. By Lemma \ref{lem:regionsClose}, the inclusion of $w$ in both $V(H)\setminus (V(C_{R^2,P,u})\cup V(C_{R^2,P,v}))$ and $V(H)\setminus (V(C_{R^2,\widehat{P},\widehat{u}})\cup V(C_{R^2,\widehat{P},\widehat{v}}))$ implies that the two following inequalities are satisfied:
\begin{itemize}\setlength\itemsep{0em}
\item $\dist_H(w,V(R^2)\setminus (\widetilde{A}_{R^2,P,u}\cup\widetilde{A}_{R^2,P,v}))\leq \alpha_{\mathrm{dist}}(k)+\alpha_{\mathrm{sep}}(k)$.
\item $\dist_H(w,V(R^2)\setminus (\widetilde{A}_{R^2,\widehat{P},\widehat{u}}\cup\widetilde{A}_{R^2,\widehat{P},\widehat{v}}))\leq \alpha_{\mathrm{dist}}(k)+\alpha_{\mathrm{sep}}(k)$.
\end{itemize}
From this, we derive the following inequality:
\[\dist_H(V(R^2)\setminus (\widetilde{A}_{R^2,P,u}\cup\widetilde{A}_{R^2,P,v}),V(R^2)\setminus (\widetilde{A}_{R^2,\widehat{P},\widehat{u}}\cup\widetilde{A}_{R^2,\widehat{P},\widehat{v}}))\leq 2(\alpha_{\mathrm{dist}}(k)+\alpha_{\mathrm{sep}}(k)).\]
In particular, this means that there exist vertices $x\in V(P_{u'_P,v'_P})$ and $y\in V(\widehat{P}_{\widehat{u}'_{\widehat{P}},\widehat{v}'_{\widehat{P}}})$ and a path $Q$ in $H$ between them whose length is at most $2(\alpha_{\mathrm{dist}}(k)+\alpha_{\mathrm{sep}}(k))$. Note that the unique path in $R^2$ between $x$ and $y$ traverses exactly one vertex in $\{u,v\}$.
Suppose w.l.o.g.~that this vertex is $u$. Then, consider the walk $W$ (which might be a path) obtained by traversing $P$ from $v$ to $x$ and then traversing  $Q$ from $x$ to $y$. Now, notice that
\[\begin{array}{ll}
\smallskip
|E(P)|-|E(W)| & \geq |E(P'_{u,u'_P})|-2(\alpha_{\mathrm{dist}}(k)+\alpha_{\mathrm{sep}}(k))\\
\smallskip
& \geq \alpha_{\mathrm{pat}}(k)/2-2(\alpha_{\mathrm{dist}}(k)+\alpha_{\mathrm{sep}}(k))\\
& = 50\cdot 2^{ck} - 2(4\cdot 2^{ck} + \frac{7}{2}\cdot 2^{ck}+2) > 0.
\end{array}\]
Here, the inequality $|E(P'_{u,u'_P})|\geq \alpha_{\mathrm{pat}}(k)/2$ followed from Lemma \ref{lem:threeParts}. As $|E(P)|-|E(W)|>0$, we have that $u,v$ and any subpath of the walk $W$ between $u$ and $v$ witness that $R^2$ has a detour (see Fig.~\ref{fig:treeflow2}).
This is a contradiction, and hence we conclude that the intersection of $V(H)\setminus (V(C_{R^2,P,u})\cup V(C_{R^2,P,v}))$ with $V(H)\setminus (V(C_{R^2,\widehat{P},\widehat{u}})\cup V(C_{R^2,\widehat{P},\widehat{v}}))$ is empty.

Having proved that the intersection is empty, we know that
\[V(H)\setminus (V(C_{R^2,P,u})\cup V(C_{R^2,P,v})) \subseteq V(C_{R^2,\widehat{P},\widehat{u}}) \cup V(C_{R^2,\widehat{P},\widehat{v}}).\]
Thus, it remains to show that $V(H)\setminus (V(C_{R^2,P,u})\cup V(C_{R^2,P,v}))$ cannot contain vertices from both $V(C_{R^2,\widehat{P},\widehat{u}})$ and $V(C_{R^2,\widehat{P},\widehat{v}})$. Since $H[V(H)\setminus (V(C_{R^2,P,u})\cup V(C_{R^2,P,v}))]$ is a connected graph (by Corollary \ref{cor:connectivityMidRegion}), if $V(H)\setminus (V(C_{R^2,P,u})\cup V(C_{R^2,P,v}))$ contains vertices from both $V(C_{R^2,\widehat{P},\widehat{u}})$ and $V(C_{R^2,\widehat{P},\widehat{v}})$, then it must also contain at least one vertex from $\Sep_{R^2}(\widehat{P},\widehat{u})\cup\Sep_{R^2}(\widehat{P},\widehat{v})\subseteq V(H)\setminus (V(C_{R^2,\widehat{P},\widehat{u}})\cup V(C_{R^2,\widehat{P},\widehat{v}}))$, which we have already shown to be impossible. Thus, the proof is complete.
\end{proof}

We are now ready to prove that $R^3$ is a Steiner tree.

\begin{lemma}\label{lem:R*Steiner}
Let $(G,S,T,g,k)$ be a good instance of \pdp. Let $R^2$ be a Steiner tree that has no detour, and $R^3$ be the subgraph constructed from $R^2$ in Step IV. Then, $R^3$ is a Steiner tree with the following properties.
\begin{itemize}
\item $R^3$ has the same set of vertices of degree at least $3$ as $R^2$.
\item Every short maximal degree-2 path $P$ of $R^2$ is a short maximal degree-2 path of $R^3$.
\item For every long maximal degree-2 path $P$ of $R^2$ with endpoints $u$ and $v$, the paths $P_{u,u'_P}$ and $P_{v,v'_P}$ are subpaths of the maximal degree-2 path of $R^3$ with endpoints $u$ and $v$.
\end{itemize}
\end{lemma}

\begin{proof}
To prove that $R^3$ is a Steiner tree, we only need to show that $R^3$ is acyclic. Indeed, the construction of $R^3$ immediately implies that it is connected and has the same set of degree-1 vertices as $R^2$, which together with an assertion that $R^3$ is acyclic, will imply that it is a Steiner tree. The other properties in the lemma are  immediate consequences of the construction of $R^3$.

By its construction, to show that $R^3$ is acyclic, it suffices to prove two conditions:
\begin{itemize}
\item For every long maximal degree-2 path $P$ of $R^2$ with endpoints $u$ and $v$, it holds that $V(P^\star_{u'_P,v'_P})\cap (V(R^2)\setminus V(P_{u'_P,v'_P}))=\emptyset$.
\item For every two distinct long maximal degree-2 paths $P$ and $\widehat{P}$ of $R^2$, it holds that $V(P^\star_{u'_P,v'_P})$ $\cap V(\widehat{P}^\star_{\widehat{u}',\widehat{v}'})=\emptyset$, where $u$ and $v$ are the endpoints of $P$, $\widehat{u}$ and $\widehat{v}$ are the endpoints of $\widehat{P}$, $u'=u'_P, v'=v'_P, \widehat{u}'=\widehat{u}'_{\widehat{P}}$ and $\widehat{v}'=\widehat{v}'_{\widehat{P}}$.
\end{itemize}
The first condition follows directly from the fact that $P^\star_{u'_P,v'_P}$ is a path in $H-(V(C_{R^2,P,u})\cup V(C_{R^2,P,u}))$ while $V(R^2)\setminus V(P_{u'_P,v'_P})\subseteq V(C_{R^2,P,u})\cup V(C_{R^2,P,u})$.

For the second condition, note that Lemma \ref{lem:distinctMidRegions} implies that $V(H)\setminus (V(C_{R^2,P,u})\cup V(C_{R^2,P,v}))$ $\subseteq V(C_{R^2,\widehat{P},\widehat{u}})\cup V(C_{R^2,\widehat{P},\widehat{v}})$. Thus, we have that $V(P^\star_{u'_P,v'_P})\subseteq V(C_{R^2,\widehat{P},\widehat{u}}) \cup V(C_{R^2,\widehat{P},\widehat{v}})$.  However, $V(\widehat{P}^\star_{\widehat{u}',\widehat{v}'})\cap (V(C_{R^2,\widehat{P},\widehat{u}}) \cup V(C_{R^2,\widehat{P},\widehat{v}}))=\emptyset$, and hence $V(P^\star_{u'_P,v'_P})\cap V(\widehat{P}^\star_{\widehat{u}',\widehat{v}'})=\emptyset$.
\end{proof}

We remark that $R^3$ might have detours. (These detours are restricted to $P_{u'_P,v'_P}$ for some long path $P=\pathT_{R^3}(u,v)$ in $R^3$.)
However, what is important for us is that we can still use the same small separators as before. 
To this end, we first define the appropriate notations, in particular since later we will like to address objects corresponding to $R^3$ directly (without referring to $R^2$). The validity of these notations follows from Lemma \ref{lem:R*Steiner}. Recall that $u'_P$ and $P_{u,u'_P}$ refer to the vertex and path in Lemma \ref{lem:threeParts}.

\begin{definition}[{\bf Translating Notations of $R^2$ to $R^3$}]\label{lem:translateRtoR*}
Let $(G,S,T,g,k)$ be a good instance of \pdp. Let $R^2$ and $R^3$ be the Steiner trees constructed in Steps II and IV. For any long maximal degree-2 path $\widehat{P}$ of $R^3$ and for each endpoint $u$ of $\widehat{P}$:
\begin{itemize}
\item Define $\widehat{P}_{R^2}$ as the unique (long maximal degree-2) path in $R^2$ with the same endpoints~as~$\widehat{P}$.
\item Let $P=\widehat{P}_{R^2}$. Then, denote $u'_{\widehat{P}}=u'_P$, $\widehat{P}_{u,u'_{\widehat{P}}}=P_{u,u'_P}$ and $\Sep_{R^3}(\widehat{P},u)=\Sep_{R^2}(P,u)$.
\item Define $A^\star_{R^3,\widehat{P},u}$ as the union of $V(\widehat{P}_{u,u'_{\widehat{P}}})$ and the vertex set of the connected component  of $R^3-(V(\widehat{P}_{u,u'_{\widehat{P}}})\setminus \{u\})$ containing $u$.
\item Define $B^\star_{R^3,\widehat{P},u}=V(R^3)\setminus A^\star_{R^3,\widehat{P},u}$.
\end{itemize}
\end{definition}

In the context of Definition \ref{lem:translateRtoR*}, note that by Lemma \ref{lem:threeParts}, $u'\in\Sep_{R^3}(\widehat{P},u)\cap V(\widehat{P})$ where $u'=u'_{\widehat{P}}$, $\widehat{P}_{u,u'}$ is the subpath of $\widehat{P}$ between $u$ and $u'$, $\widehat{P}_{u,u'}$ has no internal vertex from $\Sep_{R^3}(\widehat{P},u)\cup\Sep_{R^3}(\widehat{P},v)$, and $\alpha_{\mathrm{pat}}(k)/2\leq |V(\widehat{P}_{u,u'})|\leq \alpha_{\mathrm{pat}}(k)$. Additionally, note that $A^\star_{R^3,\widehat{P},u}$ might {\em not} be equal to $A_{R^2,P,u}$ where $P=\widehat{P}_{R^2}$. When $R^3$ is clear from context, we omit it from the subscripts.

Now, let us argue why, in a sense, we can still use the same small separators as before. Recall that a backbone Steiner tree is a Steiner tree constructed in Step IV.

\begin{lemma}\label{lem:separatorsUnchanged}
Let $(G,S,T,g,k)$ be a good instance of \pdp. Let $R^3$ be a backbone Steiner tree. Additionally, let $\widehat{P}$ be a long maximal degree-2 path of $R^3$, and $u$ be an endpoint of $\widehat{P}$. Then, $\Sep(\widehat{P},u)$ separates $A^\star_{\widehat{P},u}$ and $B^\star_{\widehat{P},u}$ in $H$.
\end{lemma}

\begin{figure}
    \begin{center}
        \includegraphics[width=0.8\textwidth]{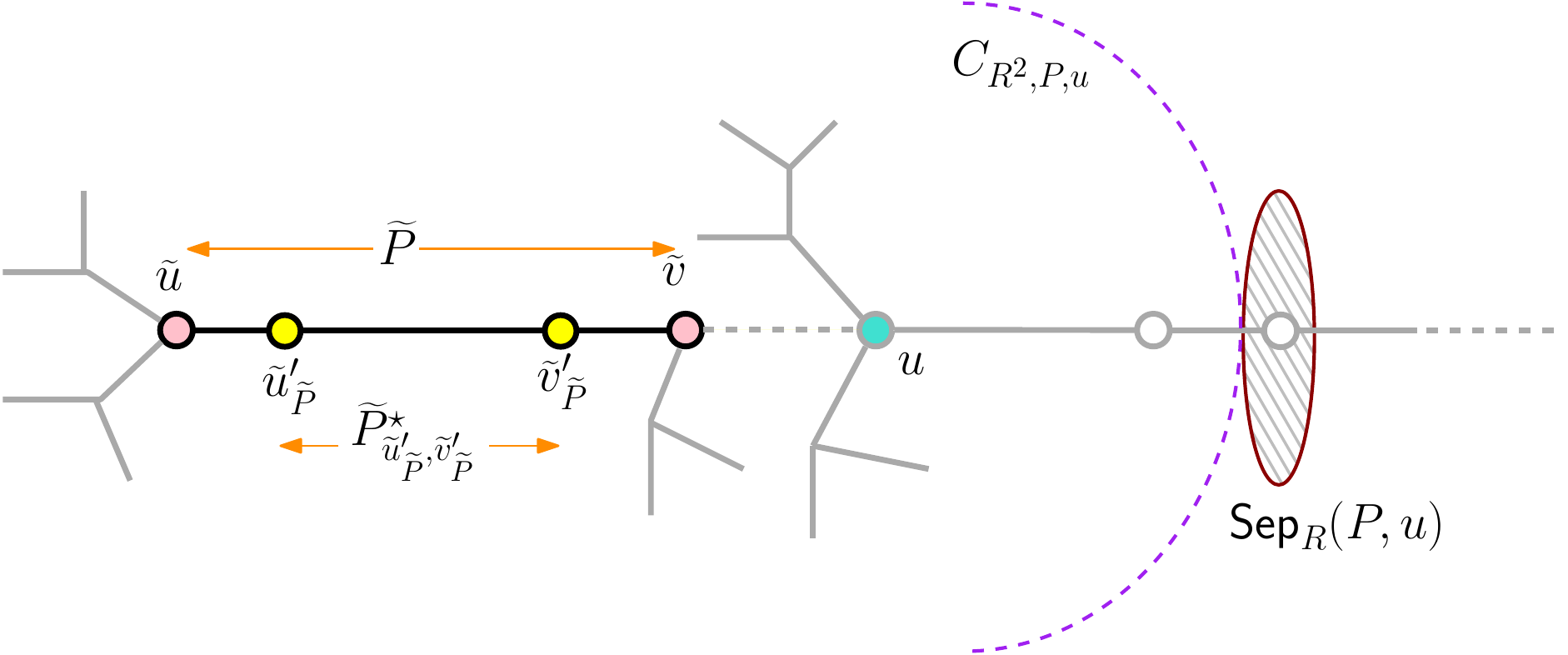}
        \caption{Illustration of Lemma~\ref{lem:separatorsUnchanged}.}
        \label{fig:treeflow3}
    \end{center}
\end{figure}

\begin{proof}
Denote $P=\widehat{P}_{R^2}$ where $R^2$ is the Steiner tree computed in Step II to construct $R^3$. Then, $\Sep(\widehat{P},u)$ separates $V(C_{R^2,P,v})$ and $V(C_{R^2,P,u})$ in $H$. Thus, to prove that $\Sep(\widehat{P},u)$ separates $A^\star_{\widehat{P},u}$ and $B^\star_{\widehat{P},u}$ in $H$, it suffices to show that $A^\star_{\widehat{P},u}\subseteq V(C_{R^2,P,u})$ and $A^\star_{\widehat{P},v}\subseteq V(C_{R^2,P,v})$ (because $B^\star_{\widehat{P},u}\setminus A^\star_{\widehat{P},v}\subseteq V(P^\star_{u'_P,v'_P})$ and $P^\star_{u'_P,v'_P}\cap V(C_{R^2,P,u}=\emptyset$ by the construction of $P^\star_{u'_P,v'_P}$). We only prove that $A^\star_{\widehat{P},u}\subseteq V(C_{R^2,P,u})$. The proof of the other containment is symmetric.

Clearly, $A_{R^2,P,u}\cap A^\star_{\widehat{P},v}\subseteq V(C_{R^2,P,u})$. Thus, due to Lemma \ref{lem:R*Steiner}, to show that $A^\star_{\widehat{P},u}\subseteq V(C_{R^2,P,u})$, it isuffices to show the following claim: For every long maximal degree-2 path $\widetilde{P}$ of $R^2$ whose vertex set is contained in $V(C_{R^2,P,u})$, it holds that the vertex set of $\widetilde{P}^\star_{\widetilde{u}'_{\widetilde{P}},\widetilde{v}'_{\widetilde{P}}}$ (computed by Lemma \ref{lem:pathThroughFlow}) is contained in $V(C_{R^2,P,u})$ as well, where $\widetilde{u}$ and $\widetilde{v}$ are the endpoints of $\widetilde{P}$. We refer the reader to Fig.~\ref{fig:treeflow3} for an illustration of this statement. For the purpose of proving it, consider some long maximal degree-2 path $\widetilde{P}$ of $R^2$ whose vertex set is contained in $V(C_{R^2,P,u})$.

By Lemma \ref{lem:distinctMidRegions}, we know that either $V(H)\setminus (V(C_{R^2,\widetilde{P},\widetilde{u}})\cup V(C_{R^2,\widetilde{P},\widetilde{v}})) \subseteq V(C_{R^2,P,u})$ or $V(H)\setminus (V(C_{R^2,\widetilde{P},\widetilde{u}})\cup V(C_{R^2,\widetilde{P},\widetilde{v}})) \subseteq V(C_{R^2,P,v})$. Moreover, by the definition of $\widetilde{P}^\star_{\widetilde{u}'_{\widetilde{P}},\widetilde{v}'_{\widetilde{P}}}$, its vertex set is contained in $V(H)\setminus (V(C_{R^2,\widetilde{P},\widetilde{u}})\cup V(C_{R^2,\widetilde{P},\widetilde{v}}))$. Thus, to conclude the proof, it remains to rule out the possibility that $V(H)\setminus (V(C_{R^2,\widetilde{P},\widetilde{u}})\cup V(C_{R^2,\widetilde{P},\widetilde{v}})) \subseteq V(C_{R^2,P,v})$. For this purpose, recall that we chose $\widetilde{P}$ such that $V(\widetilde{P})\subseteq V(C_{R^2,P,u})$, and that $V(\widetilde{P})\cap V(C_{R^2,P,u})\neq\emptyset$. Because $V(C_{R^2,P,u})\cap V(C_{R^2,P,v})=\emptyset$, we derive that the containment $V(H)\setminus (V(C_{R^2,\widetilde{P},\widetilde{u}})\cup V(C_{R^2,\widetilde{P},\widetilde{v}})) \subseteq V(C_{R^2,P,v})$ is indeed impossible.
\end{proof}

\subsection{Enumerating Parallel Edges with Respect to $R^3$}\label{sec:enumParallel}

Recall that $H$ is enriched with $4n+1$ parallel copies of each edge of the (standard) radial completion of $G$.
While the copies did not play a role in the construction of $R$, 
they will be important in how we relate a solution of the given instance of \pdp\ to a weak linkage in $H$.  We remind that for a pair of adjacent vertices $u,v\in V(H)$, we denoted the $4n+1$ parallel copies of edges between them by $e_{-2n},e_{-2n+1},\ldots,e_{-1},e_0,e_1,e_2,\ldots,e_{2n}$ where $e=\{u,v\}$, such that when the edges incident to $u$ (or $v$) are enumerated in cyclic order, the occurrences of $e_i$ and $e_{i+1}$ are consecutive (that is, $e_i$ appears immediately before $e_{i+1}$ or vice versa) for every $i\in\{-2n,-2n+1,\ldots,2n-1\}$, and $e_{-2n}$ and $e_{2n}$ are the outermost copies of $e$. We say that such an embedding is {\em valid}. Now, we further refine the embedding of $H_G$ (so that it remains valid)---notice that for each edge, there are two possible ways to order its copies in the embedding so that it satisfies the condition above. Here, we will specify for the edges of $R$ a particular choice among the two to embed their copies. Towards this, for a vertex $v \in V(R)$, let $\wh{E}_R(v) = \{ e \in E_H(v) \mid e \text{ is parallel to an edge } e' \in E(R) \}$. (The set $E_H(v)$ contains all edges in $E(H)$ incident to $v$.)

\begin{figure}
    \begin{center}
        \includegraphics[scale=0.8]{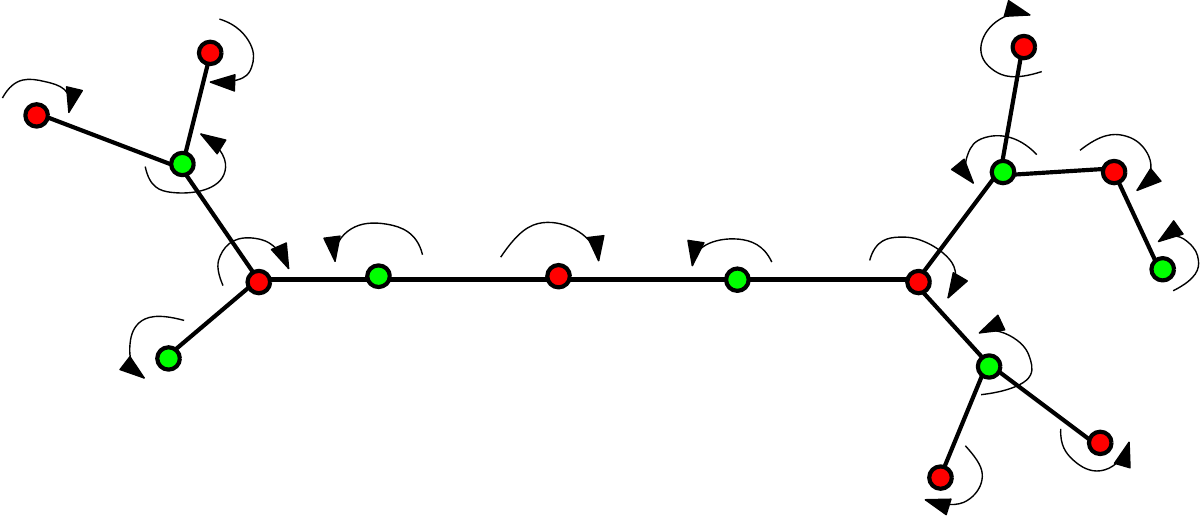}
        \caption{The clockwise / anti-clockwise enumeration of parallel edges with respect to $R$.}
        \label{fig:enumParallel}
    \end{center}
\end{figure}

For the definition of the desired embedding, we remind that any tree can be properly colored in two colors (that is, every vertex is assigned a color different than the colors of its neighbors), and that in such a coloring, for every vertex, all the neighbors of the vertex get the same color. We let ${\sf color}: V(R) \rightarrow \{ \text{red}, \text{green} \}$ be some such coloring of $R$. Then, we embed parallel copies such that for every $v \in V(R)$, the following conditions hold (see Fig.~\ref{fig:enumParallel}).
    \begin{itemize}
        \item If ${\sf color}(v) = \text{red}$, then when we enumerate $\wh{E}_R(v)$ in {\em clockwise} order, for every $e \in E_R(v)$, the $4n+1$ copies of $e$ are enumerated in this order: $e_{-2n}, e_{-2n+1}, \ldots, e_0, \ldots, e_{2n}$. We let $\order_v$ denote such an enumeration starting with an edge indexed $-2n$.
                
        \item If ${\sf color}(v) = \text{green}$, then when we enumerate $\wh{E}_R(v)$ in {\em counter-clockwise} order, for every $e \in E_R(v)$, the $4n+1$ copies of $e$ are enumerated in this order: $e_{-2n}, e_{-2n+1}, \ldots, e_0, \ldots, e_{2n}$. We let $\order_v$ denote such an enumeration starting with an edge indexed $-2n$.
    \end{itemize}

Let us observe that the above scheme is well defined. 
\begin{observation}\label{obs:enumParallelTime}
    Let $(G,S,T,g,k)$ be a good instance of \pdp\ with a backbone Steiner tree $R$. Then, there is a valid embedding of $H$ such that, for every $v\in V(R)$, the enumeration $\order_v$ is well defined with respect to some proper coloring ${\sf color}: V(R) \rightarrow \{ \text{red}, \text{green} \}$. Furthermore, such an embedding can be computed in time $\OO(n^2)$.
\end{observation}
\begin{proof}
    Since $e=\{u,v\} \in E(R)$, $u$ and $v$ get different colors under ${\sf color}$.
    Let us assume that ${\sf color}(u) = \text{red}$ and ${\sf color}(v) = \text{green}$. Then the parallel copies of $e$ are enumerated in clockwise order in $\order_v$ and anti-clockwise order in $\order_v$. Hence, they agree and by construction, it is a good enumeration. 
    Finally, to bound the time required to obtain such an embedding, observe that it cam be obtained by starting with any arbitrary embedding of $H$ and then renaming the edges. Since 
    the total number of edges in $E(H)$ (including parallel copies) is at most $\OO(n^2)$, this can be done in $\OO(n^2)$ time.
\end{proof}

From now on, we assume that $H$ is embedded in a way so that the enumerations $\order_v$ are well defined. We also remind that $R$ only contains the $0$-th copies of edges in $H$. Finally, we have the following observation.

\begin{observation}\label{obs:enumParallelEdges}
    Let $(G,S,T,g,k)$ be a good instance of \pdp\ with a backbone Steiner tree $R$. For every $v \in V(H)$, $\order_v$ is an enumeration of $\wh{E}_R(v)$ in either clockwise or counter-clockwise order around $v$ (with a fixed start). Further, for any pair $e,e' \in E_R(v)$ such that $e$ occurs before $e'$ in ${\sf order}_v$, the edges $e_0,e_1,\ldots, e_{2n}$ occur before $e'_0,e'_1,\ldots, e'_{2n}$. 
\end{observation}


\newcommand{\cqed}{\renewcommand\qedsymbol{$\diamond$}}

\section{Existence of a Solution  with Small Winding Number}
\label{sec:winding}

In this section we show that if the given instance admits a solution, then it admits a ``nice solution''.
The precise definition of nice will be in terms of ``winding number'' of the solution, which counts the number of times the solution spirals around the backbone steiner tree. Our goal is to show that there is a solution of small winding number.

\subsection{Rings and Winding Numbers}


Towards the definition of a ring, let us remind that $H$ is the triangulated plane multigraph obtained by introducing $4n+1$ parallel copies of each edge to the a radial completion of the input graph $G$. Hence, each face of $H$ is either a triangle or a 2-cycle.  

\begin{definition}[{\bf Ring}]\label{def:ring}
    Let $I_\tin$, $I_\tout$ be two disjoint cycles in $H$ such that the cycle $I_\tin$ is drawn in the strict interior of the cycle $I_\tout$. Then, $\ring(I_\tin, I_\tout)$ is the plane subgraph of $H$ induced by the set of vertices that are either in $V(I_\tin) \cup V(I_\tout)$ or drawn between $I_\tin$ and $I_\tout$
    (i.e. belong to the exterior of $I_\tin$ and the interior of $I_\tout$). 
\end{definition}

We call $I_\tin$ and $I_\tout$ are the \emph{inner} and \emph{outer interfaces} of $\ring(I_\tin, I_\tout)$. We also say that this ring is induced by $I_\tin$ and $I_\tout$.
Recall the notion of self-crossing walks defined in Section~\ref{sec:discreteHomotopy}. Unless stated otherwise, all walks considered here are \emph{not self-crossing}.
A walk $\alpha$ in $\ring(I_\tin,I_\tout)$ is {\em{traversing}} the ring if one of its endpoints lies in $I_\tin$ and the other lies in $I_\tout$.
A walk $\alpha$ is \emph{visiting} the ring if both its endpoints together lie in  either $I_\tin$ or in $I_\tout$; moreover $\alpha$ is an \emph{inner visitor} if both its endpoints lie in $I_\tin$, and otherwise it is an \emph{outer visitor}.

\begin{definition}[{\bf Orienting Walks}]\label{def:orientCurve}
Fix an arbitrary ordering of all vertices in $I_\tin$ and another one for all vertices in $I_\tout$.
Then for a walk $\alpha$ in $\ring(I_\tin, I_\tout)$ with endpoints in $V(I_\tin) \cup V(I_\tout)$, orient $\alpha$ from one endpoint to another as follows.
If $\alpha$ is a traversing walk, then orient it from its endpoint in $I_\tin$ to its endpoint in $I_\tout$.
If $\alpha$ is a visiting walk, then both its endpoints lie either in $I_\tin$ or in $I_\tout$; then, orient $\alpha$ from its smaller endpoint to its greater endpoint.
\end{definition}

Observe that if $\alpha$ is a traversing path in the ring, then the orientation of $\alpha$ also defines its \emph{left-side} and \emph{right-side}. These are required for the following definition. 

\begin{definition}[{\bf Winding Number of a Walk w.r.t. a Traversing Path}]
	Let $\alpha$ be an a walk in $\ring(I_\tin, I_\tout)$ with endpoints in $V(I_\tin) \cup V(I_\tout)$, and let $\beta$ be a traversing path in this ring, such that $\alpha$ and $\beta$ are edge disjoint. The \emph{winding number, $\wnorig(\alpha,\beta)$, of $\alpha$ with respect to $\beta$} is the signed number of crossings of $\alpha$ with respect to $\beta$. That is, while walking along $\alpha$ (according to the orientation in Definition~\ref{def:orientCurve}, for each intersection of $\alpha$ and $\beta$ record $+1$ if $\alpha$ crosses $\beta$ from left to right, $-1$ if $\alpha$ crosses $\beta$ from right to left, and $0$ if it $\alpha$ does not cross $\beta$. Then, the winding number $\wnorig(\alpha,\beta)$ is the sum of the recorded numbers.
\end{definition}

Observe that if $\alpha$ and $\beta$ are edge-disjoint traversing paths, then both $\wnorig(\alpha, \beta)$ and $\wnorig(\beta,\alpha)$ are well defined. We now state some well-known properties of the winding number. 
We sketch a proof of these properties in Appendix~\ref{sec:app:wn}, using homotopy.
\begin{proposition}\label{prop:wn-prop}
  Let $\alpha$, $\beta$ and $\gamma$ be three edge-disjoint paths traversing $\ring(I_\tin,I_\tout)$.  
  Then,
\begin{itemize}

\item[(i)] $\wnorig(\beta,\gamma)=-\wnorig(\gamma,\beta)$.

\item[(ii)] $\Big| \left| \wnorig(\alpha,\beta) - \wnorig(\alpha,\gamma) \right| - \left|\wnorig(\beta,\gamma) \right| \Big| \leq 1$.
\end{itemize}
\end{proposition}

We say that $\ring(I_\tin, I_\tout)$ is \emph{rooted} if it is equipped with some fixed path $\eta$ that is traversing it, called the \emph{reference path} of this ring. In a rooted ring $\ring(I_\tin,I_\tout)$, we measure all winding numbers with respect to $\eta$, hence we shall use the shorthand $\wnorig(\alpha)=\wnorig(\alpha,\eta)$ when $\eta$ is implicit or clear from context.
Here, we implicitly assume that the walk $\alpha$ is edge disjoint from $\eta$. This requirement will always be met by the following assumptions: $(i)$ $H$ is a plane multigraph where we have $4n+1$ parallel copies of every edge, and we assume that the reference path $\eta$ consists of only the $0$-th copy $e_0$;\; and $(ii)$~whenever we consider the winding number of a walk $\alpha$, it will edge-disjoint from the reference curve $\eta$ as it will not contain the $0$-th copy of any edge. (In particular, the walks of the (weak) linkages that we consider will always satisfy this property.)

Note that any visitor walk in $\ring(I_\tin,I_\tout)$ with both endpoints in $I_\tin$ is discretely homotopic to a segment of $I_\tin$, and similarly for $I_\tout$. Thus, we derive the following observation.
\begin{observation}\label{obs:vis_wn}
    Let $\alpha$ be a visitor in $\ring(I_\tin,I_\tout)$.
    Then, $|\wnorig(\alpha)| \leq 1$.
\end{observation}


Recall the notion of a weak linkage defined in Section~\ref{sec:discreteHomotopy}, which is a collection of edge-disjoint non-crossing walks. 
When we use the term \emph{weak linkage of order $k$ in $\ring(I_\tin, I_\tout)$}, we  refer to a weak linkage such that each walk has both endpoints in $V(I_{\tin})\cup V(I_{\tout})$. For brevity, we abuse the term `weak linkage' to mean a weak linkage in a ring when it is clear from context. 
Note that every walk in a weak linkage $\Pp$ is an inner visitor, or an outer visitor, or a traversing walk.
This partitions $\Pp$ into  $\Pp_\tin, \Pp_\tout, \Pp_\crs$. 
A weak linkage is {\em{traversing}} if it consists only of traversing walks.
Assuming that $\ring(I_\tin,I_\tout)$ is rooted, we define the {\em{winding number}} of a traversing weak linkage $\Pp$ as $\wnorig(\Pp)=\wnorig(P_1)$.
Recall that any two walks in a weak linkage are non-crossing.
Then as observed in~\cite[Observation 4.4]{DBLP:conf/focs/CyganMPP13},\footnote{This inequality also follows from the second property of Proposition~\ref{prop:wn-prop} by setting $\alpha$ to be the reference path, $\beta = P_1$ and $\gamma = P_i$ and noting that $\wnorig(\beta,\gamma) = 0$.}

$$|\wnorig(P_i)-\wnorig(\Pp)|\leq 1\qquad\textrm{for all }i=1,\ldots,k.$$

The above definition is extended to any weak linkage $\Pp$ in the ring as follows: if there is no walk in $\Pp$ that traverses the ring, then $\wnorig(\Pp) = 0$, otherwise $\wnorig(\Pp) = \wnorig(\Pp_\crs)$. 
Note that, two aligned weak linkages $\Pp$ and $\Qq$ in the ring may have different winding numbers (with respect to any reference path). Replacing a linkage $\Pp$ with an aligned linkage $\Qq$ having a ``small'' winding number will be the main focus of this section.

%
Lastly, we define a labeling of the edges based on the winding number of a walk (this relation is made explicit in the observation that follows).
\begin{definition}\label{def:spiralLabel}
    Let $(G,S,T,g,k)$ be a good instance of \pdp, and $H$ be the radial completion of $G$. Let $\alpha$ be a (not self-crossing) walk in $H$, and let $\beta$ be a path in $H$ such that $\alpha$ and $\beta$ are edge disjoint.
    Let us fix (arbitrary) orientations of $\alpha$ and $\beta$, and define the left and right side of the path $\beta$ with respect to its orientation.
    The \emph{labeling} $\lab_\beta^\alpha$ of each ordered pair of consecutive edges, $(e,e') \in E_H(\alpha) \times E_H(\alpha)$ by $\{-1, 0 , +1\}$ with respect to $\beta$, where $e$ occurs before $e'$ when traversing $\alpha$ according to its orientation is defined as follows.
    \begin{itemize}
        \item The pair $(e,e')$ is labeled $+1$ if $e$ is on the left of $\beta$ while $e'$ is on the right of $\beta$.
        
        \item else, $(e,e')$ is labeled $-1$ if 
        $e$ is on the right while $e'$ is on the left of $\beta$;
        
        \item otherwise $e$ and $e'$ are on the same side of $\beta$ and $(e,e')$ is labeled $0$.
    \end{itemize}   
\end{definition}
Note that in the above labeling only pairs of consecutive edges may get a non-zero label, depending on how they cross the reference path.
For the ease of notation, we extend the above labeling function to all ordered pairs of edges in $\alpha$ (including pairs of non-consecutive edges), by labeling them $0$. 
Then we have the following observation, when we restrict $\alpha$ to $\ring(I_\tin, I_\tout)$ and set $\beta$ to be the reference path of this ring.
\begin{observation}\label{obs:spiralLabel}
     Let $\alpha$ be a (not self-crossing) walk in $\ring(I_\tin, I_\tout)$ with reference path $\eta$. Then  $|\wnorig(\alpha, \eta)| = |\sum_{(e,e') \in E(\alpha) \times E(\alpha)} \lab_\eta^\alpha(e,e')|$.
\end{observation}

\subsection{Rerouting in a Ring}
We now address the question of rerouting a solution to reduce its winding number with respect to the backbone Steiner tree. As a solution is linkage in the graph $G$ (i.e. a collection of vertex disjoint paths), we first show how to reroute linkages within a ring. In the later subsections, we will apply this to reroute a solution in the entire plane graph.
We remark that from now onwards, our results are stated and proved only for linkages 
(rather than weak linkages). 
Further, 
define a \emph{linkage of order $k$ in a $\ring(I_\tin,I_\tout)$} as a collection of $k$ vertex-disjoint paths in $G$ such that each of these paths belongs to $\ring(I_\tin,I_\tout)$ and its endpoints belong to $V(I_\tin) \cup V(I_\tout)$. As before, we simply use the term `linkage' when the ring is clear from context.
We will use the following proposition proved by Cygan et al.~\cite{DBLP:conf/focs/CyganMPP13} using earlier results of Ding et al.~\cite{ding1992disjoint}. Its statement has been rephrased to be compatible with our notation.

\begin{proposition}[Lemma 4.8 in~\cite{DBLP:conf/focs/CyganMPP13}]\label{prop:ring-rerouting}
Let $\ring(I_\tin,I_\tout)$ be a rooted ring in $H$ and let $\Pp$ and $\Qq$ be two traversing linkages of the same order in this ring. Then, there exists a traversing linkage $\Pp'$ in this ring that is aligned with $\Pp$ and such that $|\wnorig(\Pp')-\wnorig(\Qq)|\leq 6$.
\end{proposition}


The formulation of~\cite{DBLP:conf/focs/CyganMPP13} concerns directed paths in directed graphs and assumes a fixed pattern of in/out orientations of paths that is shared by the linkages $\Pp,\Qq$ and $\Pp'$.
The undirected case (as expressed above) can be reduced to the directed one by replacing every undirected edge in the graph by two oppositely-oriented arcs with same endpoints, 
and asking for any orientation pattern (say, all paths should go from $I_\tin$ to $I_\tout$).
Moreover, the setting itself is somewhat more general, where rings and reference paths are defined by curves and (general) homotopy.

\paragraph*{Rings with Concentric Cycles.}


Let $\Cc = (C_1, C_2, \ldots, C_p)$ concentric sequence of cycles in $\ring(I_\tin, I_\tout)$ 
(then, $C_{i}$ is in the strict interior of $C_{i+1}$ for $i \in \{1, 2, \ldots p-1\}$). If $I_\tin$ is in the strict interior of $C_1$ and $C_p$ is in the strict interior of $I_\tout$, then we cay that $\Cc$ is \emph{encircling}.
An encircling concentric sequence $\Cc$ in $\ring(I_\tin, I_\tout)$ is {\em{tight}} if every $C \in \Cc$ is a cycle in $G$, and there exists a path $\eta$ in $H$ traversing $\ring(I_\tin,I_\tout)$ such that the set of internal vertices of $\eta$ contain exactly $|\Cc|$ vertices of $V(G)$, one on each each cycle in $\Cc$. Let us fix one such encircling tight sequence in the ring $\ring(I_\tin, I_\tout)$ along with the path $\eta$ witnessing the tightness. 
Then, we set the path $\eta$ as the reference path of the ring. Here, we assume w.l.o.g. that $\eta$ contains only the $0$-th copy of each of the edges comprising it. 
Any paths or linkages that we subsequently consider will not use the $0$-th copy of any edge, and hence their winding numbers (with respect to $\eta$) will be well-defined. 
This is because that they arise from $G$,
and when we consider them in $H$, we choose a `non-$0$-th' copy out of the $4n+1$ copies of any (required) edge.

A linkage $\Pp$ in $\ring(I_\tin,I_\tout)$ is {\em{minimal}} with respect to $\Cc$ if among the linkages aligned with $\Pp$, 
it minimizes the total number of edges traversed that do not lie on the cycles of $\Cc$.
The following proposition is essentially Lemma~3.7 of~\cite{DBLP:conf/focs/CyganMPP13}.

\begin{proposition}\label{lem:shallow-visitors}
    Let $G$ be a plane graph, and with radial completion $H$.
    Let $\ring(I_\tin,I_\tout)$ be a rooted ring in $H$.  Suppose $|I_\tin|,|I_\tout|\leq \ell$, for some integer $\ell$.
    Further, let $\Cc=(C_1,\ldots,C_p)$ be an encircling tight concentric sequence of cycles in $\ring(I_\tin,I_\tout)$.
    Finally, let $\Pp$ be a linkage in $\ring(I_\tin, I_\tout)$ that is minimal with respect to $\Cc$.
    Then, every inner visitor of $\Pp$ intersects less than $10\ell$ of the first cycles in the sequence $(C_1,\ldots,C_p)$, while every outer visitor of $\Pp$ intersects less than $10\ell$ of the last cycles in this sequence.
\end{proposition}
A proof of this proposition can be obtained by first ordering the collection of inner and outer visitors by their `distance' from the inner and outer interfaces, respectively, and the `containment' relation between the cycles formed by them with the interfaces. This gives a partial order on the set of inner visitors and the set of outer visitors. Then if the proposition does not hold for $\Pp$, then the above ordering and containment relation can be used to reroute these paths along a suitable cycle.
This will contradict the minimality of $\Pp$, since the rerouted linkage is aligned with it but uses strictly fewer edges outside of $\Cc$.
The main result of this section can be now formulated as follows. (Its formulation and proof idea are  based on Lemma~8.31 and Theorem~6.45 of~\cite{DBLP:conf/focs/CyganMPP13}.)

%
\begin{lemma}\label{lemma:winding}
    Let $G$ be a plane graph with radial completion $H$.
    Let $\ring(I_\tin,I_\tout)$ be a ring in $H$. 
    Suppose that $|I_\tin|,|I_\tout|\leq \ell$ for some integer $\ell$, and that in $\ring(I_\tin,I_\tout)$ there is an encircling tight concentric sequence of cycles $\Cc$ of size larger than $40\ell$. Let $\eta$ be a traversing path in the ring witnessing the tightness of $\Cc$, and fix $\eta$ as the reference path.
    Finally, let $\Pp=\Pp_{\crs}\uplus \Pp_{\vis}$ be a linkage in $G$, where $\Pp_{\crs}$ is a traversing linkage comprising the paths of $\Pp$ traversing $\ring(I_\tin,I_\tout)$, while $\Pp_{\vis}=\Pp\setminus \Pp_{\crs}$ consists of the paths whose both endpoints lie in either $V(I_\tin)$ or $V(I_\tout)$.
    Further, suppose that $\Pp$ is minimal with respect to $\Cc$.
    Then, for every traversing linkage $\Qq$ in $G$ that is minimal with respect to $\Cc$ such that every path in $\Qq$ is disjoint from $\eta$ and $|\Qq| \geq |\Pp_{\crs}|$, there is a traversing linkage $\Pp_{\crs}'$ in $G$ such that
    \begin{enumerate}[(a)]
        \item $\Pp_{\crs}'$ is aligned with $\Pp_{\crs}$,
        \item the paths of $\Pp_{\crs}'$ are disjoint from the paths of $\Pp_{\vis}$, and
        \item $|\wnorig(\Pp_{\crs}')-\wnorig(\Qq)|\leq 60\ell+6$.
    \end{enumerate}
\end{lemma}
\begin{proof}
    Let $\Cc=(C_1,\ldots,C_p)$, 
    where $p >40\ell$.
    Recall that $\Cc$ is a collection of cycles in $G$, and the path $\eta$ that witnesses the tightness of $\Cc$ contains $|\Cc|$ vertices of $V(G)$, one on each cycle of $\Cc$.  Let $v_i$ denote the vertex where $\eta$ intersects the cycle $C_i \in \Cc$ for all $i \in \{1,2,\ldots, p\}$. 
    %
    Since $\Pp$ is minimal with respect to $\Cc$, Proposition~\ref{lem:shallow-visitors} implies that the paths in $\Pp_{\vis}$ do not intersect any of the cycles $C_{10\ell},C_{10\ell+1},\ldots,C_{p-10\ell+1}$ (note that since $p>40\ell$, this sequence of cycles is non-empty).     
    Call a vertex $x \in V(\ring(I_\tin, I_\tout))$ in the ring \emph{non-separated} if there exists a path from $x$ to $C_{10 \ell}$ whose set of internal vertices is disjoint from $\bigcup_{P \in \Pp_\vis} V(P)$. Otherwise, we say that the vertex $x$ is \emph{separated}.
    Observe that every path in $\Pp_\crs$ is disjoint from the paths in $\Pp_\vis$ and intersects $C_{10\ell}$, hence all vertices 
    on the paths of $\Pp_{\crs}$ 
    are non-separated.
    Let $X$ denote the set of all non-separated vertices in the ring, and consider the graph $H[X]$. Observe that $H[X]$ is an induced subgraph of $\ring(I_\tin, I_\tout)$, since it is obtained from $\ring(I_\tin, I_\tout)$ by deleting the separated vertices.  
    Further, observe that $H[X]$ is a ring of $H$. Indeed, the inner interface of $H[X]$ is the cycle $\wh{I}_\tin$ obtained as follows: let $\Pp_\tin$ be the set of inner visitors in $\Pp$; then, $\wh{I}_\tin$ is the outer face of the plane graph $H[ V(I_\tin) \cup\,\, \bigcup_{P \in \Pp_\tin} V(P)]$. It is easy to verify that all vertices on $\wh{I}_\tin$ are non-separated, and any vertex of $\ring(I_\tin,I_\tout)$ that lies in the strict interior of this cycle is separated. We then symmetrically  obtain the outer interface $\wh{I}_\tout$ of $H[X]$ from the set $\Pp_\tout$ of outer visitors of $\Pp$. Then, $H[X] = \ring(\wh{I}_\tin, \wh{I}_\tout)$.
    Here, $\wh{I}_\tin$ is composed alternately of subpaths of $I_\tin$ and inner visitors from $\Pp_\vis$, and symmetrically for $\wh{I}_\tout$.
    
    Note that the paths of $\Pp_\crs$ are completely contained in $\ring(\wh{I}_\tin,\wh{I}_\tout)$ and they all traverse this ring.
    Thus, $\Pp_\crs$ can be regarded also as a traversing linkage in $\ring(\wh{I}_\tin,\wh{I}_\tout)$. While $\Pp_\crs$ may have a different winding number in $\ring(\wh{I}_\tin,\wh{I}_\tout)$ than in $\ring(I_\tin,I_\tout)$, the difference is ``small'' as we show below.
    (Note that the two winding numbers in the following claim are computed in two different rings.)
    
    \begin{claim}\label{cl:P-trim}
        Let $P$ be a path in $G$ that is disjoint from all paths in $\Pp_\vis$, such that $P$ belongs to $\ring(I_\tin, I_\tout)$ and traverses it.
        Then, $P$ also belongs to $\ring(\wh{I}_\tin, \wh{I}_\tout)$, and $|\wnorig(P,\eta)-\wnorig(P,\wh{\eta})|\leq 20\ell$. \footnote{Here $\wh{\eta}$ is the reference path of $\ring(\wh{I}_\tin, \wh{I}_\tout)$, which is a subpath of $\eta$ in this ring.}
    \end{claim}
    \begin{proof}
        Since $P$ traverses $\ring(I_\tin, I_\tout)$, it must intersect the cycle $C_{10\ell}$. Therefore, as $P$ is disjoint from $\Pp_{\vis}$, all vertices in $V(P)$ are non-separated. Hence, $P$ is present in $\ring(\wh{I}_\tin, \wh{I}_\tout)$.
        Next, observe that there are at most $20\ell$ vertices of $G$ that are visited by $\eta$ but not visited by $\wh{\eta}$; indeed, these are vertices in the intersection of $\eta$ with $I_\tin, I_\tout$ and the first and last $10\ell -1$ cycles of $\Cc$.
        It follows that any path 
        in $G$ has at most $20\ell$ more crossings with $\eta$ than with $\wh{\eta}$.
        Since each such crossing contributes $+1$ or $-1$ to the winding number of $P$ with respect to $\eta$, the winding numbers of $P$ with respect to $\eta$ and $\wh{\eta}$ differ by at most $20\ell$.
        \cqed\end{proof}
    
    We now turn our attention to the linkage $\Qq$. 
    In essence, our goal is show that every path in $\Qq$ can be ``trimmed'' to a path traversing $\ring(\wh{I}_\tin,\wh{I}_\tout)$ such that their winding numbers are not significantly different.
    First, however, we prove that the paths in $\Qq$ cannot ``oscillate'' too much in $\ring(I_\tin,I_\tout)$, based on the supposition that $\Qq$ is minimal with respect to $C$.
    
    \begin{claim}\label{cl:no-oscillators}
        Let $Q\in \Qq$, and let $u\in V(Q)$ such that it also lies on an inner visitor from $\Pp_\vis$. 
        Then, the prefix of $Q$ between its endpoint on $I_\tin$ and $u$ does not intersect the cycle $C_{20\ell}$.
    \end{claim}

    \begin{figure}
        \begin{center}
            \includegraphics[scale=0.4]{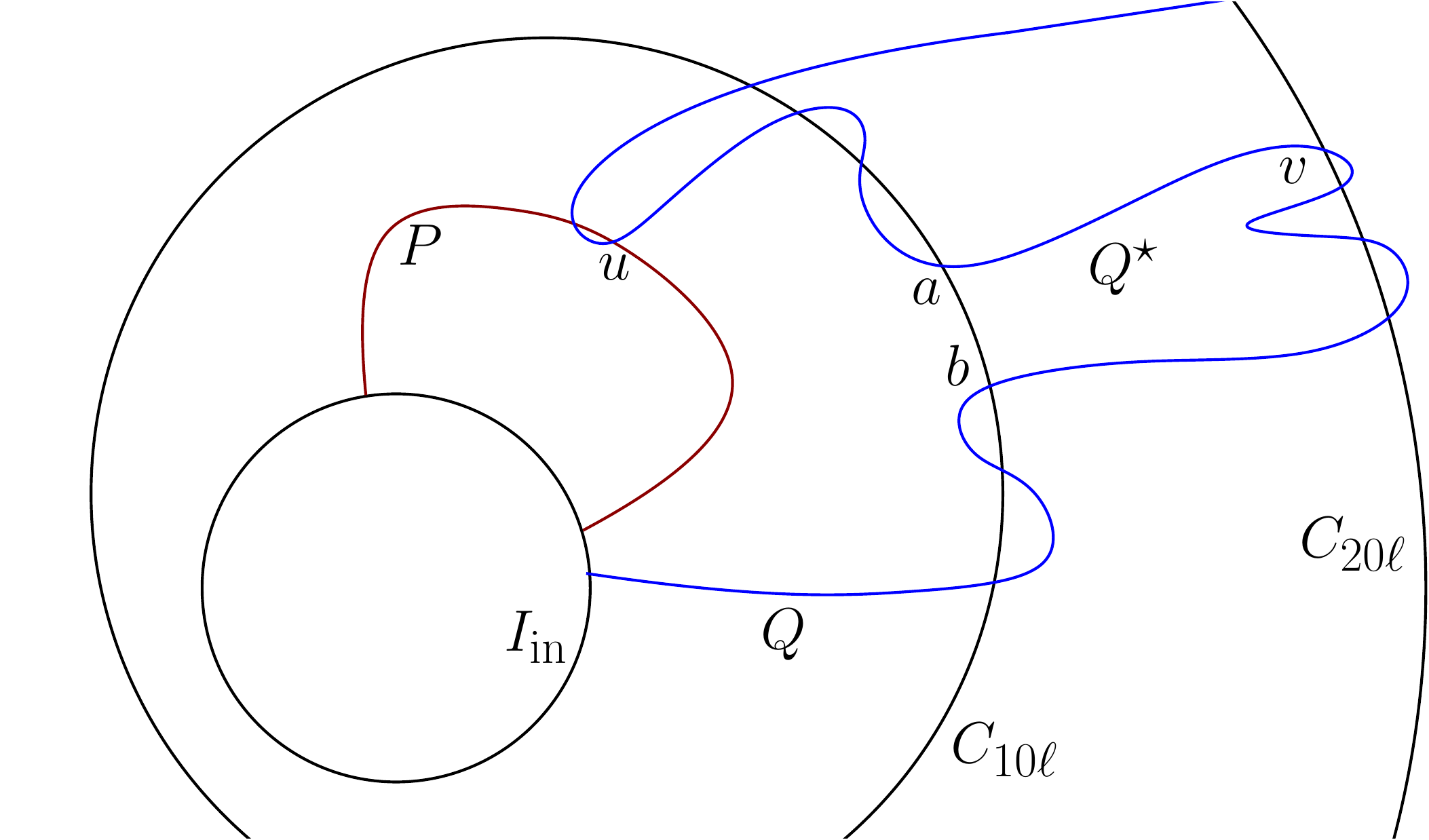}
            \caption{Illustration of Claim~\ref{cl:no-oscillators}.}
            \label{fig:claim-osc}
        \end{center}
    \end{figure}
    \begin{proof}
        Suppose, for the sake of contradiction, that the considered prefix contains some vertex $v$ that lies on $C_{20\ell}$.
        Since $u$ lies on an inner visitor $P \in \Pp_\vis$ and Proposition~\ref{lem:shallow-visitors} states that an inner visitor cannot intersect $C_{10\ell}$, 
        we infer that on the infix of $Q$ between $v$ and $u$ there exists a vertex 
        that lies on the intersection of $Q$ and $C_{10\ell}$. Let $a$ be the first such vertex.
        Similarly, on the prefix of $Q$ from its endpoint on $I_\tin$ to $v$ there exists a vertex 
        that lies on the intersection of $Q$ and $C_{10\ell}$.
        Let $b$ be the last such vertex.
        Then the whole infix of $Q$ between $a$ and $b$ does not intersect $C_{10\ell}$ internally (see Fig.~\ref{fig:claim-osc}), and hence, apart from endpoints,
        completely lies in the exterior of $C_{10\ell}$.
        Call this infix $Q^\star$.
        
        Now consider $\ring(C_{10\ell},I_\tout)$, the ring induced by $I_\tout$ and $C_{10\ell}$. 
        Moreover, consider the graph $G'$ obtained from $G$ by removing all vertices that are not 
        in $\ring(C_{10\ell},I_\tout)$ and edges 
        that are not in the strict interior of $\ring(C_{10\ell},I_\tout)$; in particular, the edges of $C_{10\ell}$ are removed, but the vertices are not. 
        Note that $G'$ is 
        a subgraph of $\ring(C_{10\ell},I_\tout)$.
        Finally, let $\Cc'=\Cc\setminus \{C_1,C_2,\ldots,C_{10\ell}\}$; then, $\Cc'$ is an encircling tight sequence of concentric cycles in $\ring(C_{10\ell}, I_\tout)$. 
        
        Let $\Qq'$ be the linkage in $G$ obtained by restricting paths of $\Qq$ to $G'$. 
        Here, a path in $\Qq$ may break into several paths in $\Qq'$ (that are its maximal subpaths contained in $G'$).
        Since $\Qq$ is minimal with respect to $\Cc$, it follows that $\Qq'$ is minimal with respect to $\Cc'$.
        Now, observe that $Q^\star$ belongs to $\Qq'$, hence it is an inner visitor of $\ring(C_{10\ell},I_\tout)$. 
        However, $Q^\star$ intersects the first $10\ell+1$ concentric cycles $C_{10\ell+1},\ldots,C_{20\ell}$ in the family $\Cc'$, which contradicts Proposition~\ref{lem:shallow-visitors}.
        \cqed\end{proof}
    
    Clearly, an analogous claim holds for outer visitors and the cycle $C_{p-20\ell+1}$.
    We now proceed to our main claim about the restriction of $\Qq$ to $\ring(\wh{I}_\tin,\wh{I}_\tout)$.
    
    \begin{claim}\label{cl:Q-trim}
        For every path $Q\in \Qq$, there exists a subpath $\wh{Q}$ of $Q$ that traverses $\ring(\wh{I}_\tin,\wh{I}_\tout)$ and such that $|\wnorig(\wh{Q},\wh{\eta})-\wnorig(Q,\eta)|\leq 40\ell$.
    \end{claim}
    \begin{proof}
        We think of $Q$ as oriented from its endpoint on $I_\tin$ to its endpoint on $I_\tout$.
        Let $a$ be the last vertex on $Q$ that lies on $\wh{I}_\tin$ and $b$ be the first vertex on $Q$ that lies on $\wh{I}_\tout$.
        Further, let $Q_\tin$ be the prefix of $Q$ from its start to $a$, and $Q_\tout$ be the suffix of $Q$ from $b$ to its end.
        By Claim~\ref{cl:no-oscillators}, $Q_\tin$ is entirely contained in the ring induced by $I_\tin$ and $C_{20\ell}$.
        By the claim analogous to Claim~\ref{cl:no-oscillators}, $Q_\tout$ is entirely contained in the ring induced by $I_\tout$ and $C_{p-20\ell+1}$.
        Since $p>40\ell$, it follows that $Q_\tin$ and $Q_\tout$ are disjoint, and in particular $a$ appears before $b$ on $Q$.
        Let $\wh{Q}$ be the infix of $Q$ between $a$ and $b$. 
        Then, $\wh{Q}$ is a path in $\ring(\wh{I}_\tin,\wh{I}_\tout)$ that traverses this ring, so it suffices to check that $|\wnorig(\wh{Q},\wh{\eta})-\wnorig(Q,\eta)|\leq 40\ell$.
        
        Observe that every crossing of $Q$ and $\eta$ that is not a crossing of $\wh{Q}$ and $\wh{\eta}$ has to occur on either $Q_\tin$ or $Q_\tout$.
        However, $Q_\tin$ and $Q_\tout$ can have at most $40\ell$ vertices in common with $\eta$, because these must be among the intersections of $\eta$ with cycles $C_1,\ldots,C_{20\ell},C_{p-20\ell+1},\ldots,C_p$,
        of which there are $40\ell$. Each such crossing can contribute $+1$ or $-1$ to the difference between the winding numbers $\wnorig(\wh{Q},\wh{\eta})$ and $\wnorig(Q,\eta)$,
        hence the difference between these winding numbers is at most $40\ell$.
        \cqed\end{proof}
    
    For every path $Q\in \Qq$, fix the path $\wh{Q}$ provided by Claim~\ref{cl:Q-trim}, and let $\wh{\Qq} \subseteq \{\wh{Q} \mid Q\in \Qq\}$ such that $|\wh{\Qq}| = |\Pp_\crs|$.
    Then, $\Qq'$ is a traversing linkage in $\ring(\wh{I}_\tin,\wh{I}_\tout)$. 
    Apply Proposition~\ref{prop:ring-rerouting} to the linkages $\Pp_\crs$ and $\wh{\Qq}$ in $\ring(\wh{I}_\tin,\wh{I}_\tout)$, 
    yielding a linkage $\Pp_\crs'$ that is aligned from $\Pp_\crs$ and such that 
    $$|\wnorig(\Pp_\crs',\wh{\eta})-\wnorig(\wh{\Qq},\wh{\eta})|\leq 6.$$
    Clearly, by construction we have that the paths of $\Pp_\crs'$ are disjoint with the paths of $\Pp_\vis$. Furthermore, the paths in $\Pp_{\crs}'$ traverse $\ring(I_\tin, I_\tout)$ since they are aligned with $\Pp_\crs$ (i.e. they have the same endpoints).
    Finally, by Claim~\ref{cl:Q-trim} we have
    $$|\wnorig(\wh{\Qq},\wh{\eta})-\wnorig(\Qq,\eta)|\leq 40\ell,$$
    and by Claim~\ref{cl:P-trim} (applied to paths in $\Pp'_\crs$) we have
    $$|\wnorig(\Pp_\crs',\eta)-\wnorig(\Pp_\crs',\wh{\eta})|\leq 20\ell.$$
    By the above we conclude that
    $$|\wnorig(\Pp_\crs',\eta)-\wnorig(\Qq,\eta)|\leq 60\ell+6,$$
    which completes the proof.
\end{proof}

\subsection{Rings of the Backbone Steiner Tree}
\label{sec:STrings}

Based on results in previous subsections, we proceed to show that if the given instance admits a solution, then it also admits a solution of small winding number. 
Recall the backbone Steiner tree $R$ constructed in Section~\ref{sec:steiner}.
Let $P = \pathT_{R}(u,v)$ be a long maximal degree-2 path in $R$, where $u,v \in V_{=1}(R) \cup V_{\geq 3}(R)$, and assume without loss of generality (under the supposition that we are given a {\sf Yes} instance) that the subtree of $R - V(P) - \{u,v\})$ containing $v$ also contains the terminal $t^\star \in T$ lying on the outer face of $H$. Recall the (minimal) separators $S_u = \Sep_R(P,u)$ and $S_v = \Sep_R(P,v)$ in $H$. 
Hence $H[S_u]$ and $H[S_v]$ form two cycles in $H$, and $H[S_u]$ is contained in the strict interior of $H[S_v]$.
Further, recall that $|S_u|,|S_v| \leq \alpha_{\rm sep}(k)$.
Consider the ring \emph{induced} by $H[S_u]$ and $H[S_v]$, i.e. $\ring(S_u,S_v) := \ring(H[S_u], H[S_v])$. 
Let $V(S_u,S_v)$ denote the set of all vertices (in $V(H)$) that lie in this ring, including those in $S_u$ and $S_v$. Then, $\ring(S_u,S_v) = H[V(S_u,S_v)]$.
Note that by definition it contains $\Sep_R(P,u)$ and $\Sep_R(P,v)$.
Let $G_{u,v}$ denote the restriction of this graph to $G$, i.e. $G_{u,v} = G[V(G) \cap V(S_u,S_v)]$.
Additionally, recall that there are two distinct vertices $u'$ and $v'$ in $P$ that lie in $S_u$ and $S_v$, respectively, such that $P = \pathT_R(u,u') {-} \pathT_R(u',v')  {-} (v',v)$ (by Lemma~\ref{lem:threeParts} and Definition~\ref{lem:translateRtoR*}).
Lastly, we remind that $A^\star_{R,P,u}$ and $A^\star_{R,P,v}$ are the two components of $R - V(\pathT_R(u',v') - \{u',v'\})$ that contain $u$ and $v$, respectively (Definition~\ref{lem:translateRtoR*}).
Then, the following observation is immediate.

\begin{observation}
        %
        %
        The path $\pathT_R(u',v')$ is contained in $\ring(S_u,S_v)$ and it traverses $\ring(S_u,S_v)$ from $u' \in S_u$ to $v' \in S_v$.
        %
        Moreover, $A^\star_{R,P,u} - \{u'\}$ is contained in the bounded region of $\Rtwo - H[S_u]$, and similarly $A^\star_{R,P,v} - \{v'\}$ is contained in the unbounded region of $\Rtwo - H[S_v]$.
\end{observation}

Now, recall the encircling tight sequence of concentric cycles $\Cc(u,v)$ and the path witnessing its tightness, which we denote by $\eta(u,v)$ (constructed in Lemma~\ref{lemma:CC_cons}). 
We assume that $\eta(u,v)$ consists of only the $0$-th copies of the edges comprising it, and fix $\eta(u,v)$ as the reference path of the ring $\ring(S_u,S_v)$.
Moreover, observe that the subpath $\pathT_R(u',v')$ of $R$ also traverses $\ring(S_u,S_v)$. 
We assume w.l.o.g. that $\pathT_R(u',v')$ consists of only the $0$-th copy of the edges comprising it. This will allow us to later define winding numbers with respect to $\pathT_{R}(u',v')$ in $\ring(S_u,S_v)$. 
Note that $\eta(u,v)$ and $\pathT_R(u',v')$ may be different paths with common edges, and further they may not be paths in $G$.

Let us first argue that we can consider the rings corresponding to each of the long degree-2 paths in $R$ independently.
To this end, consider another long maximal degree-2 path in $R$ different from $P$,
denoted by $\pathT_{R}(\wh{u},\wh{v})$, where $\wh{u},\wh{v} \in V_{=1}(R) \cup V_{\geq 3}(R)$ and $t^\star$ lies in the subtree containing $\wh{v}$ in $R - V(\pathT_{R}(\wh{u},\wh{v}) - \{\wh{u},\wh{v}\})$. Let $S_{\wh{u}} = \Sep_R(\pathT_{R}(\wh{u},\wh{v}),\wh{u})$ and $S_{\wh{v}} = \Sep_R(\pathT_{R}(\wh{u},\wh{v}),\wh{v})$. 
The $\ring(S_{\wh{u}},S_{\wh{v}})$ is symmetric to $\ring(S_u,S_v)$ above.
The following corollary follows from Lemma~\ref{lem:distinctMidRegions}. 

\begin{corollary}\label{cor:rings-pairwise-disjoint}
    The rings $\ring(S_u,S_v)$ and $\ring(S_{\wh{u}}, S_{\wh{v}})$ have no common vertices.
\end{corollary}

Let $t$ be the number of pairs of near vertices in $V_{=1}(R) \cup V_{\geq 3}(R)$ such that the path between them in $R$ is long. As above, we have a ring corresponding to each of these paths. Then, we partition the plane graph $H$ into $t+1$ regions, one for each of the $t$ rings of the long maximal degree-2 path in $R$, and the remainder of $H$ that is not contained in any of these rings.
\begin{lemma}\label{lemma:R-outside-rings}
    Let $\{u_1,v_1\}, \{u_2,v_2\}, \ldots, \{u_t,v_t\}$ denote pairs of near vertices in $V_{=1}(R) \cup V_{\geq 3}(R)$ such that $\pathT_R(u_i,v_i)$ is a long maximal degree-2 path in $R$ for all $i \in \{1,2,\ldots t\}$. Then, the corresponding rings $\ring(S_{u_1}, S_{v_1}), \ring(S_{u_2}, S_{v_2}), \ldots, \ring(S_{u_t},S_{v_t})$ are pairwise disjoint. Further, the number of vertices of $R$ lying outside this collection of rings, $\big | V(R) \setminus \bigcup_{i=1}^t V(S_{u_i}, S_{v_i}) \big|$, is upper bounded by 
    $\alpha_{\rm nonRing}(k) = 10^4 k \cdot 2^{ck}$.
\end{lemma}
\begin{proof}
    The first statement follows from Corollary~\ref{cor:rings-pairwise-disjoint} applied to each pair of rings.
    For the second statement, consider the vertices in $V(R) \setminus \bigcup_{i=1}^t V(S_{u_i}, S_{v_i})$.
    Each such vertex either belongs to a short degree-2 path between a pair of near vertices in $V_{=1}(R) \cup V_{\geq 3}(R)$, or it is a vertex that lies on a maximal degree-2 path of $R$ at distance at most $\alpha_{\rm pat}(k) = 100 \cdot 2^{ck}$ in $R$ from $V_{=1}(R) \cup V_{\geq 3}(R)$.
    Recall that there are fewer than $4k$ vertices in $V_{=1}(R) \cup V_{\geq 3}(R)$ by Observation~\ref{obs:leaIntSteiner}, and thus at most $4k-2$ maximal degree-2 paths in $R$.
    Therefore, the number of vertices in $V_{=1}(R) \cup V_{\geq 3}(R)$ is upper bounded by $(4k -2) \cdot \max\{ 2\alpha_{\mathrm{pat}}(k), \alpha_{\mathrm{long}}(k) \}\;\leq 10^4 k \cdot 2^{ck}$.
\end{proof}

\subsection{Solutions with Small Winding Number}
Having established all the required definitions in the previous subsections, we are  ready to exhibit a solution with a small winding number.
Consider a ring, say $\ring(S_u, S_v)$ with an encircling tight concentric sequence of cycles $\Cc(u,v)$ and reference path $\eta(u,v)$.
Consider a linkage $\Pp$ in $G$, and let $\Pp(u,v) = \Pp[V(S_u,S_v) \cap V(\Pp)]$. 
Observe that $\Pp(u,v)$ is a linkage in $\ring(S_u,S_v)$, and its endpoints lie in $(S_u \cup S_v) \cap V(G)$. 
Therefore, $\Pp(u,v)$ is a flow from $S_u \cap V(G)$ to $S_v \cap V(G)$ in $G_{u,v}$.
Assume w.l.o.g. that the paths in $\Pp(u,v)$ use the $1$-st copy of each edge in $H$.

\begin{definition}
    Let $u,v \in V_{=1}(R) \cup V_{\geq 3}(R)$ be a pair of near vertices such that $\pathT_R(u,v)$ is long maximal degree-2 path in $R$. Let $\Pp$ be a path system in $G$.
    Then , \emph{winding number of $\cal P$ in $\ring(S_u,S_v)$} is defined as $\wn(\Pp, \ring(S_u,S_v)) = \max_{P \in \Pp(u,v)}\left|\wnorig(P, \pathT_R(u',v'))\right|$ 
\end{definition}

We remark that the notation above differs from our earlier use of $\wnorig(\cdot)$ in the choice of the reference path in $\ring(S_u,S_v)$. For $\wnorig(\cdot)$, the path $\eta(u,v)$ was the reference path, whereas for $\wn(\cdot, \ring(S_u,S_v))$, we choose $\pathT_R(u',v')$ as the reference path.

We will now apply Lemma~\ref{lemma:winding} in each ring of the form $\ring(S_{u_i}, S_{v_i})$ to obtain a solution of bounded winding number in that ring. Towards this, 
fix one pair $(u,v) := (u_i,v_i)$ for some $1 \leq i \leq t$.
We argue that $\Cc(u,v)$ and $\flow_{R}(u,v)$ are suitable for the roles of $\Cc$ and $\cal Q$ in the premise of Lemma~\ref{lemma:winding}. 
Recall that $\flow_{R}(u,v)$ is a collection of vertex disjoint paths in the subgraph $G_{u_i,v_i}$ of $G$, and assume w.l.o.g.~that $\flow_{R}(u,v)$ uses only $2$-nd copies of edges in $H$.

\begin{observation}\label{obs:pre-rerouting}
    Let $(G,S,T,g,k)$ be a good {\sf Yes}-instance of \pdp. Let $R$ be a backbone Steiner tree. Let $u,v$ be a pair of near vertices in $V_{=1}(R) \cup V_{\geq 3}(R)$ such that $\pathT_R(u,v)$ is a long degree-2 path in $R$. Then, the following hold.
\begin{itemize}
    \item[(i)] $\Cc(u,v)$ is a encircling tight sequence of concentric cycles in  $\ring(S_u,S_v)$ where each cycle lies in $G_{u,v}$, $S_u$ is in its strict interior and $S_v$ is in its strict exterior. Further, $|\Cc(u,v)| \geq 40 \alpha_{\rm sep}(k)$ cycles.
    
    \item[(ii)] $\flow_{R}(u,v)$ is a maximum flow between $S_u \cap V(G)$ and $S_v \cap V(G)$ in $G_{u,v}$ that is minimal with respect to $\Cc(u,v)$.
    
    \item[(iii)] Each path $Q \in  \flow_R(u,v)$ traverses $\ring(S_u,S_v)$, and $|\wnorig(\pathT_R(u',v'), Q)| \leq 1$.
\end{itemize}
\end{observation}
\begin{proof}
    The first statement directly follows from the construction of $\Cc(u,v)$ (see Lemma~\ref{lemma:CC_cons}).
    Similarly, the second statement directly follows from the construction of $\flow_R(u,v)$ using $\Cc(u,v)$ (see Observation~\ref{obs:disjPTime}). 
    Additionally the construction implies that each path $Q \in \flow_{R}(u,v)$ traverses $\ring(S_u,S_v)$.
    For the second part of the third statement,
    first note any $Q \in \flow_R(u,v)$ and $\pathT_R(u',v')$ are two edge-disjoint traversing paths in  $\ring(S_u,S_v)$. 
    Hence, $\wnorig(\pathT_R(u',v'), P)$ is well defined.
    Since for any $Q \in \flow_R(u,v)$, there are at most two edges in $\pathT_{R}(u',v')$ with only one endpoint in $V(Q)$ (by Corollary~\ref{cor:pathThroughFlow}, the absolute value of the signed sum of crossing between these two paths is upper-bounded by $1$, i.e. $|\wnorig(Q, \pathT_R(u',v'))| \leq 1$.
\end{proof}

Finally we are ready to prove that the existence of a solution of small winding number.
\begin{lemma}\label{lemma:nice-solution}
    Let $(G,S,T,g,k)$ be a good {\sf Yes}-instance of \pdp. Let $R$ be a backbone Steiner tree. 
    Let $\{u_1,v_1\}, \{u_2,v_2\}, \allowbreak \ldots, \{u_t,v_t\}$ 
    be the pairs of near vertices in $V_{=1}(R) \cup V_{\geq 3}(R)$ such that $\pathT_R(u_i,v_i)$ is a long maximal degree-2 path in $R$ for all $i \in \{1,2,\ldots t\}$. Let $\ring(S_{u_1}, S_{v_1}), \ring(S_{u_2}, S_{v_2}), \allowbreak \ldots, \ring(S_{u_t},S_{v_t})$ be the corresponding rings.
    Then, there is a solution $\Pp^\star$ to $(G,S,T,g,k)$ such that $|\wn(\Pp^\star, \ring(S_{u_i},S_{v_i}))|  \leq \alpha_{\rm winding}(k)$ for all $i \in \{1,2,\ldots t\}$, where  $\alpha_{\rm winding}(k) = 60 \alpha_{\rm sep}(k) + 11 < 300 \cdot 2^{ck}$.
\end{lemma}
\begin{proof}
    Consider a solution $\Pp$ to $(G,S,T,g,k)$.
    Fix a pair $(u,v) := (u_i,v_i)$ 
    for some $i \in \{1,2, \ldots t\}$.
    Recall that $|S_u|,|S_v| \leq \ell$ where $\ell = \alpha_{\rm sep}(k)$.
    Consider $\ring(S_u,S_v)$, and
    the linkage $\Pp(u,v)$ in $G$. 
    Our goal is to modify $\Pp(u,v)$ to obtain another linkage $\Pp'(u,v)$ that is aligned with it and has a small winding number with respect to $\pathT_R(u,v)$ 
    
    Recall that by Observation~\ref{obs:pre-rerouting},
    we have a encircling tight sequence of concentric cycles $\Cc(u,v)$ in $\ring(S_u,S_v)$ that contains at least $40 \ell$ cycles. Let $\eta(u,v)$ be the path in this ring witnessing the tightness of $\Cc(u,v)$.
    Further, recall the linkage $\flow_{R}(u,v)$ between $S_u$ and $S_v$ in the subgraph $G_{u,v}$ of $G$ (the restriction of $G$ to $\ring(S_u,S_v)$), and that $\flow_{R}(u,v)$ is minimal with respect to $\Cc(u,v)$.
    We can assume w.l.o.g. that $\Pp(u,v)$ is minimal with respect to $\Cc(u,v)$. Otherwise, there is another solution $\wh{\Pp}$ such that it is identical to $\Pp$ in $G - V(S_u,S_v)$ and $\wh{\Pp}(u,v)$ is a minimal linkage with respect to $\Cc(u,v)$ (that is aligned with $\Pp(u,v)$). Then, we can consider $\wh{\Pp}$ instead of $\Pp$. Let $\Pp_\crs(u,v)$ be the set of traversing paths in $\Pp(u,v)$. Since $\Pp_\crs(u,v)$ is a flow between $S_u$ and $S_v$ in $G_{u,v}$ and $\flow_R(u,v)$ is a maximum flow, clearly $|\flow_{R}(u,v)| \geq |\Pp_\crs(u,v)|$.
    
    Now, apply Lemma~\ref{lemma:winding} to $\Pp(u,v)$, $\Cc(u,v)$ and $\flow_R(u,v)$ in $\ring(S_u,S_v)$. We thus obtain a linkage $\Pp'_\crs(u,v)$ disjoint from $\Pp_\vis(u,v)$ that is aligned with $\Pp_\crs(u,v)$.
    Hence, $\Pp'(u,v) = \Pp_\vis(u,v) \cup \Pp'_\crs(u,v)$ is a linkage in $G$ aligned with $\Pp(u,v)$.
    Assume w.l.o.g. that $\Pp'(u,v)$ uses the $3$-rd copy of each edge in $H$.  
    Let us now consider the winding number of $\Pp'_\crs(u,v)$ with respect to $\pathT_R(u',v')$.
    By Lemma~\ref{lemma:winding}(c), $|\wnorig(\Pp'_\crs(u,v)) - \wnorig(\flow_{R}(u,v))| \leq 60 \ell + 6$. 
    Now, note that for any path $P' \in \Pp'_\crs(u,v)$,
    $|\wnorig(P') - \wnorig(\Pp'_\crs(u,v))| \leq 1$, (recall that the winding number of paths in a linkage differ by at most $1$).
    Similarly, for any path $Q \in \flow_{R}(u,v)$, $|\wnorig(\flow_{R}(u,v)) - \wnorig(Q)| \leq 1$.
    Therefore, it follows that for any two paths $P' \in \Pp'_\crs(u,v)$ and $Q \in \flow_{R}(u,v)$
    $|\wnorig(P') - \wnorig(Q)| \leq 60\ell + 8$.
    Recall that we chose $\eta(u,v)$ as the reference path of $\ring(S_u,s_v)$ in the above expression. Hence, we may rewrite it as $|\wnorig(P', \eta(u,v)) - \wnorig(Q, \eta(u,v))| \leq 60 \ell + 8$. Note that $P', Q$ and $\eta(u,v)$ are three edge-disjoint paths traversing $\ring(S_u,S_v)$. Hence $\wnorig$ is well defined in $\ring(S_u,S_v)$ for any pair of them. By Proposition~\ref{prop:wn-prop}, $|\wnorig(P',Q)| \leq |\wnorig(P', \eta(u,v)) - \wnorig(Q, \eta(u,v))| + 1$.
    We have so far established that for any $P' \in \Pp'_\crs(u,v)$ and $Q \in \flow_{R}(u,v)$, $|\wnorig(P',Q)| \leq 60\ell +9$.
    Now, consider $\pathT_{R}(u',v')$ and recall that for any $Q \in \flow_{R}(u,v)$, $|\wnorig(\pathT_R(u',v'), \allowbreak Q])|  \leq 1$ by Observation~\ref{obs:pre-rerouting}.
    Furthermore $\pathT_R(u',v')$ uses the $0$-th copies of edges in $H$.
    Hence, for each $P' \in \Pp'_\crs(u,v)$  $\wnorig(P',\pathT_R(u',v'))$ is well defined in $\ring(S_u,S_v)$, and by Proposition~\ref{prop:wn-prop}, $|\wnorig(P', \allowbreak \pathT_R(u',v'))| \leq |\wnorig(P',Q) - \wnorig(\pathT_R(u',v'),Q)| + 1 \leq 60\ell + 11$.
    
    
    Finally, consider the paths in $\Pp'_\vis = \Pp_\vis$. For each $P \in \Pp_\vis$ the absolute value of the winding number of $\wnorig(P, \pathT_R(u',v'))$ is bounded by $1$ (by Observation~\ref{obs:vis_wn}).
    Hence, we conclude that $\wn(\Pp', \ring(S_u,S_v)) \leq 60 \alpha_{\rm sep}(k) + 11 < 300 \cdot 2^{ck}$.
    %
\end{proof}


\newcommand{\Segg}{\ensuremath{\mathsf{SegGro}}}
\newcommand{\Potential}{\ensuremath{\mathsf{Potential}}}

\section{Pushing a Solution onto the Backbone Steiner Tree}\label{sec:pushing}

In this section, we push a linkage that is a solution of small winding number onto the backbone Steiner tree to construct a ``pushed weak linkage'' with several properties that will make its reconstruction (in Section \ref{sec:reconstruction}) possible. 
Let us recall that we have an instance $(G,S,T,g,k)$ and $H$ is the radial completion of $G$ enriched with $4n+1$ parallel copies of each edge. Then we construct a backbone Steiner tree $R$ (Section~\ref{sec:steiner}), which uses the $0$-th copy of each edge. Formally, a pushed weak linkage is defined as follows.

\begin{definition}[{\bf Pushed Weak Linkage}] Let $(G,S,T,g,k)$ be an instance of \pdp, and $R$ be a Steiner tree. A weak linkage $\cal W$ in $H$ is {\em pushed (onto $R$)} if $E({\cal W})\cap E(R)=\emptyset$ and every edge in $E({\cal W})$ is parallel to some edge of $R$.
\end{definition}

In what follows, we first define the properties we would like the pushed weak linkage to satisfy. Additionally, we partition any weak linkage into segments that give rise to a potential function whose maintenance will be crucial while pushing a solution of small winding number onto $R$. Afterwards, we show how to push the solution, simplify it, and analyze the result. We remark that the simplification (after having pushed the solution) will be done in two stages.

\subsection{Desired Properties and Potential of Weak Linkages}

The main property we would like a pushed weak linkage to satisfy is to use only ``few'' parallel copies of any edge. This quantity is captured by the following definition of multiplicity, which we will eventually like to upper bound by a small function of $k$.

\begin{definition}[{\bf Multiplicity of Weak Linkage}] Let $H$ be a plane graph. Let  $\cal W$ be a weak linkage in $H$. Then, the {\em multiplicity} of $\cal W$ is the maximum, across all edges $e$ in $H$, of the number of edges parallel to $e$ that belong to $E({\cal W})$.
\end{definition}

Towards bounding the multiplicity of the pushed weak linkage we construct, and also an important ingredient on its own in the reconstruction in Section \ref{sec:reconstruction}, we need to repeatedly eliminate U-turns in the weak linkage we deal with. Here, U-turns are defined as follows.

\begin{definition}[{\bf U-Turn in Weak Linkage}] Let $(G,S,T,g,k)$ be an instance of \pdp, and $R$ be a Steiner tree. Let  $\cal W$ be a weak linkage in $H$. Then, a {\em U-turn} in $\cal W$ is a pair of parallel edges $\{e,e'\}$ visited consecutively by some walk $W\in{\cal W}$ such that the strict interior of the cycle formed by $e$ and $e'$ does not contain the first or last edge of any walk in $\cal W$.  We say that $\cal W$ is {\em U-turn-free} if it does not have any U-turn.
\end{definition}

Still, having a pushed weak linkage of low multiplicity and no U-turns does not suffice for faithful reconstruction due to ambiguity in which edge copies are being used. This ambiguity will be dealt with by the following definition.

\begin{definition}[{\bf Canonical Weak Linkage}] 
    Let $(G,S,T,g,k)$ be an instance of \pdp, and $R$ be a Steiner tree. Let $\cal W$ be a weak linkage in $H$ pushed onto $R$. Then, $\cal W$ is {\em canonical} if $(i)$ for every edge $e_i\in E({\cal W})$, $i\geq 1$ $(ii)$ if $e_i\in E({\cal W})$, then all the parallel edges $e_j$, for $1 \leq j < i$, are also~in~$E({\cal W})$.
\end{definition}

For brevity, we say that a weak linkage $\cal W$ in $H$ is {\em simplified} if it is sensible, pushed onto $R$, canonical, U-turn-free and has multiplicity upper bounded by $\alpha_{\rm mul}(k):=2\alpha_{\rm potential}(k)$ where $\alpha_{\rm potential}(k)=2^{\OO(k)}$ will be defined precisely in Lemma \ref{lem:solPotential}.
For the process that simplifies a pushed weak linkage, we will maintain a property that requires multiplicity at most $2n$ as well as a relaxation of canonicity. This property is defined as follows.

\begin{definition}[{\bf Extremal Weak Linkage}] Let $(G,S,T,g,k)$ be an instance of \pdp, and $R$ be a Steiner tree. Let $\cal W$ be a weak linkage in $H$ that is pushed onto $R$. Then, $\cal W$ is {\em extremal} if its multiplicity is at most $2n$ and for any two parallel edges $e_i,e_j\in E({\cal W})$ where $i\geq 1$ and $j\leq -1$, we have $(i-1)+|j+1|\geq 2n$.
\end{definition}

Additionally, we will maintain the following property.

\begin{definition}[{\bf Outer-Terminal Weak Linkage}] Let $(G,S,T,g,k)$ be a nice instance of \pdp, and $R$ be a Steiner tree. Let $\cal W$ be a weak linkage in $H$. Then, $\cal W$ is outer-terminal if it uses exactly one edge incident to $t^\star$. 
\end{definition}

\paragraph{Segments, Segment Groups and Potential.} To analyze the ``complexity'' of a weak linkage, we partition it into segments and segment groups, and then associate a potential function with it based on this partition. Intuitively, a segment of a walk is a maximal subwalk that does not cross $R$ (see Fig.~\ref{fig:segment}). Formally, it is defined as follows.

\begin{figure}
    \begin{center}
        \includegraphics[scale=0.8]{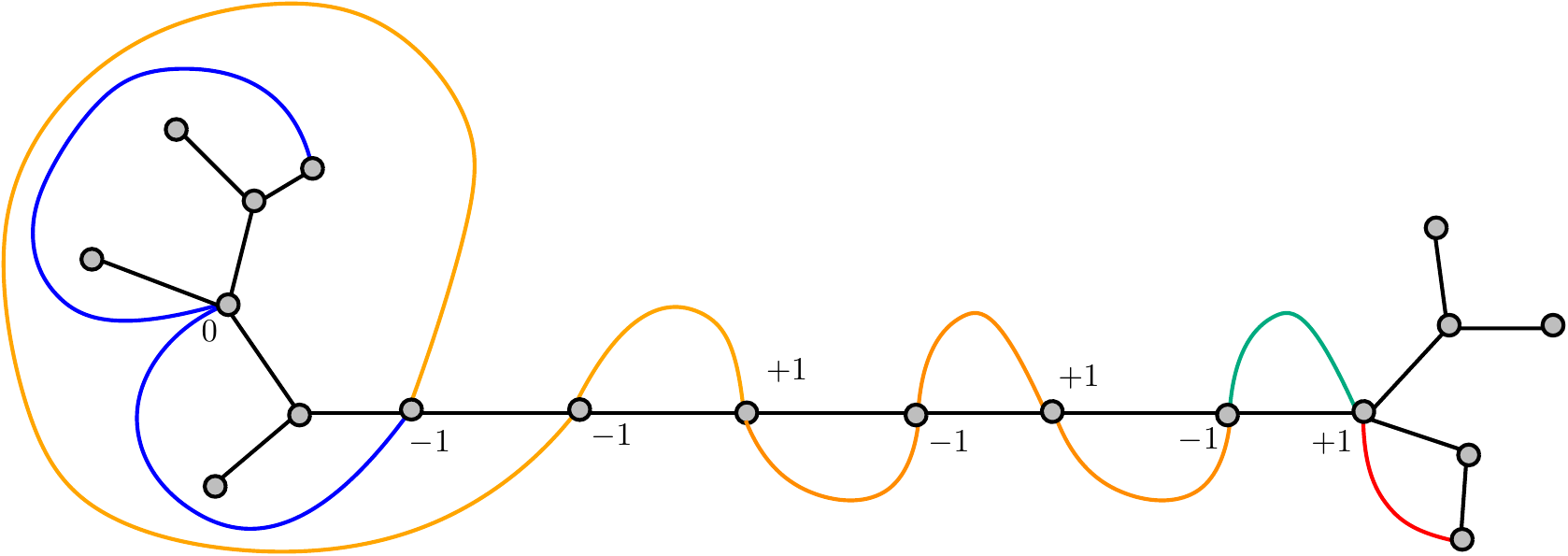}
        \caption{Segments, segment groups and their labeling.}
        \label{fig:segment}
    \end{center}
\end{figure}

\begin{definition}[{\bf Segment}] Let $(G,S,T,g,k)$ be an instance of \pdp, and $R$ be a Steiner tree. Let $W$ be a walk in $H$ that is edge-disjoint from $R$. A {\em crossing} of $W$ with $R$ is a crossing $(v,e,\widehat{e},e',\widehat{e}')$ of $W$ and some path in $R$.\footnote{The path might not be a maximal degree-2 path, thus $(v,e,\widehat{e},e',\widehat{e}')$ may concern a vertex $v\in V_{\geq 3}(R)$.} Then, a {\em segment} of $W$ is a maximal subwalk of $W$ that has no crossings with $R$. Let $\Seg(W)$ denote the set\footnote{Because we deal with walks that do not repeat edges, $\Seg(W)$ is necessarily a set rather than a multiset.} of segments of $W$.
\end{definition}

We remind that $R$ only contains $0$-copies of edges, hence we can ensure that we deal with walks that are edge-disjoint from $R$ by avoiding the usage of $0$-copies. Towards the definition of potential for a weak linkage, we group segments together as follows (see Fig.~\ref{fig:segment}).

\begin{definition}[{\bf Segment Group}] Let $(G,S,T,g,k)$ be an instance of \pdp, and $R$ be a Steiner tree. Let $W$ be a walk in $H$ that is edge-disjoint from $R$. A {\em segment group} of $W$ is a maximal subwalk $W'$ of $W$ such that either {\em (i)} $\Seg(W')\subseteq \Seg(W)$ and all of the endpoints of all of the segments in $\Seg(W')$ are internal vertices of the same maximal degree-2 path of $R$, or {\em (ii)} $W'\in\Seg(W)$ and the two endpoints of $W'$ are not internal vertices of the same maximal degree-2 path in $R$.\footnote{That is, the two endpoints of $W'$ are internal vertices in different maximal degree-2 paths in $R$ or at least one endpoint of $W'$ is a vertex in $V_{=1}(R)\cup V_{\geq 3}(R)$.} The set of segment groups of $W$ is denoted by $\Segg(W)$.
\end{definition}

Observe that the set of segments, as well as the set of segment groups, define a partition of a walk. We define the ``potential'' of a segment group based on its winding number in the ring that corresponds to its path (in case it is a long path where a ring is defined). To this end, recall the labeling function in Definition \ref{def:spiralLabel}. 
Note that the labeling is defined for any two walks irrespective of the existence of a ring.

\begin{definition}[{\bf Potential of Segment Group}] Let $(G,S,T,g,k)$ be an instance of \pdp, and $R$ be a Steiner tree. Let $W$ be a walk in $H$ that is edge-disjoint from $R$ and whose endpoints are in $V_{=1}(R)$. Let $W'\in\Segg(W)$. If $|\Seg(W')|=1$, then the {\em potential} of $W'$, denoted by $\Potential(W')$, is defined to be $1$. Otherwise, it is defined as follows.
\[\Potential(W') = 1+|\sum_{(e,e')\in E(W^\star) \times E(W^\star)}\lab_P^{W'}(e,e')|,\]
where $W^\star$ is the walk obtained from $W'$ be adding two edges to $W'$---the edge consecutive to the first edge of $W'$ in $W$ and the edge consecutive to the last edge of $W'$ in $W$, and $P$ is the maximal degree-2 path of $R$  such that all of the endpoints of all of the segments in $\Seg(W')$ are its internal vertices. 
\end{definition}

The potential of a segment group is well defined as we use the function $\lab$ only for edges incident to internal vertices of the maximal degree-2 paths in $R$. For an example of a potential of a segment groups, see Fig.~\ref{fig:segment}.
Now, we generalize the notion of potential from segment groups to weak linkages as follows.

\begin{definition}[{\bf Potential of Weak Linkage}] Let $(G,S,T,g,k)$ be an instance of \pdp, and $R$ be a Steiner tree. Let $\cal W$ be a weak linkage. Then, the {\em potential} of $\cal W$ is
\[\Potential({\cal W}) = \sum_{W'\in\Segg({\cal W})}\Potential(W'),\]
where $\Segg({\cal W})=\bigcup_{W\in{\cal W}}\Segg(W)$.
\end{definition}

To upper bound the potential of a solution of a small winding number, we first upper bound the number of segment groups.

\begin{lemma}\label{lem:numSegGro}
Let $(G,S,T,g,k)$ be a good \yes-instance of \pdp, and $R$ be a backbone Steiner tree. Then, there exists a solution $\cal P$ to $(G,S,T,g,k)$ such that $|\Segg({\cal P})|\leq \alpha_{\rm segGro}(k):=10^5 k \cdot 2^{ck} $.
\end{lemma}

\begin{proof}
First, notice that for every path $P\in{\cal P}$, every segment group of $P$ whose endpoints are both internal vertices of some maximal degree-2 path of $R$ has the following property: it is neither a prefix nor a suffix of $P$ (when $P$ is oriented arbitrarily, say, from its endpoint in $S$ to its endpoint in $T$), and the segment groups that appear immediately before and after it necessarily satisfy that each of them does not have both of its endpoints being internal vertices of some maximal degree-2 path of $R$. Let $s$ denote the number of segment groups of $\cal P$ that do not have both of their endpoints being internal vertices of some maximal degree-2 path of $R$. Then, we have that $|\Segg({\cal P})|\leq 2s-1$, and therefore to complete the proof, it suffices to show that $s\leq \alpha_{\rm segGro}(k)/2$. 

Let $\{u_1,v_1\}, \{u_2,v_2\}, \ldots, \{u_t,v_t\}$ denote the pairs of near vertices in $V_{=1}(R) \cup V_{\geq 3}(R)$ such that  for each $i \in \{1,\ldots t\}$, $\pathT_R(u_i,v_i)$ is a long maximal degree-2 path in $R$. Let $\ring(S_{u_1}, S_{v_1}),$ $ \ring(S_{u_2}, S_{v_2}), \ldots, \ring(S_{u_t},S_{v_t})$ be the corresponding rings. By Lemma \ref{lemma:R-outside-rings}, the number of vertices of $R$ lying outside these rings, $|V(R) \setminus \bigcup_{i=1}^t V(S_{u_i}, S_{v_i})|$, is upper bounded by $\alpha_{\rm nonRing}(k) = 10^4 k \cdot 2^{ck}$. We further classify each segment group $S$ of $\cal P$ that does not have both of its endpoints being internal vertices of some maximal degree-2 path of $R$ as follows. (We remark that, by the definition of a segment group, $S$ must consist of a single segment.)
\begin{itemize}\setlength\itemsep{0em}
\item $S$ has at least one endpoint in $V(R) \setminus \bigcup_{i=1}^t V(S_{u_i}, S_{v_i})$. Denote the number of such segment groups by $s_1$.
\item $S$ has one endpoint in $V(S_{u_i}, S_{v_i})$ and another endpoint in $V(S_{u_j}, S_{v_j})$ for some $i,j\in\{1,\ldots,t\}$ such that $i\neq j$. Denote the number of such segment groups by $s_2$.
\end{itemize}

Then, $s_1+s_2=s$. Now, notice that the paths in $\cal P$ are pairwise vertex-disjoint, and the collection of segment groups of each path $P\in{\cal P}$ forms a partition of $P$. Thus, because $|V(R) \setminus \bigcup_{i=1}^t V(S_{u_i}, S_{v_i})|\leq \alpha_{\rm nonRing}(k)$ and each vertex in $V(R)$ is shared as an endpoint by at most two segment groups, we immediately derive that $s_1\leq \alpha_{\rm nonRing}(k)$. To bound $s_2$, note that each segment group of the second type traverses at least one vertex in $S_{u_i}\cup S_{v_i}$ for some $i\in\{1,\ldots,t\}$ (due to Lemma \ref{lem:separatorsUnchanged}). By Lemma \ref{lem:sepSmall}, for every $i\in\{1,\ldots,t\}$, $|S_{u_i}|,|S_{v_i}|\leq \alpha_{\mathrm{sep}}(k)=\frac{7}{2}\cdot 2^{ck}+2$. Moreover, by Observation \ref{obs:leaIntSteiner}, $t<4k$.  Thus, we have that $s_2\leq 4\cdot 4k\cdot \alpha_{\mathrm{sep}}(k)$, where the multiplication by $4$ is done because two segment groups can share an endpoint and each maximal degree-2 path is associated with two separators. From this, we conclude that $s\leq 10^4 k \cdot 2^{ck} + 16k(\frac{7}{2}\cdot 2^{ck}+2)\leq \alpha_{\rm segGro}(k)/2$.
\end{proof}

Now, based on Observation \ref{obs:spiralLabel} and Lemmas \ref{lemma:nice-solution} and \ref{lem:numSegGro}, we derive the existence of a solution with low potential (if there exists a solution).

\begin{lemma}\label{lem:solPotential}
Let $(G,S,T,g,k)$ be a good \yes-instance of \pdp, and $R$ be a backbone Steiner tree. Then, there exists a solution $\cal P$ to $(G,S,T,g,k)$ such that $\Potential({\cal P})\leq \alpha_{\rm potential}(k):=(10^4\cdot 4^{ck}+1)\cdot \alpha_{\rm segGro}(k)$.
\end{lemma}

\begin{proof}
Let $\cal Q$ be a solution with the property in Lemma \ref{lemma:nice-solution}.  
Let ${\cal S}$ denote the segment groups in $\Segg({\cal Q})$ whose both endpoints belong to the same maximal degree-2 path of $R$. For each segment group $S\in{\cal S}$, denote $\mathsf{labPot}(S)=|\sum_{(e,e')\in E(S) \times E(S)}\lab_{P_S}^{S}(e,e')|$ where $P_S$ is the path in $\cal Q$ that has $S$ as a segment group. Furthermore, denote $M=\max_{S\in{\cal S}}\mathsf{labPot}(S)$.
Now, by the definition of potential,
\[\Potential({\cal Q}) = |\Segg({\cal Q})| + \sum_{S\in {\cal S}}\mathsf{labPot}(S)\leq (M+1)|\Segg({\cal Q})|.\]
By Lemma \ref{lem:numSegGro}, $|\Segg({\cal Q})|\leq \alpha_{\rm segGro}(k)$. Thus, to complete the proof, it suffices to show that $M\leq 10^4\cdot 4^{ck}$. To this end,  consider some segment group $S\in{\cal S}$. Notice that $\mathsf{labPot}(S)$ is upper bounded by the number of crossings of $S$ with $P_S$. Thus, if $P_S$ is short,  it follows that $\mathsf{labPot}(S)<\alpha_{\mathrm{long}}(k) = 10^4\cdot 2^{ck}$.

Now, suppose that $P_S$ is long, and let $\ring(S_u,S_v)$ be the ring that corresponds to $P_S$. By the choice of $\cal Q$, $|\wn({\cal Q}, \ring(S_{u},S_{v}))|  \leq \alpha_{\rm winding}(k) < 300 \cdot 2^{ck}$. Let $\widehat{\cal A}$ be collection of maximal subpaths of $S$ that are fully contained within $\ring(S_{u},S_{v})$, and let $\widehat{P}_S$ be the maximal subpath of $P_S$ that is fully contained within $\ring(S_{u},S_{v})$. Then, by Observation \ref{obs:spiralLabel}, 
\[\mathsf{labPot}(S) \leq |V(P_S)\setminus V(\widehat{P}_S)| + \sum_{\widehat{A}\in\widehat{\cal A}}|\wnorig(\widehat{A}, \widehat{P}_S)|.\]
By the definition of $\wn$, we have that $|\wnorig(\widehat{A}, \widehat{P}_S)|\leq |\wn({\cal Q}, \ring(S_{u},S_{v}))|$ for each $\widehat{A}\in\widehat{\cal A}$. Moreover, by Lemma \ref{lem:sepSmall}, $|\widehat{\cal A}|\leq |S_u|+|S_v| \leq 2\alpha_{\mathrm{sep}}(k)=7\cdot 2^{ck}+4$. Additionally, by Lemmas \ref {lem:threeParts} and \ref{lem:R*Steiner}, we have that $|V(P_S)\setminus V(\widehat{P}_S)|\leq 2\alpha_{\mathrm{pat}}(k)=200\cdot 2^{ck}$. From this, 
\[\mathsf{labPot}(S) \leq 200\cdot 2^{ck} + (7\cdot 2^{ck}+4)\cdot 300 \cdot 2^{ck}\leq 10^4\cdot 4^{ck}.\]

Thus, because the choice of $S$ was arbitrary, we conclude that $M\leq 10^4\cdot 4^{ck}$.
\end{proof}

\subsection{Pushing a Solution onto $R$}

Let us now describe the process of pushing a solution onto $R$.
To simplify this process, 
we define two ``non-atomic'' operations that encompass sequences of atomic operations in discrete homotopy. We remind that we only deal with walks that do not repeat~edges.

\begin{definition}[{\bf Non-Atomic Operations in Discrete Homotopy}]\label{def:nonAtomicDiscreteHomotopy}
Let $G$ be a triangulated plane graph with a weak linkage ${\cal W}$, and $C$ be a cycle\footnote{A pair of parallel edges is considered to be a cycle.} in $G$. Let $W\in {\cal W}$. 
\begin{itemize}
\item {\bf Cycle Move.} Applicable to $(W,C)$ if there exists a subpath $Q$ of $C$ such that {\em (i)} $Q$ is a subpath of $W$, {\em (ii)} $1\leq |E(Q)|\leq |E(C)|-1$, {\em (iii)} no edge in $E(C)\setminus E(Q)$ belongs to any walk in $\cal W$, and {\em (iii)} no edge drawn in the strict interior of $C$ belongs to any walk in $\cal W$. Then, the cycle move operation replaces $Q$ in $W$ by the unique subpath of $C$ between the endpoints of $Q$ that is edge-disjoint from $Q$. 

\item {\bf Cycle Pull.} Applicable to $(W,C)$ if {\em (i)} $C$  is a subwalk $Q$ of $W$, and {\em (ii)} no edge drawn in the strict interior of $C$ belongs to any walk in $\cal W$. Then, the cycle pull operation replaces $Q$ in $W$ by a single occurrence of the first vertex in $Q$. 
\end{itemize}
\end{definition}

We now prove that the operations above are compositions of atomic operations.

\begin{lemma}\label{lem:nonAtomicDiscreteHomotopy}
Let $G$ be a triangulated plane graph with a weak linkage  ${\cal W}$, and $C$ be a cycle in $G$. Let $W\in {\cal W}$ with a non-atomic operation applicable to $(W,C)$. Then,  the result of the application is a weak linkage that is discretely homotopic to $\cal W$.
\end{lemma}

\begin{proof}
We prove the claim by induction on the number of faces of $G$ in the interior of $C$. In the basis, where $C$ encloses only one face, then the cycle move and cycle pull operations are precisely the face move and face pull operations, respectively, and therefore the claim holds. Now, suppose that $C$ encloses $i\geq 2$ faces and that the claim is correct for cycles that enclose at most $i-1$ faces. We consider several cases as follows.

First, suppose that $C$ has a path $P$ fully drawn in its interior whose endpoints are two (distinct) vertices $u,v\in V(C)$, and whose internal vertices and all of its edges do not belong to $C$. (We remark that $P$ might consist of a single edge, and that edge might be parallel to some edge of $C$.) Now, notice that $P$ partitions the interior of $C$ into the interior of two cycles $C_1$ and $C_2$ that share only $P$ in common as follows:  one cycle $C_1$ consists of one subpath of $C$ between $u$ and $v$ and the path $P$, and the other cycle $C_2$ consists of the second subpath of $C$ between $u$ and $v$ and the path $P$. 
Notice that $C_1$ encloses less faces than $C$, and so does $C_2$. At least one of these two cycles, say, $C_1$, contains at least one edge of $Q$. Then, the cycle move operation is applicable to $(W,C_1)$. Indeed, let $\widehat{Q}$ be the subpath of $Q$ that is a subpath of $C$, and notice that $E(P)\subseteq E(C_1)\subseteq E(C)\cup E(P)$ and $E(P)\cap E({\cal W})=\emptyset$ (because the cycle move/pull operation is applicable to $(W,C)$). Therefore, $\widehat{Q}$ is a subpath of $W$, $1\leq |E(\widehat{Q})|\leq |E(C_1)|-1$, and no edge in $E(C_1)\setminus E(\widehat{Q})$ belongs to any walk in $\cal W$. Moreover, because $C_1$ belongs to the interior (including the boundary) of $C$, no edge drawn in the strict interior of $C$ belongs to any walk in $\cal W$. Now, notice that after the application of the cycle move operation for $(W,C_1)$, $C_2$ also has at least one edge used by the walk $W'$ into which $W$ was modified---in particular, $E(P)\subseteq E(W')$. Moreover,  consider the subpath (or subwalk that is a cycle) $Q'$ of $W'$ that results from the replacement of $\widehat{Q}$ in $Q$ by the subpath of $C_1$ between the endpoints of $Q'$ that does not belong to $W$. Then, $Q'$ traverses some subpath (possible empty) of $C_1$ or $C_2$, then traverses $P$, and next traverses some other subpath of $C_1$ or $C_2$. So, the restriction of $Q'$ to $C_2$ is a non-empty path or cycle $Q^\star$ that is a subwalk of $W'$. Furthermore, because $C_2$ is drawn in the interior of $C$ and the cycle move/pull operation is applicable to $(W,C)$, we have that no edge of $E(C_2)\setminus E(Q^\star)$ or the strict interior of $C_2$ belongs to $E({\cal W})$. Thus, the cycle move/pull operation is applicable to $(W',C_2)$. Now, the result of the application of this operation is precisely the result of the application of the original cycle move or pull operation applicable to $(W,C)$. To see this, observe that the edges of $E(C)\setminus E(W)$ that occur in $C_1$ along with $E(P)$ have replaced the edges of $E(C)\cap E(W)$ that occur in $C_1$ in the first operation, and the edges of $E(C)\setminus E(W)$ that occur in $C_2$ have replaced the edges of $E(C)\cap E(W))$ that occur in $C_1$ along with $E(P)$ in the second operation. Thus, by the inductive hypothesis with respect to $(W,C_1)$ and $(W',C_2)$, and because discrete homotopy is transitive, the claim follows.

Thus, it remains to prove that $C$ has a path $P$ fully drawn in its interior whose endpoints are two (distinct) vertices $u,v\in V(C)$, and whose internal vertices and all of its edges do not belong to $C$. In case $C$ has a chord (that is, an edge in $G$ between two vertices of $C$ that does not belong to $C$), then the chord is such a path $P$. Therefore, we now suppose that this is not the case. Then, $C$ does not contain in its interior an edge parallel to an edge of $C$. In turn, because $G$ is triangulated), when we consider some face $f$ in the interior of $C$ that contains an edge $e$ of $C$, this face must be a triangle. Moreover, the vertex of $f$ that is not incident to $e$ cannot belong to $C$, since otherwise we obtain a chord in $C$. Thus, the subpath (that consists of two edges) of $f$ between the endpoints of $e$ that does not contain $e$ is a path $P$ with the above mentioned properties.
\end{proof}

In the process of pushing a solution onto $R$, we push parts of the solution one-by-one. We refer to these parts as sequences, defined as follows  (see Fig.~\ref{fig:seq}).

\begin{figure}
    \begin{center}
        \includegraphics[scale=0.7]{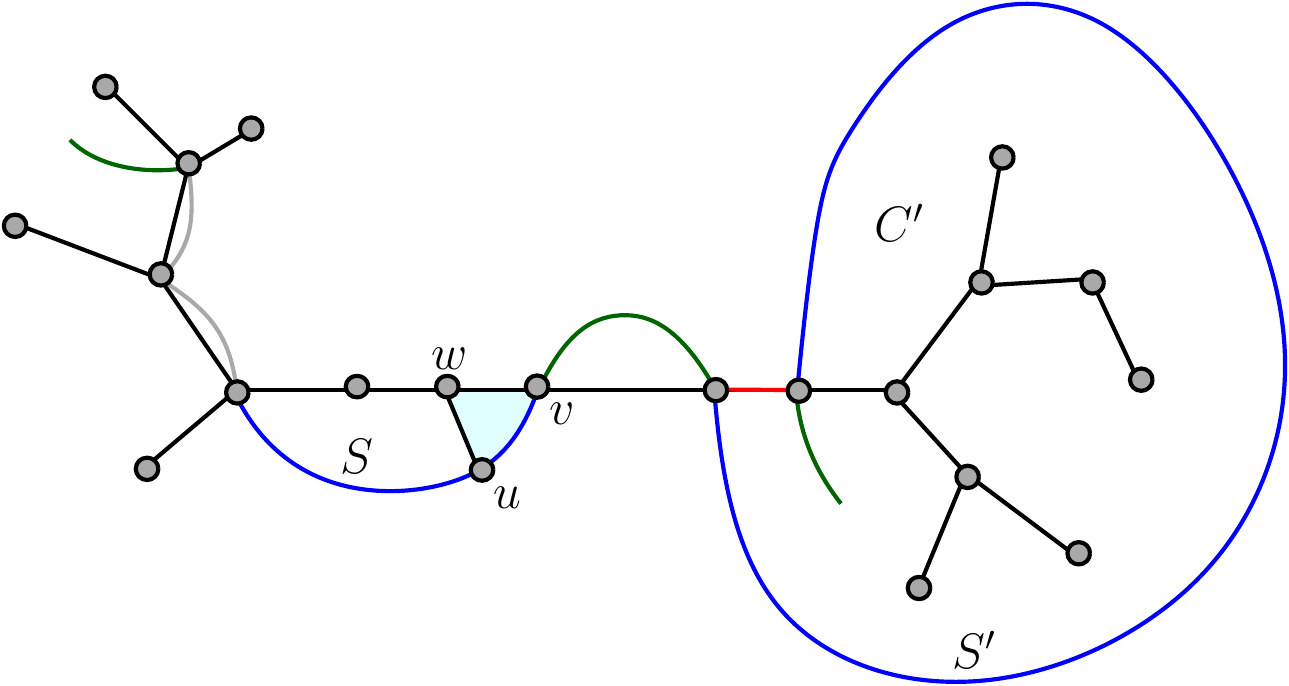}
        \caption{A Sequence, its projecting cycle and a shrinking cycle.}
        \label{fig:seq}
    \end{center}
\end{figure}

\begin{definition}[{\bf Sequence}]
Let $(G,S,T,g,k)$ be an instance of \pdp, and $R$ be a Steiner tree. Let $W$ be a walk. Then, a {\em sequence} of $W$ is a maximal subwalk of $W$ whose internal vertices (if any exists) do not belong to $R$ and which contains at least one edge that is not parallel to an edge of $R$. The set of sequences of $W$ is denoted by $\Seq(W)$. For a weak linkage $\cal W$, the set of sequences of $\cal W$ is defined as $\Seq({\cal W})=\bigcup_{W'\in{\cal W}}\Seq(W')$.
\end{definition}

Notice that the set of sequences of a walk does not necessarily form a partition of the walk because the walk can traverse edges parallel to the edges of $R$ and these edges do not belong to any sequence. Moreover, for sensible weak linkages, the endpoints of every sequence belong to $R$. To deal only with sequences that are paths, we need the following definition.

\begin{definition}[{\bf Well-Behaved Weak Linkage}]
Let $(G,S,T,g,k)$ be an instance of \pdp, and $R$ be a Steiner tree.  A weak linkage $\cal W$ is {\em well-behaved} if every sequence in $\Seq({\cal W})$ is a path or a cycle.
\end{definition}

When we will push sequences one-by-one, we ensure that the current sequence to be pushed can be handled by the cycle move operation in Definition \ref{def:nonAtomicDiscreteHomotopy}. To this end, we define the notion of an innermost sequence, based on another notion called a projecting cycle  (see Fig.~\ref{fig:seq}) We remark that this cycle will not necessarily be the one on which we apply a cycle move operation, since this cycle might contain in its interior edges of some walks of the weak linkage.

\begin{definition}[{\bf Projecting Cycle}]
Let $(G,S,T,g,k)$ be an instance of \pdp, and $R$ be a Steiner tree. Let $\cal W$ be a sensible well-behaved weak linkage, and $S\in\Seq({\cal W})$.  The {\em projecting cycle} of $S$ is the cycle $C$ formed by $S$ and the subpath $P$ of $R$ between the endpoints of $S$. Additionally, $\mathsf{Volume}(S)$ denotes the number of faces enclosed by the projecting cycle of~$S$.
\end{definition}

Now, we define the notion of an innermost sequence. 

\begin{definition}[{\bf Innermost Sequence}]
Let $(G,S,T,g,k)$ be an instance of \pdp, and $R$ be a Steiner tree. Let $\cal W$ be a sensible well-behaved weak linkage, and $S\in\Seq({\cal W})$. Then, $S$ is {\em innermost} if every edge in $E({\cal W})$ that is drawn in the interior of its projecting cycle  is parallel to some edge of $R$.
\end{definition}

We now argue that, unless the set of sequences is empty, there must exist an innermost one.

\begin{lemma}\label{lem:innermostSeqExists}
Let $(G,S,T,g,k)$ be an instance of \pdp, and $R$ be a Steiner tree. Let $\cal W$ be a sensible well-behaved weak linkage such that $\Seq({\cal W})\neq\emptyset$. Then, there exists an innermost sequence in $\Seq({\cal W})$.
\end{lemma}

\begin{proof}
Let $S\in\Seq({\cal W})$ be a sequence that minimizes $\mathsf{Volume}(S)$. We claim that $S$ is innermost. Suppose, by way of contradiction, that this claim is false. Then, there exists an edge $e\in E({\cal W})$ that is drawn in the interior of the projecting cycle of $S$ and is not parallel to any edge of $R$. Thus, $e$ belongs to some sequence $S'\in\Seq({\cal W})$. Because $\cal W$ is well-behaved, $S$ and $S'$ are vertex disjoint. This implies that the projecting cycle of $S'$ is contained in the interior of the projecting cycle of $S$. Because $H$ is triangulated, this means that $\mathsf{Volume}(S')<\mathsf{Volume}(S)$, which is a contradiction to the choice of $S$.
\end{proof}

When we push the sequence onto $R$, we need to ensure that we have enough copies of each edge of $R$ to do so. To this end, we need the following definition.

\begin{definition}[{\bf Shallow Weak Linkage}]\label{def:shallow}
Let $(G,S,T,g,k)$ be an instance of \pdp, and $R$ be a Steiner tree. Let $\cal W$ be a sensible well-behaved weak linkage. Then, $\cal W$ is {\em shallow} if for every edge $e_0\in E(R)$, the following condition holds. Let $\ell$ (resp.~$h$) be the number of sequences $S\in\Seq({\cal W})$ whose projecting cycle encloses $e_1$ (resp.~$e_{-1}$). Then, $e_i$ is not used by $\cal W$ for every $i\in\{-n,-n+1,\ldots,-n+\ell-1\}\cup \{0\}\cup\{n-h+1,n-h+2,\ldots,n\}$.
\end{definition}

To ensure that we make only cycle moves/pulls as in Definition \ref{def:nonAtomicDiscreteHomotopy}, we do not necessarily push the sequence at once, but gradually shrink the area enclosed by its projecting cycle.\footnote{Instead, we could have also always pushed a sequence at once by defining moves and pulls for closed walks, which we find somewhat more complicated to analyze formally.}

\begin{definition}[{\bf Shrinking Cycle}]\label{def:shrinkCyc}
Let $(G,S,T,g,k)$ be an instance of \pdp, and $R$ be a Steiner tree. Let $\cal W$ be a sensible well-behaved weak linkage, and $S\in\Seq({\cal W})$ with an endpoint $v\in V(R)$.  Then, a cycle $C$ in $H$ is a {\em shrinking cycle} for $(S,v)$ if it has no edge of $R$ in its interior and it can be partitioned into three paths where the first  has at least one edge and the last has at most one edge: {\em (i)} a subpath $P_1$ of $S$ with $v$ as an endpoint; {\em (ii)} a subpath $P_2$ from the other endpoint $u$ of $P_1$ to a vertex on $R$ that consists only of edges drawn in the strict interior of the projecting cycle of $S$ and no vertex of $S$ apart from $u$; {\em (iii)} a subpath $P_3$ that has $v$ as an endpoint and whose edge (if one exists) is either not parallel to any edge of $R$ or it is the $i$-th copy of an edge parallel to some edge of $R$ for some $i\in\{-n+\ell-1,n-h+1\}$, where $\ell$ and $h$ are as in Definition \ref{def:shallow}.
\end{definition}

With respect to shrinking cycles, we prove two claims. First, we assert their existence.

\begin{figure}
    \begin{center}
        \includegraphics[scale=0.7]{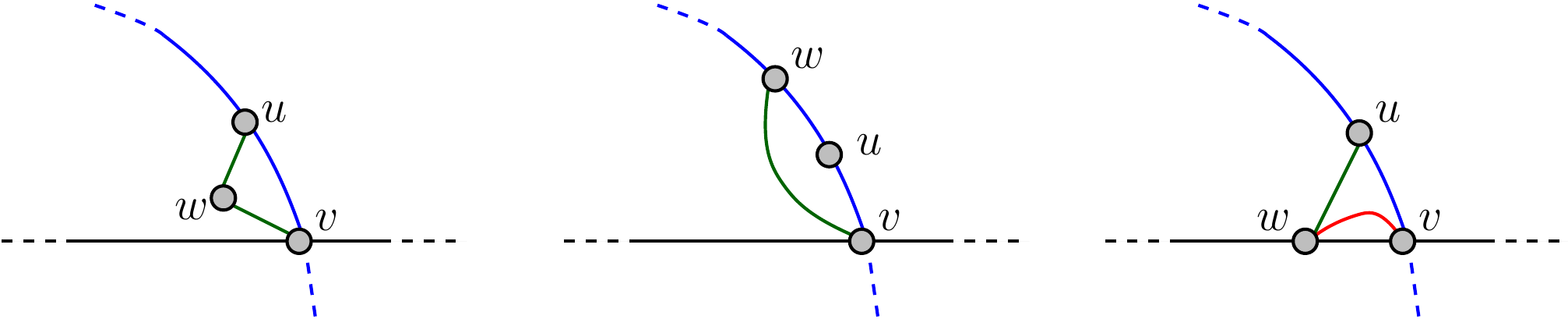}
        \caption{An illustration of Lemma~\ref{lem:shrinkCycExists}}
        \label{fig:shrinkingCycle}
    \end{center}
\end{figure}

\begin{lemma}\label{lem:shrinkCycExists}
Let $(G,S,T,g,k)$ be an instance of \pdp, and $R$ be a Steiner tree. Let $\cal W$ be a sensible well-behaved weak linkage, and $S\in\Seq({\cal W})$ with an endpoint $v\in V(R)$.  Then, there exists a shrinking cycle for $(S,v)$.
\end{lemma}
\begin{proof}
Let $e$ be an edge in $S$ incident to $v$ (if there are two such edges, when $S$ is a cycle, pick one of them arbitrarily), and denote the other endpoint of $e$ by $u$. Because $H$ is triangulated, $e$ belongs to the boundary $B$ of a face $f$ of $H$ in the interior of $C$ such that $B$ is a cycle that consists of only two or three edges. If $B$ does not contain any vertex of $V(R)\cup V(S)$ besides $u$ and $v$, then it is clearly a shrinking cycle (see Fig.~\ref{fig:shrinkingCycle}) Thus, we now suppose that $B$ is a cycle on three edges whose third vertex, $w$, belongs to $V(R)\cup V(S)$.  If $w\in V(S)$, then the cycle that consists of the subpath of $S$ from $v$ to $w$ and the edge in $E(B)$ between $v$ and $w$ is also clearly a shrinking cycle (see Fig.~\ref{fig:shrinkingCycle}). Thus, we now also suppose that $w\in V(R)$.

We further distinguish between two cases. First, suppose that $w$ is not adjacent to $v$ on $R$. In this case, $B$ does not enclose any edge of $R$ as well as any edge parallel to an edge of $R$. Moreover, $B$ can be partitioned into $P_1,P_2$ and $P_3$ that are each a single edge, where $P_1$ consists of the edge in $B$ between $v$ and $u$, $P_2$ consists of the edge in $B$ between $u$ and $w$, and $P_3$ consists of the edge in $B$ between $w$ and $v$, thereby complying with Definition \ref{def:shrinkCyc}. Thus, $B$ is a shrinking cycle for $(S,v)$. 
Now, suppose that $w$ is adjacent to $v$ on $R$. 
Then, define $P_1,P_2$ and $P_3$ similarly to before except that to $P_3$, we do not take the edge of $B$ between $v$ and $w$ but its parallel $i$-th copy where $i\in\{-n+\ell-1,n-h+1\}$ such that $\ell$ and $h$ are as in Definition \ref{def:shallow}. The choice of whether $i=-n+\ell-1$ or $i=n-h+1$ is made so that the cycle $B'$ consisting of $P_1,P_2$ and $P_3$ does not enclose any edge of $R$. (Such a choice necessarily exists, see~Fig.~\ref{fig:shrinkingCycle}). 
\end{proof}

Now, we prove that making a cycle move/pull operation using a shrinking cycle is valid and  maintains some properties of weak linkages required for our analysis.

\begin{lemma}\label{lem:shrinkCycMaintain}
Let $(G,S,T,g,k)$ be an instance of \pdp, and $R$ be a Steiner tree. Let $\cal W$ be a sensible, well-behaved, shallow and outer-terminal weak linkage, and $S\in\Seq({\cal W})$ be innermost with an endpoint $v\in V(R)\setminus\{t^\star\}$. Let $C$ be a shrinking cycle for $(S,v)$ that encloses as many faces as possible. Then, the cycle move/pull operation is applicable to $(W,C)$ where $W\in{\cal W}$ is the walk having $S$ as a sequence. Furthermore, the resulting weak linkage ${\cal W}'$ is sensible, well-behaved, shallow and outer-terminal having the same potential as $\cal W$, and $\sum_{\widehat{S}\in \Seq({\cal W}')}\mathsf{Volume}(\widehat{S})<\sum_{\widehat{S}\in \Seq({\cal W})}\mathsf{Volume}(\widehat{S})$.
\end{lemma}

\begin{proof}
We first argue that  the cycle move/pull operation is applicable to $(W,C)$. Let $P_1,P_2$ and $P_3$ be the partition of $C$ in Definition \ref{def:shrinkCyc}. Note that because $S$ is innermost, its projecting cycle does not contain any edge in $E({\cal W})$ that is not parallel to some edge of $R$, and hence so does $C$ as it is drawn in the interior (including the boundary) of the projecting cycle. Furthermore, the only edges parallel to $R$ that $C$ can contain are those parallel to the edge $e_i$ of $P_3$ whose subscripts have absolute value larger than $|i|$. However, none of these edges belong to $\cal W$ because $i\in\{-n+\ell-1,n-h+1\}$ where $\ell$ and $h$ are as in Definition \ref{def:shallow}, and $\cal W$ is shallow. Lastly, note that $P_1$ is either $S$ (which might be a cycle) or a subpath of $S$, and hence it is a subwalk of $W$. Thus, the cycle move or pull (depending on whether $P_1$ is a cycle) is applicable to $(W,C)$.  Furthermore, the new walk $W'$ that results from the application is the modification of $W$ that replaces $P_1$ by the path consisting of $P_2$ and $P_3$.

Because $\cal W$ is sensible and the endpoints of no walk in $\cal W$ are changed in ${\cal W}'$, we have that ${\cal W}'$ is sensible as well. Moreover, the vertices of $P_2$ are not used by any walk in $\cal W$ apart from $W'$ and only in its subwalk that traverses $P_2$, and therefore, as $\cal W$ is well-behaved, so is ${\cal W}'$. Additionally, note that $\cal W$ is shallow and that each edge belongs to at most as many projecting cycles of sequences in ${\cal W}'$ as it does in $\cal W$. Thus, if $P_3$ does not contain an edge (the only edge parallel to an edge of $R$ that might be used by ${\cal W}'$ but not by $\cal W$ is the edge $e_i$ of $P_3$, if it exists), it is immediate that ${\cal W}'$ is shallow. Now, suppose that $e_i$ exists. Let $b\in\{-1,1\}$ have the same sign as $i$. Recall that $i\in\{-n+\ell-1,n-h+1\}$ where $\ell$ and $h$ are as in Definition \ref{def:shallow}, thus to conclude that ${\cal W}'$ is shallow, we only need to argue that $e_b$ belongs to the interior of fewer projecting cycles of sequences in ${\cal W}'$ as it does in $\cal W$. However, this holds since the only difference between the sequences of $\cal W$ and ${\cal W}'$ is that the sequence $S$ occurs in $\cal W$ (and contains $e_b$ in the interior of its projecting cycle), but is transformed into (one or two) other sequences in ${\cal W}'$, and these new sequences, by the definition of $W'$, no longer contain $e_b$ in their projecting cycles. In this context, also note that the projecting cycles of the (one or two) new sequences enclose disjoint areas contained in the area enclosed by the projecting cycle of $S$, and the projecting cycles of the new sequences do not enclose the faces enclosed by $C$, but the projecting cycle of $S$ does enclose them. Thus, $\sum_{\widehat{S}\in \Seq({\cal W}')}\mathsf{Volume}(\widehat{S})<\sum_{\widehat{S}\in \Seq({\cal W})}\mathsf{Volume}(\widehat{S})$.

It remains to show that ${\cal W}'$ is outer-terminal and that has the same potential as $\cal W$. The second claim is immediate since $\cal W$ and ${\cal W}'$ have precisely the same crossings with $R$. For the first claim, note that since $\cal W$ is outer-terminal, it uses exactly one edge incident to $t^\star$.  The only vertex of $R$ that can possibly be incident to more edges in $E({\cal W}')$ that in $E({\cal W})$ is the other endpoint, say, $w$, of the edge of $P_3$ in the case where $P_3$ contains an edge. So, suppose that $P_3$ does contain an edge and that $w=t^\star$, else we are done. Since $t^\star$ is a leaf or $R$ that belongs to the boundary of the outer-face of $H$, it cannot be enclosed in the strict interior of the projecting cycle of $S$ and therefore it must be a vertex of $S$. However, this together with the maximality of the number of cycles enclosed by the shrinking cycle $C$ implies that the $C$ is equal to the projecting cycle of $S$. Thus, by the definition of $W'$, the only difference between the edges incident to $t^\star$ in ${\cal W}$ compared to ${\cal W}'$ is that in $\cal W$ it is incident to an edge of $S$, while in ${\cal W}'$ it is incident to the edge of $P_3$. In particular, this means that ${\cal W}'$ has exactly one edge incident to $t^\star$ and therefore it is outer-terminal. 
\end{proof}

Having Lemmas \ref{lem:solPotential}, \ref{lem:innermostSeqExists}, \ref{lem:shrinkCycExists} and \ref{lem:shrinkCycMaintain} at hand, we are ready to push a solution onto $R$. Since this part is only required to be existential rather than algorithmic, we give a simpler proof by contradiction rather than an explicit process to push the solution. Notice that once the solution has already been pushed, rather than using the notion of shallowness, we only demand to have multiplicity at most $2n$. 

\begin{lemma}\label{lem:pushSequencesFinal}
Let $(G,S,T,g,k)$ be a good \yes-instance of \pdp, and $R$ be a backbone Steiner tree. Then, there exists a sensible outer-terminal weak linkage $\cal W$ in $H$ that is pushed onto $R$, has multiplicity at most $2n$, discretely homotopic in $H$ to some solution of $(G,S,T,g,k)$ and $\Potential({\cal W})\leq \alpha_{\rm potential}(k)$.
\end{lemma}

\begin{proof}
By Lemma \ref{lem:solPotential}, there exists a solution $\cal P$ to $(G,S,T,g,k)$ such that $\Potential({\cal P})\leq \alpha_{\rm potential}(k)$. Because the paths in $\cal P$ are pairwise vertex-disjoint and $\cal P$ has a path where $t^\star$ is an endpoint, it is clear that $\cal P$ is a sensible, well-behaved, shallow and outer-terminal weak linkage. Since $\cal P$ is discretely homotopic to itself, it is well defined to let $\cal W$ a weak linkage that, among all sensible, well-behaved, shallow and outer-terminal weak linkages that are discretely homotopic to $\cal P$, minimizes $\sum_{S\in\Seq({\cal W})}\mathsf{Volume}(S)$. Notice that shallowness is a stronger demand than having multiplicity at most $2n$, and that being pushed onto $R$ is equivalent to having an empty set of sequences. Thus, to conclude the proof, it suffices to argue that $\Seq({\cal W})=\emptyset$.

Suppose, by way of contradiction, that $\Seq({\cal W})\neq\emptyset$. Then, by Lemmas \ref{lem:innermostSeqExists} and \ref{lem:shrinkCycExists}, there exist an innermost sequence $S\in\Seq({\cal W})$ and a shrinking cycle $C$ for $(S,v)$ where we pick $v$ as an endpoint of $S$ that is not $t^\star$ (because $\cal W$ is outer-terminal, not both endpoints of $S$ can be $t^\star$), and we pick a shrinking cycle enclosing as many faces as possible. By Lemma \ref{lem:shrinkCycMaintain}, the cycle move/pull operation is applicable to $(W,C)$ where $W\in{\cal W}$ is the walk having $S$ as a sequence. Furthermore, the resulting weak linkage ${\cal W}'$ is sensible, well-behaved, shallow and outer-terminal having the same potential as $\cal W$, and $\sum_{\widehat{S}\in \Seq({\cal W}')}\mathsf{Volume}(\widehat{S})<\sum_{\widehat{S}\in \Seq({\cal W})}\mathsf{Volume}(\widehat{S})$. Since discrete homotopy is an equivalence relation, ${\cal W}'$ is discretely homotopic to $\cal P$. However, this contradicts the choice of $\cal W$.
\end{proof}

\subsection{Bounding the Total Number of Segments}

Having pushed the solution onto $R$, we further need to make the resulting weak linkage simplified, which requires to make it have low multiplicity, be U-turn-free and canonical. We first show that we can focus only on the first two properties, as being canonical can be easily derived using cycle move operations on cycles consisting of two parallel edges.

\begin{lemma}\label{lem:makingCanonical}
Let $(G,S,T,g,k)$ be an instance of \pdp, and $R$ be a Steiner tree. Let $\cal W$ be a weak linkage  in $H$ that is sensible and pushed onto $R$, whose multiplicity is at most $2n$. Then, there exists a weak linkage ${\cal W}'$  that is sensible, pushed onto $R$, canonical, discretely homotopic to $\cal W$, and whose multiplicity is upper bounded by the multiplicity of $\cal W$.
\end{lemma}

\begin{proof}
Let us first argue that, there is a weak linkage $\cal W'$ that is sensible, pushed onto $R$, discretely homotopic to $\cal W$, and whose multiplicity is upper bounded by the multiplicity of $\cal W$ such that for all edges $e_i \in E({\cal W'})$ $i \geq i$. In other words, $\cal W'$ contains only edges of positive subscript.
Consider all weak linkages in $H$ that are sensible, pushed onto $R$, discretely homotopic to $\cal W$, and whose multiplicity is upper bounded by the multiplicity of $\cal W$. Among these weak linkages, let ${\cal W}'$ be one such that 
the sum of the subscripts of the edge copies in $E({\cal W}')$ is maximized. 
Let us argue that for every $e_i \in E({\cal W})$, $i \geq 1$ and for every $j > i$,
$e_i \in E({\cal W})$.
Suppose not, and there exists an edge $e_i$ that is used by ${\cal W}'$ such that {\em (i)} $i\leq 0$ (in fact, $i=0$ is not possible as ${\cal W}'$ is pushed onto $R$), or {\em (ii)} there exists an edge $e_j$ (parallel to $e_i$) for some $1\leq j < i$ that is not used by $E({\cal W}')$. Since ${\cal W}'$ has multiplicity at most $2n$, the satisfaction of the first condition implies the satisfaction of the second one, thus there exists such an edge $e_j$. 
Let $e_t$ be the edge (parallel to $e_j$ and $e_i$) of largest  $j$ that is used by $E({\cal W}')$. Moreover, let $C$ be the cycle (which might be the boundary of a single face) that consists of two edges: $e_j$ and $e_t$. By the choice of $e_t$, the strict interior of $C$ does not contain any edge of $E({\cal W}')$. Thus, the cycle move operation is applicable to $(W,C)$ where $W$ is the walk in ${\cal W}'$ that uses $e_t$. Let ${\cal W}^\star$ be the result of the application of this operation. Then, the only difference between ${\cal W}^\star$ and ${\cal W}'$ is the replacement of $e_t$ by $e_j$.

Because ${\cal W}^\star$ is discretely homotopic to ${\cal W}'$, ${\cal W}'$ is discretely homotopic to $\cal W$ and discrete homotopy is transitive, we derive that ${\cal W}^\star$ is discrete  Moreover, the endpoints of the walks in ${\cal W}'$ were not changed when the cycle move operation was applied. Thus, because ${\cal W}'$ is sensible, so is ${\cal W}^\star$. Moreover, it is clear that ${\cal W}^\star$ is pushed onto $R$ and has the same multiplicity as ${\cal W}'$. However, as $j>t$, the sum of the subscripts of the edge copies in $E({\cal W}^\star)$ is larger than that of $E({\cal W}')$, which contradicts the choice of ${\cal W}'$. 

Hence, there exist weak linkages $\cal W'$ that are sensible, pushed onto $R$, discretely homotopic to $\cal W$, whose multiplicity is upper bounded by the multiplicity of $\cal W$, and for all edges $e_i \in E({\cal W'})$ $i \geq i$. 
Consider the collection of all such weak linkages, and let $\cal W^\star$ be the one maximizing $w(E({\cal W^\star})) = \sum{e \in E({\cal W^\star})} w(e)$ where 
Let us begin by defining a weight function $w:E(H) \rightarrow \mathbb{Z}$ on the parallel copies of edges in $H$ as follows.
$$ w(e) = \begin{cases}
      -2n   & e \text{ is not parallel to any edge in } E(R) \\
      -2n   & e = e_i \text{ is parallel to an edge in } E(R) \text{ and } i \leq 0 \\
       2n-i & e = e_i \text{ is parallel to an edge in } E(R) \text{ and } i \geq 1
           \end{cases} 
$$
We claim that $\cal W^\star$ is canonical, i.e. for every edge $e_i \in E({\cal W^\star})$ the subscript $i \geq 1$, and for every parallel edge $e_j$ where $i \leq j < i$, $e_j \in E({\cal W})$. The first property is ensured by the choice of $\cal W^\star$. For the second property, we argue as before. Suppose not, then choose $i$ and $j$ such that $i-j$ is minimized. Then clearly $j = i-1$, since any parallel copy $e_t$ with $j < t < i$ is either in $E({\cal W^\star})$ contradicting the choice of $i$, or not in $E({\cal W^\star})$ contradicting the choice if $j$. Therefore, the edges $e_i$ and $e_j$ form a cycle $C$ such that the interior of $C$ contains no edge of any walk in $\cal W^*$. Let $W \in {\cal W}^\star$ be the walk containing $e_i$, and observe that the cycle move operation is applicable to $(W,C)$.
Let $\wh{\cal W}$ be the result of this operation. Then observe that $w(E(\wh{\cal W})) > w(E({\cal W^\star}))$, since $w(e_j) > w(e_i)$ and $E(\wh{\cal W}) \setminus \{e_j\} = E({\cal W^\star}) \setminus \{e_i\}$. And, $\wh{\cal W}$ is discretely homotopic to $\cal W'$ which is in turn discretely homotopic to $\cal W$, as before we can argue that $\wh{\cal W}$ contradicts the choice of $\cal W^\star$. Hence, $\cal W^\star$ must be canonical.
\end{proof}

In case we are interested only in extremality rather than canonicity, we can use the following lemma that does not increase potential. The proof is very similar to the proof of Lemma \ref{lem:makingCanonical}, except that now we can ``move edges in either direction'', and hence avoid creating new crossings.

\begin{lemma}\label{lem:makingExtremal}
Let $(G,S,T,g,k)$ be an instance of \pdp, and $R$ be a Steiner tree. Let $\cal W$ be a weak linkage  in $H$ that is sensible and pushed onto $R$, whose multiplicity is at most $2n$. Then, there exists a weak linkage ${\cal W}'$  that is sensible, pushed onto $R$, extremal, discretely homotopic to $\cal W$, and whose potential is upper bounded by the potential of $\cal W$.
\end{lemma}

\begin{proof}
Consider all weak linkages in $H$ that are sensible, pushed onto $R$, discretely homotopic to $\cal W$, and whose potential and multiplicity are upper bounded by the potential and multiplicity, respectively, of $\cal W$. Among these weak linkages, let ${\cal W}'$ be one such that the sum of the absolute values of the subscripts of the edge copies in $E({\cal W}')$ is maximized. We claim that ${\cal W}'$ is extremal, which will prove the lemma. To this end, suppose by way of contradiction that ${\cal W}'$ is not extremal. Thus, there exist an edge $e_i,e_j\in{\cal W}'$ where $i\geq 1, j\leq -1$ and $(i-1)+|j+1|\leq 2n-1$. Because the multiplicity of ${\cal W}'$ is at most $2n$, this means that there exists an edge $e_p$ (parallel to $e_i$ and $e_j$) for some $p>i\geq 1$ that is not used by $E({\cal W}')$. Let $e_t$ be the edge (parallel to $e_j$ and $e_i$) of largest subscript smaller than $p$ that is used by $E({\cal W}')$. Moreover, let $C$ be the cycle (which might be the boundary of a single face) that consists of two edges: $e_p$ and $e_t$. By the choice of $e_t$, the strict interior of $C$ does not contain any edge of $E({\cal W}')$. Thus, the cycle move operation is applicable to $(W,C)$ where $W$ is the walk in ${\cal W}'$ that uses $e_t$. Let ${\cal W}^\star$ be the result of the application of this operation. Then, the only difference between ${\cal W}^\star$ and ${\cal W}'$ is the replacement of $e_t$ by $e_p$.

Because ${\cal W}^\star$ is discretely homotopic to ${\cal W}'$, ${\cal W}'$ is discretely homotopic to $\cal W$ and discrete homotopy is transitive, we derive that ${\cal W}^\star$ is discrete  Moreover, the endpoints of the walks in ${\cal W}'$ were not changed when the cycle move operation was applied. Thus, because ${\cal W}'$ is sensible, so is ${\cal W}^\star$. Moreover, it is clear that ${\cal W}^\star$ is pushed onto $R$, because ${\cal W}^\star$ and ${\cal W}'$ cross $R$ exactly in the same vertices and in the same direction (indeed, we have only replaced one edge of positive subscript by another parallel edge of positive subscript), they have the same potential. However, as $p>t\geq 1$, the sum of the absolute values of the subscripts of the edge copies in $E({\cal W}^\star)$ is larger than that of $E({\cal W}')$, which contradicts the choice of ${\cal W}'$.
\end{proof}

To achieve the properties of having low multiplicity and being U-turn-free, we perform two stages. In the first stage, that is the focus of this subsection, we make modifications that bound the total number of segments (rather than only the number of segments groups). The second stage, where we conclude the two properties, will be performed in the next subsection. The first stage in itself is partitioned into two phases as follows.

\paragraph{Phase I: Eliminating Special U-Turns.} We eliminate some of the U-turns, but not all of them. Specifically, the elimination of some U-turns may result in too major changes in the segment groups, and hence we only deal with them after we bound the total number of segments, in which case classification into segment groups becomes immaterial. The U-turns we eliminate now are defined as follows  (see Fig.~\ref{fig:uturn}).

\begin{figure}
    \begin{center}
        \includegraphics[scale=0.7]{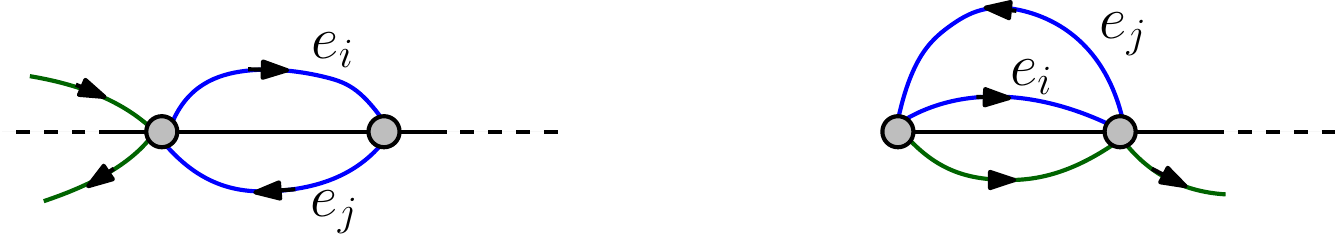}
        \caption{Special U-Turns.}
        \label{fig:uturn}
    \end{center}
\end{figure}

\begin{definition}[{\bf Special U-Turn}]
Let $(G,S,T,g,k)$ be an instance of \pdp, and $R$ be a Steiner tree. Let $\cal W$ be weak linkage that is pushed onto $R$. Let $U=\{e_i,e_j\}$ be a U-turn. Then, $U$ is {\em special} if at least one among the following conditions hold:  {\em (i)} $i$ and $j$ have the same sign, i.e. they are on the same side of $R$; {\em (ii)} both endpoints of $e_i$ and $e_j$ do not belong to $V_{=1}(R)\cup V_{\geq 3}(R)$.
\end{definition}

We eliminate special U-turns one-by-one, where the U-turn chosen to eliminate at each step is an innermost one, defined as follows. 

\begin{definition}[{\bf Innermost U-Turn}]\label{def:innermostUTurn}
Let $(G,S,T,g,k)$ be an instance of \pdp, and $R$ be a Steiner tree. Let $\cal W$ be weak linkage that is pushed onto $R$. Let $U=\{e_i,e_j\}$ be a U-turn. Then, $U$ is {\em innermost} if there does not exist a parallel edge $e_\ell\in E({\cal W})$ such that $\min\{i,j\}\leq \ell\leq \max\{i,j\}$. We say that $U$ is {\em crossing} is the signs of $i$ and $j$ are different, i.e. $e_i$ and $e_j$ are on opposite sides of $R$.
\end{definition}

We argue that if there is a (special) U-turn, then there is also an innermost (special) one.

\begin{lemma}\label{lem:existsInnermostUTurn}
Let $(G,S,T,g,k)$ be an instance of \pdp, and $R$ be a Steiner tree. Let $\cal W$ a weak linkage pushed onto $R$ with at least one U-turn $U=\{e_i,e_j\}$. Then, $\cal W$ has at least one innermost U-turn $U'=\{e_x,e_y\}$ whose edges lies in the interior (including the boundary) of the cycle $C$ formed by $e_i$ and $e_j$. 
\end{lemma}

\begin{proof}
Denote the endpoints of $e_i$ and $e_j$ by $u$ and $v$. Among all U-turns whose edges lies in the interior (including the boundary) of the cycle $C$ formed by $e_i$ and $e_j$ (because $U$ satisfies these conditions, there exists at least one such U-turn), let $U'=\{e_x,e_y\}$ be one whose edges $e_x$ and $e_y$ form a cycle $C'$ that contains minimum number of edges of $H$ in its interior. Let $W'$ be the walk in ${\cal W}$ that traverses $e_x$ and $e_y$ consecutively. Without loss of generality, suppose that when we traverse $W'$ so that we visit $e_x$ and then $e_y$, we first visit $u$, then $v$, and then $u$ again.

 We claim that $U'$ is innermost. To this end, suppose by way of contradiction that $U'$ is not innermost. Thus, by Definition \ref{def:innermostUTurn}, this means that $C'$ contains an edge $e_\ell$ in its strict interior that belongs to some walk $\widehat{W}\in{\cal W}$ (possibly $\widehat{W}=W'$). Because $U'$ is a U-turn, $e_\ell$ is neither the first nor the last edge of $\widehat{W}$. Thus, when we traverse $\widehat{W}$ so that when we visit $e_\ell$, we first visit $u$ and then $v$, we next visit an edge $e'$. Because $\cal W$ is a weak linkage, this edge must belong to the strict interior of $C'$ (because otherwise we obtain that $(v,e_x,e_y,e_\ell,)$ is a crossing or an edge is used more than once). However, this implies that $e'$ is parallel to the edges $e_\ell, e_x, e_y$, and $\widehat{U}=\{e_\ell,e'\}$ is a U-turn whose edges lies in the interior (including the boundary) of the cycle $C$ and which forms a cycle $\widehat{C}$ that contains fewer edge of $H$ than $C'$ in its interior. This is a contradiction to the choice of $U'$.
\end{proof}

We now prove that an innermost U-turn corresponds to a cycle on which we can perform the cycle pull operation, and consider the result of its application.

\begin{lemma}\label{lem:eliminateInnermostUTurn}
Let $(G,S,T,g,k)$ be an instance of \pdp, and $R$ be a Steiner tree. Let $\cal W$ be a sensible, outer-terminal, extremal weak linkage that is pushed onto $R$, and let $U=\{e_i,e_j\}$ be an innermost U-turn. Let $W$ be the walk in $\cal W$ that uses $e_i$ and $e_j$, and $C$ be the cycle in $H$ that consists of $e_i$ and $e_j$. Then, the cycle pull operation is applicable to $(W,C)$. Furthermore, the resulting weak linkage ${\cal W}'$ is sensible, outer-terminal, extremal, pushed onto $R$, having fewer edges than $\cal W$, $|\Seg({\cal W}')|\leq |\Seg({\cal W})|$, and if $U$ is special, then also its potential is upper bounded by the potential of $\cal W$.
\end{lemma}

\begin{proof}
Because $U$ is innermost, there does not exist an edge in the strict interior of $C$ that belongs to $E({\cal W})$. Therefore, the cycle pull operation is applicable to $(W,C)$. The only difference between ${\cal W}'$ and ${\cal W}$ is that ${\cal W}'$ does not use the edge $e_i$ and $e_j$ and hence the vertex, say, $v$, that $W$ visits between them. Therefore, because $\cal W$ be a sensible, outer-terminal , extremal and  pushed onto $R$, so is ${\cal W}'$. Moreover, the walks in ${\cal W}'$ have the same endpoints as their corresponding walks in $\cal W$, and thus because $\cal W$ is sensible, so is ${\cal W}'$. Let $u$ be the other endpoints of the edges $e_i$ and $e_j$, and let $W'$ be the walk in ${\cal W}'$ that resulted from $W$. Observe that $W'$ has at most as many crossings with $R$ as $W$ has---indeed, if the elimination of $e_i$ and $e_j$ created a new crossing at $u$ (this is the only new crossing that may be created), then $W$ crosses $R$ between $e_i$ and $e_j$ and this crossing does not occur in $W'$. Thus, $|\Seg({\cal W}')|\leq |\Seg({\cal W})|$.

\begin{figure}
    \begin{center}
        \includegraphics[scale=0.8]{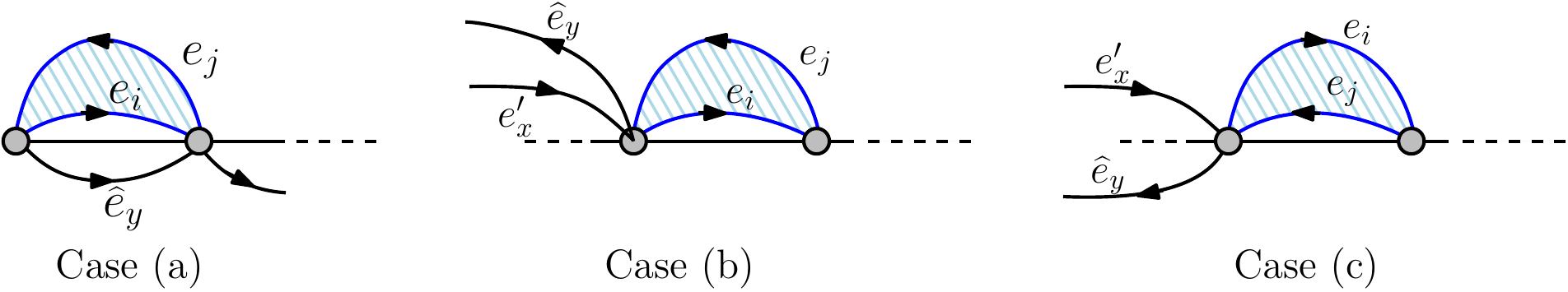}
        \caption{U-Turn with $e_i$ and $e_j$ on the same side}
        \label{fig:uturn1}
    \end{center}
\end{figure}

Now, suppose that $U$ is special. Then, at least one among the following conditions holds:  {\em (i)} $i$ and $j$ have the same sign; {\em (ii)} $u,v\notin V_{=1}(R)\cup V_{\geq 3}(R)$.  We first consider the case where $i$ and $j$ have the same sign, say positive (without loss of generality; see Fig.~\ref{fig:uturn1}). Let $e'_x,\widehat{e}_y\in E(W)$ be such that $e'_x$ and $e_i$ are consecutive in $W$ (if such an edge $e'_x$ exists) and denote the segment that contains $e'_x$ by $S_x$, and $\widehat{e}_y$ and $e_j$ are consecutive in $W$ (if such an edge $\widehat{e}_y$ exists) and denote the segment that contains $\widehat{e}_y$ by $S_y$. Possibly some edges among $e'_x,\widehat{e}_y$ and $e_i$ are parallel. In case $e'_x$ does not exist (see Fig.~\ref{fig:uturn1}(a)), then $u\in V_{=1}(R)$ and therefore $e_i$ and $e_j$ belong to a segment that is in a singleton segment group. When we remove $e_i$ and $e_j$, either this segment shrinks (and remains in a singleton group) or it is removed  completely together with its segment group. If the segment shrinks, the potential clearly remains unchanged, and otherwise the reduction of segment groups makes the potential decrease by $1$ (the potential of the consecutive segment group remains unchanged as it is a singleton segment group because it has an endpoint in $V_{=1}(R)$). The case where $\widehat{e}_y$ does not exist is symmetric, thus we now assume that both $e'_x$ and $\widehat{e}_y$ exist. In case both $x$ and $y$ are on the same side as $e_i$ and $e_j$ 
(see Fig.~\ref{fig:uturn1}(b)), then the removal of $e_i$ and $e_j$ only shrinks the segment $S_x=S_y$ where all of the four edges $e_i,e_j,e'_x$ and $\widehat{e}_y$  lie, and thus does not change the potential. Similarly, if $e'_x$ is on the same side as $e_i, e_j$ 
and $\wh{e}_y$ is on the opposite side 
(see Fig.~\ref{fig:uturn1}(c)), then we only shrink the segment where $e'_x, e_i$ and $e_j$ lie, and rather than crossing from $e_j$ to $\widehat{e}_y$, we cross from $e'_x$ to $\widehat{e}_y$ (which have the same label as both cross from the positive side to the negative side). The case where $e'_x$ is on the opposite side of $e_i,e_j$ and $\wh{e}_y$ is on the same side is not possible, since then $W$ crosses itself.

Now, consider the case where both $e'_x$ and $\wh{e}_y$ are both on the opposite side of $e_i,e_j$  (see Fig.~\ref{fig:uturn2}) Then, $e_i$ and $e_j$ form a complete segment, which we call $S_{ij}$, and the segments $S_x,S_{ij}$ and $S_y$ are different. Notice that the two crossings with $R$, one consisting of $e'_x$ with $e_i$, and the other consisting of $e_j$ with $\widehat{e}_y$, cross in different directions. If both $S_x$, $S_{ij}$ and $S_y$ belong to the same segment group (see Fig.~\ref{fig:uturn2}(a)), then the removal of $S_{ij}$ removes the contributions of its two crossings (mentioned above), which sum to $0$ as their direction is opposite. Then, the potential remains unchanged. Now, suppose that exactly one among $S_x$ and $S_y$ belongs to the same group as $S_{ij}$. Without loss of generality, suppose that it is $S_x$ (the other case is symmetric; see Fig.~\ref{fig:uturn2}(b)). Then, when we remove $S_{ij}$, the segments $S_x$ and $S_y$ merge into one segment that has endpoints in different maximal degree-2 paths in $R$ (or on vertices of degree other than $2$) and hence forms its own group. This group replaces the singleton group of $\cal W$ that contained only $S_y$. Furthermore, the labeling of all crossings remain the same (as well as all other associations into segment groups), apart from the two crossings consisting of $e'_x$ with $e_i$, and of $e_j$ with $\widehat{e}_y$, which are both eliminated but have previously contributed together $0$ (as they cross in opposite directions). We remark that the size of the segment group that previously contained $S_x$ might become $1$ or completely removed, but this does not increase potential. Thus, overall the potential does not increase.

\begin{figure}
    \begin{center}
        \includegraphics[scale=0.8]{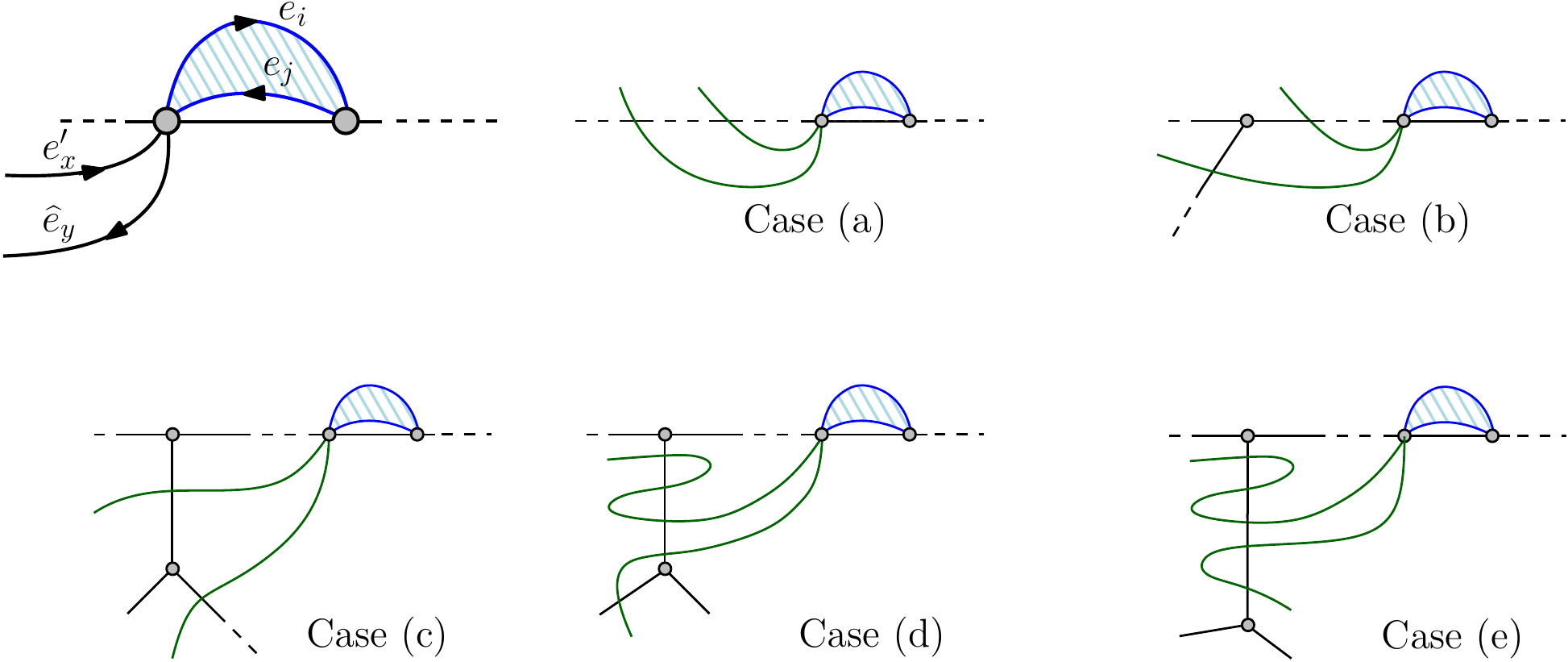}
        \caption{U-Turn with $e_i,e_j$ on one side and $e'_x,\wh{e}_y$ is on the opposite side.}
        \label{fig:uturn2}
    \end{center}
\end{figure}

Lastly, suppose that $S_x,S_{ij}$ and $S_y$ belong to different segment groups. Then, $S_x$ and $S_y$ had endpoints in different maximal degree-2 paths in $R$ (or on vertices of degree other than $2$), hence each one among $S_x,S_{ij}$ and $S_y$ belonged to a singleton segment group of $\cal W$. The removal of $S_{ij}$ eliminates all of these three segment groups, which results in a decrease of $3$ in the potential. However, now $S_x$ and $S_y$ belong to the same segment group. If they form a singleton segment group (see Fig.~\ref{fig:uturn2}(c)), the overall the potential decreases by $2$. Else, they join an existing segment group, and we have several subcases as follows. In the first subcase, suppose that they join only the group that contains the segment $\widetilde{S}_x$ of ${\cal W}$ consecutive to $S_x$ (in the walk in $\cal W$ to which $S_x,S_{ij}$ and $S_y$ belong; see Fig.~\ref{fig:uturn2})(d)). Then, the crossing at the endpoint of $S_y$ is now contributing (1 or -1) to the sum of labels in the potential, and if $\widetilde{S}_x$ was in a singleton group in $\cal W$, then so are its crossings. Overall, this results in a contribution of at most $3$, so in total the potential does not increase. The subcase where they join only the group that contains the segment $\widetilde{S}_y$ of ${\cal W}$ consecutive to $S_y$ is symmetric. Now, consider the subcase where they join both of these groups and hence merge them (see Fig.~\ref{fig:uturn2}(e)). In this subcase, we have four new crossings that may contribute to the sum of labels, but we have also merged two groups which makes the potential decreased by at  least $1$, so overall the potential does not increase.

Now, suppose that only case {\em (ii)} holds. That is, $u,v\notin V_{=1}(R)\cup V_{\geq 3}(R)$, and $i$ and $j$ have the different signs (see Fig.~\ref{fig:uturn3}). Without loss of generality, suppose that $i\geq 1$ and $j\leq -1$.  Because $u,v\notin V_{=1}(R)\cup V_{\geq 3}(R)$ and $\cal W$ is sensible, $e'_x$ and $\widehat{e}_y$ exist. The case where $e'_x$ is on the opposite side of $e_i$ and $\wh{e}_y$ is on the same side as $e_i$ cannot occur since then $W$ crosses itself. Thus, we are left with three cases: {\em (a)} $e'_x$ is on the same side as $e_i$ and $\wh{e}_y$ is on the opposite side of $e_i$; {\em (b)} both $e'_x$ and $\wh{e}_y$ are on the opposite side of $e_i$; and {\em (c)} both $e'_x, \wh{e}_y$ are on the same side as $e_i$. The cases {\em (b)} and {\em (c)} are symmetric, therefore we will only consider cases {\em (a)} and {\em (b)}. In case {\em (a)}, $e'_x$ and $e_i$ belong to one segment, and $e_j$ and $\widehat{e}_y$ belong to a different segment (see Fig.~\ref{fig:uturn3}(a)). The removal of $e_i$ and $e_j$ only shirks these two segments by one edge each, and does not change the labeling of the crossings at their endpoints---previously, we crossed from $e_i$ to $e_j$, and now we cross from $e'_x$ to $\widehat{e}_y$, which are both crossings from the side of $e_i$ to the opposite side. Thus, the potential does not increase.

Lastly, consider case {\em (b)} (see Fig.~\ref{fig:uturn3}(b)). In this case, $e_i$ belongs to a segment $S_i$ containing only $e_i$, and $e_j$ belongs to $S_y$. Further, when crossing from $e_x$ to $e_i$, we cross from the opposite side of $e_i$ to the same side, and when we cross from $e_i$ to $e_j$, we cross from the side of $e_i$ to the opposite side. Additionally, notice that the elimination of $e_i$ and $e_j$ results in the elimination of $S_i$, and the merge of $S_x$ and $S_y$ with $e_j$ removed. We consider several subcases as follows. In the subcase where $S_x,S_i$ and $S_y$ belong to the same group (see Fig.~\ref{fig:uturn3}(c)), then the only possible effect with respect to the potential of this group is the cancellation of the two crossings (of $e'_x$ with $e_i$ and of $e_i$ with $e_j$), but these two crossings together contribute $0$ to the sum of labels because they cross in opposite directions. Possibly the size of the segment group shrunk to $1$, but this does not increase potential. Thus, in this subcase, the potential does not increase. 

\begin{figure}
    \begin{center}
        \includegraphics[scale=0.8]{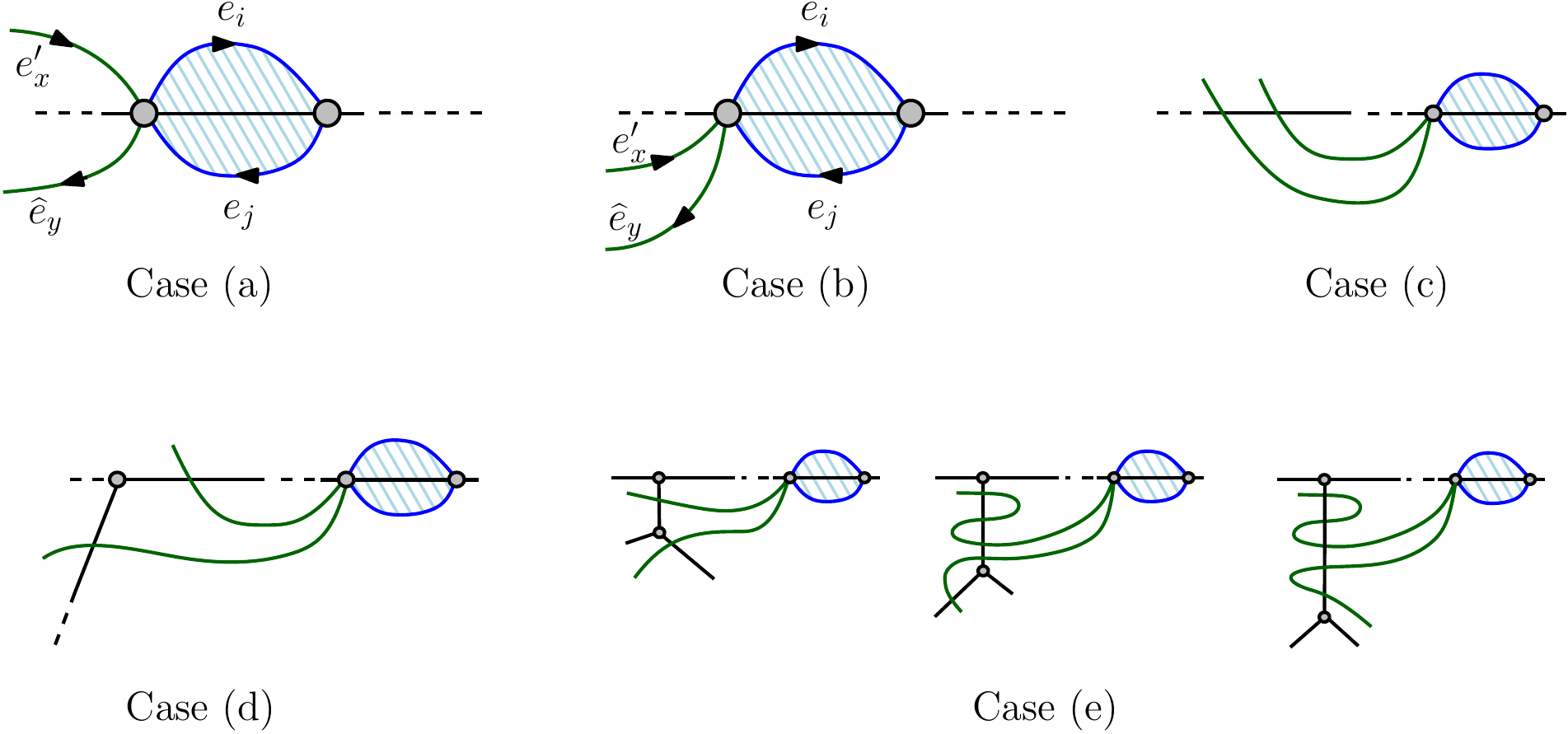}
        \caption{U-Turn with $e_i$ on one side and $e_j$ on the opposite side.}
        \label{fig:uturn3}
    \end{center}
\end{figure}

Now, consider the subcase where $S_x$ and $S_i$ are in the same segment group, and $S_y$ is in a different segment group (see Fig.~\ref{fig:uturn3}(d)).  Then, $S_y$ is in a singleton segment group because its endpoints belong to different maximal degree-2 path of $R$ (or vertices of degree other than $2$). In ${\cal W}'$, the segment group that resulted from the merge of $S_x$ and $S_y$ is also a singleton segment group.  Furthermore, the segment group of $\cal W$ that contained $S_x$ and $S_i$ does not change in terms of its labeled sum since the crossings at the endpoints of $S_i$ crossed in opposite directions. Possibly the size of the segment group shrunk to $1$, but this does not increase potential. Thus, in this subcase, the potential does not increase.  Next, we note that the analysis of the subcase where $S_i$ and $S_y$ are in the same segment group, and $S_x$ is in a different segment group, is symmetric.  
Lastly, suppose that $S_x,S_{i}$ and $S_y$ belong to different segment groups (see Fig.~\ref{fig:uturn3}(e)).  The analysis of this case is the same as the analysis of the last subcase of case {\em (i)} (i.e., the subcase where $S_x,S_{ij}$ and $S_y$ belong to the same segment group, where now we have $S_i$ instead of $S_{ij}$).
\end{proof}

We are now ready to assert that all special U-turns can be eliminated as follows.

\begin{lemma}\label{lem:noUTurns}
Let $(G,S,T,g,k)$ be a good \yes-instance of \pdp, and $R$ be a backbone Steiner tree. Then, there exists a sensible, outer-terminal, extremal weak linkage $\cal W$ in $H$ that has no special U-turns, is pushed onto $R$ and discretely homotopic in $H$ to some solution of $(G,S,T,g,k)$ and $\Potential({\cal W})\leq \alpha_{\rm potential}(k)$.
\end{lemma}

\begin{proof}
By Lemma \ref{lem:pushSequencesFinal},  there exists a sensible outer-terminal weak linkage in $H$ that is pushed onto $R$, has multiplicity at most $2n$, discretely homotopic in $H$ to some solution of $(G,S,T,g,k)$ and has potential at most $\alpha_{\rm potential}(k)$. By Lemma \ref{lem:makingExtremal} and because discrete homotopy is an equivalence relation, there also exists such a weak linkage ${\cal W}'$ that is extremal. Thus, there exists a weak linkage $\cal W$ that among all sensible, outer-terminal, extremal weak linkages in $H$ that are pushed onto $R$, discretely homotopic in $H$ to some solution of $(G,S,T,g,k)$ and satisfy $\Potential({\cal W})\leq \alpha_{\rm potential}(k)$, the weak linkage $\cal W$ is one that minimizes the number of edges that it uses. To conclude the proof, it suffices to argue that $\cal W$ has no special U-turns.

Suppose, by way of contradiction, that $\cal W$ has at least one special U-turn. Then, by Lemma \ref{lem:existsInnermostUTurn}, $\cal W$ has an innermost special U-turn $U=\{e_i,e_j\}$. Let $W$ be the walk in $\cal W$ that uses $e_i$ and $e_j$, and $C$ be the cycle in $H$ that consists of $e_i$ and $e_j$. Then, by Lemma \ref{lem:eliminateInnermostUTurn}, the cycle pull operation is applicable to $(W,C)$. Furthermore, by Lemma \ref{lem:eliminateInnermostUTurn}, the resulting weak linkage ${\cal W}'$ is sensible, outer-terminal, extremal, pushed onto $R$, has fewer edges than $\cal W$, and its potential is upper bounded by the potential of $\cal W$. Since discrete homotopy is an equivalence relation, ${\cal W}'$ is discretely homotopic to some solution of $(G,S,T,g,k)$. However, this is a contradiction to the choice of $\cal W$.
\end{proof}

\paragraph{Phase II: Eliminating Swollen Segments.} The goal of the second phase is to eliminate the existence crossings with opposing ``signs'' for each segment and thereby, as the potential is bounded, bound the number of segments (rather than only the number of segment groups). We remark that one can show, even without this step, that the multiplicity is bounded, however this complicates the analysis. Towards this, we eliminate ``swollen'' segments (see Fig.~\ref{fig:swollen-seg}).
\begin{figure}
    \begin{center}
        \includegraphics[scale=0.7]{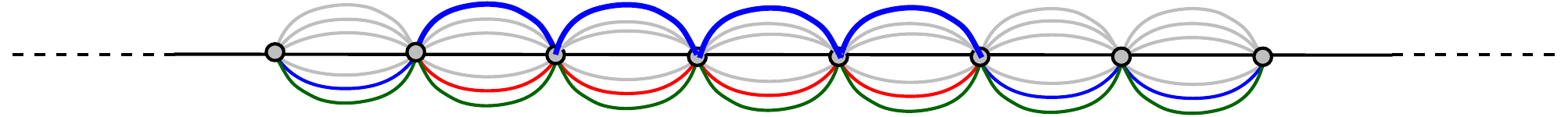}
        \caption{A swollen segment and the cycles in its move-through tuple.}
        \label{fig:swollen-seg}
    \end{center}
\end{figure}

\begin{definition}[{\bf Swollen Segment}]
Let $(G,S,T,g,k)$ be an instance of \pdp, and $R$ be a Steiner tree. Let $\cal W$ be weak linkage that is pushed onto $R$. Consider segment $S\in\Seg(W)$ for some $W\in{\cal W}$ such that $S$ does not contain two the extreme edges of $W$. Let $e$ and $e'$ be the extreme edges of $S$ (possibly $e=e'$), and let $\widehat{e}$ and $\widehat{e}'$ be the edges of $E(W)\setminus E(S)$ that are consecutive to $e$ and $e'$ on $W$, respectively. Then, $S$ is {\em swollen} if its endpoints are internal vertices of the same maximal degree-2 path $P$ of $R$ and $\lab_P^{W}(e,e')\neq \lab_P^{W}(\widehat{e},\widehat{e}')$.
\end{definition}

We show that due to the first phase, when we deal with outer-terminal weak linkages, the swollen segments have a ``clean appearance'' as stated in the following lemma.

\begin{lemma}\label{lem:cleanSwollenSegments}
Let $(G,S,T,g,k)$ be a nice instance of \pdp, and $R$ be a Steiner tree. Let $\cal W$ be an outer-terminal weak linkage that is pushed onto $R$ and has no special U-turns, and $S\in\Seg({\cal W})$ be swollen. Then, $S$ is parallel to a subpath of a maximal degree-2 path of $R$.  
\end{lemma}

\begin{proof}
Let $e_i$ and $e'_j$ be the first and last edges of $S$, and denote the endpoints of $S$ by $u$ and $v$ where $u$ is an endpoint of $e_i$. Because $S$ is a swollen segment, 
both these edges are on the same side of $\pathT_{R}(u,v)$.
Let $P$ be the unique subpath in $R$ between $u$ and $v$. Because  $\cal W$ does not have any special U-turn, we have that one of the following cases occurs (see Fig.~\ref{fig:swollenSegPath}): {\em (i)} $S$ traverses a path that starts at $e_i$, consists of edges parallel to $P$ and ends at $e'_j$; {\em (ii)} $S$ traverses a path that starts at $e_i$, consists of edges parallel to $P$ but does not end at $e'_j$, and hence (to reach $e'_j$ without having U-turns) $S$ traverses at least two copies (on opposite sides) of every edge of $R$; {\em (iii}) the first edge that $S$ traverses after $e_i$ is not parallel to an edge of $P$, and hence (to reach $e'_j$ without having U-turns) $S$ traverses at least two copies (on opposite sides) of every edge of $R$ except possibly for the edges of $P$. In the first case, we are done. In the other two cases, we have that $E({\cal W}$ contains more than one copy of the edge incident to $t^\star$ in $R$, which contradicts the assumption that $\cal W$ is outer-terminal.
\end{proof}

The segment chosen to move at each step is an innermost one, formally defined as follows.

\begin{definition}[{\bf Innermost Swollen Segment}]
Let $(G,S,T,g,k)$ be an instance of \pdp, and $R$ be a Steiner tree. Let $\cal W$ be weak linkage that is pushed onto $R$. Let $S\in\Seg({\cal W})$ be swollen. Then, $S$ is {\em innermost} if there do not exist parallel edges $e_i\in E(S)$ and $e_j\in E({\cal W})\setminus E(S)$ such that $i$ and $j$ have the same sign and $|j|<|i|$.
\end{definition}

We now argue that if there is a swollen segment, then there is also an innermost one.

\begin{lemma}\label{lem:existsInnermostSwollenSeg}
Let $(G,S,T,g,k)$ be an instance of \pdp, and $R$ be a Steiner tree. Let $\cal W$ an outer-terminal weak linkage that has no special U-turns and is pushed onto $R$, such that $\Seg({\cal W})$ contains at least one swollen segment. Then, $\Seg({\cal W})$ contains at least one innermost swollen segment.
\end{lemma}

\begin{proof}
Let $S$ be a swollen segment of $\cal W$ such that the sum of the absolute values of the indices of the edge copies it uses is minimized. By Lemma \ref{lem:cleanSwollenSegments}, $S$ is parallel to a subpath $P$ of a maximal degree-2 path of $R$. Thus, because $S$ is a segment, all the edge copies it uses are on the same side.
We claim that $S$ is innermost. Suppose, by way of contradiction, that this claim is false. Thus there exist parallel edges $e_i\in E(S)$ and $e_j\in E({\cal W})\setminus E(S)$ such that $i$ and $j$ have the same sign and $|j|<|i|$. Let $S'$ be the segment of ${\cal W}'$ to which $e_j$ belongs. Because $\cal W$ has no special U-turns and because weak linkages contain neither crossings not repeated edges, it follows that $S'$ is parallel to a subpath $Q$ of $P$ and consists only of edge copies whose indices strictly smaller absolute value than the edges of $P$ they are parallel to. However, this implies that $S'$ is a swollen segment of $\cal W$ such that the the sum of the absolute value of the indices of the edge copies it uses is smaller than the sum of $S$. This contradicts the choice of $S$.
\end{proof}

Given an innermost swollen segment whose copies have, on one side of $R$, we would like to move the segment to ``the other side'' of $R$. 
We know that these copies will be free in case we handle an extremal weak linkage.  We now define a tuple of cycles on which we will perform move operations  (see Fig.~\ref{fig:swollen-seg}). The fact that this notion is well-defined  (in the sense that the indices $\ell$ in the definition exist) will be argued in the lemma that follows it.

\begin{definition}[{\bf Move-Through Tuple}]
Let $(G,S,T,g,k)$ be an instance of \pdp, and $R$ be a Steiner tree. Let $\cal W$ an outer-terminal extremal weak linkage that has no special U-turns and is pushed onto $R$, and let $S\in\Seg({\cal W})$ be an innermost swollen segment. Let $e^1_{i_1},e^2_{i_2},\ldots,e^t_{i_t}$, where $t=|E(S)|$, be the edges of $S$ in the order occurred when $S$ is traversed from one endpoint to another.\footnote{To avoid ambiguity in the context of this definition, suppose that we have a fixed choice (e.g., lexicographic) of which endpoint is traversed first.} Then, the {\em move-through tuple} of $S$ is $T=(C_1,\ldots,C_t)$ where for every $j\in\{1,\ldots,t\}$, $C_j$ is a cycle that consists of two parallel edges: $e^j_{i_j}$ and $e^j_\ell$ where $\ell$ is the index of sign opposite to $i_j$ that has the largest absolute value such that all indices $r$ of the same sign as $\ell$ and whose absolute value is upper bounded by $|\ell|$ satisfy that $e^j_r\notin E({\cal W})$.

The {\em application of $T$} is done by applying the cycle move operation to $(W,C_i)$ for $i$ from $1$ to~$t$ (in this order) where $W$ is the walk that contains $S$ as a segment.\footnote{Note that $W$ changes in each application, thus by $W$ we mean the current walk with the same endpoints as the original walk in $\cal W$ that had $S$ as a segment.}
\end{definition}
\begin{figure}
    \begin{center}
        \includegraphics[scale=0.5]{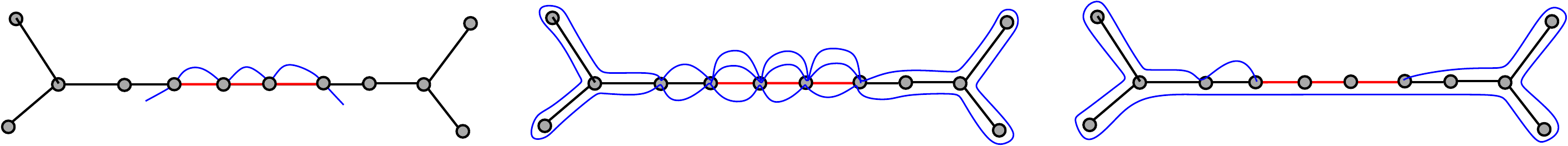}
        \caption{Illustration of Lemma~\ref{lem:cleanSwollenSegments}.}
        \label{fig:swollenSegPath}
    \end{center}
\end{figure}
Now, we prove that application of a move-through tuple is valid.

\begin{lemma}\label{lem:moveThroughTupleValid}
Let $(G,S,T,g,k)$ be an instance of \pdp, and $R$ be a Steiner tree. Let $\cal W$ an outer-terminal extremal weak linkage that has no special U-turns and is pushed onto $R$, and let $S\in\Seg({\cal W})$ be an innermost swollen segment. Then, the move-through tuple $T=(C_1,\ldots,C_t)$ of $S$ is well-defined, and the application of $T$ is valid (that is, the cycle move operation is applicable to $(W,C_i)$ when it is done).
\end{lemma}

\begin{proof}
Let $e^1_{i_1},e^2_{i_2},\ldots,e^t_{i_t}$, where $t=|E(S)|$, be the edges of $S$ in the order occurred when $S$ is traversed from one endpoint to another. To assert that $T$ is valid, consider some $j\in\{1,\ldots,t\}$. 
We need to show that $e^j_b$, where $b\in\{-1,1\}$ has sign opposite to the sign of $i_j$, does not belong to $E({\cal W})$. 
Indeed, then there also exists an index $\ell_j$ of sign opposite to $i_j$ that has the largest absolute value such that all indices $r$ of the same sign as $\ell_j$ and whose absolute value is upper bounded by $|\ell_j|$ satisfy that $e^j_r\notin E({\cal W})$. To this end, suppose by contradiction that $e^j_b\in E({\cal W})$. Without loss of generality, suppose that $b=-1$ (the other case is symmetric). Then, we have $(i_j-1)+|-b+1| = i_j-1\leq 2n-1$, which contradicts the assumption that $\cal W$ is extremal. Thus, $T$ is valid.

Notice that for every $j\in\{1,\ldots,t\}$, all of the edges parallel to $e^j_{i_j}$  (and $e^j_{\ell_j}$)  whose index is of the same sign as $i_j$ and has absolute value is smaller than $|i_j|$ do not belong to $E({\cal W})$ (because $S$ is innermost). Additionally for every $j\in\{1,\ldots,t\}$, all of the edges parallel to $e^j_{i_j}$ (and $e^j_{\ell_j}$) whose index is of sign opposite to $i_j$ and has absolute value is smaller or equal to  than $|\ell_j|$ do not belong to $E({\cal W})$ (by the choice of $\ell_j$. Thus, for each $j\in\{1,\ldots,t\}$ the interior of $C_j$ does not contain any edge of $\cal W$, and in particular the cycle move operation is applicable to it. When we apply the cycle move operation to some cycle $C_j$, it replaces $e^j_{i_j}$ by $e^j_{i_j}$. By Lemma \ref{lem:existsInnermostSwollenSeg}, these replacements are done on edges not parallel to one another---that is, for every pair of distinct $j,j'\in\{1,\ldots,t\}$, the edges $e^j_{i_j}$ and $e^{j'}_{i_{j'}}$ are not parallel. Thus, the application of one cycle move operation in the application of $T$ does not effect the applicability of any other cycle move operation in the application of $T$. Therefore,  the application of $T$ is valid.
\end{proof}

Now, we consider the properties of the weak linkage that results from the application of a move-through tuple.

\begin{lemma}\label{lem:moveThroughTupleProperties}
Let $(G,S,T,g,k)$ be an instance of \pdp, and $R$ be a Steiner tree. Let $\cal W$ be a sensible, outer-terminal, extremal weak linkage that has no special U-turns and is pushed onto $R$, and let $S\in\Seg({\cal W})$ be an innermost swollen segment. Let $T$ be the move-through tuple of $S$, and let ${\cal W}'$ be the weak linkage that results from the application of $T$. Then, ${\cal W}'$ is a sensible, outer-terminal, extremal weak linkage that has no special U-turns and is pushed onto $R$, whose potential is upper bounded by the potential of $\cal W$ and which has fewer segments than $\cal W$.
\end{lemma}

\begin{proof}
Let $u$ and $v$ be the endpoints of $S$. Because $S$ is swollen, $u$ and $v$ are internal vertices of the same maximal degree-2 path $P$ of $R$. Thus, because $\cal W$ is sensible, there exists a segment $S_1$ and a segment $S_2$ that the walk $W\in{\cal W}$ that has $S$ as a segment traverses immediately before visiting the endpoint $u$ of $S$, and immediately after visiting the endpoint $v$ of $S$, respectively. (Here, we supposed without loss of generality that $W$ is traversed from one endpoint to another such that the endpoint $u$ of $S$ is visited before the endpoint $v$ of $S$.) By Lemma \ref{lem:existsInnermostSwollenSeg}, $S$ is parallel to the subpath $Q$ of $P$ with endpoints $u$ and $v$, and without loss of generality, we suppose that it uses the positive copies of the edges of $Q$. Then, the application of $T$ replaces each one of these positive copies by a negative copy. Thus, the segments $S_1,S,S_2$ are removed and replaced by one segment $S^\star$ that consists of $S_1$, the new negative copies of the edges of $Q$, and $S_2$. Hence, ${\cal W}'$ has fewer (by $2$) segments than $\cal W$. Therefore, because $\cal W$ is sensible, outer-terminal, has no special U-turn and is pushed onto $R$, it is clear that ${\cal W}'$ also has these properties. Now, we show that ${\cal W}'$ also has the property that it is extremal. Clearly, because $\cal W$ is extremal, the multiplicity of ${\cal W}'$ is also upper bounded by $2n$, and for any two parallel edges $e_i,e_j\in E({\cal W}')$ that are not parallel to an edge of $Q$, where $i\geq 1$ and $j\leq -1$, we have $(i-1)+|j+1|\geq 2n$. Now, consider some two edges $e_i,e_j\in E({\cal W}')$ that are parallel to an edge $e_0$ of $Q$, where $i\geq 1$ and $j\leq -1$. Let $e_t$ be the edge of $S$ parallel to $e_i$ and $e_j$, and let $e_r$ be the edge with whom $e_t$ is replaced in the application of $t$ (thus, $e_t\in E({\cal W})\setminus E({\cal W}')$ and $e_r\in E({\cal W}')\setminus E({\cal W})$). Without loss of generality, suppose that $t\geq 1$. Then, by the applicability of $T$, we have that $i\geq t+1$ and $j\leq r$. Thus, $(i-1)+|j+1|\geq t+|r+1|=(t-1)+|(r-1)+1|$. Furthermore, by the definition of a move-through tuple, either $e_{r-1}\in E({\cal W})$ or $r=-2n$. In the first case, because ${\cal W}$ is extremal, we obtain that $(t-1)+|(r-1)+1|\geq 2n$, and in the second case we obtain that $(t-1)+|(r-1)+1|\geq 2n$ as well (because $t\geq 1$).

It remains to prove that the potential of ${\cal W}'$ is upper bounded by the potential of $\cal W$. For this purpose, we consider several cases as follows. First, suppose that $S_1,S$ and $S_2$ belong to the same segment group in $\cal W$, then the only effect on the potential that might increase it is the removal of the two crossings at the endpoints of $S$. However, by the definition of a swollen segment, these crossings have opposite labels, and hence their removal does not effect the potential. (Possibly the segment group that contained $S_1,S$ and $S_2$ has shrunk to a singleton group with respect to ${\cal W}'$, but this does not increase the potential.) Now, suppose that only one among $S_1$ and $S_2$ is in the same segment group as $S$ in $\cal W$, and without loss of generality, suppose that it is $S_1$. Then, in ${\cal W}'$ the segment $S^\star$ belongs to a singleton group (as its endpoints belong to different maximal degree-2 paths of $R$ or it has an endpoint in $V_{=1}(R)\cup V_{\geq 3}(R)$). Moreover, in ${\cal W}$ the segment $S_2$ belongs to a singleton segment group. Thus, one singleton segment group has been replaced by another, and the size of existing segment groups might have shrunk. Since the only change in terms of crossing is that the crossing at the endpoints of $S$ were eliminated, and as in the previous case, this does not effect the potential, we conclude that the potential does not increase.

Lastly, we consider the case where $S_1,S$ and $S_2$ belong to different segment groups in $\Segg({\cal W})$. These three groups are singleton groups ---$S_1$ and $S_2$ have endpoints that belong to different maximal degree-2 paths of $R$ (or an endpoint in $V_{=1}(R)\cup V_{\geq 3}(R)$), and $S$ lies in between them. In ${\cal W}'$, these three groups are eliminates, which results in a decrease of $3$ in the potential. If $S^\star$ forms a singleton segment group, the overall the potential decreases by $2$. Else, $S^\star$ joins an existing segment group, and we have several subcases as follows. In the first subcase, suppose that $S^\star$ joins only the group that contains the segment $\widetilde{S}_1$ of ${\cal W}$ consecutive to $S_1$ in $W$ (that is not $S$). Then, the crossing at the endpoint of $S_2$ is now contributing (1 or -1) to the sum of labels in the potential of the group, and if $\widetilde{S}_1$ was in a singleton group in $\cal W$, then so are its two crossings. Overall, this results in a contribution of at most $3$, so in total the potential does not increase. The subcase where $S^\star$ joins only the group that contains the segment $\widetilde{S}_2$ of ${\cal W}$ consecutive to $S_2$ is symmetric. Now, consider the subcase where $S^\star$ joins both of these groups and hence we merge them with respect to ${\cal W}'$. In this subcase, we have four new crossings that may contribute to the sum of labels, but we have also merged two groups which makes the potential decreased by at  least $1$, so overall the potential does not increase.
\end{proof}

Lastly, we assert that all swollen segments can be eliminated as follows.

\begin{lemma}\label{lem:noSwollenSegments}
Let $(G,S,T,g,k)$ be a good \yes-instance of \pdp, and $R$ be a backbone Steiner tree. Then, there exists a sensible, outer-terminal, extremal weak linkage $\cal W$ in $H$ that is pushed onto $R$, has no special U-turns and swollen segments, is discretely homotopic in $H$ to some solution of $(G,S,T,g,k)$ and $\Potential({\cal W})\leq \alpha_{\rm potential}(k)$.
\end{lemma}

\begin{proof}
By Lemma \ref{lem:noUTurns},  there exists a sensible, outer-terminal, extremal weak linkage in $H$ that has no special U-turns, is pushed onto $R$, discretely homotopic in $H$ to some solution of $(G,S,T,g,k)$ and whose potential is upper bounded by $\alpha_{\rm potential}(k)$. Among all such weak linkages, let $\cal W$ be one with minimum number of segments. To conclude the proof, it suffices to argue that $\cal W$ has no swollen segments.

Suppose, by way of contradiction, that $\cal W$ has at least one swollen segment. Then, by Lemma \ref{lem:existsInnermostSwollenSeg}, $\cal W$ has an innermost swollen segment $S$.  By Lemma \ref{lem:moveThroughTupleValid}, the move-through tuple $T$ of $S$ is well-defined, and its application is valid. Furthermore, by Lemma \ref{lem:moveThroughTupleProperties}, the resulting weak linkage ${\cal W}'$ is sensible, outer-terminal, extremal, pushed onto $R$, has no special U-turns and fewer segments than $\cal W$, and its potential is upper bounded by the potential of $\cal W$. Since discrete homotopy is an equivalence relation, ${\cal W}'$ is discretely homotopic to some solution of $(G,S,T,g,k)$. However, this is a contradiction to the choice of $\cal W$.
\end{proof}

Lastly, we prove that having eliminated all swollen segments indeed implies that the total number of segments is small.

\begin{lemma}\label{lem:fewSegments}
Let $(G,S,T,g,k)$ be a nice instance of \pdp, and $R$ be a Steiner tree. Let $\cal W$ be a sensible, outer-terminal, extremal weak linkage that is pushed onto $R$ and has no special U-turns and swollen segments. Then, $|\Seg({\cal W})|\leq \Potential({\cal W})$.
\end{lemma}

\begin{proof}
To prove that $|\Seg({\cal W})|\leq \Potential({\cal W})$, it suffices to show that for every segment group $W\in\Segg({\cal W})$, we have that $|\Seg(W)|\leq \Potential(W)$. For segment groups of size $1$, this inequality is immediate. Thus, we now consider a segment group $W\in\Segg({\cal W})$ of size at least $2$. Then, 
\[\Potential(W) = 1+|\sum_{(e,e')\in E(W) \times E(W)}\lab_P^{W}(e,e')|,\]
where $P$ is the maximal degree-2 path of $R$  such that all of the endpoints of all of the segments in $\Seg(W)$ are its internal vertices. Because there do not exist swollen segments, we have that $\lab_P^{W}$ assigns either only non-negative values (0 or 1) or only non-positive values (0 or -1). Without loss of generality, suppose that it assigns only non-negative values. Now, notice that every pair of edges consecutively visited by $W$ that below to different segments of $W$ is assigned $1$ (because it creates a crossing with $P$). However, the number of segments of $\cal W$  is upper bounded by one plus the number of such pairs of edges. Thus, $|\Seg(W)|-1\leq \sum_{(e,e')\in E(W) \times E(W)}\lab_P^{W}(e,e')$. From this, we conclude that $|\Seg(W)|\leq \Potential(W)$. 
\end{proof}

\subsection{Completion of the Simplification}

The purpose of this section is to prove the following lemma. 
%

\begin{lemma}\label{lem:pushOutcome}
Let $(G,S,T,g,k)$ be a good \yes-instance of \pdp, and $R$ be a backbone Steiner tree. Then, there exists a simplified weak linkage $\cal W$ in $H$ that is discretely homotopic in $H$ to some solution of $(G,S,T,g,k)$.
\end{lemma}

To this end, we first eliminate all remaining (non-special) U-turns based on Lemmas \ref{lem:eliminateInnermostUTurn}, \ref{lem:noSwollenSegments} and \ref{lem:fewSegments}, similarly to the proof of Lemma \ref{lem:noUTurns}.

\begin{lemma}\label{lem:noUTurnsFinal}
Let $(G,S,T,g,k)$ be a good \yes-instance of \pdp, and $R$ be a backbone Steiner tree. Then, there exists a sensible, outer-terminal, U-turn-free weak linkage $\cal W$ in $H$ that  is pushed onto $R$, discretely homotopic in $H$ to some solution of $(G,S,T,g,k)$, and $|\Seg({\cal W})|\leq \alpha_{\rm potential}(k)$.
\end{lemma}

\begin{proof}
By Lemma \ref{lem:noSwollenSegments},  there exists a sensible, outer-terminal, extremal weak linkage in $H$ that is pushed onto $R$, has no special U-turns and swollen segments, is discretely homotopic in $H$ to some solution of $(G,S,T,g,k)$ and whose potential is upper bounded by $\alpha_{\rm potential}(k)$. By Lemma \ref{lem:fewSegments}, its number of segments is also upper bounded $\alpha_{\rm potential}(k)$. Among all such weak linkages that are sensible, outer-terminal, pushed onto $R$ and whose number of segments is upper bounded $\alpha_{\rm potential}(k)$, let $\cal W$ be one with minimum number of edges. To conclude the proof, it suffices to argue that $\cal W$ is U-turn-free.

Suppose, by way of contradiction, that $\cal W$ has at least one U-turn. Then, by Lemma \ref{lem:existsInnermostUTurn}, $\cal W$ has an innermost U-turn $U=\{e_i,e_j\}$. Let $W$ be the walk in $\cal W$ that uses $e_i$ and $e_j$, and $C$ be the cycle in $H$ that consists of $e_i$ and $e_j$. Then, by Lemma \ref{lem:eliminateInnermostUTurn}, the cycle pull operation is applicable to $(W,C)$. Furthermore, by Lemma \ref{lem:eliminateInnermostUTurn}, the resulting weak linkage ${\cal W}'$ is sensible, outer-terminal, pushed onto $R$, has fewer edges than $\cal W$, and its number of segments is upper bounded by the number of segments of $\cal W$. Since discrete homotopy is an equivalence relation, ${\cal W}'$ is discretely homotopic to some solution of $(G,S,T,g,k)$. However, this is a contradiction to the choice of $\cal W$.
\end{proof}

Now, we prove that having no U-turns implies that each segment can use only two parallel copies of every edge.

\begin{lemma}\label{lem:eachSegContribTwo}
Let $(G,S,T,g,k)$ be a nice instance of \pdp, and $R$ be a Steiner tree. Let $\cal W$ be an outer-terminal, U-turn-free weak linkage that is pushed onto $R$. Then,  each segment $S\in\Seg({\cal W})$ uses at most two copies of every edge in $E(R)$.
\end{lemma}

\begin{proof}
Consider some segment $S\in\Seg({\cal W})$. Suppose, by way of contradiction, that there exists some edge $e_0\in E(R)$ such that $S$ contains at least three edges parallel to $e_0$ (but it cannot contain $e_0$ as it is pushed onto $R$). Then, without loss of generality, suppose that it contains two copies with positive subscript, $e_i$ and $e_j$, and let $S'$ be the subwalk of $S$ having these edge copies as the extreme edges. Then, because $S$ is U-turn free, when we traverse $S'$ from $e_i$ to $e_j$, then we must visit the positive and negative copy of every other edge in $E(R)$ exactly once. 
However, this means that $E({\cal W}$ contains more than one copy of the edge incident to $t^\star$ in $R$, which contradicts the assumption that $\cal W$ is outer-terminal.
\end{proof}

Having established Lemmas \ref{lem:makingCanonical}, \ref{lem:noUTurnsFinal} and \ref{lem:eachSegContribTwo}, we are ready to prove Lemma \ref{lem:pushOutcome}.

\begin{proof}[Proof of Lemma \ref{lem:pushOutcome}.]
By Lemma \ref{lem:noUTurnsFinal}, there exists a sensible, outer-terminal, U-turn-free weak linkage ${\cal W}'$ in $H$ that  is pushed onto $R$, discretely homotopic in $H$ to some solution of $(G,S,T,g,$ $k)$, and whose number of segments is upper bounded by $\alpha_{\rm potential}(k)$. By Lemma \ref{lem:eachSegContribTwo}, the multiplicity of ${\cal W}'$ is upper bounded by $2\alpha_{\rm potential}(k)=\alpha_{\rm mul}(k)$. By Lemma \ref{lem:makingCanonical}, there exists a weak linkage ${\cal W}$  that is sensible, pushed onto $R$, canonical, discretely homotopic to ${\cal W}'$, and whose multiplicity is upper bounded by the multiplicity of ${\cal W}'$. Thus, $\cal W$ is simplified. Moreover, since discrete homotopy is an equivalence relation, $\cal W$ is discretely homotopic to some solution of $(G,S,T,g,k)$.
\end{proof}


\section{Reconstruction of Pushed Weak Linkages from Templates}\label{sec:reconstruction}

In this section, based on the guarantee of Lemma \ref{lem:pushOutcome}, we only attempt to reconstruct simplified weak linkages. Towards this, we introduce the notion of a template (based on another notion called a pairing). Roughly speaking, a template  indicates how many parallel copies of each edge incident to a vertex in $V_{=1}(R)\cup V_{\geq 3}(R)$ are used by the walks in the simplified weak linkage $\cal W$ under consideration, and how many times, for each pair $(e,e')$ of non-parallel edges sharing a vertex, the walks in $\cal W$ traverse from a copy of $e$ to a copy of $e'$. 
Observe that a template does not indicate which edge copy is used by each walk, but only specifies certain numbers. Nevertheless, we will show that this is sufficient for faithful reconstruction of simplified weak linkages. The great advantage of templates, proved later, is that there are only few~of~them.

\subsection{Generic Templates and Templates of Simplified Weak Linkages}

We begin with the definition of the notion of a pairing, which will form the basis of a template. Let $V^\star(R) = V_{=1}(R) \cup V_{\geq 3}(R) \cup V^\star_2(R)$ where $ V^\star_2(R) = \{ v \in V_{=2}(R) \mid \exists u \in V_{=1}(R) \cup V_{\geq 3}(R) \text{ such that } \{u,v\} \in E(R) \}$. Observe that $|V^\star_2(R)| \leq 2(|V_{=1}(R)| + |V_{\geq 3}(R)| - 1) \leq 8k$, by Observation~\ref{obs:leaIntSteiner}. Therefore, $|V^\star(R)| \leq 12k$.
Let $E^\star(R)$ denote the set of edges in $E(R)$ that are incident on a vertex of $V^\star(R)$, and observe that $|E^\star(R)| \leq 24k$ (since $R$ is a tree).

\begin{definition}[{\bf Pairing}]\label{def:pairing}
Let $(G,S,T,g,k)$ be an instance of \pdp\ with a Steiner tree $R$. For a vertex $v\in V_{\geq 3}(R)$, a {\em pairing at $v$} is a set $\pairing_v$ of unordered pairs of distinct edges in $R$ incident to $v$. For a vertex $v \in V^\star_2(R)$, a \emph{ pairing at $v$} is a collection of pairs (possibly non-distinct) edges in $E_R(v)$. And, for a vertex $v\in V_{=1}(R)$, it is the empty set or singleton set of the pair where the (unique) edge incident to $v$ in $R$ occurs twice. A collection $\{\pairing_u\}|_{u\in V^\star(R)}$, where $\pairing_u$ is a pairing at $u$ for every vertex $u\in V^\star(R)$, is called a {\em pairing}.
\end{definition}

As we will see later, simplified weak linkages can only give rise to a specific type of pairings, which we call non-crossing pairings.

\begin{definition}[{\bf Non-Crossing Pairing}]\label{def:noncrossingPairing}
Let $(G,S,T,g,k)$ be an instance of \pdp. Let $R$ be a Steiner tree. Consider a vertex $v\in V^\star(R)$, and let $e^1,e^2,\ldots,e^r$ be the edges in $E(R)$ incident to $v$ in clockwise order where the first edge $e^1$ is chosen arbitrarily. A pairing $\pairing_v$ at $v$ is {\em non-crossing} if there do not exist two pairs $(e^i,e^j)$ and $(e^x,e^y)$ in $\pairing_v$, where $i<j$ and $x<y$, such that $i<x<j<y$ or $x<i<y<j$. More generally, a pairing $\{\pairing_u\}|_{u\in V^\star(R)}$ is {\em non-crossing} if, for every $u\in V^\star(R)$, the pairing $\pairing_u$ is non-crossing.
\end{definition}

We now show that a non-crossing pairing can contain only $\OO(k)$ pairs, which is better than a trivial bound of $\OO(k^2)$. This bound will be required to attain a running time of $2^{\OO(k^2)}n^{\OO(1)}$ rather than~$2^{\OO(k^3)}n^{\OO(1)}$.

\begin{lemma}\label{lem:numNonCrossingLinK}
Let $(G,S,T,g,k)$ be an instance of \pdp. Let $R$ be a Steiner tree. Let $\{\pairing_v\}|_{v\in V^\star(R)}$ be a non-crossing pairing. Then, $|\bigcup_{v\in V^\star(R)}\pairing_v|\leq \alpha_{\mathrm{npair}}(k):=48k$. 
\end{lemma}
\begin{proof}
Towards the bound on $|\bigcup_{v\in V^\star(R)}\pairing_v|$, we first obtain a bound on each individual set $\pairing_v$. To this end, consider some vertex $v\in  V_{\geq 3}(R)$, and let $e^0,e^1,\ldots,e^{r-1}$ are the edges in $E(R)$ incident to $v$ in clockwise order. Consider the undirected graph $C$ on vertex set $\{u_{e^i} \mid i\in\{0,1,\ldots,r-1\}\}$ and edge set $\{\{u_{e^i},u_{e^{(i+1)\mod r}}\} \mid i\in\{0,1,\ldots,r-1\}\}\cup\{\{u_{e^i},u_{e^j}\} \mid (e^i,e^j)\in \pairing_v\}$. Now, notice that $C$ is an outerplanar graph (Fig.~\ref{fig:pairOuterplanar}). To see this, draw the vertices of $C$ on a circle on the plane, so that the curves on the cycle that connect them correspond to the drawing of the edges in  $\{\{u_{e^i},u_{e^{(i+1)\mod r}}\} \mid i\in\{0,1,\ldots,r-1\}\}$. Now, for each edge in $\{\{u_{e^i},u_{e^j}\} \mid (e^i,e^j)\in \pairing_v\}$, draw a straight line segment inside the circle that connects $u_{e^i}$ and $u_{e^j}$. The condition that asserts that $\pairing_v$ is non-crossing ensures that no two lines segments among those drawn previously intersect (except for at their endpoints). As an outerplanar graph on $q$ vertices can have at most $2q-3$ edges, we have that $|E(C)|<2|V(C)|=2r$. Because $|\pairing_v|\leq |E(C)|$, we have that $|\pairing_v|\leq 2r$. For a vertex in $v \in V^\star_2(R)$, since it has only two edges incident on it, $|\pairing_v| \leq 3$. Finally, for $v \in V_{=1}(R)$, $|\pairing_v| \leq 1$.

Thus, for every vertex $v\in V^\star(R)$, $|\pairing_v|$ is bounded by twice the degree of $v$ in $R$. 
Since $|V^\star(R)| \leq 12k$, the sum of the degrees in $R$ of the vertices in $V^\star(R)$ is upper bounded by $24k$. From this, we conclude that $|\bigcup_{v\in V^\star(R)}\pairing_v|\leq \alpha_{\mathrm{npair}}(k)$.
\end{proof}

\begin{figure}
    \begin{center}
        \includegraphics[scale=0.6]{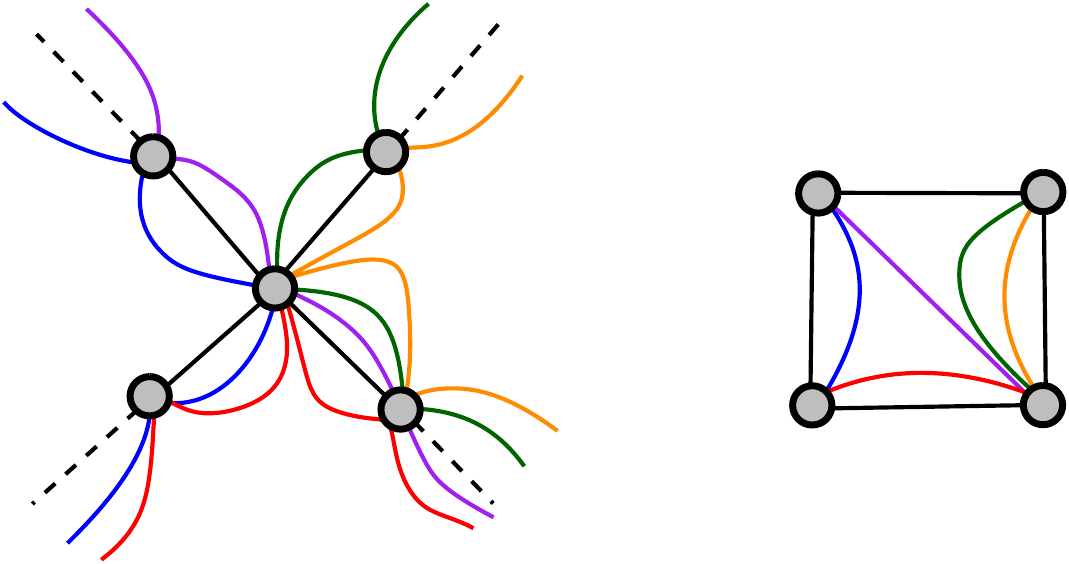}
        \caption{Illustration of the outerplanar graph in Lemma~\ref{lem:numNonCrossingLinK}.}
        \label{fig:pairOuterplanar}
    \end{center}
\end{figure}

Now, we associate a pairing with a pushed weak linkage.

\begin{definition}[{\bf Pairing of a Weak Linkage}]\label{def:pairingOfStitching}
Let $(G,S,T,g,k)$ be an instance of \pdp. Let $R$ be a Steiner tree, and let $\cal W$ be a weak linkage pushed onto $R$. For a vertex $v\in V^\star(R)$, the {\em pairing of $\cal W$ at $v$} is the set that contains every pair of edges $(e,e')$ in $R$ that are incident to $v$ and such that there exists at least one walk in $\cal W$ where $e_i$ and $e'_j$ occur consecutively, where $e_i$ and $e'_j$ are parallel copies of $e$ and $e'$, respectively. More generally, the {\em pairing of $\cal W$} is the collection $\{\pairing_u\}|_{u\in V^\star(R)}$, where $\pairing_u$ is the pairing of $\cal W$ at $u$ for every vertex $u\in V^\star(R)$.
\end{definition}

Apart from a pairing, to be able to reconstruct a simplified weak linkage, we need additional information in the form of an assignment of numbers to the pairs in the pairing. To this end, we have the definition of a template.

\begin{definition}[{\bf Template}]\label{def:template}
Let $(G,S,T,g,k)$ be an instance of \pdp. Let $R$ be a Steiner tree. Let $\pairing_v$ be a pairing at some vertex $v\in V^\star(R)$. A {\em template for $\pairing_v$} is a function $\template_v: \pairing_v\rightarrow \mathbb{N}$. If the maximum integer assigned by $\template_v$ is upper bounded by $N$, for some $N\in\mathbb{N}$, then it is also called an \emph{$N$-template}.
More generally, a {\em template (resp.~$N$-template) of a pairing $\{\pairing_u\}|_{u\in V^\star(R)}$} is a collection $\{\template_u\}|_{u\in V^\star(R)}$, where $\template_u$ is a template (resp.~$N$-template) for $\pairing_u$ for every vertex $u\in V^\star(R)$.
\end{definition}

We proceed to associate a template with a weak linkage.

\begin{definition}[{\bf Template of  Weak Linkage}]\label{def:templateOfStitch}
Let $(G,S,T,g,k)$ be an instance of \pdp. Let $R$ be a Steiner tree, and let $\cal W$ be a weak linkage pushed onto $R$. Let $\{\pairing_u\}|_{u\in V^\star(R)}$ be the pairing of $\cal W$. For a vertex $v\in V^\star(R)$, the {\em template of $\cal W$ at $v$} is the function $\template_v: \pairing_v\rightarrow \mathbb{N}$ such that for every $(e,e')\in \pairing_v$, we have
\[\begin{array}{ll}
\template_v((e,e'))  = |\big\{\{\widehat{e},\widehat{e}'\} \mid & \widehat{e}\ \mathrm{is\ parallel\ to}\ e, \widehat{e}'\ \mathrm{is\ parallel\ to}\ e',\\

& \exists W\in{\cal W}\ \mathrm{s.t.}\ W\ \mathrm{traverses}\ \widehat{e}\ \mathrm{and}\ \widehat{e}'\ \mathrm{consecutively}\big\}|.
\end{array}\]
More generally, the {\em template of $\cal W$} is the collection $\{\template_u\}|_{u\in V^\star(R)}$, where $\template_u$ is the template of $\cal W$ at $u$ for every vertex $u\in V^\star(R)$.
\end{definition}

Now, we claim that the pairing of a pushed weak linkage non-crossing.

\begin{lemma}\label{lem:pairingOfFlowIsNoncrossing}
Let $(G,S,T,g,k)$ be an instance of \pdp. Let $R$ be a Steiner tree. Then, the pairing of any weak linkage $\cal W$ pushed onto $R$ is non-crossing.
\end{lemma}

\begin{proof}
Let $\{\pairing_v\}|_{v\in V^\star(R)}$ be the pairing of $\cal W$. Suppose, by way of contradiction, that $\{\pairing_v\}|_{v\in V^\star(R)}$ is crossing. If $v\in V_{=1}(R) \cup V^\star_2(R)$, then $\pairing_v$ is trivially non-crossing. Thus, there exists a vertex $v\in V_{\geq 3}(R)$ such that $\pairing_v$ is crossing. Let $e^1,e^2,\ldots,e^r$ be the edges in $E(R)$ incident to $v$ in clockwise order. Because $\pairing_v$ is crossing, there exist two pairs $(e^i,e^j)$ and $(e^x,e^y)$ in $\pairing_v$, where $i<j$ and $x<y$, such that $i<x<j<y$ or $x<i<y<j$. By the definition of $\pairing_v$, this means that there exist walks $\widehat{W},\overline{W}\in {\cal W}$ (possibly $\widehat{W}=\overline{W}$) and edges $\widehat{e}^i,\widehat{e}^j,\overline{e}^x$ and $\overline{e}^y$ that are parallel to $e^i,e^j,e^x$ and $e^y$, respectively, such that $\widehat{W}$ traverses $\widehat{e}^i$ and $\widehat{e}^j$ consecutively, and $\overline{W}$ traverses $\overline{e}^x$ and $\overline{e}^y$ consecutively. However, because $i<x<j<y$ or $x<i<y<j$, and because parallel edges incident to each vertex appear consecutively in its cyclic order,  we derive that $(v,\widehat{e}^i,\widehat{e}^j,\overline{e}^x,\overline{e}^y)$ is a crossing of $\widehat{W}$ and~$\overline{W}$.
\end{proof}

From Lemmas \ref{lem:numNonCrossingLinK} and \ref{lem:pairingOfFlowIsNoncrossing}, we obtain the following corollary.

\begin{corollary}\label{cor:pairingOfFlowIsLinear}
Let $(G,S,T,g,k)$ be an instance of \pdp\ with a backbone Steiner tree $R$. Let $\cal W$ be a weak linkage pushed onto $R$ with a pairing $\{\pairing_v\}|_{v\in V^\star(R)}$. Then, $|\bigcup_{v\in V^\star(R)}\pairing_v|$ $\leq \alpha_{\mathrm{npair}}(k)$.
\end{corollary}

Additionally, we claim that we can focus our attention on pushed flows whose templates are $\alpha_{\mathrm{weight}}(k)$-templates.

\begin{lemma}\label{lem:templateOfFlowIsBounded}
Let $(G,S,T,g,k)$ be a good \yes-instance of \pdp\ with a backbone Steiner tree $R$. Then, there exists a simplified weak linkage that is discretely homotopic in $H$ to some solution of $(G,S,T,g,k)$, whose template is an $\alpha_{\mathrm{mul}}(k)$-template.
\end{lemma}

\begin{proof}
By Lemma~\ref{lem:pushOutcome}, there exists a simplified weak linkage $\cal W$ that is discretely homotopic in $H$ to some solution of $(G,S,T,g,k)$. Because $\cal W$ is simplified, its multiplicity is upper bounded by $\alpha_{\rm mul}(k)$. Let $\{\pairing_v\}|_{v\in V^\star(R)}$ and $\{\template_v\}|_{v\in V^\star(R)}$ be the pairing and template of $\cal W$, respectively. Consider some vertex $v\in V^\star(R)$ and pair $(e,e')\in \pairing_v$. To complete the proof, we need to show that $\template_v((e,e'))\leq \alpha_{\mathrm{mul}}(k)$.  By the  definition of a template, 
\[\begin{array}{ll}
\template_v((e,e'))  = |\{\{\widehat{e},\widehat{e}'\} \mid & \widehat{e}\ \mathrm{is\ parallel\ to}\ e, \widehat{e}'\ \mathrm{is\ parallel\ to}\ e',\\

& \exists W\in{\cal W}\ \mathrm{s.t.}\ W\ \mathrm{traverses}\ \widehat{e}\ \mathrm{and}\ \widehat{e}'\ \mathrm{consecutively}\}|.
\end{array}\]
Thus, because every the walks in $\cal W$ are edge-disjoint and each walk in $\cal W$ visits distinct edges (by the definition of a weak linkage), $\template_v((e,e'))$ is upper bounded by the number of edges parallel to $e$  that belong to $E({\cal W})$. Thus, by the definition of the multiplicity of a linkage, we conclude that $\template_v((e,e'))\leq \alpha_{\mathrm{mul}}(k)$.
\end{proof}

In light of Corollary \ref{cor:pairingOfFlowIsLinear} and Lemma \ref{lem:templateOfFlowIsBounded}, we define the set of all pairings and templates in which we will be interested as follows.

\begin{definition}[{\bf The Set $\ALL$}]\label{ref:allTemplates}
Let $(G,S,T,g,k)$ be an instance of \pdp. Let $R$ be a Steiner tree. The set $\ALL$ contains every collection $\{(\pairing_v,\template_v)\}|_{v\in V^\star(R)}$ where $\{\pairing_v\}|_{v\in V^\star(R)}$ is a non-crossing pairing that satisfies $|\bigcup_{v\in V^\star(R)}\pairing_v|\leq \alpha_{\mathrm{npair}}(k)$, and $\{\template_v\}|_{v\in V^\star(R)}$ is an $\alpha_{\mathrm{mul}}(k)$-template for $\{\pairing_v\}|_{v\in V^\star(R)}$. 
\end{definition}

From Corollary \ref{cor:pairingOfFlowIsLinear} and Lemma \ref{lem:templateOfFlowIsBounded}, we have the following result.

\begin{corollary}\label{cor:existsInAll}
Let $(G,S,T,g,k)$ be a good \yes-instance of \pdp. Let $R$ be a backbone Steiner tree. Then, there exists a simplified weak linkage that is discretely homotopic in $H$ to some solution of $(G,S,T,g,k)$ and satisfies the following property: There exists $\{(\pairing_v,\template_v)\}|_{v\in V^\star(R)}\in\ALL$ such that $\{\pairing_v\}|_{v\in V^\star(R)}$ is the pairing of $\cal W$, and $\{\template_v\}|_{v\in V^\star(R)}$ is the template of $\cal W$.
\end{corollary}

Because we only deal with pairings having just $\OO(k)$ pairs and upper bound the largest integer assigned by templates by $2^{\OO(k)}$, the set $\ALL$ is ``small'' as asserted by the following~lemma.

\begin{lemma}\label{lem:enumerateTemplates}
Let $(G,S,T,g,k)$ be an instance of \pdp. Let $R$ be a Steiner tree. Then, $|\ALL|=2^{\OO(k^2)}$. Moreover, $\ALL$ can be computed in time $2^{\OO(k^2)}$.
\end{lemma}

\begin{proof}
First, we upper bound the number of pairings $\{\pairing_v\}|_{v\in V^\star(R)}$ that satisfy $|\bigcup_{v\in V^\star(R)}\pairing_v|\leq \alpha_{\mathrm{npair}}(k)$. By Observation \ref{obs:leaIntSteiner}, the number of edges in $E^\star(R)$ is at most $24k$, and hence the number of pairs of edges in $E^\star(R)$ is at most $(24k)^2$. Thus, the number of choices for $\bigcup_{v\in V^\star(R)}\pairing_v$ is at most $\binom{(24k)^2}{\alpha_{\mathrm{npair}}(k)}$. Note that each pair of edges in this union belongs to $\pairing_v$ for at most one vertex $v\in V^\star(R)$.
From this, we conclude that the number of pairings $\{\pairing_v\}|_{v\in V^\star(R)}$ that satisfy $|\bigcup_{v\in V^\star(R)}\pairing_v|\leq \alpha_{\mathrm{npair}}(k)$ is at most $\binom{(24k)^2}{\alpha_{\mathrm{npair}}(k)}\cdot 2^{\alpha_{\mathrm{npair}}(k)}=2^{\OO(k\log k)}$ (as $\alpha_{\mathrm{npair}}(k)=\OO(k)$).

Now, fix some pairing $\{\pairing_v\}|_{v\in V^\star(R)}$ that satisfies $|\bigcup_{v\in V^\star(R)}\pairing_v|\leq \alpha_{\mathrm{npair}}(k)$. We test if this pairing is a non-crossing pairing, at each vertex $v \in V^\star(R)$, by testing all possible $4$-tuples of edges in $E_R(v)$. This takes $k^{\OO(1)}$ time in total.
Then, the number of $\alpha_{\mathrm{mul}}(k)$-templates for $\{\pairing_v\}|_{v\in V^\star(R)}$ is upper bounded by $(\alpha_{\mathrm{mul}}(k))^{\alpha_{\mathrm{npair}}(k)}=2^{\OO(k^2)}$ (as $\alpha_{\mathrm{mul}}(k)=2^{\OO(k)}$ and $\alpha_{\mathrm{npair}}(k)=\OO(k)$). Thus, we have that $|\ALL|=2^{\OO(k^2)}$. It should also be clear that these arguments, by simple enumeration, imply that $\ALL$ can be computed in time $2^{\OO(k^2)}$.
\end{proof}

\paragraph{Extension of Pairings and Templates.}
To describe the reconstruction of simplified weak linkages from their pairings and templates, we must extend them from $V^\star(R)$ to all of $V(R)$. Intuitively, this extension is based on the observation that $\cal W$ is U-turn free and sensible. Therefore, if a walk in $\cal W$, visits a maximal degree-2 path in $R$, then it must traverse the entirety of this path. Hence, the pairings and templates at any internal vertex of a degree-2 path can be directly obtained from the endpoint vertices of the path.
We begin by identifying which collections of pairings and templates can be extended.

\begin{definition}
\label{def:extpairtemplCheck}
    Let $(G,S,T,g,k)$ be an instance of \pdp, and let $R$ be a backbone Steiner tree.
    And let ${\cal A} = \{(\pairing_v, \template_v)\}_{v\in V^\star(R)}$. Then $\cal A$ is \emph{extensible to all of $V(R)$} if the following conditions are true for every maximal degree-2 path in $R$.
    \begin{itemize}
        \item Let $u,v \in V^\star_2(R)$ such that they lie on the same maximal degree-2 path of $R$. Consider the subpath $\pathT_{R}(u,v)$ with endpoints $u$ and $v$, and let $e^u$ and $e^v$ be the edges in $E(\pathT_{R}(u,v))$ incident on $u$ and $v$, respectively. Suppose that $e^u$ and $e^v$ are distinct edges. Then $(e,e^u) \in \pairing_u$ if and only if $(e',e^v) \in \pairing_v$, where $E_R(u) = \{e,e^u\}$ and $E_R(v)=\{e',e^v\}$.
        
        \item Assuming that the above condition is true, furthermore $\template_u(e,e^u) = \template_v(e',e^v)$.
    \end{itemize}
\end{definition}

The following lemma shows that pairings and templates of simplified weak linkages are extensible to all of $V(R)$. 
\begin{lemma}
    Let $(G,S,T,g,k)$ be a good instance of \pdp. Let $R$ be a backbone Steiner tree. Let $\cal W$ be a simplified weak linkage with pairing and template ${\cal A} = \{(\pairing_v, \allowbreak \template_v)\}|_{v\in V^\star(R)}$.
    Then, $\cal A$ is extensible to all of $V(R)$.
\end{lemma}
\begin{proof} 
    Let $u,v \in V^\star_2(R)$ such that they lie on the same maximal degree-2 path of $R$. Consider the subpath $\pathT_{R}(u,v)$ with endpoints $u$ and $v$, and let $e^u$ and $e^v$ be the edges in $E(\pathT_{R}(u,v))$ incident on $u$ and $v$, respectively. Suppose that $e^u$ and $e^v$ are distinct edges.
    Then $\pathT_{R}(u,v)$ must have an internal vertex. Let $V(\pathT_{R}(u,v)) = \{ u = w_0, w_1, w_2, \ldots, w_p, w_{p+1} = v \}$ and let $E(\pathT_{R}(u,v)) = \{e^i=\{w_i,w_{i+1}\} \mid 0 \leq i \leq p\}$ where $e^0 = e^u$ and $e^p = e^v$. 
    Consider an internal vertex $w_i$ of $\pathT_{R}(u,v)$ and note that $E_R(w_i) = \{e^{i-1},e^i\}$. Observe that, as $\cal W$ is U-turn free, $w^i \notin N_R(V_{=1}(R))$ and the endpoints of every the walk in $\cal W$ lies in $V_{=1}(R)$, there is no walk in $\cal W$ that visits two parallel copies of an edge in $E_R(w_i)$ consecutively, as that would constitute a U-turn. Therefore, any walk in $\cal W$ that visits $e^{i-1}_{j_{i-1}}$ must also visit $e^i_{j_i}$, where $e^{i-1}_{j_{i-1}}$ and $e^i_{j_i}$ are parallel copies of $e^{i-1}$ and $e^i$ respectively. This holds for all vertices $w_1, w_2, \ldots, w_p$. 
    Let $E^\star(\pathT_{R}(u,v))$ be the collection of all the parallel copies of every edge in $E(\pathT_{R}(u,v))$.
    Then, for a walk $W \in \cal{W}$, 
    let $\wh{W}_1, \wh{W}_2, \ldots, \wh{W}_t$ be the maximal subwalks of $W$ restricted to $E^\star(\pathT_{R}(u,v))$. Then each $\wh{W}_i$ is a path from $u$ to $v$ along the parallel copies of $e^0,e^1, \ldots, e_p$. Hence, if $(e, e^u) = (e,e_0) \in \pairing_u$ if and only if $(e',e^v) = (e',e_p) \in \pairing_v$ where $e$ and $e'$ are the edges in $E(R) \setminus E(\pathT_{R}(u,v))$ that are incident on $u$ and $v$, respectively.
    Further, observe that, each time a walk $W \in \cal W$ visits a parallel copy of $e^u = e^0$ (immediately after visiting a parallel copy of $e$) it must traverse a parallel copy of $e^1, e^2, \ldots e^p = e^v$ consecutively (and then immediately visit a parallel copy of $e'$), and vice-versa.
    Therefore by definition, $\template_u(e,e^u) = \template_{w_1}(e^u,e^1) = \template_{w_2}(e^1, e^2) = \ldots = \template_v(e_{p-1},e^v) = \template_v(e',e^v)$.
\end{proof}

We now have the following definition.

\begin{definition}[\bf Extension of Pairings and Templates.]
\label{def:extendPairTemplate}
    Let $(G,S,T,g,k)$ be a good instance of \pdp. Let $R$ be a backbone Steiner tree. Let ${\cal A} = \{(\pairing_v,\template_v)\}|_{v\in V^\star(R)}$. Then, the \emph{extension of $\cal A$ to $V(R)$} is the collection $\wh{\cal A} = \{(\wh{\pairing}_v,\wh{\template}_v)\}|_{v(R)}$ such that:
    \begin{itemize}
        \item If $\cal A$ is not extensible (Definition~\ref{def:extpairtemplCheck}), then $\wh{\cal A}$ is invalid.
        
        \item Otherwise, $\cal A$ is extensible and we have two cases.
        \begin{itemize}
            \item If $v \in V_{=1}(R) \cup V_{\geq 3}(R)$, then $\wh{\pairing_v} = \pairing_v$ and $\wh{\template_v} = \template_v$.
        
            \item Otherwise, $v \in V_{=2}(R)$ and let $u,w \in V_{=1}(R) \cup V_{\geq 3}(R)$ such that $v \in V(\pathT_{R}(u,v))$. Let $e_u$ and $e_w$ be the two edges in $\pathT_{R}(u,v)$ incident on $u$ and $w$, respectively.
            If $(e,e_u) \notin \pairing_u$ for any $e \in E_R(u)$, then $\pairing_v = \emptyset$. Otherwise, there is some $e \in E_R(u)$ such that $(e,e_u) \in \pairing_u$. Then $\pairing_v = \{(e',e'')\}$ where $e'$ and $e''$ are the two edges in $E(R)$ incident on $v$. Further, $\template_v(e',e'') =         \template_u(e,e_u)$.
        \end{itemize}
    \end{itemize}
%
\end{definition}

Let $\wh{\ALL}$ denote the collection of extensions of all the pairings and templates in $\ALL$. Then we have the following corollary of Lemma~\ref{lem:enumerateTemplates}.

\begin{lemma}\label{lem:enumerateExtTemplates}
    Let $(G,S,T,g,k)$ be an instance of \pdp. Let $R$ be a Steiner tree. Then, $|\wh{\ALL}|=2^{\OO(k^2)}$. Moreover, $\wh{\ALL}$ can be computed in time $2^{\OO(k^2)} n$.
\end{lemma}
\begin{proof}
    Given ${\cal A} = \{(\pairing_v, \template_v)\}|_{v \in V^\star(R)} \in \ALL$, we apply Definition~\ref{def:extendPairTemplate} to obtain the extension $\wh{\cal A}$. Note that, for every vertex $v\in V^\star(R)$, we have that $|\pairing_v|=\OO(k)$ and the numbers assigned by $\template_v$ are bounded by $2^{\OO(k)}$. Further, since $|V^\star(R)| \leq 12k$, we can test the conditions in Definition~\ref{def:extendPairTemplate} in $k^{\OO(1)}$ time. Finally we can construct $\wh{\cal A} = \{(\wh{\pairing}_v, \wh{\template}_v)\}|_{v \in V(R)}$ in total time $2^{\OO(k)} n$, as $|V(R)| \leq n$. 
    Since $|\ALL| = 2^{\OO(k^2)}$ and it can be enumerated in time $2^{\OO(k^2)}$, it follows that $|\wh{\ALL}| = 2^{\OO(k^2)}$ and it can be enumerated in $2^{\OO(k^2)} n$ time.
\end{proof}

In the rest of this section, we only require pairings and templates that are extended to all of $V(R)$.
For convenience, we abuse the notation to denote the extension of a collection of pairings and templates, ${\cal A} \in \ALL$ to all of $V(R)$, by $\{\pairing_v\}|_{v \in V(R)}$ and $\{\template_v\}|_{v \in V(R)}$, respectively. 
%
%
The following corollary follows from the definition of $\wh{\ALL}$ and Corollary~\ref{cor:existsInAll}.

\begin{corollary}\label{cor:existsInAll2}
    Let $(G,S,T,g,k)$ be a good \yes-instance of \pdp. Let $R$ be a backbone Steiner tree. Then, there exists a simplified weak linkage that is discretely homotopic in $H$ to some solution of $(G,S,T,g,k)$ and satisfies the following property: There exists $\{(\pairing_v,\template_v)\}|_{v\in V(R)}\in\wh{\ALL}$ such that $\{\pairing_v\}|_{v\in V(R)}$ is the pairing of $\cal W$, $\{\template_v\}|_{v\in V(R)}$ is the template of $\cal W$.
\end{corollary}

\paragraph{Stitching of Weak Linkages.}
Let us now introduce the notion of a stitching, which gives a localized view of a weak linkage pushed onto $R$ at each vertex in $V(R)$. 
Intuitively, the stitching at a vertex $v \in V(R)$ is function on the set of edges incident on $v$, that maps each edge to the next (or the previous) edge in a weak linkage. Note that, only a subset of the edges incident on $v$ may participate in a weak linkage. Therefore, we introduce the following notation: an edge is mapped to $\bot$ to indicate that it is not part of weak linkage. Also recall that, for a vertex $v \in V(R)$, $\order_v$ is an enumeration of the edges in $\wh{E}_R(v)$ in either clockwise or anticlockwise order, where $\wh{E}_R(v) = \{ e \in E_H(v) \mid e \text{ is parallel to an edge } e' \in E(R) \}$.

\begin{definition}[\bf Stitching at a Vertex]
    Let $(G,S,T,g,k)$ be a good {\sf Yes}-instance of \pdp, and let $R$ be a backbone Steiner tree. For a vertex $v \in V(R)$, a function $f_v: E_H(v) \rightarrow E_H(v) \cup \bot$ is a \emph{stitching at $v$} if it satisfies the following conditions.
    \begin{itemize}
        \item For any edge $e \in E_H(v) \setminus \wh{E}_R(v)$, $f_v(e) = \bot$.
        
        \item For a pair of (possibly non-distinct) edges $e, e' \in \wh{E}_R(v)$, $f_v(e) = e'$ if and only if $f_v(e') = e$.
        
        \item If $v \in S \cup T$, then there is exactly one edge such that $f_v(e) = e$. Otherwise, there is no such edge.       
        
        \item If $e_1, e_2, e_3, e_4 \in \wh{E}_R(v)$ such that $f_v(e_1) = e_2$ and $f_v(e_3) = e_4$. Then $\{e_1,e_2\}$ and $\{e_3, e_4\}$ are disjoint and non-crossing in $\order_v$.\footnote{That is, in a clockwise (or anticlockwise) enumeration of $\wh{E}_R(v)$ starting from $e_1$, these edges occur as either $e_1, e_2, e_3, e_4$ or $e_1, e_3, e_4, e_2$, where without loss of generality we assume that $e_3$ occurs before $e_4$ in this ordering.}
    \end{itemize}

    Let $\{f_v\}|_{v \in V(R)}$ be a collection of functions such that $f_v$ is a stitching at $v$ for each $v \in V(R)$. Then, this collection is called a \emph{stitching} if 
        for every edge  $e=\{u,v\} \in E_H(R)$, $f_u(e) = \bot$ if and only if $f_v(e) = \bot$.         
\end{definition}

Let us now describe the stitching of a weak linkage that is pushed onto $R$.

\begin{definition}[\bf Stitching of a Weak Linkage Pushed onto $R$]
\label{def:linkageStitching}
    Let $(G,S,T,g,k)$ be a good {\sf Yes}-instance of \pdp, and let $R$ be a backbone Steiner tree. Let $\cal W$ be a weak linkage pushed onto $R$. Then we define the \emph{stitching of $\cal W$} as the collection of functions $\{\stitch_v\}|_{v \in V(R)}$, where $\stitch_v: E_H(v) \rightarrow E_H(v) \cup \bot$ satisfies the following.
    \begin{itemize}
        \item If there is $W\! \in\! {\cal W}$ where $e\!=\!\{v,w\}$ is the first edge of $W$, then $v \!\in\! S \!\cup\! T$ and~$\stitch_v(e) \!=\! e$. 
        \item If there is $W\! \in\!{\cal W}$ where $e\!=\!\{u,v\}$ is the last edge of $W$, then $v \!\in\! S \!\cup\! T$ and~$\stitch_v(e) \!=\! e$.
        \item If there is a walk $W \in {\cal W}$ such that $e, e' \in E_H(v)$ are consecutive edges with a common endpoint $v \in V(R)$, then $\stitch_v(e) = e'$ and $\stitch_v(e') = e$.
        \item If $e \in E_H(v)$ is not part of any walk in $\cal W$, then $\stitch_v(e) = \bot$.
    \end{itemize}
\end{definition}

It is easy to verify that $\{\stitch_v\}|_{v \in V(R)}$ is indeed a stitching. Let us make a few more observations on the properties of this stitching.
\begin{observation}\label{obs:stitchingProp}
    Let $(G,S,T,g,k)$ be a good {\sf Yes}-instance of \pdp, and let $R$ be a backbone Steiner tree. Let $\cal W$ be a weak linkage pushed onto $R$ and let $\{\stitch_v\}|_{v \in V(R)}$ be the stitching of $\cal W$. Let $\{\pairing_v\}|_{v \in V(R)}$ and $\{\template_v\}|_{v \in V(R)}$ be the pairing and template of $\cal W$, respectively. Then the following holds.
    \begin{itemize}
        \item Let $e_i,e'_j \in E_H(v)$, then $\stitch_v(e_i) = e'_j$ if and only if $\stitch_v(e'_j) = e_i$.
        
        \item Let $e,e' \in E_R(v)$. Then $(e,e') \in \pairing_v$ if and only if there is a pair $e_i,e'_j$ of edges in $E_H(R)$, where $e_i$ is parallel to $e$ and $e'_j$ is parallel to $e'$, such that $\stitch_v(e_i) = e'_j$ and $\stitch_v(e'_j) = e_i$. Further, the number of pairs of parallel edges is equal to $\template_v(e,e')$.
        
        \item If $e_i,e'_j$ and $e^\star_p,\wh{e}_q$ are pairs of edges in $E_H(v)$ such that $\stitch_v(e_i) = e'_j$ and $\stitch_v(e^\star_p) = \wh{e}_q$, then the pairs $e_i,e'_j$ and $e^\star_p,\wh{e}_q$ are non-crossing in $\order_v$.
        
        \item If the multiplicity of $\cal W$ is upperbounded by $\ell$, then for each edge $e=\{u,v\} \in E(R)$, $|\{e' \in E(H) \mid e' \text{ is parallel to } e \text{ and } \stitch_v(e') \neq \bot \}| \leq k \cdot \ell$.
    \end{itemize}
\end{observation}

\subsection{Translating a Template Into a Stitching}

Given {\em (i)} an instance $I=(G,S,T,g,k)$ of \pdp, {\em (ii)} a backbone Steiner tree $R$, and {\em (iii)} a collection $\{(\pairing_v,\template_v)\}|_{v \in V(R)}\in \wh{\ALL}$, our current objective is to either determine that $\{(\pairing_v,\template_v)\}|_{v \in V(R)}$ is invalid or construct a multiplicity function $\ell$ and a stitching $\{f_v\}|_{v \in V(R)}$ to reconstruct the a weak linkage. The cases where we determine that $\{(\pairing_v,\template_v)\}|_{v \in V(R)}$ is invalid will be (some of the) cases where there exists no simplified weak linkage whose pairing and template are $\{\pairing_v\}|_{v \in V(R)}$ and $\{\template_v\}|_{v \in V(R)}$, respectively. Let us begin with the notion of multiplicity function $\ell$ of a collection of pairings and templates, as follows.

\begin{definition}[{\bf Multiplicity Function}]\label{def:locaWLofTemplate}
    Let $(G,S,T,g,k)$ be an instance of \pdp, and $R$ be a Steiner tree. Let ${\cal A}=\{(\pairing_v,\template_v)\}|_{v \in V(R)}\in \wh{\ALL}$. For every vertex $v\in V^\star(R)$, let $\ell_v$ be the function that assigns $\sum_{e': (e,e') \in\pairing_v}\template_v((e,e'))$ to every edge $e\in E(R)$ incident to $v$. If one of the following conditions is satisfied, then the {\em multiplicity function extracted from ${\cal A}$} is {\em invalid}.
    \begin{enumerate}
        \item There exists an edge $e=\{u,v\}$ such that $u,v\in V^\star(R)$ and $\ell_u(e)\neq \ell_v(e)$.
        \item There exists a terminal $v\in S\cup T$ such that $\pairing_v=\emptyset$.
    \end{enumerate}
    Otherwise, the {\em multiplicity function extracted from ${\cal A}$} is {\em valid} and it is the function $\ell: E_{1,3+}(R)\rightarrow \mathbb{N}_0$ such that for each $e\in E_{1,3+}(R)$, $\ell(e)=\ell_v(e)$ where $v$ is an endpoint of $e$ in $V^\star(R)$.\footnote{The choice of the endpoint when both belong to $V^\star(R)$ is immaterial by the definition of invalidity.}
\end{definition}

Let $\cal W$ be a weak linkage pushed onto $R$. The \emph{multiplicity function of a $\cal W$} is defined as the multiplicity function $\ell$ extracted from $\cal A$, the pairings and templates of $\cal W$. It is clear that the multiplicity of $\cal W$ is $\max_{e \in E(R)} \ell(e)$. 
\begin{observation}\label{obs:weaklinkmult}
    Let $(G,S,T,g,k)$ be an instance of \pdp\ with a simplified weak linkage $\cal W$, and let $\ell$ is the multiplicity function of $\cal W$. 
    For any $e \in E(R)$, $\ell(e) \leq \alpha_{\mathrm{mul}}(k)$.
\end{observation}

Having extracted a multiplicity function, we turn to extract a stitching.
Towards this, recall the embedding of $H$ with respect to $R$, and the resulting enumeration of edges around vertices in $V(R)$ (see Section~\ref{sec:enumParallel}). 
%
Let us now describe the stitching extraction at a terminal vertex.
\begin{definition}[{\bf Stitching Extraction at Terminals}]
\label{def:locaStitchExtractTerminal}
    Let $(G,S,T,g,k)$ be a nice instance of \pdp. Let $R$ be a Steiner tree. Consider a collection ${\cal A}= \break \{(\pairing_v,\template_v)\}|_{v \in V(R)}\in \wh{\ALL}$. Let $\ell$ be the multiplicity function extracted from ${\cal A}$, and suppose that $\ell$ is valid. Let $v\in S\cup T$, and let $e^\star$ be the unique edge in $E(R)$ incident to $v$. If $\ell(e^\star)$ is an even number, then the {\em local stitching extracted from ${\cal A}$ at $v$} is {\em invalid}. Otherwise, the {\em stitching extracted from ${\cal A}$ at $v$} is {\em valid} and it is the involution $f_v: E_H(v) \rightarrow E_H(v) \cup \bot$ defined as follows.
\begin{equation*}
    f_v(e) =
    \begin{cases*}
  	  e^\star_{\ell(e)+1 - i }  & if $e = e^\star_i$ and $1 \leq i \leq \ell(e)$ \\
      \bot  & otherwise.
    \end{cases*}
  \end{equation*}
\end{definition}

Next, we describe how to extract a stitching at a vertex $v\in V_{=2}(R) \cup V_{\geq 3}(R)$. 


\begin{definition}[{\bf Stitching Extraction at Non-Terminals}]
\label{def:locaStitchExtractNonTerminal}
    Let $(G,S,T,g,k)$ be a nice instance of \pdp. Let $R$ be a backbone Steiner tree. Consider a collection ${\cal A}= \break \{(\pairing_v,\template_v)\}|_{v \in V(R)}\in \wh{\ALL}$. Let $\ell$ be the multiplicity function extracted from ${\cal A}$, and suppose that $\ell$ is valid. Let $v\in V_{=2}(R) \cup V_{\geq 3}(R)$, and
    suppose that  let $e^1,e^2,\ldots,e^{r}$ denote the arcs in $E(R)$ incident to $v$ enumerated as per $\order_v$ starting from $e^1$. Then the define a function $f_v: E_H(v) \rightarrow E_H(v) \cup \bot$ as follows. 
    \begin{itemize}
        \item For each $(e,e') \in \pairing_v$ such that $\template_v(e,e') > 0$, where $e$ occurs before $e'$ in $\order_v$, let ${\sf inner}(e,e') = \{e^\star \in E_R(v) \mid  e^\star \text{ occurs between } e \text{ and } e' \text{ in } \order_v\}$, and ${\sf outer}(e,e') = \{e^\star \in E_R(v) \mid \text{ either } e^\star \text{ occurs before } e \text{ or occurs after } e' \text{ in } \order_v\}$.
        
        \item Then, for each $i \in \{1,\ldots, \template_v(e,e')\}$,
        let $f_v(e_{i+x}) = e'_{y-i}$ and $f_v(e'_{y-i}) = e_{i+x}$ where
        $$x = \sum_{e^\star \in {\sf outer}(e,e')} \template_v(e,e^\star),$$ and 
        $$y = 1 + \template_v(e,e') + \sum_{e^\star \in {\sf inner}(e,e')} \template_v(e,e^\star)$$
        
        \item For all other edges in $E_H(v)$, define $f_v(e) = \bot$.
    \end{itemize}
%
    If the assignment $f_v$ is fixed point free,
    then $f_v$ is the {\em stitching extracted from ${\cal A}$ at $v$}, which is said to be {\em valid}. Otherwise, it is {\em invalid}
\end{definition}

Lastly, based on Definitions \ref{def:locaStitchExtractTerminal} and \ref{def:locaStitchExtractNonTerminal}, we extract the stitching as follows,

\begin{definition}[{\bf Stitching Extraction}]\label{def:locaStitchExtract}
    Let $(G,S,T,g,k)$ be a nice instance of \pdp. Let $R$ be a Steiner tree. Consider a collection ${\cal A}= \{(\pairing_v,\template_v)\}|_{v \in V(R)}\in \wh{\ALL}$. 
    For each $v \in V(R)$, let $f_v$ be the stitching extracted from ${\cal A}$ at $v$.
    Then the \emph{stitching extracted from $\cal A$ is invalid} if it satisfies one of the following conditions.
    \begin{itemize}
        \item There is a vertex $v\in V_{=1}(R))$ such that the stitching extracted from ${\cal A}$ at $v$ is invalid.
        \item There is an edge $e=\{u,v\} \in E(H)$ parallel to an edge in $E(R)$ such that $f_u(e) = \bot$ and $f_v(e) \neq \bot$.
    \end{itemize}
       Otherwise, the {\em stitching extracted from $\cal A$} is {\em valid} and defined as the collection $\{f_v\}|_{v \in V(R)}$ where $f_v$ is the stitching extracted from $\cal A$ at $v$ for every $v\in V(R)$. 
\end{definition}


Less obviously, we also show that in case we are given a collection in $\wh{\ALL}$ that corresponds to weak linkage, not only is the stitching extracted from that collection valid, but also, most crucially, it is the stitching we were originally given (under the assumption that the pair of flow and stitching we deal with is simplified). In other words, we are able to faithfully reconstruct a stitching from the template of weak linkage. The implicit assumption in this lemma that $\{(\pairing_v,\template_v)\}|_{v \in V(R)}$ belongs to $\wh{\ALL}$ is supported by Corollary \ref{cor:existsInAll2}.

\begin{lemma}\label{lem:locaStitchofTemplate1}
    Let $(G,S,T,g,k)$ be a good {\sf Yes}-instance of \pdp, and let $R$ be a backbone Steiner tree. Let $\cal W$ be a simplified weak linkage in $H$ and let $\ell$ be the multiplicity function of $\cal W$. 
    Consider the collection ${\cal A}= \{(\pairing_v,\template_v)\}|_{v \in V(R)}\in \wh{\ALL}$, such that it is the collection of pairings and templates of $\cal W$. 
    Let $\{f_v\}|_{v \in V(R)}$ be the stitching extracted from $\cal A$.
    Then for every vertex $v \in V_{=1}(R)$, $\stitch_v(e) = f_v(e)$ for every edge $e \in E_H(v)$.
\end{lemma}
\begin{proof}
    Let $E^\star(v) = \{ e^\star_i \mid 1 \leq i \leq \ell(e^\star) \}$ where $e^\star_i$ denotes the $i$-th parallel copy of $e^\star$, in the enumeration in $\order_v$. Observe that, since $\cal W$ is simplified, $E^\star(v)$ is exactly the set of edges from $E_H(v)$ that appear in $\cal W$. Hence, for any edge $e \in E_H(v)$, if $e \notin E^\star(v)$ then $\stitch_v(e) = \bot$.
    
    Let us now consider the edges in $E^\star(v)$.
    Since $\cal W$ is a weak linkage, exactly one walk, say $W_1$, that has the vertex $v$ as an endpoint. Any other walk in $\cal W$ contains an even number of edges from $E_H(v)$, and since $\cal W$ is pushed onto $R$, these edges are all parallel copies of $e^\star$. Hence, the walks in $\cal W$ contain an odd number of parallel copies of $e^\star$ in total, i.e. $\ell(e^\star)$ is an odd number.
    Since $v$ is an endpoint of $W_1 \in {\cal W}$, there is exactly one edge in $e^\star_{z} \in E^\star(v)$ such that $\stitch_v(e^\star_z) = e^\star_z$. 
    
    We claim that $z = \frac{\ell(e^\star)+1}{2}$.
    Towards this, let us argue that for any edge $e^\star_i$, where $i < z$, if $\stitch_v(e^\star_i) = e^\star_j$ then $j > z$. Suppose not, and without loss of generality assume that $i < j < z$. Let us choose $i$ such that $|j-i|$ is minimized, and note that $j \neq i$. Consider the collection of edges $e^\star_p$ such that $i < p < j$. If this collection is empty, i.e. $j = i+1$, then observe that the edges $(e_i,e_j)$ form a U-turn, since $\stitch_v(e_i) = e_j$ only if they were consecutive edges of some walk in $\cal W$ and there is no edge in the strict interior of the cycle formed by the parallel edges $e_i$ and $e_j$. Otherwise this collection is non-empty, then observe that if $\stitch_v(e^\star_p) = e^\star_q$ then $i < q < j$. Indeed, if this were not the case then the pairs $e^\star_i,e^\star_j$ and $e^\star_p,e^\star_q$ are crossing at $v$, since they occur as $e^\star_i < e^\star_p < e^\star_j < e^\star_q$ in $\order_v$. This contradicts the weak linkage $\cal W$ is non-crossing. Otherwise, $i < q < j$ and hence $|q-p| < |j-i|$. But this contradicts the choice of $i$. Hence, for every $i < z$, $\stitch_v(e^\star_i) = e^\star_j$ where $j > z$. A symmetric argument holds for the other case, i.e. if $i > z$ then $\stitch_v(e^\star_i) = e^\star_j$ where $j < z$. Therefore, we can conclude that $z = \frac{\ell(e^\star)+1}{2}$,
    and hence $stitch_v(e^\star_z) = e^\star_{\ell(e^\star)+1 - z}$
    
    Let us now consider the other edges in $E^\star(v)$.
    Suppose that there exist integers $1 \leq i,p \leq \ell(e^\star)$ such that $\stitch_v(e^\star_i) = e^\star_j$, $\stitch_v(e^\star_p) = e^\star_q$ such that $i < p < z$ and $z < j < q$. Then it is clear that the pairs $e^\star_i,e^\star_j$ and $(e^\star_p,e^\star_q)$ are crossing at $v$, which is a contradiction.
    Therefore, if $i<p<z$ then $z < q < j$, and this holds for every choice of $i$ and $p$. A symmetric arguments holds in the other direction, i.e. if $i > p > z$ and $\stitch_v(e^\star_i) = e^\star_j$ and $\stitch_v(e^\star_p) = e^\star_q$, then $j < q < z$.
    Now we claim that for any $i \in \{1,2, \ldots, \ell{e^\star}\}$, if $\stitch_v(e^\star_i) = e^\star_j$ then $j = \ell(e^\star)+1 - i$.
    Suppose not, and consider the case $i < z$, and further let $j < \ell(e^\star)+1 - i$. Then observe that, for any edge $e^\star_p \in \{ e^\star_{i+1}, \ldots, e^\star_{z-1} \}$, $\stitch_v(e^\star_p) \in \{e^\star_{z+1}, \ldots, e^\star{j-1}\}$. But $\left| \{ e^\star_{i+1}, \ldots, e^\star_{z-1} \} \right|$ is strictly larger than $\left| \{e^\star_{z+1}, \ldots, e^\star{j-1}\} \right|$, which is a contradiction to the definition of $\stitch_v$.
    Hence, $j \geq \ell(e^\star)+1 - i$. A symmetric argument implies that $j \leq \ell(e^\star)+1 - i$. Therefore, for any $i < z$, $\stitch_v(e^\star_i) = e^\star_{\ell(e^\star)+1 -i}$. We can similarly argue that for $i > z$ $\stitch_v(e^\star_i) = e^\star_{\ell(e^\star)+1 -i}$. Since we have already shown that $\stitch_v(e^\star_z) = e^\star_z$, this concludes the proof of this lemma.
\end{proof}

\begin{lemma}\label{lem:locaStitchofTemplate2}
    Let $(G,S,T,g,k)$ be a good {\sf Yes}-instance of \pdp, and let $R$ be a backbone Steiner tree. Let $\cal W$ be a simplified weak linkage in $H$ and let $\ell$ be the multiplicity function of $\cal W$. 
    Consider the collection ${\cal A}= \{(\pairing_v,\template_v)\}|_{v \in V(R)}\in \wh{\ALL}$, such that it is the collection of pairings and templates of $\cal W$. 
    Let $\{f_v\}|_{v \in V(R)}$ be the stitching extracted from $\cal A$. Then, for every vertex $v\in V_{=2}(R) \cup V_{\geq 3}(R)$, $stitch_v(e) = f_v(e)$ for all edges $e \in E_H(v)$.
\end{lemma}
\begin{proof}
    Let $\ell$ be the multiplicity function of the simplified weak linkage $\cal W$. 
    Then, as $\cal W$ is canonical, for each edge $e \in E_R(v)$ with a parallel copy $e_i$, $\stitch_v(e_i) \neq \bot$ if and only if $i \in \{1,2, \ldots, \ell(e^x)\}$.
    Since $\cal W$ is a sensible and $v \not\in V_1(R)$, it cannot be the endpoint of any walk in $\cal W$. Hence, any walk contains an even number of edges from $E_H(v)$, and further any such edge is a parallel copy of an edge in $E_R(v) = \{e^1,e^2, \ldots, e^r\}$, where these edges are enumerated according to $\order_v$.
    Note that, the collections of parallel copies of theses edges also occur in the same manner in $\order_v$.
    We present our arguments in three steps.
    
    \begin{claim}
        Consider a pair of edges $(e,e') \in \pairing_v$, such that $\template_v(e,e') > 0$. 
        Then $\stitch_v$ maps each edge in $\{e_{(x_{e,e'}+1)}, \ldots, e_{(x_{e,e'}+\template_v(e,e'))}\}$ to some edge in $\{e'_1, e'_2, \ldots, e'_{\ell(e')} \}$, and vice versa. 
    \end{claim}
    \begin{proof}
    Suppose not, and consider the case where $e$ occur before $e'$ in $\order_v$.
    Consider a parallel copy of $e$, say $e_i \in \{e_{(x_{e,e'}+1)}, \ldots, e_{(x_{e,e'}+\template_v(e,e'))}\}$ such that $\stitch_v(e_i) = \wh{e}_j$, where $\wh{e} \in E_R(v)$ and $\wh{e}_j$ is the $j$-th parallel copy of $\wh{e}$.
    Let us choose $e'$ (with respect to $e$) so that the $x_{e,e'}$ is minimized, and then choose $e_i$ such that $i$ is minimized. Here, note that $i > x_{e,e'}$. 
    Let us argue that $\wh{e} = e'$. Suppose not, and note that $(e,\wh{e}) \in \pairing_v$ and $\template_v(e,\wh{e}) > 0$. Then we have three cases depending on the position of these edges in $\order_v$, either $e < \wh{e} < e'$, or $\wh{e} < e < e'$, or $e < e' < \wh{e}$.
    Consider the first case, and note that every parallel copy of $\wh{e}$ occurs before all parallel copies of $e'$ and after all parallel copies of $e$ in $\order_v$.
    We claim that for any $e_p \in \{ e_{i+1}, \ldots, e_{\ell(e)}\}$, $\stitch_v(e_p) \notin \{e'_{1}, \ldots, e'_{\ell(e')}\}$. If this claim were false, then observe that, as $e < \wh{e} < e'$, we have $ e_i < e_p < \wh{e}_j < \stitch_v(e_p)$ in $\order_v$. Hence $e_i,\wh{e}_j$ and $e_p, \stitch_v(e_p)$ are crossing pairs at $v$, in the weak linkage $\cal W$, which is a contradiction.
    On the other hand, if $\stitch_v(e_p) \notin \{e'_{1}, \ldots, e'_{\ell(e')}\}$ for any $e_p \in \{ e_{i+1}, \ldots, e_{\ell(e)}\}$, then we claim that $\stitch_v$ maps strictly fewer than $\template_v(e,e')$ edges from $\{e_{1}, \ldots, e_{\ell(e)} \}$ to $\{e'_1, \ldots, e'_{\ell(e')}\}$. Indeed, we choose $e'$ such that $x_{e,e'}$ is minimized, and hence the edges in $\{e_1, \ldots, e_{x_{e,e'}}\}$ are not mapped to any edge in $\{e'_1, \ldots, e'_{\ell(e')}\}$.
    And since, no edge in $\{ e_i, e_{i+1}, \ldots, e_{\ell(e)}\}$ maps to $\{e'_{1}, \ldots, e'_{\ell(e')}\}$, only the edges in $\{ e_{(x_{e,e'}+1)}, \ldots, e_{i-1}\}$ remain, which is strictly fewer than $\template_v(e,e')$.
    But this contradicts the definition of $\template_v(e,e')$. Hence, it cannot be the case that $e < \wh{e} < e'$ in $\order_v$. 
    Next, consider the case when $e < e' < \wh{e}$. 
    Note that $\wh{e} \in {\sf outer}(e,e')$, and by definition $x_{e,\wh{e}} < x_{e,e'}$.
    Since we choose $e'$ to minimize $x_{e,e'}$, and we didn't choose $e' = \wh{e}$,
    $\stitch_v$ maps the edges in $\{e_{(x_{e,\wh{e} + 1})}, \ldots, e_{(x_{e,\wh{e}+\template_v(e,\wh{e})})} \}$ to $\template_v(e,\wh{e})$ edges in $\{\wh{e}_1, \ldots, \wh{e}_{\ell(\wh{e})}\}$. Therefore, if $\stitch_v(e_i) = \wh{e}_j$, then there are $\template_v(e,\wh{e})+1$ parallel copies of $e$ that are mapped to parallel copies $\wh{e}$, which is a contradiction.
    Hence it is not possible that $e < e' < \wh{e}$. The last case, $\wh{e} < e < e'$ is similar to the previous case, since $\wh{e} \in {\sf outer}(e,e')$ in this case as well. Hence, we conclude that if $\stitch_v(e_i) = \wh{e}_j$ then $\wh{e} = e'$. Therefore, when $e$ occurs before $e'$ in $\order_v$, $\stitch_v$ maps each edge in $\{e_{(x_{e,e'}+1)}, \ldots, e_{(x_{e,e'}+\template_v(e,e'))}\}$ to some edge in $\{e'_1, e'_2, \ldots, e'_{\ell(e')} \}$.
    %
    
    By a symmetric argument, we obtain that for any edge in $\{e'_{(y_{e,e'} - \template_v(e,e'))}, \ldots, e'_{(y_{e,e'} - 1)}\}$ maps to an edge in $\{e_1, e_2, \ldots, e_{\ell(e)} \}$.\footnote{Note that, this is equivalent to the case when $e'$ occurs  before $e$ in $\order_v$. Here, we obtain a contradiction by choosing $e$ (with respect to $e'$) that maximizes $y_{e,e'}$, and then choosing the maximum $i$ such that $y_{e,e'} - \template_v(e,e') \leq i \leq y_{e,e'}$ and $\stitch_v(e'_i) \notin \{ e_1, \ldots, e_{\ell(e)}\}$.}
    \cqed\end{proof}

    We now proceed to further restrain the mapping of edges to the ranges determined by $x_{e,e'}, y_{e,e'}$ and $\template_v(e,e')$.
        
    \begin{claim} 
        Consider a pair $(e,e') \in \pairing_v$ such that $\template_v(e,e') > 0$
        Then, $\stitch_v$ maps $\{e_{(x_{e,e'}+1)}, \allowbreak \ldots, e_{(x_{e,e'}+\template_v(e,e'))}\}$ to $\{e'_{(y_{e,e'} - \template_v(e,e'))}, \allowbreak \ldots, e'_{(y_{e,e'}-1)}\}$ and vice-versa\footnote{Note that, by definition of $\stitch_v$, this immediately implies the other direction.}.
    \end{claim}
    \begin{proof}
    Suppose not, and without loss of generality assume that $e$ occurs before $e'$ in $\order_v$. Then consider the case when there is an edge $e_i \in \{e_{(x_{e,e'}+1)}, \ldots, e_{(x_{e,e'}+\template_v(e,e'))}\}$ such that $\stitch_v(e_i) = e'_j$, where either $j > y_{e,e'}-1 $ or $j < y_{e,e'} - \template_v(e,e')$. Then consider the collection $\{e'_j \} \cup \{e'_{(y_{e,e'} - \template_v(e,e'))}, \ldots, e'_{(y_{e,e'}-1)}\}$, and observe that each edge in this collection is mapped to a distinct edge in $\{e_1, e_2, \ldots, e_{\ell(e)}\}$.
    But then there are $\template_v(e,e') + 1$ edges in $\{e'_1, \ldots e'_{\ell(e')}\}$ that map to an edge in $\{e_1, \ldots, e_{\ell(e)}\}$ under $\stitch_v$. This is a contradiction to the definition of $\template_v(e,e')$. 
    \cqed\end{proof}
    
    Finally, we show that $\stitch_v$ is equal to $f_v$.
    \begin{claim}
        Consider a pair $(e,e') \in \pairing_v$ such that $\template_v(e,e') > 0$.
        Then for each $i \in \{1,2,\ldots, \template_v(e,e')\}$,
        $\stitch_v(e_{x_{e,e'}+i}) = e'_{y_{e,e'}-i}$ and $\stitch_v(e'_{y_{e,e'}-i}) = e_{x_{e,e'}+i}$. 
    \end{claim}
    \begin{proof}
    Suppose not and consider the case when $\stitch_v(e_{x_{e,e'}+i}) = e'_j$ where $ j \neq y_{e,e'}-i$. Note that $e'_j \in \{e'_{(y_{e,e'} - \template_v(e,e'))}, \ldots, \allowbreak e'_{(y_{e,e'} - 1)}\}$ by previous arguments. 
    Consider the case when $j > y_{e,e'} - i$. We claim that, the edges in $\{e'_{(y_{e,e'} - \template_v(e,e'))}, \ldots, e'_{(j-1)} \}$ must map to the edges in $\{e_{(x_{e,e'} + i+1)}, \ldots, e_{(x_{e,e'}+ \template_v(e,e'))}\}$. If not, then consider an edge $e'_p \in \{e'_{(y_{e,e'} - \template_v(e,e'))}, \ldots, e'_{(j-1)} \}$ such that $\stitch_v(e'_p) = e_q$ where $q < x_{e,e'}+i$. 
    Then consider the pairs $e_{x_{e,e'}+i}, e_j$ and $e_p, e_q$ in $\order_v$, and observe that $e_q < e_{x_{e,e'}+i} < e'_p < e'_j$ in $\order_v$.
    Then these pairs of edges are crossing at $v$, which is a contradiction to the fact that $\cal W$ is weak linkage. 
    On the other hand, $\left| \{e'_{(y_{e,e'} - \template_v(e,e'))}, \ldots, e'_{(j-1)} \} \right|$ is strictly larger than $\left|\{e_{(x_{e,e'} + i+1)}, \ldots, e_{(x_{e,e'}+ \template_v(e,e'))}\} \right|$,
    which is again a contradiction, since all edges in $\{e'_1, \ldots, e'_{\ell(e')}\}$ are mapped to distinct edges by $\stitch_v$, and they are not mapped to $\bot$. By symmetric arguments, the case when $j < y_{e,e'}-i$ also leads to a contradiction. 
    \cqed\end{proof}
    
    Now, by considering all pairs in $\pairing_v$ and applying the above claims, we obtain that  $\stitch_v = f_v$ for all $v \in V_{=2}(R) \cup V_{\geq e}(R)$. This concludes the proof of this lemma.
\end{proof}

The following lemma is a corollary of Lemma~\ref{lem:locaStitchofTemplate1} and Lemma~\ref{lem:locaStitchofTemplate2}

\begin{lemma}\label{lem:locaStitchofTemplate}
    Let $(G,S,T,g,k)$ be a nice instance of \pdp. Let $R$ be a Steiner tree. 
    Let ${\cal W}$ be a simplified weak linkage, and let ${\cal A} =\{(\pairing_v,\template_v)\}|_{v \in V(R)} \in \wh{\ALL}$ be pairing and template of ${\cal W}$. Let $\{f_v\}|_{v \in V(R)}$ be the stitching extracted from ${\cal A}$. Then $f_v = \stitch_v$ for every vertex $v \in V(R)$. 
\end{lemma}

Now, we consider the computational aspect of the definitions considered so far in this section.

\begin{lemma}\label{lem:computeLocStitchTime}
Let $(D,S,T,g,k)$ be a nice instance of \dpdp. Let $R$ be a Steiner tree. Let ${\cal A}=\{(\pairing_v,\template_v)\}|_{v \in V(R)}\in \wh{\ALL}$. Then, the multiplicity function extracted from ${\cal A}$ can be computed in time $k^{\OO(1)} n$, and the stitching extracted from $\cal A$ can be computed in time $2^{\OO(k)} n$.
\end{lemma}

\begin{proof}
    First we consider the computation of the multiplicity function $\ell = \{\ell_v\}|_{v \in V(R)}$ extracted from ${\cal A}$ according to Definition \ref{def:locaWLofTemplate}.
    Note that, for every vertex $v\in V(R)$, we have that $|\pairing_v|=\OO(k)$ and the numbers assigned by $\template_v$ are bounded by $2^{\OO(k)}$. 
    Therefore, $\ell_v$ can be computed in $2^{\OO(k)}$ time for each $v \in V(R)$, taking a total of $2^{\OO(k)} n$ time. 
    %
    Now, note that for any vertex $v\in V(R)$, it holds that $|\ell_v(e)|=2^{\OO(k)}$ for any edge $e \in E_H(v)$ (because $\template_v$ is a $2^{\OO(k)}$-template). 
    
    Let $\{f_v\}|_{v\in V(R)}|$ be the stitching extracted from $\cal A$ by Definition \ref{def:locaStitchExtractTerminal}.
    Observe that when describing the stitching $f_v$ extracted at a vertex $v \in V(R)$, we only need to describe it for the parallel copies of edges in $E(R)$, and then only for the parallel copies $\{e_1, e_2, \ldots, e_{\ell(v)} \}$ of $e \in E(R)$, where $\ell$ is the multiplicity function extracted from $\cal A$.
    For all other edges and parallel copies, the stitching maps them to $\bot$.
    Since $\ell(e) \leq \alpha_{\mathrm{mul}}(k)$, and the tree $R$ has at most $2k$ leaves, the stitching at each vertex can be described by a collection of $\OO(k \cdot \alpha_{\mathrm{mul}}) = 2^{\OO(k)}$ pairs of edges in $E_H(v) \times E_H(v)$. Further, by the construction described in Definitions \ref{def:locaStitchExtractTerminal} and \ref{def:locaStitchExtractNonTerminal}, 
    the stitching $f_v$ at each vertex $v\in V(R)$ can be constructed in time $2^{\OO(k)}$ time. Therefore, the collection $\{f_v\}_{v \in V(R)}$ can be constructed in $2^{\OO(k)} n$ time. Finally, we need to test if this collection is a valid stitching, as described in Definition~\ref{def:locaStitchExtract}, which can be done by picking each edge $e \in E(R)$ and testing the parallel copies $\{e_1, e_2, \ldots, e_{\ell(e)}\}$ one by one, which again takes $2^{\OO(k)}n$ time. Hence the total time required to extract the stitching is $2^{\OO(k)}n$.
    %
\end{proof}

\subsection{Reconstruction of  Weak Linkages from Templates}

Now we describe the construction of a weak linkage from a valid stitching.

\begin{definition}[Weak Linkage of a stitching.]
\label{def:stitchweaklinkage}
    Let $(G,S,T,g,k)$ be a good instance of \pdp, and let $R$ be a backbone Steiner tree. Let $\{f_v\}|_{v \in V(R)}$ be a stitching and suppose that is is valid.
    Then the \emph{weak linkage $\cal W$ constructed from $f_v$} is obtained as follows.
    \begin{itemize}
        \item For each $v \in S \cup T$, Let $e_v \in E_H(v)$ be the unique edge such that $f_v(e_v) = e_v$.
         
        \item Then the walk $W_v$ is defined as the sequence of edges $e_{0}, e_{1}, e_{2}, \ldots, e_{p_v}$, where $e_{v_0} = e_v$,
        and for each $i \in \{0, \ldots, p_v-1\}$, the edge $e_{i} =\{v_i, v_{i+1}\}$ satisfies $(i)$~$f_{v_{i+1}}(e_i) = e_{i+1}$ where $v_0 = v$;\; and $(ii)$ $f_{v_{p_v}}(e_{p_v}) = p_v$. 
        
        \item We iteratively construct a sequence of walks $W_{v_1}, W_{v_2}, \ldots,$ where the walk $W_{v_i}$ starts from a vertex $v_i$ that is not the endpoint of any of the previous walks.
        Finally, we output $\cal W$ as the collection of these walks.
    \end{itemize}
\end{definition}

It is clear that running time for the construction of a weak linkage from a stitching $\{f_v\}_{v \in V(R)}$ is upperbounded by the number of pairs of edges in $E_H(V(R))$ such that they are images of each other in the stitching.
The following observation is follows directly from Definition~\ref{def:stitchweaklinkage} and Definition~\ref{def:linkageStitching}.
and the fact that 

\begin{observation}\label{obs:weaklinkagestitching}
    Let $(G,S,T,g,k)$ be a good instance of \pdp, and let $R$ be a backbone Steiner tree. Let $\cal W$ be a weak linkage that is pushed onto $R$, and let $\{\stitch_v\}|_{v \in V(R)}$ be the stitching of $\cal W$.
    Then the weak linkage constructed from this stitching is equal to $\cal W$.
\end{observation}

Let ${\cal W}_\ALL$ denote the collection of weak linkages extracted from $\wh{\ALL}$.
The following Lemma is the main result of this section. 

\begin{lemma}\label{lemma:enumerateSimplifiedWeakLinkages}
    Let $(G,S,T,g,k)$ be a good {\sf Yes}-instance of $\pdp$, and let $R$ be a backbone steiner tree. Then, there exists a collection of $2^{\OO(k^2)}$ simplified linkages ${\cal W}_\ALL$ such that there is a weak linkage ${\cal W} \in {\cal W}_\ALL$ that is discretely homotopic in $H$ to some solution of $(G,S,T,g,k)$.  Further, collection can be enumerated in $2^{\OO(k^2)} n$ time.
\end{lemma}
\begin{proof}
    By Lemma~\ref{lem:pushOutcome}, the given instance has a solution that is discretely homotopic to some simplified weak linkage ${\cal W}$, and by Corollary~\ref{cor:existsInAll2}, the pairing and template of $\cal W$, denoted by  ${\cal A}$ lies in the collection $\wh{\ALL}$. Then, by Lemma~\ref{lem:locaStitchofTemplate}, the stitching $\{f_v\}|_{v \in V(R)}$ is equal to $\{\stitch_v\}|_{v \in V(R)}$, the stitching of $\cal W$, and it can be computed in $2^{\OO(k)} n$ time by Lemma~\ref{lem:computeLocStitchTime}. Finally, we can construct a weak linkage ${\cal W}'$ from the stitching $\{f_v\}|_{v \in V(R)}$, and by Observation~\ref{obs:weaklinkagestitching}, ${\cal W'} = {\cal W}$. 
    Note that, as $\cal W$ is a simplified weak linkage, its multiplicity is upperbounded by $\alpha_{\mathrm{mul}}(k) = 2^{\OO(k)}$. Hence by Observation~\ref{obs:stitchingProp} and ~\ref{obs:weaklinkmult}, the number of pairs in $(e,e') \in E_H(v) \times E_H(v)$ such that $f_v(e) = e'$ is upperbounded by $k \cdot \alpha_{\mathrm{mul}}(k) = 2^{\OO(k)}$. Hence, it is clear that we reconstruct $\cal W$ from the stitching $\{f_v\}_{v \in V(R)}$ in time $2^{\OO(k)} n$ time.
    
    To enumerate the collection ${\cal W}_\ALL$, we iterate over $\wh{\ALL}$. For each ${\cal A} \in \wh{\ALL}$, we attempt to construct a stitching and if it returns an invalid stitching, we move on to next iteration. Otherwise, we construct a weak linkage from this stitching and output it. Observe that for each ${\cal A} \in \wh{\ALL}$ we can compute the corresponding weak linkage $\cal W$, if it exists, in time $2^{\OO(k)} n$ time. Since $|\wh{\ALL}| = 2^{\OO(k^2)}$, clearly $|{\cal W}_\ALL| = 2^{\OO(k^2)}$ and it can be enumerated in $2^{\OO(k^2)} n$ time, we can enumerate ${\cal W}_\ALL$ in $2^{\OO(k^2)} n$ time.
\end{proof}

\section{The Algorithm}\label{sec:algorithm}

Having set up all required definitions and notions, we are ready to describe our algorithm. Afterwards, we will analyze its running time and prove its correctness.

\subsection{Execution of the Algorithm}
We refer to this algorithm as \pdpAlg. It takes as input an instance $(\widetilde{G},\widetilde{S},\widetilde{T},\widetilde{g},k)$ 
of \pdp, and its output is the decision whether this instance is a \yes-instance. The specification of the algorithm is as follows.
\vspace{0.5em}

\noindent{\bf Step I: Preprocessing.} First, \pdpAlg\ invokes Corollary \ref{cor:twReduction} to transform $(\widetilde{G},\widetilde{S},\widetilde{T},\widetilde{g},k)$ into an equivalent good instance $(G,S,T,g,k)$ of \pdp\ where $|V(G)|=\OO(|V(\widetilde{G})|)$. 
\vspace{0.5em}

\noindent{\bf Step II: Computing a Backbone Steiner Tree.} Second, \pdpAlg\ invokes Lemma \ref{lem:goodSteinerTreeComputeTime} with respect to $(G,S,T,g,k)$ to compute a backbone Steiner tree, denoted by $R$. Then, \pdpAlg\ computes the embedding of $H$ with respect to $R$ (Section~\ref{sec:enumParallel}).
\vspace{0.5em}

\noindent{\bf Step III: Looping on ${\cal W}_\ALL$.} Now, \pdpAlg\ invokes Lemma \ref{lemma:enumerateSimplifiedWeakLinkages} to enumerate ${\cal W}_\ALL$. 
For each weak linkage ${\cal W} \in {\cal W}_\ALL$, the algorithm applies the algorithm of Corollary~\ref{cor:discreteHomotopy} to $(G,S,T,g,k)$ and $\cal W$; 
if the algorithm finds a solution, then \pdpAlg\ determines that $(\widetilde{G},\widetilde{S},\widetilde{T},\widetilde{g},k)$ is {\sf Yes}-instance and terminates.
\vspace{0.5em}

\noindent{\bf Step IV: Reject.} \pdpAlg\ determines that $(\widetilde{G},\widetilde{S},\widetilde{T},\widetilde{g},k)$ is a \no-instance and terminates.

\subsection{Running Time and Correctness}
Let us first analyze the running time of \pdpAlg.

\begin{lemma}\label{lem:pdpAlgTime}
\pdpAlg\ runs in time $2^{\OO(k^2)}n^{\OO(1)}$.
\end{lemma}

\begin{proof}
By Corollary \ref{cor:twReduction}, the computation of $(G,S,T,g,k)$ in Step I is performed in time $2^{\OO(k)}n^2$. 
By Lemma \ref{lem:goodSteinerTreeComputeTime}, the computation of $R$ in Step II is performed in time $2^{\OO(k)}n^{3/2}\log^3 n$. Moreover, the computation of $R$ can clearly be done in time linear in $n$. 
Let $H = H_G$ be the radial completion of $G$ enriched with $4|V(G)|+1$ parallel copies of each edge, and note that $|V(H)| = \OO(|V(G)|)$. 
By Observation~\ref{obs:enumParallelTime}, we can compute the embedding of $H$ with respect to $R$ in time $\OO(n^2)$.
By Lemma \ref{lemma:enumerateSimplifiedWeakLinkages}, $|{\cal W}_\ALL|=2^{\OO(k^2)}$, and 
it can be enumerated in $2^{\OO(k^2)} n$ time. Finally, for each ${\cal W} \in {\cal W}_\ALL$, Corollary~\ref{cor:discreteHomotopy} takes $n^{\OO(1)}$ time to test if there is a solution that is discretely homotopic to $\cal W$. 
Thus, \pdpAlg\ runs in time $2^{\OO(k^2)}n^{\OO(1)}$.
\end{proof}

The reverse direction of the correctness of \pdpAlg\ is trivially true.

\begin{lemma}\label{lem:pdpAlgReverse}
Let $(\widetilde{G},\widetilde{S},\widetilde{T},\widetilde{g},k)$ be an instance of \pdp. If $(\widetilde{G},\widetilde{S},\widetilde{T},\widetilde{g},k)$ is accepted by \pdpAlg, then $(\widetilde{G},\widetilde{S},\widetilde{T},\widetilde{g},k)$ is a \yes-instance.
\end{lemma}
%

Now, we handle the forward direction of the correctness of \pdpAlg.

\begin{lemma}\label{lem:pdpAlgForward}
Let $(\widetilde{G},\widetilde{S},\widetilde{T},\widetilde{g},k)$ be a instance of \pdp. If $(\widetilde{G},\widetilde{S},\widetilde{T},\widetilde{g},k)$ is a \yes-instance, then $(\widetilde{G},\widetilde{S},\widetilde{T},\widetilde{g},k)$ is accepted by \pdpAlg.
\end{lemma}

\begin{proof}
Suppose that $(\widetilde{G},\widetilde{S},\widetilde{T},\widetilde{g},k)$ is a \yes-instance of \pdp. Then, by Corollary \ref{cor:twReduction}, $(G,S,T,g,k)$ is a \yes-instance of \pdp. Then by Lemma~\ref{lem:pushOutcome} and Lemma~\ref{lemma:enumerateSimplifiedWeakLinkages}, there is a collection ${\cal W}_\ALL$ of $2^{\OO(k^2)}$ simplified weak linkages containing at least one simplified weak linkage ${\cal W}^\star$, that is discretely homotopic in $H$ to some solution of $(G,S,T,g,k)$. Here $H$ is the radial completion of $G$ enriched with $4|V(G)|+1$ parallel copies of each edge. Then in Step~III, Corollary~\ref{cor:discreteHomotopy} ensures that we obtain a solution to $(G,S,T,g,k)$ in the iteration we consider ${\cal W}^\star$. Hence \pdpAlg\ accepts the instance $(\widetilde{G},\widetilde{S},\widetilde{T},\widetilde{g},k)$.
\end{proof}

Lastly, we remark that \pdpAlg\ can be easily modified not only to reject or accept the instance $(\widetilde{G},\widetilde{S},\widetilde{T},\widetilde{g},k)$, but to return a solution in case of acceptance, within time $2^{\OO(k^2)}n^{\OO(1)}$.

\bibliographystyle{siam}
\bibliography{references,dpsaket,introDaniel}

\begin{thebibliography}{10}

\bibitem{AdlerOpen13}
{\sc I.~Adler}, {\em 6th workshop on graph classes, optimization, and width
  parameters 2013, list of open problems}.
\newblock
  \url{http://www.cs.upc.edu/~sedthilk/grow/Open_Problems_GROW_2013.pdf}, 2013.

\bibitem{DBLP:journals/jct/AdlerKKLST17}
{\sc I.~Adler, S.~G. Kolliopoulos, P.~K. Krause, D.~Lokshtanov, S.~Saurabh, and
  D.~M. Thilikos}, {\em Irrelevant vertices for the planar disjoint paths
  problem}, J. Comb. Theory, Ser. {B}, 122 (2017), pp.~815--843.

\bibitem{DBLP:journals/corr/abs-1011-2136}
{\sc I.~Adler and P.~K. Krause}, {\em A lower bound for the tree-width of
  planar graphs with vital linkages}, CoRR, abs/1011.2136 (2010).

\bibitem{DBLP:journals/tcs/BasteS15}
{\sc J.~Baste and I.~Sau}, {\em The role of planarity in connectivity problems
  parameterized by treewidth}, Theor. Comput. Sci., 570 (2015), pp.~1--14.

\bibitem{DBLP:journals/jacm/BodlaenderFLPST16}
{\sc H.~L. Bodlaender, F.~V. Fomin, D.~Lokshtanov, E.~Penninkx, S.~Saurabh, and
  D.~M. Thilikos}, {\em {(Meta)} kernelization}, J. {ACM}, 63 (2016),
  pp.~44:1--44:69.

\bibitem{DBLP:journals/jacm/ChekuriC16}
{\sc C.~Chekuri and J.~Chuzhoy}, {\em Polynomial bounds for the grid-minor
  theorem}, J. {ACM}, 63 (2016), pp.~40:1--40:65.

\bibitem{ChuzhoyK15}
{\sc J.~Chuzhoy and D.~H.~K. Kim}, {\em On approximating node-disjoint paths in
  grids}, in Approximation, Randomization, and Combinatorial Optimization.
  Algorithms and Techniques, {APPROX/RANDOM} 2015, August 24-26, 2015,
  Princeton, NJ, {USA}, vol.~40 of LIPIcs, Schloss Dagstuhl - Leibniz-Zentrum
  fuer Informatik, 2015, pp.~187--211.

\bibitem{ChuzhoyKL16}
{\sc J.~Chuzhoy, D.~H.~K. Kim, and S.~Li}, {\em Improved approximation for
  node-disjoint paths in planar graphs}, in Proceedings of the 48th Annual
  {ACM} {SIGACT} Symposium on Theory of Computing, {STOC} 2016, Cambridge, MA,
  USA, June 18-21, 2016, {ACM}, 2016, pp.~556--569.

\bibitem{ChuzhoyKN17}
{\sc J.~Chuzhoy, D.~H.~K. Kim, and R.~Nimavat}, {\em New hardness results for
  routing on disjoint paths}, in Proceedings of the 49th Annual {ACM} {SIGACT}
  Symposium on Theory of Computing, {STOC} 2017, Montreal, QC, Canada, June
  19-23, 2017, {ACM}, 2017, pp.~86--99.

\bibitem{ChuzhoyKN18}
\leavevmode\vrule height 2pt depth -1.6pt width 23pt, {\em Almost polynomial
  hardness of node-disjoint paths in grids}, in Proceedings of the 50th Annual
  {ACM} {SIGACT} Symposium on Theory of Computing, {STOC} 2018, Los Angeles,
  CA, USA, June 25-29, 2018, {ACM}, 2018, pp.~1220--1233.

\bibitem{DBLP:conf/icalp/ChuzhoyKN18}
\leavevmode\vrule height 2pt depth -1.6pt width 23pt, {\em Improved
  approximation for node-disjoint paths in grids with sources on the boundary},
  in 45th International Colloquium on Automata, Languages, and Programming,
  {ICALP} 2018, July 9-13, 2018, Prague, Czech Republic, vol.~107 of LIPIcs,
  2018, pp.~38:1--38:14.

\bibitem{DBLP:conf/stoc/Cohen-AddadVKMM16}
{\sc V.~Cohen{-}Addad, {\'{E}}.~C. de~Verdi{\`{e}}re, P.~N. Klein, C.~Mathieu,
  and D.~Meierfrankenfeld}, {\em Approximating connectivity domination in
  weighted bounded-genus graphs}, in Proceedings of the 48th Annual {ACM}
  {SIGACT} Symposium on Theory of Computing, {STOC} 2016, Cambridge, MA, USA,
  June 18-21, 2016, 2016, pp.~584--597.

\bibitem{CK12}
{\sc S.~Cornelsen and A.~Karrenbauer}, {\em Accelerated bend minimization}, J.
  Graph Algorithms Appl., 16 (2012), pp.~635--650.

\bibitem{DBLP:conf/focs/CyganMPP13}
{\sc M.~Cygan, D.~Marx, M.~Pilipczuk, and M.~Pilipczuk}, {\em The planar
  directed k-vertex-disjoint paths problem is fixed-parameter tractable}, in
  54th Annual {IEEE} Symposium on Foundations of Computer Science, {FOCS} 2013,
  26-29 October, 2013, Berkeley, CA, {USA}, 2013, pp.~197--206.

\bibitem{DBLP:journals/jacm/DemaineFHT05}
{\sc E.~D. Demaine, F.~V. Fomin, M.~T. Hajiaghayi, and D.~M. Thilikos}, {\em
  Subexponential parameterized algorithms on bounded-genus graphs and
  \emph{H}-minor-free graphs}, J. {ACM}, 52 (2005), pp.~866--893.

\bibitem{DBLP:journals/combinatorica/DemaineH08}
{\sc E.~D. Demaine and M.~Hajiaghayi}, {\em Linearity of grid minors in
  treewidth with applications through bidimensionality}, Combinatorica, 28
  (2008), pp.~19--36.

\bibitem{ding1992disjoint}
{\sc G.~Ding, A.~Schrijver, and P.~D. Seymour}, {\em Disjoint paths in a planar
  graph—a general theorem}, SIAM Journal on Discrete Mathematics, 5 (1992),
  pp.~112--116.

\bibitem{DBLP:conf/focs/FominLMPPS16}
{\sc F.~V. Fomin, D.~Lokshtanov, D.~Marx, M.~Pilipczuk, M.~Pilipczuk, and
  S.~Saurabh}, {\em Subexponential parameterized algorithms for planar and
  apex-minor-free graphs via low treewidth pattern covering}, in {IEEE} 57th
  Annual Symposium on Foundations of Computer Science, {FOCS} 2016, New
  Brunswick, New Jersey, {USA}, {IEEE} Computer Society, 2016, pp.~515--524.

\bibitem{DBLP:journals/jacm/FominLS18}
{\sc F.~V. Fomin, D.~Lokshtanov, and S.~Saurabh}, {\em Excluded grid minors and
  efficient polynomial-time approximation schemes}, J. {ACM}, 65 (2018),
  pp.~10:1--10:44.

\bibitem{DBLP:journals/tcs/FortuneHW80}
{\sc S.~Fortune, J.~E. Hopcroft, and J.~Wyllie}, {\em The directed subgraph
  homeomorphism problem}, Theor. Comput. Sci., 10 (1980), pp.~111--121.

\bibitem{frank1990packing}
{\sc A.~Frank}, {\em Packing paths, cuts, and circuits-a survey}, Paths, Flows
  and VLSI-Layout,  (1990), pp.~49--100.

\bibitem{DBLP:conf/soda/GroheKR13}
{\sc M.~Grohe, K.~Kawarabayashi, and B.~A. Reed}, {\em A simple algorithm for
  the graph minor decomposition - logic meets structural graph theory}, in
  Proceedings of the Twenty-Fourth Annual {ACM-SIAM} Symposium on Discrete
  Algorithms, {SODA} 2013, New Orleans, Louisiana, USA, January 6-8, 2013,
  {SIAM}, 2013, pp.~414--431.

\bibitem{DBLP:conf/soda/JansenLS14}
{\sc B.~M.~P. Jansen, D.~Lokshtanov, and S.~Saurabh}, {\em A near-optimal
  planarization algorithm}, in Proceedings of the Twenty-Fifth Annual
  {ACM-SIAM} Symposium on Discrete Algorithms, {SODA} 2014, Portland, Oregon,
  USA, January 5-7, 2014, 2014, pp.~1802--1811.

\bibitem{DBLP:conf/coco/Karp72}
{\sc R.~M. Karp}, {\em Reducibility among combinatorial problems}, in
  Proceedings of a symposium on the Complexity of Computer Computations, held
  March 20-22, 1972, at the {IBM} Thomas J. Watson Research Center, Yorktown
  Heights, New York, {USA}, 1972, pp.~85--103.

\bibitem{karp1975computational}
{\sc R.~M. Karp}, {\em On the computational complexity of combinatorial
  problems}, Networks, 5 (1975), pp.~45--68.

\bibitem{DBLP:journals/jct/KawarabayashiKR12}
{\sc K.~Kawarabayashi, Y.~Kobayashi, and B.~A. Reed}, {\em The disjoint paths
  problem in quadratic time}, J. Comb. Theory, Ser. {B}, 102 (2012),
  pp.~424--435.

\bibitem{DBLP:conf/stoc/KawarabayashiW10}
{\sc K.~Kawarabayashi and P.~Wollan}, {\em A shorter proof of the graph minor
  algorithm: the unique linkage theorem}, in Proceedings of the 42nd {ACM}
  Symposium on Theory of Computing, {STOC} 2010, Cambridge, Massachusetts, USA,
  5-8 June 2010, {ACM}, 2010, pp.~687--694.

\bibitem{DBLP:conf/stoc/KawarabayashiW11}
\leavevmode\vrule height 2pt depth -1.6pt width 23pt, {\em A simpler algorithm
  and shorter proof for the graph minor decomposition}, in Proceedings of the
  43rd {ACM} Symposium on Theory of Computing, {STOC} 2011, San Jose, CA, USA,
  6-8 June 2011, {ACM}, 2011, pp.~451--458.

\bibitem{DBLP:conf/soda/KleinM14}
{\sc P.~N. Klein and D.~Marx}, {\em A subexponential parameterized algorithm
  for subset {TSP} on planar graphs}, in Proceedings of the Twenty-Fifth Annual
  {ACM-SIAM} Symposium on Discrete Algorithms, {SODA} 2014, Portland, Oregon,
  USA, January 5-7, 2014, {SIAM}, 2014, pp.~1812--1830.

\bibitem{kramer1984complexity}
{\sc M.~R. Kramer and J.~van Leeuwen}, {\em The complexity of wirerouting and
  finding minimum area layouts for arbitrary {VLSI} circuits}, Advances in
  computing research, 2 (1984), pp.~129--146.

\bibitem{Lipton2013}
{\sc R.~J. Lipton and K.~W. Regan}, {\em David Johnson: Galactic Algorithms},
  Springer Berlin Heidelberg, Berlin, Heidelberg, 2013, pp.~109--112.

\bibitem{DBLP:journals/siamcomp/LokshtanovMS18}
{\sc D.~Lokshtanov, D.~Marx, and S.~Saurabh}, {\em Slightly superexponential
  parameterized problems}, {SIAM} J. Comput., 47 (2018), pp.~675--702.

\bibitem{lynch1975equivalence}
{\sc J.~F. Lynch}, {\em The equivalence of theorem proving and the
  interconnection problem}, ACM SIGDA Newsletter, 5 (1975), pp.~31--36.

\bibitem{DBLP:books/daglib/0030489}
{\sc B.~Mohar and C.~Thomassen}, {\em Graphs on Surfaces}, Johns Hopkins series
  in the mathematical sciences, Johns Hopkins University Press, 2001.

\bibitem{DBLP:journals/jacm/PerlS78}
{\sc Y.~Perl and Y.~Shiloach}, {\em Finding two disjoint paths between two
  pairs of vertices in a graph}, J. {ACM}, 25 (1978), pp.~1--9.

\bibitem{DBLP:journals/talg/PilipczukPSL18}
{\sc M.~Pilipczuk, M.~Pilipczuk, P.~Sankowski, and E.~J. van Leeuwen}, {\em
  Network sparsification for steiner problems on planar and bounded-genus
  graphs}, {ACM} Trans. Algorithms, 14 (2018), pp.~53:1--53:73.

\bibitem{DBLP:journals/jct/RobertsonS95b}
{\sc N.~Robertson and P.~D. Seymour}, {\em Graph minors .xiii. the disjoint
  paths problem}, J. Comb. Theory, Ser. {B}, 63 (1995), pp.~65--110.

\bibitem{DBLP:journals/jct/RobertsonS03a}
\leavevmode\vrule height 2pt depth -1.6pt width 23pt, {\em Graph minors. {XVI.}
  excluding a non-planar graph}, J. Comb. Theory, Ser. {B}, 89 (2003),
  pp.~43--76.

\bibitem{DBLP:journals/jct/RobertsonS12}
\leavevmode\vrule height 2pt depth -1.6pt width 23pt, {\em Graph minors.
  {XXII.} irrelevant vertices in linkage problems}, J. Comb. Theory, Ser. {B},
  102 (2012), pp.~530--563.

\bibitem{DBLP:journals/jct/RobertsonST94}
{\sc N.~Robertson, P.~D. Seymour, and R.~Thomas}, {\em Quickly excluding a
  planar graph}, J. Comb. Theory, Ser. {B}, 62 (1994), pp.~323--348.

\bibitem{scheffler1994practical}
{\sc P.~Scheffler}, {\em A practical linear time algorithm for disjoint paths
  in graphs with bounded tree-width}, TU, Fachbereich 3, 1994.

\bibitem{DBLP:journals/siamcomp/Schrijver94}
{\sc A.~Schrijver}, {\em Finding k disjoint paths in a directed planar graph},
  {SIAM} J. Comput., 23 (1994), pp.~780--788.

\bibitem{schrijver2003combinatorial}
{\sc A.~Schrijver}, {\em Combinatorial optimization: polyhedra and efficiency},
  vol.~24, Springer Science \& Business Media, 2003.

\end{thebibliography}

\appendix

\section{Properties of Winding Number}\label{sec:app:wn}

In this section we sketch a proof of Proposition~\ref{prop:wn-prop} using homotopy. 
Towards this, we introduce some notation that are extensions of the terms introduced Section~\ref{sec:winding} in the continuous setting.
Recall that we have a plane graph $\ring(I_\tin, I_\tout)$ and we are interested in the winding number of paths in this graph, where $I_\tin$ and $I_\tout$ are two cycles such that $I_\tout$ is the outer-face, and there are no vertices or edges in the interior of $I_\tin$. Let us denote the (closed) curves defined by these two curves by $\rho_\tout$ and $\rho_\tin$, respectively. Then consider the collection of all the points in the plane that lie in the exterior of $\rho_\tin$ and interior of $\rho_\tout$. The closure of this set of points defines a surface called a \emph{ring}, which we denote by $\ring(\rho_\tin, \rho_\tout)$ by abusing notation.
Observe that the graph $\ring(I_\tin, I_\tout)$ is embedded in this ring, where the vertices of $I_\tin$ and $I_\tout$ lie on $\rho_\tin$ and $\rho_\tout$ respectively.
 
A curve $\alpha$ in the $\ring(\rho_\tin, \rho_\tout)$  \emph{traverses} it if it has one endpoint in $\rho_\tin$ and the other in $\rho_\tout$. We then orient this curve from its endpoint in $\rho_\tin$ to its endpoint $\rho_\tout$.
A curve $\beta$ \emph{visits} $\ring(\rho_\tin,\rho_\tout)$ if both its endpoints line on either $\rho_\tin$ or $\rho_\tout$. In this case we orient this curve as follows. We first fix an arbitrary ordering of all points in the curve $\rho_\tin$ and another one for all the points in the curve $\rho_\tout$. We then orient $\beta$ from the smaller endpoint to the greater one.
 
Consider two curves $\alpha,\alpha'$ that are either traversing or visiting 
$\ring(\rho_\tin,\rho_\tout)$ are {\em{homotopic}} if there exists a homotopy of the ring that fixes $\rho_\tin$ and $\rho_\tout$ and transforms $\alpha$ to $\alpha'$.
Note that homotopic curves have same endpoints. Two curves $\beta, \beta'$ are \emph{transversally intersecting} if $\beta \cap \beta'$ is a finite collection of points. Let us remark that the above orientation is preserved under homotopy, since any two homotopic curves have the same endpoints. 
Furthermore, when we speak of oriented curves in a ring, it is implicit that such curves are either visitors or traversing the ring.  Now we are ready to define the winding number of oriented curves in a ring.

\begin{definition}[Winding Number of Transversally Intersecting Curves]
    Two curves $\alpha,\beta$ 
    in $\ring(\rho_\tin,\rho_\tout)$ {\em{intersect transversally}} if $\alpha\cap \beta$ 
    is a finite set of points.
    For two curves $\alpha$ and $\beta$ 
    in $\ring(\rho_\tin,\rho_\tout)$ that intersect transversally, we define the {\em{winding number}} $\wnorig(\alpha,\beta)$ as the signed number of traversings of $\beta$ along $\alpha$.
    That is, 
    for every intersection point of $\alpha$ and $\beta$ we record $+1$ if $\beta$ crosses $\alpha$ from left to right, and $-1$ if it crosses from right to left
    (of course, with respect to the chosen direction of traversing $\beta$) and $0$ if it does not cross at that point. The winding number $\wnorig(\alpha,\beta)$ is the sum of the recorded numbers.
\end{definition}

It can be easily observed that if $\alpha$ and $\alpha'$ are homotopic curves traversing $\ring(\rho_\tin,\rho_\tout)$ and both intersect $\beta$ transversally, then $\wnorig(\alpha,\beta)=\wnorig(\alpha',\beta')$.
Here we rely on the fact that two homotopic curves have the same end-points.
Observe that that every intersection point of the curves $\alpha'$ and $\beta'$ is a traversing point, i.e. the point assigned either $+1$ or $-1$. 
Therefore, we can extend the notion of the winding number to pairs of curves not necessarily intersecting transversally as follows.
\begin{definition}[Winding Number]
    If $\alpha,\beta$ are two curves in $\ring(\rho_\tin,\rho_\tout)$, then we define $\wnorig(\alpha,\beta)$ to be the winding number $\wnorig(\alpha',\beta')$ for any $\alpha',\beta'$ such that
    $\alpha$ and $\alpha'$ are homotopic, $\beta$ and $\beta'$ are homotopic, and $\alpha',\beta'$ intersect transversally, and each common point of $\alpha'$ and $\beta'$ is a traversing point.
\end{definition}
Note that such $\alpha',\beta'$ always exist and the definition of $\wnorig(\alpha,\beta)$ does not depend on the particular curves.
Let us now proceed towards a proof of Proposition~\ref{prop:wn-prop}.

\begin{lemma}\label{lem:additivity-clean}
Suppose $\alpha,\beta,\gamma$ are curves traversing a ring $\ring(\rho_\tin,\rho_\tout)$ that pairwise intersect transversally.
Further, letting $a,b,c$ and $a',b',c'$ be the endpoints of $\alpha,\beta,\gamma$ on $\rho_\tin$ and $\rho_\tout$ respectively, suppose that $a,b,c$ are different and appear in the clockwise order on $\rho_\tin$, and that $a',b',c'$ are different and appear in the clockwise order on $\rho_\tout$. Then
$$\wnorig(\alpha,\beta)+\wnorig(\beta,\gamma)=\wnorig(\alpha,\gamma).$$
\end{lemma}
\begin{proof}
First, we argue that we may assume that $\wnorig(\alpha,\gamma)=0$. This can be done as follows. Let $k=\wnorig(\alpha,\gamma)$. Glue a ring $\ring(\rho_\tout,\rho_\tout')$ to $\ring(\rho_\tin,\rho_\tout)$ along $\rho_\tout$, for some non-self-traversing closed curve $\rho_\tout'$ that encloses $\rho_\tout$, thus obtaining ring $\ring(\rho_\tin,\rho_\tout')$. Pick $a'',b'',c''$ in the clockwise order on $\rho_\tout'$. Extend $\alpha$ to a curve $\alpha'$ traversing $\ring(\rho_\tin,\rho_\tout')$ using any curve within $\ring(\rho_\tout,\rho_\tout')$ connecting $a'$ with $a''$. Next, extend $\beta$ to $\beta'$ in the same way, but choose the extending segment so that it does not cross $\alpha'$. Finally, extend $\gamma$ to $\gamma'$ in the same way, but choose the extending segment so that it crosses $\alpha'$ (and thus also $\beta'$) exactly $-k$ times (where we count signed traversings). Thus we have $\wnorig(\alpha',\gamma')=0$ and if we replace $\alpha,\beta,\gamma$ with $\alpha',\beta',\gamma'$, both sides of the postulated equality are decremented by $k$. Hence, it suffices to prove this equality for $\alpha',\beta',\gamma'$, for which we known that 
$\wnorig(\alpha',\gamma')=0$.

Having assumed that $\wnorig(\alpha,\gamma)=0$, it remains to prove that $\wnorig(\alpha,\beta)+\wnorig(\beta,\gamma)=0$, or equivalently 
\begin{equation}\label{eq:after-parallel}
\wnorig(\alpha,\beta)+\wnorig(\gamma^{-1},\beta)=0.
\end{equation}
Since $\wnorig(\alpha,\gamma)=0$, we may further replace $\alpha$ and $\gamma$ with homotopic curves that do not cross at all.
Note that this does not change the winding numbers in the postulated equality.
Hence, from now on we assume that $\alpha$ and $\gamma$ are disjoint.

Let us connect $c$ with $a$ using an arbitrary curve $\epsilon$ through the interior of the disk enclosed by $\rho_\tin$, and let us connect $a'$ with $c'$ using an arbitrary curve $\epsilon'$ outside of the disk enclosed by $\rho_\tout$. Thus, the concatenation of $\alpha,\epsilon',\gamma^{-1},\epsilon$ is a closed curve in the plane without self-traversings; call it $\delta$. Then $\delta$ separates the plane into two regions $R_1,R_2$. Since $a,b,c$ appear in the same order on $\rho_\tin$ as $a',b',c'$ on $\rho_\tout$, it follows that $b$ and $b'$ are in the same region, say $R_1$.

Now consider travelling along $\beta$ from the endpoint $b$ to the endpoint $b'$. Every traversing of $\alpha$ or $\gamma$ along $\beta$ is actually a traversing of $\delta$ that contributes to the left hand side of~\eqref{eq:after-parallel} with $+1$ if on the traversing $\beta$ passes from $R_1$ to $R_2$, and with $-1$ if it passes from $R_2$ to $R_1$. Since $\beta$ starts and ends in $R_1$, the total sum of those contributions has to be equal to $0$, which proves~\eqref{eq:after-parallel}.
\end{proof}

\begin{lemma}\label{lem:wn-prop}
For any curves $\alpha,\beta,\gamma$ traversing a ring $\ring(\rho_\tin,\rho_\tout)$, it holds that
$$|(\wnorig(\alpha,\beta)+\wnorig(\beta,\gamma)) - \wnorig(\alpha,\gamma)|\leq 1.$$
\end{lemma}
\begin{proof}
By slightly perturbing the curves using homotopies, we may assume that they pairwise intersect transversally.
Further, we modify the curves in the close neighborhoods of $\rho_\tin$ and $\rho_\tout$ so that we may assume that the endpoints of $\alpha,\beta,\gamma$ on $\rho_\tin$ and $\rho_\tout$ are pairwise different and appear in the same clockwise order on both cycles; for the latter property, we may add one traversing between two of the curves, thus modifying one of the numbers $\wnorig(\alpha,\beta),\wnorig(\beta,\gamma),\wnorig(\alpha,\gamma)$ by one.
It now remains to use Lemma~\ref{lem:additivity-clean}.
\end{proof}

\paragraph{Proof of Proposition~\ref{prop:wn-prop}.}
Recall that the graph $\ring(I_\tin, I_\tout)$ is embedded in the ring.
The first property of  follows directly from the definition of winding numbers.
For the second property, we apply Lemma~\ref{lem:wn-prop} to the curves defined by the paths $\alpha, \beta$ and $\gamma$ in $\ring(\rho_\tin, \rho_\tout)$.
\qed

\end{document}